\documentclass[final,onefignum,onetabnum]{siamart190516}

\usepackage{lipsum}
\usepackage{amsfonts}
\usepackage{graphicx}
\usepackage{epstopdf}
\usepackage{algorithmic}
\ifpdf
  \DeclareGraphicsExtensions{.eps,.pdf,.png,.jpg}
\else
  \DeclareGraphicsExtensions{.eps}
\fi

\newcommand{\creflastconjunction}{, and~}

\newsiamremark{remark}{Remark}
\newsiamremark{hypothesis}{Hypothesis}
\crefname{hypothesis}{Hypothesis}{Hypotheses}
\newsiamthm{claim}{Claim}

\headers{Two variable logic with ultimately periodic counting
}{M. Benedikt, E. Kostylev, and T. Tan}

\title{Two variable logic with ultimately periodic counting \thanks{
\funding{The first author acknowledges generous financial support of EPSRC under grant no. EP/M005852/1.
The third author acknowledges generous financial support of Taiwan Ministry of Science and Technology under grant no. 107-2221-E-002-026-MY2  and NTU under grant no. 108L891908.}}}

\author{Michael Benedikt \thanks{University of Oxford (\email{michael.benedikt@cs.ox.ac.uk}) }
 \and Egor V. Kostylev \thanks{University of Oslo (\email{egor.kostylev@cs.ox.ac.uk})  }
\and Tony Tan \thanks{National Taiwan University (\email{tonytan@csie.ntu.edu.tw} )}
}

\usepackage{amsopn}

\ifpdf
\hypersetup{
  pdftitle={Two variable logic with ultimately periodic counting}
  pdfauthor={M. Benedikt, E. Kostylev, and T. Tan}
}
\fi

\usepackage{mysty}
\usepackage{paralist}
\usepackage{amssymb}
\usepackage{tikz}

\usetikzlibrary {arrows.meta}

\title{Two variable logic with ultimately periodic counting\thanks{This paper extends an abstract appearing in ICALP 2020 \cite{usicalp}.}} %TODO Please add

\author{Michael Benedikt\thanks{University of Oxford (\email{michael.benedikt@cs.ox.ac.uk})} 
\and Egor V. Kostylev\thanks{University of Oslo (\email{egork@ifi.uio.no})}
\and Tony Tan\thanks{National Taiwan University (\email{tonytan@csie.ntu.edu.tw})}
}

\begin{document}

\maketitle

\bibliographystyle{siamplain}% the mandatory bibstyle


% REQUIRED

% abstract
\begin{abstract}
We consider the extension of $\FOtwo$ with quantifiers 
that state that the number of elements where  a formula
holds should belong to a given ultimately periodic set.
We show that both satisfiability and finite satisfiability of
the logic are decidable. We also show that the spectrum of any sentence,
i.e., the set of the sizes of its finite models, is definable in Presburger
arithmetic. In the process we present several refinements to  the ``biregular graph method''.
In this method, decidability issues concerning two-variable logics are reduced
 to questions about Presburger definability  of
integer vectors associated with partitioned
graphs,  where nodes in a partition satisfy certain constraints on their in- and out-degrees.
\end{abstract}

% intro

\section{Introduction} 
\label{sec:intro}

In the search for expressive logics with decidable satisfiability problem, two-variable
logic, denoted here as $\fotwo$,  is one yardstick. This logic is expressive enough to subsume basic modal
logic and many  description logics, while satisfiability and finite satisfiability for this logic
 coincide,  and both are decidable \cite{scott, mortimer, gkv}.
However, $\fotwo$  lacks the ability to count. Two-variable logic
with counting, $\ctwo$, is a decidable extension of $\fotwo$ that adds \emph{counting
quantifiers}.
In $\ctwo$ one can write, for example, formulas $\exists^5 x \ P(x)$ and $\forall x \exists^{\geq 5} y\ E(x,y)$
which, respectively, express that there are exactly $5$ elements in unary relation $P$, 
and that every element in a graph has at least $5$ adjacent edges. 
Satisfiability and finite satisfiability
do not coincide for $\ctwo$, but both are decidable \cite{decidectwo,pachoskicounttwovar}.
In  \cite{pachoskicounttwovar}  the problems were
shown to be $\nexptime$-complete under a unary encoding of numbers, and this
was extended to binary encoding in \cite{ctwobinary}.
However, the numerical capabilities of $\ctwo$
are quite limited. For example, one can not express that the number of outgoing edges of each
element in the graph is even.

 A natural extension is to combine $\fotwo$
with \emph{Presburger arithmetic} where one is allowed to define collections 
of tuples of integers from addition and equality using Boolean operators  and
quantifiers. The collections of $k$-tuples that one can define in this way
are the \emph{semi-linear sets}, and the collections of integers (when $k=1$) definable
are the {\em ultimately periodic sets}.
It is natural to consider the addition of \emph{Presburger quantification} to fragments
of two-variable logic; this is in the spirit of works such as \cite{bartosztonyfsstcs,bartosztcs}.
For every definable set $\phi(x,y)$ and every ultimately periodic set $S$, one has a formula
$\exists^S y  ~ \phi(x,y)$ 
that holds at $x$ when the number of $y$ such that
$\phi(x,y)$ is in $S$.  We let $\fotwopres$ denote the logic that adds this construct
to $\fotwo$.

On the one hand,
the corresponding quantification over general $k$-tuples (allowing
semi-linear rather than only ultimately periodic sets) easily leads to undecidability \cite{undecidabletwovar, BaaderBR20}.
On the other hand, adding this quantification to modal logic has been shown
to preserve decidability \cite{baader,demrilugiez}.  
Related \emph{one-variable fragments} in which we have only a 
unary relational vocabulary and the main quantification is $\exists^S x  ~ \phi(x)$ are known to be decidable (see, e.g. \cite{bartosztcs}), and their decidability is the basis
for a number of software tools focusing on integration of relational languages
with Presburger arithmetic \cite{bapa}. The decidability of full $\fotwopres$
is, to the best of our knowledge, open.
There are a number of other extensions of $\ctwo$ that have been shown decidable; for example
it has been shown that one can allow a distinguished equivalence relation \cite{ctwocountequiv} or
a forest-structured relation \cite{ctwoandforeststocl, ctwoandforests}.
$\fotwopres$ is easily seen to be orthogonal to these other extensions.
For example, equivalence relations and forest-structure are not expressible in $\fotwopres$,
whereas modulo counting is not expressible in the logics of \cite{ctwocountequiv,ctwoandforeststocl,ctwoandforests}.

In this paper we show that both satisfiability and finite satisfiability of $\fotwopres$
are decidable. Our result uses a method based on analyzing \emph{biregular graph  constraints}, introduced for analyzing $\ctwo$ in \cite{KT15}.
In this analysis we search for the existence of graphs equipped with a partition of vertices  based on constraints
on the out- and in-degree. Such a partitioned graph can be characterized  by  the cardinalities
of each partition component, and the key step in showing these decidability results is to prove
that the set of tuples of integers representing valid sizes of partition components
is definable by a formula in Presburger arithmetic. From this ``biregular graph constraint Presburger definability''
result, one can reduce satisfiability in the logic to satisfiability of a Presburger  formula, and from there
infer decidability using known results on Presburger arithmetic.
We will also use the method to get information on the \emph{spectrum} of a $\fotwopres$ sentence: 
the set of sizes of models of the sentence.
We use the method to conclude that this set is definable in Presburger arithmetic, a result that had been demonstrated for $\ctwo$ in
\cite{KT15}.

\myparagraph{Organization}
Section \ref{sec:prelims} provides background on two variable logic and Presburger
arithmetic. Section \ref{sec:main} introduces our major results on the logic, and gives a reduction
of these logic-based problems to  results concerning the analysis of constrained biregular graphs. 
 Section \ref{sec:1type} gives some of the details behind
the core lemmas concerning Presburger definability of solutions to biregular graph
problems that underlie the proof, and provides a full proof in the  case where there is only a single kind of 
edge in the graph. We refer to this as  the ``$1$-color case''.
Section \ref{sec:proofsimple} generalizes to give a   proof in the case of an unbounded number of edge colors, but with an extra restriction
on the matrices that specify the graph constraints.  The restriction is that they are ``simple matrices''.
Section \ref{sec:simple-complete} extends the analysis in Section \ref{sec:proofsimple} to the complete graph cases --
but still with the restriction on simple matrices.
Section \ref{sec:proofnonsimple} shows how to reduce the general case to the simple case.
Section \ref{sec:complexity} provides complexity upper bounds for all problems considered in this paper.
Section \ref{sec:spectrum} gives 
an application of the graph analysis result to the spectrum problem. After a discussion of related work in Section \ref{sec:related}, the paper closes with conclusions and future
directions in Section \ref{sec:conc}.
Some  proofs  that are not required in order to follow the main line of argument in the paper are deferred to the appendix. 
In addition, to make the main line of argument clearer, we consider only the finite graph case in the body of the paper,
which already implies decidability of the finite satisfiability of $\fotwopres$.
The general case is deferred to the appendix.

\myparagraph{Acknowledgements}
We thank Bartosz Bednarczyk for many useful comments and suggestions on preliminary
drafts of this work. 
This paper extends an abstract appearing in ICALP 2020 \cite{usicalp}, and we thank
 the conference reviewers for their feedback. Most of all we are deeply grateful to the reviewers of SICOMP for their detailed
comments on the submission.

This research was funded in whole or in part by EPSRC grant EP/T022124/1.
For the purpose of Open Access, the author has applied a CC BY 
public copyright license to any Author Accepted Manuscript (AAM)
version arising from this submission.

The third author acknowledges generous financial support of Taiwan National Science and Technology Council under grant no.~109-2221-E-002-143-MY3
and NTU under grant no.~109L891808.

% notations
%!TEX root = fo2siam.tex

\section{Preliminaries}
\label{sec:prelims}

Let $\bbN=\{0,1,2,\ldots\}$
and let $\bbN_{\infty}=\bbN \cup \{\infty\}$.

\myparagraph{Linear  and ultimately periodic sets}
A set of the form $\{a+ip \mid i \in\bbN\}$, for some $a,p\in\bbN$, is a \emph{linear set}.
We will denote such a set by $\prdp{a}$,
where $a$ and $p$ are called the {\em offset} and {\em period} of the set, respectively.
Note that, by definition, $\prd{a}{0}=\{a\}$,
which is a linear set.
For convenience, we define $\emptyset$ and $\{\infty\}$ (which may be written as $\prdp{\infty}$) to also be linear sets. 
An {\em ultimately periodic set} ({\em u.p.s.}) $S$ is a finite union of linear sets.

In this paper we represent a u.p.s.~$S=\{c_1\} \cup \cdots \cup \{c_m\}\cup \prd{a_1}{p_1}\cup\cdots\cup\prd{a_n}{p_n}$,
where $p_1,\ldots,p_n\neq 0$, as a ``finite set'' $\{c_1,\ldots,c_m,\prd{a_1}{p_1},\ldots,\prd{a_n}{p_n}\}$. %, with each $\prd{a_i}{p_i}$ being represented as a pair $(a_i,p_i)$.
In such a representation the offsets in $S$ are $c_1,\ldots,c_m,a_1,\ldots,a_n$
and the (non-zero) periods are $p_1,\ldots,p_n$.
For an integer $a$, we write $a\in S$, if $a$ is in $S$ in the standard sense.
Abusing notation, we write $\prdp{a} \in S$, if $a+ip \in S$ for every $i\in \bbN$.
We also note that the set of u.p.s.'s is closed under complement, union and intersection~\cite{ginsburg-spanier}.

\myparagraph{Two-variable logic with ultimately periodic counting quantifiers}
An \emph{atomic formula} is one of the following:
\begin{compactitem}
\item  an atom $R(\vec u)$, where $R$ is a predicate, and
 $\vec u$ is a tuple of variables/constants of appropriate size, 
\item  an equality $u=u'$, with $u$ and $u'$
 variables/constants, 
\item one of the formulas $\myT$ and $\myF$ denoting the True and False values.
\end{compactitem}
The logic $\fotwopres$ is the class of formulas using only variables $x$ and $y$,
built up from atomic formulas and equalities using
the usual Boolean connectives and also
\emph{ultimately periodic counting quantification}, which is of the form
$\exists^S x \ \phi$  where $S$ is a u.p.s.~and $\phi$ is an $\fotwopres$ formula.
One special case is  where $S$ is a singleton $\{a\}$ with $a \in \bbN_{\infty}$,
which we write $\exists^a x  \ \phi$;  in case of $a \in \bbN$, these are \emph{counting quantifiers}.
The semantics of  $\fotwopres$ is defined as usual except that, 
for every $a\in \bbN$, $\exists^a x \ \phi$ holds when there are {\em exactly} $a$ number of $x$'s
such that $\phi$ holds, $\exists^\infty x \ \phi$ holds when
there are infinitely many $x$'s such that $\phi$ holds, and $\exists^S x \ \phi$  holds when there is some $a \in S$ such that
$\exists^a x \ \phi$ holds.

Formulas in $\fotwopres$ still use only two variables.
So just as in $\fotwo$ they can be normalized. If they use atomic predicates with arity $3$ or above,
they can be rewritten into an equisatisfiable formula that uses only unary and binary predicates.
See~\cite[Sect.~3]{gkv} for the details of such a rewriting.
In addition,  each constant  $c$ can be represented with a fresh unary predicate $U_c$
that contains exactly one element.
For constants $c_1,c_2$, an atomic predicate $x=c_1$ can then be rewritten as $U_{c_1}(x)$,
and predicate $c_1=c_2$ can be rewritten as $\forall x \ U_{c_1}(x)\leftrightarrow U_{c_2}(x)$.
Thus, in this paper we may assume that $\fotwopres$ formulas
use only unary and binary predicates, and do not use constants.

Note that when $S$ is $\prd{0}{1} \cup \{\infty\} = \bbN_{\infty}$, $\exists^S x \ \phi$ is equivalent
to $\myT$.
When $S$ is $\prd{0}{1} = \bbN$, $\exists^S x \ \phi$ semantically means that there are finitely many $x$ such that $\phi$ holds.
We also observe that, for every formula $\phi$, $\exists^{\emptyset} x \ \phi$ is equivalent to $\myF$, $\exists^{0}x\ \phi$ is equivalent to $\forall x\ \neg \phi$, and
$\neg \exists^S x \ \phi$ is equivalent to $\exists^{\bbN_{\infty}-S} x \ \phi$. We remark that $\bbN_{\infty}-S$ is a u.p.s., since the set of u.p.s.'s is closed under complement.

For example, we can state in $\fotwopres$ that a graph is undirected and
every node in a graph has even degree
(i.e., the graph is Eulerian in the sense that every connected component has Eulerian cycle):
$$
\forall x \forall y\ E(x,y) \leftrightarrow E(y,x) 
\quad \wedge \quad
\forall x \exists ^{S}y\ E(x,y)\wedge x\neq y,\qquad\text{where}\ S= 2^{+2}.
$$
Clearly $\fotwopres$ extends
$\ctwo$, the fragment of the logic where only  counting
quantifiers are used, and $\fotwo$,  the fragment where only the classical quantifier
$\exists  x$  is allowed (which is equivalent to $\exists^S x$ for $S = \{1^{+1},\infty\}$).

\myparagraph{Presburger arithmetic}
An \emph{existential Presburger formula}
is a first-order logic formula of the form 
$\exists x_1 \cdots \exists x_k ~ \phi$,
where $\phi$ is a quantifier-free formula
over the signature including constants $0,1$, a binary function symbol $+$, and a binary
relation $\leq$.
Such a formula is a \emph{sentence} if it has no free variables.
The notion of a  sentence holding
 in a structure interpreting the function, relation, and constants is defined
in the usual way. 
The structure ${\cal N} = (\bbN, +, \leq, 0, 1)$, 
is defined by interpreting $+, \leq, 0, 1$ in
the standard way. We will focus not on this structure, but on  ${\cal N}_\infty = (\bbN_{\infty}, +, \leq, 0, 1)$
which is the same as $\cN$, except that 
there is  an element $\infty$, with
$a + \infty=\infty$  and $a \leq \infty$ for each  
$a \in \bbN_\infty$.  Note that in ${\cal N}_\infty$ there is a unique element $n$ such that $n+1=n$, namely
$\infty$. 
 We will thus abuse notation in the sequel by writing $t=\infty$, where $t$ is a term,  
as syntactic sugar for $t=t+1$.
Since ${\cal N}$ is quantifier-free definable in ${\cal N}_\infty$, satisfaction of a formula in finite integers can still be expressed
when working over ${\cal N}_\infty$.

It is known that the problem 
of checking whether an existential
Presburger sentence holds in $\cal N$ is decidable and is $\np$-complete \cite{christos}.
Further, the analogous problem for $\cal N_\infty$
can easily be reduced to that for $\cal N$. Indeed, we can first guess which variables are mapped to $\infty$
and then which atoms should be true.  Then we can check whether each guessed atomic truth value is consistent
with other guesses,
in the sense that no two contradicting atoms are guessed to be both true or false at the same time.
We can
determine additional variables which must be infinite based on this choice. Finally we can
restrict to atoms that do not involve variables guessed to be infinite, and check that the conjunction is
satisfiable for $\cal N$. This gives us the following theorem:

\begin{theorem} \label{thm:pres} The problems of checking whether an existential
Presburger sentence holds in $\cal N_\infty$ in $\np$.
\end{theorem}

%!TEX root = fo2siam.tex

\section{From Analysis of Constrained Regular Graph Problems to Decidability of $\fotwopres$}
\label{sec:main}

In this section we prove decidability of $\fotwopres$ satisfiability via reduction to some results on analyzing
certain graph problems. These graph analysis results will  then be proven later in the paper by means of connections to existential Presburger sentences.
Our decision procedure is based on the key notion of biregular graphs.
Note that whenever we talk about graphs or digraphs (i.e., directed graphs), by default we allow both finite 
and infinite sets of vertices and edges.

\subsection{Biregular graphs and constrained biregular graph problems} \label{subsec:graphs}

We fix an integer $p\geq 0$.
Let $\bbNop$ denote the set $\bbNo \cup \{\prdp{a} \mid a\in \bbNo\}$.
For integers $t,m\geq 1$, let $\bbNop^{t\times m}$ denote the set of matrices
with $t$ rows and $m$ columns where each entry is an element of $\bbNop$.
For an integer $k\geq 1$, let $[k]$ denote the set $\{1,2,\ldots,k\}$.

A {\em $t$-color bipartite (undirected) graph} is $G=(U,V,E_1,\ldots,E_t)$,
where $U$ and $V$ are sets of vertices, and 
$E_1,\ldots,E_t$ are pairwise disjoint sets of edges between $U$ and $V$ -- that is, pairs $(u,v) \in U \times V$. 
Edges in $E_i$ are called $E_i$-{\em edges}, and we often refer to an index from $1$ to $t$ -- the type of an edge --
as a \emph{color}.
For a vertex $u \in U\cup V$, the $E_i$-\emph{degree} of $u$ is the number of $E_i$-edges adjacent to $u$.
The degree of $u$ is the sum of the $E_i$-degrees for $i=1, \ldots, t$: we use this primarily for brevity when the
there is only a single edge relation. In the context of multiple relations,  we sometimes refer to this as the \emph{total degree} to emphasize that
all relations are considered.
We say that $G$ is {\em complete} if $U\times V = \bigcup_{i=1}^{t} E_i$.

For two matrices $A\in \bbN_{\infty,+p}^{t\times m}$ and $B\in \bbN_{\infty,+p}^{t\times n}$,
graph $G$ is an {\em $A|B$-biregular} graph,
if there exist a partition\footnote{As usual, we write $U=U_{1}\uplus \cdots\uplus U_{m}$
to denote the partition of $U$ into the sets $U_1,\ldots,U_m$,
i.e., when $U=U_1\cup\cdots\cup U_m$ for pairwise disjoint sets $U_1,\ldots,U_m$.}
$U=U_{1}\uplus \cdots\uplus U_{m}$
and a partition $V=V_{1}\uplus \cdots\uplus V_n$ such that for every $i\in [t]$,
for every $k\in [m]$ and for every $\ell\in [n]$,
the $E_i$-degree of every vertex in $U_k$ is $A_{i,k}$ (i.e., the element of $A$ in the $i$-th row and $k$-th column) and the $E_i$-degree
of every vertex in $V_\ell$ is $B_{i,\ell}$; note here that, by abuse of notation,
when we say that a non-negative integer $z$ is a linear set $\prdp{a}$,
we mean that $z\in \prdp{a}$. 
For each such partition, we say that \emph{$G$ has size $\vM|\vN$}, where
$\vM=(|U_{1}|,\ldots,|U_{m}|)$ and $\vN=(|V_{1}|,\ldots,|V_{n}|)$.
The partitions $U=U_1\uplus \cdots\uplus U_m$ and $V=V_1\uplus \cdots \uplus V_n$
are called a \emph{witness partition} for $A|B$-biregularity.
We should remark that some $U_i$ and $V_i$ are allowed to be empty.
The matrices $A$ and $B$ are called {\em ($t$-color) degree matrices}
and the vectors $\vM$ and $\vN$ are called {\em size vectors}.
For convenience, we treat the empty graph (i.e., the graph with no vertex)
as a complete $A|B$-biregular graph for any degree matrices $A$ and $B$.

The above definitions can be easily adapted for the case of directed  graphs that are not necessarily bipartite.
A $t$-color {\em directed graph} (or \emph{digraph}) is a tuple $G=(V,E_1,\ldots,E_t)$,
where $E_1,\ldots,E_t$ are pairwise disjoint sets of directed edges on a  set $V$ of vertices 
such that ({\em i}) there are no self-loops -- that is, $(v,v) \notin E_i$
for every $v \in V$ and every $E_i$, and ({\em ii}) if $(u,v) \in E_i$ then $(v,u) \notin E_j$ for every $E_j$.
As before, edges in $E_i$ are called $E_i$-edges.
The $E_i$-indegree and -outdegree of a vertex $u$
is defined as the number of incoming and outgoing $E_i$-edges incident to $u$.
We say that $G$ is {\em complete}, if for every $u,v \in V$ and $u\neq v$,
either $(u,v)$ or $(v,u)$ is an $E_i$-edge, for some $E_i$.
We consider the empty digraph and the digraph with only one vertex without any edge
as  complete digraphs.

We say that $G$ is an {\em $A|B$-regular} digraph, 
for $A,B\in \bbNop^{t\times m}$,
if there exists a partition $V=V_{1}\uplus \cdots\uplus V_n$
such that, for every $i\in [t]$ and for every $k\in [m]$,
the $E_i$-outdegree and -indegree of every vertex in $V_k$ is $A_{i,k}$ and $B_{i,k}$, respectively.
We say that $G$ has size $(|V_{1}|,\ldots,|V_{m}|)$, and
call $V=V_1\uplus \cdots\uplus V_m$ a witness partition for $A|B$-regularity of $G$.
When the entries in $A$ and $B$ are all $0$ or $\prdp{0}$, 
we regard the graph with only one vertex to be a complete $A|B$-regular digraph.

In this work we will be interested in computational problems concerning
the possible sizes
 of an $A|B$-biregular graph or -regular digraph, and the possible sizes of a complete $A|B$-biregular
graph or -regular digraph. 
Biregular one-color graphs are arguably quite natural, independently of any connection with satisfiability of a logic.
Completeness, as well as disjointness of edges for different colors, is more motivated specifically by our application to logic. Intuitively, the different edge colors in a biregular graph
 represent  the possible relationships between two elements in a structure. One color might represent a binary relationship, 
and another might represent
its negation.  Since every two elements have \emph{some} relationship, we want every pairs to be colored by exactly one edge color.
This will be formalized in Subsection \ref{subsec:procedure} below.

We briefly consider the (finite) \emph{membership problem}:
given size vectors $\vM, \vN$ along with  matrices  $A$ and $B$, all without $\infty$,
decide if there is $A|B$-biregular graph $G$ with size $\vM|\vN$. 
The problem is clearly in $\np$ if the entries in $\vM$ and $\vN$ are in unary, 
since we can guess $G$ and check that it is $A|B$-biregular with size $\vM|\vN$.

The \emph{degree sequence} for a ($1$-color) bipartite graph $(U,V,E)$ 
with $k$ vertices in $U$ and $k'$ vertices in $V$, 
is the pair of sequences $d_1, \ldots, d_k$ and $d'_1, \ldots, d'_{k'}$
where $d_1, \ldots, d_k$ enumerates the degrees of elements in $U$ in non-decreasing order 
and $d'_1, \ldots, d'_{k'}$ enumerates the degrees of elements of $V$ in non-decreasing order.
Note that given $1$-color degree matrices $A,B$ with only entries from $\bbN$ and size vectors $\vM, \vN$, coded in unary,
we can generate the degree sequences of any $A|B$-biregular graph with sizes $\vM|\vN$ in linear time:
if entry with fixed degree $d$ is to have size $m$, the degree sequences will contain a contiguous subsequence
consisting of $m$ $d$'s.

It follows from the Gale-Reyser Theorem (the Main Theorem in \cite{galer}) that one can determine in polynomial time whether
a pair of sequences is the degree sequences of a bipartite graph. 
From this we derive:

\begin{proposition} 
In the case of $1$-color degree matrices with only entries from $\bbN$,
coded in unary, the membership problem is in $\ptime$.
\end{proposition}

While we will not provide a detailed analysis of the complexity of the membership problem,
we will show that, when fixing $A$ and $B$, we  can succinctly describe -- and hence efficiently 
compute -- the size vectors of partitioned graphs for which membership holds.
This will be a consequence of
the following theorem, which
will also be the main technical tool for our decidability result.

\begin{theorem}
\label{thm:main-lemma-bireg}
For every degree matrices $A\in \bbN_{\infty,+p}^{t\times m}$ and $B\in \bbN_{\infty,+p}^{t\times n}$,
there is an (effectively computable) existential 
Presburger formula $\biregc_{A|B}(\vx,\vy)$ such that, 
for every pair of size vectors
 $\vM \in \bbNo^m$  and $\vN \in \bbNo^n$,
there is a complete $A|B$-biregular graph with size $\vM|\vN$
if and only if $\biregc_{A|B}(\vM,\vN)$ holds in $\cN_{\infty}$.
\end{theorem}

We have an analogous theorem for digraphs:

\begin{theorem}
\label{thm:main-lemma-direg}
For every degree matrices $A\in \bbN_{\infty,+p}^{t\times m}$ and $B\in \bbN_{\infty,+p}^{t\times m}$,
there is an (effectively computable) existential 
Presburger formula $\diregc_{A|B}(\vx)$ such that 
for every size vector $\vM\in \bbNo^m$,
there is a complete $A|B$-regular digraph with size $\vM$
if and only if $\diregc_{A|B}(\vM)$ holds in $\cN_{\infty}$.
\end{theorem}

The proofs of these two theorems are given later in Sections \ref{sec:1type}--\ref{sec:proofnonsimple},
beginning with an overview of the ideas via an extremely special case (the ``$1$-color case'') in Section~\ref{sec:1type}.
An immediate  consequence of these results is the decidability of graph analysis problems:

\begin{corollary}
\label{cor:graph-decision}
We can decide, given
matrices $A\in \bbN_{\infty,+p}^{t\times m}$ and $B\in \bbN_{\infty,+p}^{t\times n}$,
whether there exists a complete $A|B$-biregular graph.
The analogous result holds for digraphs.
Moreover, the decision procedure runs in non-deterministic exponential time
in the size of $A$ and $B$ where the coefficients are written in binary.
\end{corollary}

\begin{proof}
By Theorems~\ref{thm:main-lemma-bireg} and~\ref{thm:main-lemma-direg}, we can reduce the graph existence problems to  checking whether the existential closures of $\biregc_{A|B}(\vx,\vy)$ and $\diregc_{A|B}(\vx)$ hold in $\cN_{\infty}$. In turn, these problems are decidable by Theorem~\ref{thm:pres}.
Moreover, the upper bound for both cases holds by Lemma~\ref{lem:algo-for-bireg}, 
which we prove in Section~\ref{sec:complexity}.
\end{proof}

\begin{remark}
\label{rem:formula-general-finite-readjustment}
Theorems \ref{thm:main-lemma-bireg} and~\ref{thm:main-lemma-direg}, as well as Corollary \ref{cor:graph-decision}, 
can be easily readjusted in the case where we are interested only in finite sizes, i.e., when $\vM \in \bbN^m$ and $\vN\in \bbN^n$,
by replacing every atom $x=\infty$ in the formulas with the False value $\myF$
and requiring them to hold in $\cN$, instead of $\cN_{\infty}$.
Alternatively, we can also state inside the formulas that 
none of the variables in $\vx$ and $\vy$ are equal to $\infty$. 
\end{remark}

The rest of this section will be devoted to proving the decidability result concerning our logic,
making use of these theorems.

\subsection{Reducing satisfiability in the logic to biregular graph problems} \label{subsec:procedure}

We are now ready to present the decidability result for two variable logic with
ultimately periodic quantifiers: 
\begin{theorem}
\label{theo:main}
For every $\fotwopres$ sentence $\phi$, 
(i) there is an (effectively computable) existential Presburger sentence $\PREB^\infty_\phi$ such that
 $\phi$ has a model iff $\PREB^\infty_\phi$ holds in $\cN_{\infty}$
and
(ii) there is an (effectively computable) existential Presburger sentence $\PREB_\phi$ such that
 $\phi$ has a finite model iff $\PREB_\phi$ holds in $\cN$.
\end{theorem}

From the decision procedure for existential Presburger formulas (Theorem \ref{thm:pres})
mentioned in Section \ref{sec:prelims}, we will immediately obtain the following corollary:

\begin{corollary} \label{cor:decide} Both  satisfiability and finite satisfiability
 for $\fotwopres$ are decidable.
\end{corollary}

We  prove Theorem \ref{theo:main} using
Theorems ~\ref{thm:main-lemma-bireg} and~\ref{thm:main-lemma-direg}.
We start by observing that satisfiability for an $\fotwopres$ sentence ---  as well as spectrum analysis, to be defined formally in Section \ref{sec:spectrum} ---
can be converted effectively into the same question for a sentence
in a variant of Scott normal form:
\begin{align}
\label{eq:snf}
\phi &\ := \
\forall x \forall y\  \alpha(x,y)
\ \wedge\ 
\bigwedge_{i=1}^{k} \forall x \exists^{S_i} y\ \beta_i(x,y) \wedge x \neq y,
\end{align}
where $\alpha(x,y)$ is a quantifier-free formula, each $\beta_i(x,y)$ is an atomic formula and 
each $S_i$ is a u.p.s.
More precisely, every $\fotwopres$ sentence
can be converted effectively into a sentence in form (\ref{eq:snf}) such
that they are equisatisfiable and have the same spectrum.
The proof, which is fairly standard, can be found in the appendix.
By taking the least common multiple, we may assume that all the non-zero periods in all $S_i$ are the same.
For example, if $S_1=\{\prd{0}{2}\}$ and $S_2=\{\prd{0}{3}\}$,
they can be rewritten as $S_1 = \{\prd{0}{6},\prd{2}{6},\prd{4}{6}\}$
and $S_2=\{\prd{0}{6},\prd{3}{6}\}$.
Here it is worth mentioning that when we write $\alpha(x,y)$ and $\beta(x,y)$,
we implicitly assume that both $x$ and $y$ occur.
For the rest of this section, we fix an $\fotwopres$ sentence $\phi$
in form \eqref{eq:snf}, with all $S_i$ as described above. The signature of structures we consider will be the signature
of $\phi$.

We recall some standard terminology.
A {\em $1$-type} is a maximally consistent set of atomic and negated atomic formulas 
using only variable $x$,
including atomic formulas such as $r(x,x)$ or $\neg r(x,x)$.
Each $1$-type can be identified with the quantifier-free formula
formed as the conjunction of its constituent formulas.
Thus, we say that an element $u$ in a structure $\cA$ has $1$-type $\pi$,
if $\pi$ holds on the element $u$.
For a structure $\cA$ with domain $A$, we let $A_{\pi}$ denote the set of elements in $\cA$ with $1$-type $\pi$.
Clearly  $A$ is partitioned into the sets $A_{\pi}$ with  $\pi$ ranging over $1$-types.
Similarly, a {\em $2$-type} is a maximally consistent set of binary atoms and negations of atoms containing $x\neq y$,
where each atom or its negation uses two variables $x$ and $y$.\footnote{Under standard definitions, such as the ones in~\cite{gkv,ctwobinary},
a $2$-type may contain unary atoms or negations of unary atoms involving variable $x$ or $y$.
In this paper we use a different definition and require that each atom and the negation of an atom 
in a $2$-type explicitly mentions both $x$ and $y$.}
The notion of a pair of elements $(u,v)$ in a structure $\cA$ having $2$-type $E$ is defined as for $1$-types.
We let $\Pi = \{\pi_1, \pi_2, \ldots, \pi_n\}$ and $\cE=\{E_1,\ldots,E_t\}$
denote the sets of all $1$-types and $2$-types (over the same signature as $\phi$), respectively.

We can now explain the connection between satisfiability in the logic and graph analysis.
This will involve associating to a model  $\cA$ for a formula $\phi$ a collection of graphs and digraphs, along
with  partitions that witnesses biregularity of the graphs and digraphs.
The following crucial definition explains the first aspect,
how to go from a structure $\cA$ to a collection of graphs and digraphs.

\begin{definition} \label{def:graph-abstraction}
Let $\cA$ be a structure.
A \emph{graph representation} of $\cA$
is a complete $t$-color digraph $G_{\cA}=(V,E_1,\ldots,E_t)$ 
where the vertices in $G_{\cA}$ are the elements in the domain of $\cA$
and for each pair of elements $(u,v)$ where $u\neq v$,
we put an arbitrary orientation between them: either from $u$ to $v$ or from $v$ to $u$.
\begin{itemize}
\item
If the orientation is from $u$ to $v$, then we set $(u,v)$ as an $E_i$-edge where $E_i$ is the $2$-type of $(u,v)$.
\item 
If the orientation is from $v$ to $u$, then we set $(v,u)$ as an $E_i$-edge where $E_i$ is the $2$-type of $(v,u)$.
\end{itemize}
For a graph representation $G_{\cA}$ of a structure $\cA$,
we will consider two kinds of subgraphs of $G_{\cA}$.
The first is the induced subgraph of $G_{\cA}$ by the set $A_{\pi}$ for a $1$-type  $\pi$,
denoted by $G_{\cA,\pi}$.
The second is the bipartite restriction of $G_{\cA}$ on the vertices in $A_{\pi}$ and $A_{\pi'}$ for different $1$-types $\pi,\pi'$,
denoted by $G_{\cA,\pi,\pi'}$.
That is, $G_{\cA,\pi,\pi'}$ is the complete bipartite graph where 
$A_{\pi}$ is the set of vertices on the left hand side,
$A_{\pi'}$ is the set of vertices on the right hand side and
the edges are between the vertices in $A_{\pi}$
and the vertices in $A_{\pi'}$.
Note that in $G_{\cA,\pi,\pi'}$ the edges are oriented.
Some edges are oriented from the vertices in $A_{\pi}$ to the vertices in $A_{\pi'}$,
and some from the vertices in $A_{\pi'}$ to the vertices in $A_{\pi}$.
It is complete since for every pair $(u,v)\in A_{\pi}\times A_{\pi'}$,
either $(u,v)$ or $(v,u)$ is an $E_i$-edge, for some $E_i$.
\end{definition}

See Figure~\ref{fig:behavior} for an illustration of a graph representation
of a structure $\cA$ with domain $\{u_1,u_2,u_3,v_1,v_2,w\}$.
The $1$-types are $\pi_1,\pi_2,\pi_3$, and $2$-types are $E_1,E_2,E_3,E_4$.
In the graph representation the edge between $u_1$ and $v_1$ is oriented from $u_1$ to $v_1$
and the $2$-type of $(u_1,v_1)$ is $E_1$.

\begin{figure}
\begin{center}

\begin{tikzpicture}

\definecolor{emerald}{rgb}{0.31, 0.78, 0.47}

%%% fake nodes on left
\node (b0) at (-2,3) {\color{red}\footnotesize$E_2$};
\node (b1) at (-1.25,2.5) {\color{red}\footnotesize$E_2$};

%%% nodes u1,u2,u3
\node[circle,fill=blue,inner sep=0pt,minimum size=3pt,label=left:{\footnotesize$u_1$}] (u1) at (0,2) {};
\node[circle,fill=blue,inner sep=0pt,minimum size=3pt,label=below:{\small $u_2$}] (u2) at (0,0) {};
\node[circle,fill=blue,inner sep=0pt,minimum size=3pt,label=below:{\small $u_3$}] (u3) at (0,-2) {};
\draw[gray!50] (0,0) ellipse (1.1cm and 2.5cm);
\node at (0,-2.8) {\small $A_{\pi_1}$};

%%% nodes v1,v2
\node[circle,fill=blue,inner sep=0pt,minimum size=3pt,label=right:{\small $v_1$}] (v1) at (7,1.5) {};
\node[circle,fill=blue,inner sep=0pt,minimum size=3pt,label=below:{\small $v_2$}] (v2) at (7,-1.5) {};
\draw[gray!50] (7,0) ellipse (.9cm and 2cm);
\node at (7,-2.3) {\small $A_{\pi_2}$};

%%% node w
\node[circle,fill=blue,inner sep=0pt,minimum size=3pt,label=above:{\small $w$}] (w) at (3.5,5) {};
\draw[gray!50] (3.5,5) ellipse (.35cm and .7cm);
\node at (3.5,6) {\small $A_{\pi_3}$};

%%% edges from u to u

\draw [-{Stealth[scale=1.1]},shorten >=6pt,blue,line width =.2mm] (u1) to [bend left=50] (u2);
\node[fill=white] at (.55,1) {\color{blue}\footnotesize$E_3$};

\draw [-{Stealth[scale=1.1]},shorten >=6pt,blue,line width=.2mm] (u1) to [bend right=40] (u3);
\node[fill=white] at (-.7,-.7) {\color{blue}\footnotesize$E_3$};

\draw [-{Stealth[scale=1.1]},shorten >=6pt,blue,line width =.2mm] (u2) to [bend left=50] (u3);
\node[fill=white] at (.55,-1) {\color{blue}\footnotesize$E_3$};

%%% edges from w to u
\draw [red,line width =.2mm] (w) to [bend right=30]  (b0);
\draw [-{Stealth[scale=1.1]},shorten >=6pt,red,line width =.2mm] (b0) to [bend right=45]  (u3);

\draw [red,line width=.2mm] (w) to [bend right=30] (b1);
\draw [-{Stealth[scale=1.1]},shorten >=6pt,red, line width =.2mm] (b1) to [bend right=30]   (u2);

\draw [-{Stealth[scale=1.1]},shorten >=6pt,emerald,line width=.2mm] (u1) -- node[color=emerald,align=center,fill=white] {\footnotesize$E_4$} (w);

%%% edges from w to v
\draw [-{Stealth[scale=1.1]},shorten >=6pt,blue,line width =.2mm] (w) to [bend left=85]   (v2);
\node[fill=white] at (7.2,3) {\color{blue}\footnotesize$E_3$};

\draw [-{Stealth[scale=1.1]},shorten >=6pt,red,line width =.2mm] (v1) -- node[color=red,align=center,fill=white] {\footnotesize$E_2$}  (w);

%%% edges from v to v
\draw [-{Stealth[scale=1.1]},shorten >=6pt,black,line width =.2mm] (v1) -- node[color=black,align=center,fill=white] {\footnotesize$E_1$}  (v2);

%%% edges from u to v

\draw [-{Stealth[scale=1.1]},shorten >=6pt,line width =.2mm] (u1) -- node[color=black,align=center,fill=white] {\footnotesize$E_1$} (v1);
\draw [-{Stealth[scale=1.1]},shorten >=6pt,red,line width =.2mm] (u1) -- node[xshift=-1.7cm,yshift=.85cm,color=red,align=center,fill=white] {\small$E_2$} (v2);

\draw [-{Stealth[scale=1.1]},shorten >=6pt,line width =.2mm] (u2) -- node[xshift=-1.7cm,yshift=.4cm,color=black,align=center,fill=white] {\small$E_1$} (v2);
\draw [-{Stealth[scale=1.1]},shorten >=6pt,emerald,line width =.2mm] (v1) -- node[xshift=.8cm,yshift=.17cm,color=emerald,align=center,fill=white] {\small$E_4$} (u2);

\draw [-{Stealth[scale=1.1]},shorten >=6pt,emerald,line width =.2mm] (v2) -- node[color=emerald,align=center,fill=white] {\small$E_4$} (u3);
\draw [-{Stealth[scale=1.1]},shorten >=6pt,emerald,line width =.2mm] (v1) -- node[xshift=1.7cm,yshift=.8cm,color=emerald,align=center,fill=white] {\small$E_4$} (u3);

\end{tikzpicture}

\end{center}
\label{fig:behavior}
\caption{Illustration of a graph representation of a structure with $1$-types $\pi_1,\pi_2,\pi_3$.
The $2$-types are $E_1,E_2,E_3,E_4$ represented by edges with color black, red, blue and green, respectively.
The vertices $u_1,u_2,u_3$ are in $A_{\pi_1}$, $v_1,v_2$ are in $A_{\pi_2}$ and $w$ is in $A_{\pi_3}$.}
\end{figure}

\begin{remark}
\label{rem:graph-rep}
It is worth noting that for a structure $\cA$,
the graph representation of $\cA$ is not unique since 
it depends on the orientation put between the vertices.
On the other hand, a graph representation uniquely defines a structure
since the information about the vertices and the edges in a graph representation,
i.e., the $1$- and $2$-types,
uniquely determines the relations in the structure.
\end{remark}

The biregular graph problem which our reduction produces will involve counting the possible sizes of 
certain partitions in the vector of graphs $G_{\cA, \pi, \pi'}$ and $G_{\cA, \pi}$,
for every graph representation $G_{\cA}$ of every structure $\cA\models \phi$.
We now explain the partitions we are looking for.

Let $g:\{\tout,\tin\}\times\TwoTypes \times \OneTypes \rightarrow \bbNop$ be a function.
We will use  $g$ to describe the ``behavior'' of elements in a graph representation $G_{\cA}$ in the 
following sense.
We say that an element $u\in A$  {\em behaves according to $g$} in a graph representation $G_{\cA}$, if, 
for every $\pi\in \OneTypes$ and for every $E\in \TwoTypes$:
\begin{itemize}
\item
the number of outgoing $E_i$-edges in the graph $G_{\cA}$ from $u$ to vertices $v \in A_{\pi}$ is $g(\tout,E,\pi)$,
\item
the number of incoming $E_i$-edges in the graph $G_{\cA}$ to $u$ from vertices $v \in A_{\pi}$ is $g(\tin,E,\pi)$.
\end{itemize}
For example, in the graph representation in Figure~\ref{fig:behavior}
the element $w$ behaves according to the following function $g_1$:
\begin{itemize}
\item
$g_1(\tout, E_2,\pi_1) = 2$, $g_1(\tout, E_3,\pi_2)  = 1$,
$g_1(\tin,E_2,\pi_2) =1$, $g_1(\tin,E_4,\pi_1) = 1$.
\item 
$g_1$ maps all the other tuples in $\{\tout,\tin\}\times\TwoTypes \times \OneTypes$ to $0$.
\end{itemize}
As another example, the element $u_1$ behaves according to the following function $g_2$:
\begin{itemize}
\item
$g_2(\tout,E_1,\pi_2) = 1$, $g_2(\tout, E_2,\pi_2)  = 1$,
$g_2(\tout,E_3,\pi_1) =2$. \\ And $g_2(\tout,E_4,\pi_3) = 1$.
\item 
The rest are mapped to $0$.
\end{itemize}

We will call a function $g:\{\tout,\tin\}\times\TwoTypes \times \OneTypes \rightarrow \bbNop$
a {\em behavior}.
The restriction of $g$ on $1$-type $\pi$ is the function  $g_{\pi}:\{\tout,\tin\}\times\TwoTypes\to \bbNop$,
where $g_{\pi}(\kappa,E)= g(\kappa,E,\pi)$ for every $\kappa\in\{\tout,\tin\}$ and $E\in \TwoTypes$.
We call the function $g_\pi$ the \emph{behavior} (function) towards $1$-type $\pi$.

We are, of course, only interested in $1$-types and behaviors that are ``allowed''
by the sentence $\phi$ we are considering.  
To formalize this,
we will use the following terminology,
where $\alpha(x,y)$, $\beta_i(x,y)$ and $S_i$ are from the fixed $\phi$.
\begin{compactitem}
\item
A $1$-type $\pi\in \Pi$
is {\em compatible} (with~$\phi$) if \footnote{As usual, we use $\models$ 
in both $\cA\models \phi$ (for ``$\cA$ satisfies $\phi$'') 
and $\phi_1\models \phi_2$ (for ``$\phi_1$ implies $\phi_2$'').}
$$
\pi(x)  \ \models\ \alpha(x,x).
$$
Otherwise, we say that $\pi$ is {\em incompatible}.
Intuitively, $\pi$ is incompatible means that
whenever $\cA\models \phi$, 
there is no element with $1$-type $\pi$.

\item 
For a $1$-type $\pi\in \Pi$, for a behavior function $g:\{\tout,\tin\}\times\TwoTypes \times\OneTypes \to \bbNop$,
we say that $(\pi,g)$ is {\em compatible} (with~$\phi$)
if, for every $E\in \TwoTypes$ and for every $\pi'\in \Pi$:

If $g(\tout,E,\pi')\neq 0$, then
$$
\pi(x)\ \wedge\ E(x,y)\ \wedge\ \pi'(y)\ \models\ \alpha(x,y)
\quad\text{and}\quad
\pi(y)\ \wedge\ E(y,x)\ \wedge\ \pi'(x)\ \models\ \alpha(x,y)
$$
and if $g(\tin,E,\pi')\neq 0$, then
$$
\pi(x)\ \wedge\ E(y,x)\ \wedge\ \pi'(y)\ \models\ \alpha(x,y)
\quad\text{and}\quad
\pi(y)\ \wedge\ E(x,y)\ \wedge\ \pi'(x)\ \models\ \alpha(x,y).
$$
Otherwise, we say that $(\pi,g)$ is {\em incompatible}.
Intuitively, $(\pi,g)$ is incompatible means that
whenever $\cA\models \phi$, 
there is no element in $A_{\pi}$ that behaves according to $g$ in any graph representation $G_{\cA}$ of $\cA$.

\item
A function $g$ is a {\em good} behavior (w.r.t.~$\phi$) if 
for every $i \in [k]$:\footnote{Here the operation $+$ on $\bbNop$
is defined to be the commutative extension of the standard addition on $\bbN$ such that
$a+ \infty = \prdp{a}+\infty = \infty$ and $\prdp{a}+b= \prdp{a}+\prdp{b} = \prdp{(a+b)}$}
\begin{align}
\label{eq:sum-good-function}
\sum_{E \ni \beta_i(x,y)} 
\ 
\sum_{\pi \in \OneTypes}\ g(\tout,E,\pi) \ + \
\sum_{E \ni \beta_i(y,x)} 
\ 
\sum_{\pi \in \OneTypes}\ g(\tin,E,\pi) \ \
 \in\ \ S_i.
\end{align}
Intuitively, for a vertex $u$ in a graph representation $G_{\cA}$ that behaves according to $g$,
the sum $\sum_{E \ni \beta_i(x,y)} \sum_{\pi \in \OneTypes}\ g(\gout,E,\pi)$ is the number of outgoing edges
that contains the relation $\beta_i(x,y)$
and the sum $\sum_{E \ni \beta_i(y,x)} \sum_{\pi \in \OneTypes}\ g(\tin,E,\pi)$ is the number of incoming edges
that contains the relation $\beta_i(y,x)$.
Their total sum is the the number of elements $v$ such that $\cA,x/u,y/v\models \beta_i(x,y)$.
Hence, when $\cA\models \phi$, it must be inside the set $S_i$.
\end{compactitem}

The notion of compatibility will be used to capture
the universal part $\forall x \forall y \alpha(x,y)$ of our formula.
The notion of good function will be used to capture
the universally and presburger quantified part: $\bigwedge_{i=1}^k \forall x \exists y ^{S_i} \beta_i(x,y)\wedge x\neq y$.

We observe that, for every structure $\cA\models \phi$, for every graph representation $G_{\cA}$ of $\cA$,
each vertex in $G_{\cA}$ behaves according to a function $g$
where the range is a subset of $\{0,\ldots,q,\prdp{0},\ldots,\prdp{q},\infty\}$ for
$q$ the maximal non-$\infty$ offset in all $S_i$ (when seen as finite sets of linear sets). 
Indeed, suppose $\cA\models \phi$ and let $G_{\cA}$ be its graph representation.
Let $u$ be an element that behaves according to $g$.
Suppose $g(\tout,E,\pi)=a$ or $\prdp{a}$ for some $a> q$, $E\in \TwoTypes$ and $\pi\in \Pi$.
We will show that $u$ also behaves according to a function $g'$ 
where $g'$ is the same function as $g$ except that $g'(\tout,E,\pi)$ is now $\prdp{(a-sp)}$
where $s$ is the minimum integer such that $a-sp \leq q$.
We consider the case where $g(\tout,E,\pi)=a$.
Suppose $\beta_i(x,y)\in E$ where $i \in [k]$.
Let $b$ denote the number of elements $v$ such that $\cA,x/u,y/v \models \beta_i(x,y)\wedge x\neq y$.
Since $u$ behaves according to $g$, we have $b\geq a$ and hence $b> q$.
Moreover, $b\in S_i$ since  $\cA\models \phi$.
Because $b> q$ there must be $\prdp{c}\in S_i$ such that $b\in \prdp{c}$.
This means that $u$ also behaves according to $g'$ where
$g'$ is the same as $g$ except that $g'(\tout,E,\pi)=\prdp{(a-sp)}$ where $s$ is the minimum integer such that $a-sp \leq q$.
The cases where $g(\tout,E,\pi)=\prdp{a}$ or $g(\tin,E,\pi)=a$ or $g(\tin,E,\pi)=\prdp{a}$ with $a> q$
can be treated in a similar manner.

So we may concentrate on only the finite set $\cG = \{g_1, g_2, \ldots,g_m\}$ of good 
behaviors whose co-domain is $\{0,\ldots,q,\prdp{0},\ldots,\prdp{q},\infty\}$,
where $q$ the maximal non-$\infty$ offset in all $S_i$.
Below we will partition elements based on their behaviors, always using good behaviors, thus the partitions will
be finite.

For $\cA\models \phi$, for a graph representation $G_\cA$ of $\cA$,
we can partition $A = A_{\pi_1,g_1}\uplus \cdots \uplus A_{\pi_n,g_m}$ 
according to the $1$-types and good behavior functions:
for every element $u\in A$, we pick a behavior function $g_j$ such that $u$ behaves according to $g$ (in $G_{\cA}$),
and declare that $u \in A_{\pi_i,g_j}$ where $\pi_i$ is the $1$-type of $u$.\footnote{In general, for an element $u\in A$,
there may be several behaviors according to which $u$ behaves; we partition the domain by picking one such behavior.}
We can then consider the vector of subgraphs $G_{\cA,\pi}$ and $G_{\cA, \pi, \pi'}$ of $G_{\cA}$. 
We call this the \emph{Type-Behavior Partitioned Graph Vector} associated to the graph $G_{\cA}$.
Intuitively, to decide whether $\phi$ is satisfiable,
we construct a Presburger formula that captures the sizes of the subgraphs $G_{\cA,\pi}$ and $G_{\cA, \pi, \pi'}$
of every possible graph representation $G_{\cA}$ of every model $\cA\models \phi$.

At this point we can expand on the intuition for the reduction of satisfiability to biregular graph problems.
We will construct a sentence $\PREB_\phi$
that ``counts'' the possible cardinalities of partitioned graphs corresponding to
a Type-Behavior Partitioned Graph Vector associated to a graph representation $G_{\cA}$ for a model $\cA$ of $\phi$.
The sentence $\PREB_\phi$ will be of the form:
\begin{align}
\label{eq:preb}
\PREB_\phi &  :=  \exists \vX \ \const_1(\vX)
 \wedge  \const_2 (\vX) \wedge (\bigvee_{i\in [n],\ j\in [m]}  X_{\pi_i, g_j} \neq 0),
\end{align}
where $\vX$ is a vector of variables $(X_{\pi_1, g_1}, X_{\pi_1, g_2}, \ldots , X_{\pi_n, g_m})$.
Intuitively, each  $X_{\pi_i,g_j}$ represents $|A_{\pi_i,g_j}|$ in some graph representation $G$.
The final conjunct ensures that the domain is non-empty.
By the formula $\const_1(\vX)$, we capture the consistency of the non-negative integers $\vX$ with the first conjunct $\forall x\forall y \ \alpha(x,y)$ of $\phi$.
By the formula $\const_2(\vX)$,
we capture the consistency of the non-negative integers $\vX$ with the second conjuncts
$\bigwedge_{i=1}^{k} \forall x \exists^{S_i} y\ \beta_i(x,y) \wedge x \neq y$. In $\const_2$ we will consider the Type-Behavior Partitioned Graph Vector
as the common solution of a set biregular graph and digraph problems, and make use of the Presburger definability of biregular graph problems.

Towards defining the formulas $\const_1$ and $\const_2$,
we define matrices that will constrain the partitions.
\begin{align*}
M_\pi  :=
\begin{pmatrix}
g_1(\tout,E_1,\pi) &  \cdots & g_m(\tout,E_1,\pi)
\\
\vdots &  \ddots & \vdots
\\
g_1(\tout,E_{t},\pi) &  \cdots & g_m(\tout,E_{t},\pi)
\end{pmatrix}
\end{align*}
and
\begin{align*}
\rev{M}_\pi  :=
\begin{pmatrix}
g_1(\tin,E_1,\pi) & \cdots & g_m(\tin,E_1,\pi)
\\
\vdots & \ddots & \vdots
\\
g_1(\tin,E_t,\pi) &  \cdots &  g_m(\tin,E_t,\pi)
\end{pmatrix}.
\end{align*}
That is, $M_\pi$ contains the information of the outgoing edges toward $1$-type $\pi$
and $\rev{M}_\pi$ contains the information of the incoming edges from $1$-type $\pi$.

Now, we explain how to capture the behavior between elements with distinct $1$-types.
Define matrices $L_\pi,\rev{L}_\pi\in \bbNop^{2t\times m}$:
\begin{align} \label{eq:Ls}
L_\pi  :=
\begin{pmatrix}
M_\pi
\\
\rev{M}_\pi
\end{pmatrix}
\quad \text{and} \quad
\rev{L}_\pi  :=
\begin{pmatrix}
\rev{M}_\pi
\\
M_\pi
\end{pmatrix};
\end{align}
that is, in $L_\pi$ the first $t$ rows come from $M_\pi$ with the next $t$ rows from $\rev{M}_\pi$.
On the other hand, in $\rev{L}_\pi$ the first $t$ rows come from $\rev{M}_{\pi}$, followed by the $t$ rows from $M_\pi$.

The intended meaning of the matrices is as follows.
For every structure $\cA$, for every graph representation $G_{\cA}$ of $\cA$,
$\cA\models \phi$ if and only if the following two sentences hold.
\begin{itemize}
\item For every $1$-type $\pi$,
the subgraph $G_{\cA,\pi}$ is a complete $M_\pi|\rev{M}_\pi$-regular digraph
\item For distinct $1$-types $\pi,\pi'$,
the subgraph $G_{\cA,\pi,\pi'}$ is a complete  \\
$L_{\pi'}|\rev{L}_{\pi}$-biregular graph.

Here the first $t$ rows in $L_{\pi'}|\rev{L}_{\pi}$ capture the edges in $G_{\cA,\pi,\pi'}$ that are oriented from left to right,
whereas the last $t$ rows capture the edges in $G_{\cA,\pi,\pi'}$ that are oriented from right to left.
\end{itemize}

We are now ready to define the formulas, beginning with   $\const_1(\vX)$.
Letting $H$ be the set of all incompatible pairs $(\pi,g)$, the formula $\const_1(\vX)$ can be defined as follows:
\begin{align}
\label{eq:consistent-1}
\const_1(\vX) & : = \ 
\bigwedge_{\pi\ \text{is incompatible, } g \in \cG} X_{\pi, g} = 0
\quad\wedge\quad
\bigwedge_{(\pi,g) \in H} X_{\pi, g} = 0.		
\end{align}

We turn to formula $\const_2(\vX)$.
Recall that we enumerated all the $1$-types as $\pi_1,\ldots,\pi_n$.
We now define $\const_2$, where 
below each $\vX_{\pi_i}$ is the vector \\
$(X_{\pi_i, g_1}, X_{\pi_i, g_2}, \ldots , X_{\pi_i, g_m})$ and each $\vX_{\pi_j}$ is defined in the same way:
\begin{multline}
\qquad \const_2(\vX) : = \\ 
\bigwedge_{1 \leq i < j \leq n}
\biregc_{L_{\pi_j}|\rev{L}_{\pi_i}}(\vX_{\pi_i},\vX_{\pi_j})
\ \wedge \
\bigwedge_{1 \leq i \leq n} 
\diregc_{M_{\pi_i}|\rev{M}_{\pi_i}}(\vX_{\pi_i}).
\end{multline}

Observe that formula $\const_1(\vX)$ is Presburger definable by inspection, while
$\const_2(\vX)$ is Presburger definable using Theorem \ref{thm:main-lemma-bireg} and Theorem \ref{thm:main-lemma-direg}. 
Thus, the sentence $\PREB_{\phi}$ is an existential Presburger sentence, and the following lemma shows that 
$\PREB_{\phi}$ is indeed the sentence required by Theorem~\ref{theo:main}.

\begin{lemma}
\label{lem:correct}
For each structure $\cA\models \phi$,
for every graph representation $G_{\cA}$ of $\cA$,
there is a partition $A= A_{\pi_1,g_1}\uplus\cdots\uplus A_{\pi_n,g_m}$
such that:
\begin{itemize}
\item
For every $i\in [n]$, for every $j\in[m]$, $A_{\pi_i,g_j}$ contains the elements with $1$-type $\pi_1$
and behaves according to $g_j$ in the graph representation $G_{\cA}$.
\item 
$\const_1(\vN)\ \wedge\ \const_2(\vN) \ \wedge \ \bigvee_{i\in [n],\ j \in [m]} |A_{\pi_i,g_j}|\neq 0$
holds, where $\vN=(|A_{\pi_1,g_1}|,\ldots,|A_{\pi_n,g_m}|)$.
\end{itemize}

Conversely, for every non-zero vector $\vN$ such that $\const_1(\vN)\wedge\const_2(\vN)$ holds,
there is a structure $\cA\models\phi$, a graph representation $G_{\cA}$ and 
a partition $A = A_{\pi_1,g_1}\uplus \cdots\uplus A_{\pi_n,g_m}$
such that:
\begin{itemize}
\item
$\vN=(|A_{\pi_1,g_1}|,\ldots,|A_{\pi_n,g_m}|)$.
\item 
For every $i\in [n]$, for every $j\in[m]$, $A_{\pi_i,g_j}$ contains the elements with $1$-type $\pi_1$
and behaves according to $g_j$ in the graph representation $G_{\cA}$.
\end{itemize} 

\end{lemma}

\begin{proof}
We  prove the first statement in the lemma, the direction from a model of the formula to a solution.
Let $\cA \models \phi$.
We fix a graph representation $G_{\cA}$, and the corresponding Type-Behavior Partitioned Graph Vector.
We will show that 
when each $X_{\pi_i,g_j}$ is assigned the value $|A_{\pi_i,g_j}|$ we have that
$$
\const_1(\vX) \ \wedge \ \const_2 (\vX)\ \wedge \
\bigvee_{i\in [n],\ j \in [m]} X_{\pi_i,g_j}\neq 0.
$$
Since $\cA$ contains at least one element,
at least one of the $A_{\pi,g}$'s is not empty.
Hence the last conjunct $\bigvee_{i\in [n],\ j \in [m]} |A_{\pi_i,g_j}|\neq 0$ holds.

Since $\cA\models \forall x\forall y\ \alpha(x,y)$,
$A_{\pi_i}=\emptyset$ whenever $\pi_i$ is incompatible
and $A_{\pi_i,g_j}=\emptyset$ whenever $(\pi_i,g_j)$ is incompatible.
Thus, $\const_1(\vX)$ holds for the assignment.

We show that $\const_2(\vX)$ holds.
Consider an arbitrary $\pi_i \in \Pi$.
As explained above, the subgraph 
$G_{\cA,\pi_i}$ is a complete $M_{\pi_i}|\rev{M}_{\pi_i}$-biregular digraph 
with size  $(|A_{\pi_i,g_1}|,\ldots,|A_{\pi_i,g_m}|)$.
By Theorem~\ref{thm:main-lemma-direg}, $\diregc_{M_{\pi_i}|\rev{M}_{\pi_i}}(\vX_{\pi_i})$ holds for the assignment.

For $\pi_i,\pi_j \in\Pi$ with $i< j$,
the subgraph $G_{\cA,\pi_i,\pi_j}$ can be viewed as a complete 
$L_{\pi_j}|\rev{L}_{\pi_i}$-biregular graph.
By Theorem~\ref{thm:main-lemma-bireg},
$\biregc_{L_{\pi_j}|\rev{L}_{\pi_i}}(\vX_{\pi_i},\vX_{\pi_j})$ holds.
Therefore, $\const_2(\vX)$ holds for the assignment, which completes this direction.

We turn to the second statement, going from a solution to a model.
Let $\vN=(N_{\pi_1,g_1},\ldots,N_{\pi_n,g_m})$ be a non-zero vector such that $\const_1(\vN)\wedge\const_2(\vN)$ holds.
For each $\pi_i\in \Pi$, let $\vN_{\pi_i} = (N_{\pi_i,g_1},\ldots,N_{\pi_i,g_m})$.

For each $(\pi_i,g_j)\in \Pi\times \cG$, we have a set $V_{\pi_i,g_j}$ with cardinality $N_{\pi_i,g_j}$.
We let $V_{\pi_i} = \bigcup_{g_j\in \cG} V_{\pi_i,g_j}$ for each $\pi_i \in \Pi$.
We construct a structure $\cA\models\phi:$ along with a particular graph representation.
\begin{itemize}
\item 
The domain is $A = \bigcup_{\pi_i\in \Pi,\ g_j\in \cG} V_{\pi_i,g_j}$.

Note that since $\vN$ is a non-zero vector, at least one $V_{\pi_i,g_j}$ is not empty,
and therefore, $A$ is not empty. 
\item 
For each $\pi_i \in \Pi$ and
for each element $u \in V_{\pi_i}$,
the predicates that hold on $u$ are defined such that
the $1$-type of $u$ is $\pi_i$.
\item 
For each $\pi_i \in \Pi$, we define the edges on each pair $(u,v)\in V_{\pi_i}\times V_{\pi_i}$, 
where $u\neq v$ as follows. Note that the binary atoms of the model for these pairs will follow.

Since $\diregc_{M_{\pi_i}|\rev{M}_{\pi_i}}(\vN_{\pi_i})$ holds,
by Theorem \ref{thm:main-lemma-direg}, there is a complete $M_{\pi_i}|\rev{M}_{\pi_i}$-regular digraph 
$G_{\pi_i} = (V_{\pi_i},E_1,\ldots,E_t)$ with size $\vN_{\pi_i}$.
Note that we can take $V_{\pi_i}$ as the domain of the graph and $V_{\pi_i}=V_{\pi_i,g_1}\uplus \cdots\uplus V_{\pi_i,g_m}$
as the witness partition since $(|V_{\pi_i,g_1}|,\ldots,|V_{\pi_i,g_m}|)=\vN_{\pi_i}$ by construction.
We then use the edges in $E_1,\ldots,E_t$ in $G_{\pi_i}$ to define the orientation and the $2$-types of each pair 
$(u,v)\in V_{\pi_i}\times V_{\pi_i}$ where $u\neq v$.

\item 
For every $\pi_i,\pi_j$ with $i< j$,
we now define the edges, and hence the binary atoms of the model, on each pair $(u,v)\in V_{\pi_i}\times V_{\pi_j}$.

Since $\biregc_{L_{\pi_j}|\rev{L}_{\pi_i}}(\vN_{\pi_i},\vN_{\pi_j})$ holds,
applying Theorem~\ref{thm:main-lemma-bireg},
there is a complete $L_{\pi_j}|\rev{L}_{\pi_i}$-biregular graph $G_{\pi_i,\pi_j} = (V_{\pi_i},V_{\pi_j},E_1,\ldots,E_t,\rev{E_1},\ldots,\rev{E_t})$
with size $\vN_{\pi_i}|\vN_{\pi_j}$.
Again, note that we can take $V_{\pi_i}$ and $V_{\pi_j}$ as
the set of vertices on the left hand side and the right hand side of the graph $G_{\pi_i,\pi_j}$ respectively,
and that $V_{\pi_i}=V_{\pi_i,g_1}\uplus \cdots\uplus V_{\pi_i,g_m}$ and  
$V_{\pi_j}=V_{\pi_j,g_1}\uplus \cdots\uplus V_{\pi_j,g_m}$ as the witness partition since 
the sizes $(|V_{\pi_i,g_1}|,\ldots,|V_{\pi_i,g_m}|)$ and $(|V_{\pi_j,g_1}|,\ldots,|V_{\pi_j,g_m}|)$
match the vectors $\vN_{\pi_i}$ and $\vN_{\pi_j}$  by construction.
We then use the edges in $E_1,\ldots,E_t$ in $G$ to define the orientation from left to right
and the edges in $\rev{E}_1,\ldots,\rev{E}_t$ in $G$ to define the orientation from right to left.

\end{itemize}
The above process produces a model $\cA$ along with a complete $t$-color graph $G$ that is a valid representation of the model.  It is easy to see that
 for every $1$-type $\pi_i$, every behavior function $g_j$,
every vertex $u\in V_{\pi_i,g_j}$ has $1$-type $\pi_i$ behaves according to the function $g_j$.

To show that $\cA\models \phi$,
we first show that $\cA\models \forall x \forall y\ \alpha(x,y)$. 
Let $u,v \in A$.
There are two cases.
\begin{itemize}
\item
When $u=v$ and $u\in V_{\pi}$.
This means $V_{\pi}\neq \emptyset$. Hence $|V_\pi|=\sum_{g\in \cG} N_{\pi,g} \neq 0$.
Therefore, $\pi$ is compatible, which by definition means $\pi(x) \models \alpha(x,x)$.
By the construction of $\cA$, we have $\cA,x/u,y/u \models \alpha(x,y)$.
\item 
When $u\neq v$ and $u\in V_{\pi}$ and $v\in V_{\pi'}$.

Suppose $(u,v)$ is an $E_i$-edge in the graph $G$, i.e., the orientation is from left to right.
This means there is $g\in \cG$ such that $g(\tout,E_i,\pi')\neq 0$ and $u \in V_{\pi,g}$,
which implies that $V_{\pi,g}\neq \emptyset$, i.e., $N_{\pi,g}\neq 0$.
Since $\const_1(\vN)$ holds, which states that $N_{\pi,g}=0$ whenever $(\pi,g)$ is incompatible,
the pair $(\pi,g)$ is compatible -- that is, 
$$
\pi(x) \wedge E(x,y) \wedge \pi'(y)\ \models\ \alpha(x,y)
\quad\text{and}\quad
\pi(y) \wedge E(y,x) \wedge \pi'(x)\ \models\ \alpha(x,y).
$$
Since $\cA$ is a structure with  representation $G$, we have
$\cA,x/u,y/v \models \alpha(x,y)$ and $\cA,x/v,y/u \models \alpha(x,y)$.
The case when the orientation is from $v$ to $u$ can be treated in similar manner.
\end{itemize}

Next, we show that $\cA\models \bigwedge_{i=1}^{k} \forall x \exists^{S_i} y\ \beta_i(x,y) \wedge x \neq y$.
To this end, let $u\in A$.
Let $\pi\in \Pi$ and $g\in \cG$ such that $u \in V_{\pi,g}$,
i.e., $u$ behaves according to $g$ in the graph $G$.
Since $g\in \cG$, the function $g$ is a good function.
By the construction of the graph $G$,
for every $i\in [k]$, the number of elements $y\neq u$ such that $\beta_i(x,y)$ belongs to the $2$-type of $(u,y)$ is the sum
\begin{align}
\label{eq:sum-in-Si}
\sum_{E \ni \beta_i(x,y)} 
\ 
\sum_{\pi' \in \OneTypes}\ g(\tout,E,\pi') \ + \
\sum_{E \ni \beta_i(y,x)} 
\ 
\sum_{\pi' \in \OneTypes}\ g(\tin,E,\pi')
\end{align}
By the definition of a good function,
for every $i\in [k]$, the sum \eqref{eq:sum-in-Si} is an element in $S_i$.
Therefore, $\cA,x/u \models \exists^{S_i} y \ \beta_i(x,y) \wedge x\neq y$ for every $i\in [k]$.
Since the choice of $u$ is arbitrary, $\cA\models \forall x \exists^{S_i} y \ \beta_i(x,y) \wedge x\neq y$.
 \end{proof}

Thus, we have shown that,
for every $\fotwopres$ sentence $\phi$ in normal form \eqref{eq:snf},
we can effectively construct an existential Presburger sentence $\PREB_\phi^{\infty}$ such that
$\phi$ has a model iff $\PREB_\phi^{\infty}$ holds in $\cN_{\infty}$.
By Remark~\ref{rem:formula-general-finite-readjustment}, the formula $\PREB_\phi^{\infty}$
can be easily rewritten to another formula $\PREB_\phi$ such that 
$\phi$ has a finite model iff $\PREB_\phi$ holds in $\cN$.
The sentences $\PREB_{\phi}^{\infty}$ and $\PREB_{\phi}$
are as required by Theorem~\ref{theo:main}.

\begin{remark}
Note that in a Type-Behavior Partitioned Graph Vector, information about $2$-types is coded in both the edge relation and in the partition, since the partition
is defined via  behavior functions. Thus there are additional dependencies on sizes for a Type-Behavior Partitioned Graph Vector of a model of $\phi$, 
beyond what will be captured in outdegree constraints. This will not be a problem for us, because these dependencies could  be captured by additional Presburger constraints. We highlight that to solve satisfiability for our logic,
it was not sufficient to know whether a biregular graph problem is solvable: we needed to get a Presburger formula for the possible cardinalities, which we combine with these additional constraints.
\end{remark}

%!TEX root = fo2siam.tex

\section{Proof ideas using a special case for the graph analysis results (Theorems~\ref{thm:main-lemma-bireg} and~\ref{thm:main-lemma-direg})}
\label{sec:1type}

We now discuss the proofs of the main (bi)regular graph theorems. 
These theorems deal with matrices that may contain infinite entries, as well
as matrices that can contain 
%period 1
periodic entries. 
Thus elements of the witness partitions can be forced to be infinite or finite.
\emph{In the body of the paper we restrict to graphs that are finite, and thus in particular ignore the possibility of an infinite entry}.
This suffices to show the claimed bounds on the finite satisfiability problem for our logic.
In the appendix we explain the extensions needed to deal with the infinite case, and thus the general satisfiability problem.

We start in this section by giving proofs only for the \emph{$1$-color case}, without the completeness requirement. While this case does not directly correspond to any formula used in the proof of Theorem~\ref{theo:main} (since matrices~\eqref{eq:Ls} have 2 rows even when there are no binary predicates), this case gives the flavor of the arguments, and will
also be used as the base cases in inductive constructions for the case with arbitrary colors. 
This will be bootstrapped to the multi-color case  in later sections. Note that the $1$-color case \emph{with} the completeness requirement is not very interesting, and also not useful for the general case: completeness states that every node on 
the left must be connected, via the unique edge relation, to every node on the right -- regardless of the matrix. We
can easily write down equations that capture this.

This section is organized as follows.
In Subsection~\ref{subsec:1type-incomplete}
we will focus on the version of Theorem~\ref{thm:main-lemma-bireg} for $1$-color biregular graphs.
 In Subsection~\ref{subsec:1type-digraphs} we
 present a brief explanation of how to modify the proof for regular digraphs (i.e., the case of Theorem~\ref{thm:main-lemma-direg}).
In this section and also in the next, we will be concerned with effectiveness but not complexity.
The complexity of our procedures will be analyzed in Section \ref{sec:complexity}.

\subsection{The case of incomplete $1$-color biregular graphs}
\label{subsec:1type-incomplete}

We will begin by proving a result for $1$-color biregular graphs without the completeness requirement:

\begin{lemma}
\label{lem:1type-incomplete}
For every pair of degree matrices $A\in \bbN_{+p}^{1\times m}$ and $B\in \bbN_{+p}^{1\times n}$,
there exists an (effectively computable) existential 
Presburger formula $\bireg_{A|B}(\vx,\vy)$ such that 
for every size vectors $\vM \in \bbN^m$ and $\vN \in \bbN^n$ 
there is an $A|B$-biregular graph with size $\vM|\vN$
if and only if $\bireg_{A|B}(\vM,\vN)$ holds in $\cN$.
\end{lemma}

Our strategy to prove Lemma~\ref{lem:1type-incomplete} is to divide it into two main cases.
The first case deals with the graphs with ``big enough'' sizes
and the second case with the graphs with ``not big enough'' sizes.
We organize the rest of Section~\ref{subsec:1type-incomplete} as follows.
In Section~\ref{subsubsec:1type-notation} we introduce some notation
and the formal definition of ``big enough'' sizes.
Then, in Section~\ref{subsubsec:1type-incomplete-big-enough}, we present the formula
that captures $A|B$-biregular graphs with ``big enough'' size.
The ``not big enough'' sizes will be handled in Section~\ref{subsubsec:1type-incomplete-not-big-enough}.

\subsubsection{Notation and terminology}
\label{subsubsec:1type-notation}

We will use the following notation.
The term ``vectors'' always refers to row vectors (of finite length).
We use $\va,\vb,\vM,\vN,\ldots$ (possibly indexed) to denote such row vectors.
For a vector $\va$, we denote by $a_j$ the $\myth{j}$ entry in $\va$.
We write $(\va,\vb)$ to denote the row vector obtained by concatenating $\va$ with $\vb$. 
We use $\cdot$ to denote the standard dot product between two vectors.
To avoid being repetitive, when vectors operations such as dot products/additions/subtractions
are performed, it is implicit that the vector lengths are the same.

We now fix notation for degree matrices.
Recall that, in our case, degree matrices are matrices with entries from $\bbNp$, where 
$p$ is a positive integer which is a common non-zero period in all the set $S_i$'s in \eqref{eq:snf}.
Obviously, $1$-row matrices can be viewed as row vectors.
Entries of the form $\prdp{a}$ in a degree matrix are called {\em periodic} entries. 
Otherwise, they are called {\em fixed} entries. 

We write $\offset(\prdp{a})$, for a periodic entry $\prdp{a}$, to denote the offset value $a$.
Note that this is consistent with the definition
of offset of the corresponding linear set from  Section \ref{sec:prelims}.
We define $\offset(a)$ for an integer $a$ to be $a$ itself.
The offset of a vector $\va$, denoted by $\offset(\va)$, is the row vector obtained by replacing
every entry $a_{j}$ with $\offset(a_{j})$.
Of course, if $\va$ does not contain any periodic entry, then
$\offset(\va)$ is $\va$ itself.

In the $1$-color case, matrices $A$ and $B$ for $A|B$-biregular graphs are in fact row vectors. So, we will often write these matrices as $\va$ and $\vb$, respectively.
To differentiate between vectors that are supposed to represent the degrees of vertices in a graph
and vectors that are supposed to represent the sizes of a graph,
we call the former {\em degree vectors} and the latter {\em size vectors}.
We usually write $\va,\vb,\ldots$ to denote degree vectors
and $\vM,\vN,\ldots$ to denote size vectors.
Note that degree vectors have entries from $\bbNp$,
whereas size vectors have entries from $\bbN$.

For degree vectors $\va$ and $\vb$ containing only fixed entries,
we write $\delta(\va,\vb)$ to denote $\max(\va,\vb)$, i.e.,
the maximal element in $\va$ and $\vb$.
When at least one of $\va$ and $\vb$ contain periodic entries,
we define $\delta(\va,\vb)$ as the maximal entry in $(\offset(\va),\offset(\vb),p)$.
For example, if $\va=(3,1)$ and $\vb=(2^{+5},4)$,
then $\delta(\va,\vb)$ is the maximal entry in $(3,1,2,4,5)$, which is $5$.

Let $\va$ be a degree vector.
We let $\nz(\va)$  denote the set of indices $j$ where $a_j$ is not $0$.
We let $\per(\va)$  denote the set of indices $j$ where $a_j$ is a periodic entry.

For a size vector $\vM$ of length $m$,
let $\normt{\vM}$ denote the sum of all the entries in $\vM$, i.e., $\sum_{j=1}^m M_j$,
that is,  the $1$-norm of the column vector $\vM^{\tT}$,
where $\vM^{\tT}$  denotes the transpose of $\vM$.
For a subset $X \subseteq [m]$, we write $\normt{\vM}_X = \sum_{j\in X} M_j$ (which includes the case $\normt{\vM}_{\emptyset}=0)$.
In this section we will only use 
$\normt{\vM}_X$, where $X$ is $\nz(\va)$ or $\per(\va)$,
for some degree vector $\va$.

The intuition is that if $G$ is $\va|\vb$-biregular graph with size $\vM|\vN$,
then the norm $\normt{\vM}_{\nz(\va)}$ denotes the number of vertices on the left of the graph with non-zero degree bound
and 
$\normt{\vM}_{\per(\va)}$ denotes the number of vertices where the corresponding entry of $\va$ is periodic.
The meaning of $\normt{\vN}_{\nz(\vb)}$ and $\normt{\vN}_{\per(\vb)}$ is analogous
with respect to the vertices on the right.

We now introduce the notion of ``big enough'' sizes,
the intuitive meaning of which will become apparent later on.

\begin{definition}
\label{def:1type-big-enough} 
Let $\va$ and $\vb$ be degree vectors and 
let $\vM$ and $\vN$ be size vectors with the same length as $\va$ and $\vb$, respectively.
We say that {\em $\vM|\vN$ is big enough for $\va|\vb$},
if each of the following holds:\footnote{$(\delta(\va,\vb))^2$ is abbreviated $\delta(\va,\vb)^2$.}
\begin{enumerate}[(a)]
\item 
$\max(\normt{\vM}_{\nz(\va)},\normt{\vN}_{\nz(\vb)})\geq 2\delta(\va,\vb)^2+1$;
\item 
$\normt{\vM}_{\per(\va)}=0$ or $\normt{\vM}_{\per(\va)}\geq \delta(\va,\vb)^2+1$;
\item 
$\normt{\vN}_{\per(\vb)}=0$ or $\normt{\vN}_{\per(\vb)}\geq \delta(\va,\vb)^2+1$.
\end{enumerate}
\end{definition}

In the following, to avoid clutter, when we say that $\vM|\vN$ is big enough for $\va|\vb$,
it is implicit that $\vM$ has the same length as $\va$
and $\vN$ has the same length as $\vb$.
As usual, when presenting a Presburger formula,
we will write $\vx,\vy,\vz,\ldots$ (possibly indexed) to denote vectors of variables,
where $x_j$ denotes the $\myth{j}$ entry in $\vx$.
We will also use the notation $\normt{\vx}$ to denote the sum of all the variables in $\vx$,
and similarly use $\normt{\vx}_X$ to denote the sum $\sum_{j\in X} x_j$.

\subsubsection{The formula for the case of big enough sizes}
\label{subsubsec:1type-incomplete-big-enough}

Note that for the conditions (b) and (c) required in the definition of ``big enough,''
there are two possible subcases:
either the norm is $0$ or at least as big as some threshold.
There are altogether 4 possible scenarios and
our formula for big enough sizes will be a disjunction of $4$ formulas,  one for each
scenario.
By symmetry, it suffices to consider the following three of these scenarios for the sizes $\vM|\vN$
of $\va|\vb$-biregular graphs:
\begin{enumerate}[(S1)]
\item $\normt{\vM}_{\per(\va)}=\normt{\vN}_{\per(\vb)}=0$ (i.e., there are only vertices with fixed degree);
\item $\normt{\vM}_{\per(\va)}\neq 0$ and $\normt{\vN}_{\per(\vb)}=0$ (i.e., there are vertices with periodic degrees on exactly one side);
\item $\normt{\vM}_{\per(\va)}\neq 0$ and $\normt{\vN}_{\per(\vb)}\neq 0$ (i.e., there are vertices with periodic degrees on both sides).
\end{enumerate}
The rest of this section is devoted to the formulas for each of the cases above.

\paragraph{The formula and argument for scenario (S1): partition on one side, merge, and swap}
Consider the formula $\psi^{1}_{\va|\vb}(\vx,\vy)$ defined as follows:
\begin{align}
\label{eq:1type-s1}
& \offset(\va)\cdot\vx = \offset(\vb)\cdot\vy
\quad
\wedge \quad
\normt{\vx}_{\per(\va)}=\normt{\vy}_{\per(\vb)}=0.
\end{align}
Note that the last conjunct simply states that the condition of (S1) holds.
The first conjunct is something we will see often, an \emph{edge counting equality}, saying that
 the number of outgoing edges from the left must equal the number of incoming edges on the right.
\begin{lemma}
\label{lem:1type-s1}
For every pair of degree vectors $\va,\vb$ and
for every   $\vM | \vN$ big enough for $\va|\vb$,
there is an $\va|\vb$-biregular graph with size $\vM|\vN$ where (S1) holds
if and only if $\psi^{1}_{\va|\vb}(\vM,\vN)$ holds in $\cN$.
\end{lemma}
\begin{proof}
Let $\va,\vb$  be degree vectors and
$\vM|\vN$ be size vectors big enough for $\va|\vb$.

For the ``only if'' direction,
note that if we have an $\va|\vb$-biregular graph $G$ with size $\vM|\vN$ where (S1) holds,
the total number of edges (by counting the edges adjacent to the vertices on the left) must be $\offset(\va)\cdot\vM$.
Similarly by considering the vertices on the right, the total number of edges must be  $\offset(\vb)\cdot\vN$.
Thus the condition $\offset(\va)\cdot\vM = \offset(\vb)\cdot \vN$ is always a necessary one,
regardless of whether $\vM|\vN$ is big enough.
Since the second conjunct of $\psi^1_{\va|\vb}(\vM,\vN)$ just says that (S1) holds the whole $\psi^1_{\va|\vb}(\vM,\vN)$ also holds.

We now prove the ``if'' direction.
Suppose $\psi^1_{\va|\vb}(\vM,\vN)$ holds.
Since $\normt{\vM}_{\per(\va)}=\normt{\vN}_{\per(\vb)}=0$, 
we may ignore all the periodic entries in $\va$ and $\vb$
and assume that $\va$ and $\vb$ contain only fixed entries, i.e., 
$\va=\offset(\va)$ and $\vb=\offset(\vb)$.

Our proof is similar to the one of~\cite[Lemma 7.2]{KT15}
which shows how to construct an $\va|\vb$-biregular graph with size $\vM|\vN$
for big enough $\vM|\vN$.
For completeness, we repeat the construction here, 
which we will also see later (e.g., in the proof of Lemma~\ref{lem:1type-s2}).

Suppose $\va\cdot\vM=\vb\cdot\vN=K$.
To construct an $\va|\vb$-biregular graph $G$ with size $\vM|\vN$, 
we  ``{\em partition on one side, merge on the other side, and swap.}''
Intuitively, this means that we first construct an $\va|1$-biregular graph $G=(U,V,E)$ with size $\vM|K$, i.e.,
the vertices on the left side are ``partitioned'' correctly to have degrees $\va$
and those on the right side all have degree $1$.
Then, we ``merge'' vertices on the right side so that they have the correct degrees $\vb$.
Since this  merging may produce parallel edges between two vertices,
we perform ``edge swapping'' to get rid of them without changing the degree of each vertex.

The details of the construction are as follows.
Since $\va\cdot\vM=K$, it is straightforward to construct 
an $\va|1$-biregular graph $G=(U,V,E)$ with size $\vM|K$.
Let $\vN=(N_1,\ldots,N_n)$ and $\vb=(b_1,\ldots,b_n)$.
To obtain an $\va|\vb$-biregular graph, we partition $V=V_1\uplus \cdots \uplus V_n$,
where $|V_j|=N_jb_j$ for each $j\in [n]$.
This is possible since $K=\vb\cdot\vN$.
Then, for each $j\in[n]$, we merge every $b_j$ vertices in $V_j$ into $1$ vertex,
thus, making its degree $b_j$.
Such merging yields an ``almost'' $\va|\vb$-biregular graph, 
except that it is possible there are parallel edges between two vertices.
Here big enough comes into play, 
where the condition (a) in Definition~\ref{def:1type-big-enough} is applied,
i.e., 
$\max(\normt{\vM}_{\nz(\va)},\normt{\vN}_{\nz(\vb)})\geq 2\delta(\va,\vb)^2+1$.
We will get rid of the parallel edges one by one.

Suppose in between vertices $u$ and $v$ there are several parallel edges.
There are only at most $\delta(\va,\vb)^2$ edges incident to the neighbors of vertex $u$ (including parallel edges).
The same holds for neighbors of $v$.
Note that there are at least 
$\max(\normt{\vM}_{\nz(\va)},\normt{\vN}_{\nz(\vb)}) \geq 2\delta(\va,\vb)^2+1$ edges in $G$.
So there is an edge $(w,w')$ such that both $w,w'$ are not adjacent to either $u$ or $v$.
To get rid of one parallel edge $(u,v)$ between $u$ and $v$, we replace it and $(w,w')$ by $(u,w')$ and $(w,v)$ (see Fig.~\ref{fig:1type-edge-swap} for an illustration).
We perform such edge swapping until there are no parallel edges.
Furthermore, such edge swapping does not change the degree of the vertices.
\end{proof}

%\noindent
\begin{figure}
\resizebox{1.0 \textwidth}{!}{%
\begin{picture}(350,100)(0,0)

\put(60,75){\circle*{3}}\put(54,78){$u$}
\put(62,75){\line(1,0){57}}
\put(120,75){\circle*{3}}\put(124,78){$v$}

\qbezier(60,75)(90,90)(120,75)
\qbezier(60,75)(90,60)(120,75)

\put(60,15){\circle*{3}}\put(54,18){$w$}
\put(62,15){\line(1,0){57}}
\put(120,15){\circle*{3}}\put(124,18){$w'$}

\put(170,50){$\Rightarrow$}

\put(220,75){\circle*{3}}\put(214,78){$u$}
\put(222,73){\line(1,-1){56}}
\put(280,75){\circle*{3}}\put(284,78){$v$}
\put(222,75){\line(1,0){57}}

\qbezier(220,75)(250,90)(280,75)

\put(220,15){\circle*{3}}\put(214,18){$w$}
\put(222,17){\line(1,1){56}}
\put(280,15){\circle*{3}}\put(284,18){$w'$}
\end{picture}  \label{fig:1type-edge-swap}
}
\caption{Edge swapping used in the proof of Lemma \ref{lem:1type-s1}.
After  swapping there is one less parallel edge between $u$ and $v$, and 
the degrees of all vertices stay the same.}
\end{figure}

\paragraph{The formula and argument for scenario (S2): creating a ``phantom  partition'' for the period, then merging}
Recall that (S2) states that 
``there are vertices with periodic degrees on exactly one side.''
By symmetry, we may assume that the vertices with periodic degrees are on the left.
Let the formula $\psi^{2}_{\va|\vb}(\vx,\vy)$ be defined as
\begin{align}
\label{eq:1type-s2}
& \exists z\ \big( \offset(\va)\cdot\vx  +  pz  = \offset(\vb)\cdot \vy\big)
\ \wedge \ \normt{\vx}_{\per(\va)}\neq 0 \ \wedge \
\normt{\vy}_{\per(\vb)}=0.
\end{align}
equality, with $pz$ representing the total number of edges added by the periodic components over all elements 
on the left hand side.

\begin{lemma}
\label{lem:1type-s2}
For every pair of degree vectors $\va$, $\vb$ and
for every  $\vM| \vN$ big enough for $\va|\vb$,
there is an $\va|\vb$-biregular graph with size $\vM|\vN$ where (S2) holds
if and only if $\psi^{2}_{\va|\vb}(\vM,\vN)$ holds in $\cN$.
\end{lemma}

\begin{proof}
Let $\va,\vb$  be degree vectors and
$\vM|\vN$ be size vectors big enough for $\va|\vb$.
We first prove the ``if'' direction.
Note that if $G=(U,V,E)$ is an $\va|\vb$-biregular graph with size $\vM|\vN$ where (S2) holds, 
then the number of edges $|E|$ should equal the sum of the degrees of the vertices in $U$,
which is $\offset(\va)\cdot\vM+ zp$, for some integer $z\geq 0$.
Since this quantity must equal the sum of the degrees of the vertices in $V$,
which is $\offset(\vb)\cdot \vN$, we conclude that the first conjunct of $\psi^{2}_{\va|\vb}(\vM,\vN)$ holds.
Since (S2) holds by assumption, the second conjuncts also hold.

We now prove the ``only if'' direction.
Assume that $\psi^{2}_{\va|\vb}(\vM,\vN)$ holds.
By \eqref{eq:1type-s2}, we have
$\normt{\vM}_{\per(\va)}\neq 0$
and $\normt{\vN}_{\per(\vb)}=0$.
Clearly we might as well  assume %that $\va$ may contain fixed and periodic entries and 
that $\vb$ contains only fixed entries, i.e., $\offset(\vb)=\vb$.

To construct an $\va|\vb$-biregular graph with size $\vM|\vN$,
we  ``{\em create a phantom partition for the period, then merge}.''
Abusing notation, we denote the value assigned to variable $z$ by $z$ itself.
Suppose $\offset(\va)\cdot \vM + pz = \vb\cdot \vN$.
Since $\vM|\vN$ is big enough for $\va|\vb$,
it follows immediately that $(\vM,z)|\vN$ is big enough for $(\offset(\va),p)|\vb$.
Applying Lemma~\ref{lem:1type-s1},
there is an $(\offset(\va),p)|\vb$-biregular graph with size $(\vM,z)|\vN$. 
That is, we have a graph that satisfies our requirements,
but there is an additional partition class $Z$ on the left of size $z$, where the  degree of elements is $p$.
Let $G=(U,V,E)$ be such a graph, and
let $U=U_0\uplus U_1\uplus Z$,
where $U_0$ is the set of vertices whose degrees are from the fixed entries in $\va$ and 
$U_1$ is the set of vertices whose degrees satisfy the periodic entries in $\va$: in fact, they will initially satisfy these using just the offset.
Note that $|U_1|=\norm{\vM^{\tT}}_{\per(\va)}$ and $|Z|=z$.

We will construct an $\va|\vb$-biregular graph with size $\vM|\vN$.
The idea is to merge the vertices in $Z$ with vertices in $U_1$.
Let $z_0\in Z$.
The number of vertices in $U_1$ reachable from $z_0$ in distance $2$ is at most $\delta(\va,\vb)^2$.
Because $\vM|\vN$ is big enough for $\va|\vb$,
 $|U_1|=\norm{\vM^\tT}_{\per(\va)}\geq\delta(\va,\vb)^2+1$.
Thus, there is a vertex $u\in U_1$ not reachable from $z_0$ in distance $2$; that is, $u$ does not share adjacent vertices
with $z_0$.
We merge $z_0$ and $u$ into one vertex.
See Figure~\ref{fig:1type-s2} for an illustration.
Since the degree of $z_0$ is $p$, 
the merging increases the degree of $u$ by $p$, which does not break our requirement.
We perform this merging for each vertex in $Z$.

\begin{figure}
\begin{center}

\begin{tikzpicture}
\draw (0,0.8) ellipse (0.4cm and .6cm);
\node  at (-.6,1.4) {$U_0$}; 

\draw (0,-1.3) ellipse (.5cm and 1.3cm);
\node  at (-.6,-.2) {$U_1$}; 

\node[circle,fill=blue,inner sep=0pt,minimum size=3pt,label=above:{\small $u$}] (a) at (0,-.6) {};

\draw[gray!40] (0,-2.3) -- (9,-2.9);
\draw[gray!40] (0,-0.9) -- node[color=black,above,sloped] {\scriptsize the neighbors of the vertices in $V'$}(9,-2.1);

\draw[gray!40,fill=gray!40] (0,-1.6) ellipse (.25cm and .7cm);
\node at (0.05,-1.6) {\small $U'$};

\draw (0,-3.7) ellipse (.5cm and .8cm);
\node  at (-.6,-3) {$Z$}; 

\draw[gray!40] (0,-3.7) -- (9,-2.1);
\draw[gray!40] (0,-3.7) -- node[color=black,below,sloped] {\scriptsize the neighbors of $z_0$} (9,-2.9);
\draw[gray!40,fill=gray!40] (9,-2.5) ellipse (.25cm and .4cm);

\node[circle,fill=blue,inner sep=0pt,minimum size=3pt,label=below:{\small $z_0$}] (a) at (0,-3.7) {};

\draw (9,-2) ellipse (.7cm and 2cm);
\node at (9,-2.5) {\small $V'$};
\node  at (10,-.6) {$V$};

\end{tikzpicture}

\end{center}
\label{fig:1type-s2}
\caption{Illustration of the choice of the vertices $z_0\in Z$ and $u\in U_1$.
The set $V'$ is the set of the neighbors of $z_0$.
The set $U'$ is the set of the neighbors of the vertices in $V'$ in set $U_1$,
i.e., the set of vertices reachable from $z_0$ in distance $2$.
Since $|U_1|\geq \delta(\va,\vb)^2+1$ and $|U'|\leq \delta(\va,\vb)^2$,
there is a vertex $u \in U_1 - U'$.
We merge $z_0$ and $u$ into one vertex.}
\end{figure}

Note that the constructed graph $G$ is $\va|\vb$-biregular,
where $\va$ contains periodic entries and $\vb$ contains only fixed entries.
Thus, (S2) holds in $G$.
\end{proof}

\paragraph{The formula and argument for scenario (S3): move a multiple of the period entries to one side}
Recall that (S3) states that ``there are vertices with (finite) periodic degrees on both sides.''
Consider the formula $\psi^{3}_{\va|\vb}(\vx,\vy)$ that is defined as follows:
\begin{align}
\label{eq:1type-s3}
& 
\exists z_1,z_2 \big(\offset(\va)\cdot\vx  +  pz_1  = \offset(\vb)\cdot \vy  +  pz_2 \big)
\wedge 
\normt{\vx}_{\per(\va)}\neq 0  \wedge  \normt{\vy}_{\per(\vb)}\neq 0.
\end{align}

\begin{lemma}
\label{lem:1type-s3}
For each  degree vectors $\va,\vb$ and
for each size vectors $\vM|\vN$ big enough for $\va|\vb$, 
there is an $\va|\vb$-biregular graph with size $\vM|\vN$ where (S3) holds
if and only if $\psi^{3}_{\va|\vb}(\vM,\vN)$ holds in $\cN$.
\end{lemma}
\begin{proof}
As before, the ``only if'' direction is straightforward, 
so we focus on the ``if'' direction.
Suppose $\psi^{3}_{\va|\vb}(\vM,\vN)$ holds.
If there are witnesses $z_1$ and $z_2$ such that $z_1\geq z_2$, we can rewrite the first conjunct as the following:
\[
\exists z_1,z_2 \big(\offset(\va)\cdot\vx  +  p(z_1-z_2)   = \offset(\vb)\cdot \vy  \big).
\]
That is, we ``\emph{move the multiple of period $p$ to one side},'' i.e., to the left side. Let
$\vb'$ denote the vector formed by taking offsets of $\vb$. Thus by definition, $\normt{\vN}_{\per(\vb')}= 0$.
After replacing $z_1-z_2$ with $z$, 
%while  using the facts that $\normt{\vM}_{\per(\va)}\neq 0$ and $\normt{\vN}_{\per(\offset(\vb))}= 0$,  
we can apply  Lemma~\ref{lem:1type-s2}, corresponding to Scenario (S2), to 
 $\va$ and $\vb'$.
Applying this tells us that
there is a $\va|\offset(\vb)$-biregular graph with size $\vM|\vN$.
This graph, of course, is also $\va|\vb$-biregular.

Note that in this case, i.e., when $z_1\geq z_2$, we are arguing, using the prior characterization and algebra,
that when the condition holds we can construct a graph where the degrees on the right hand side are exactly $\offset(\vb)$; that is,
we do not need to take advantage of the ability to have a non-trivial period.
The proof for the case $z_1\leq z_2$ is analogous.
\end{proof}

To wrap up Subsection~\ref{subsubsec:1type-incomplete-big-enough},
we define  the formula $\psi_{\va|\vb}(\vx,\vy)$ as follows:
%\egor{I wrote the formula below explicitly, because there are 4 cases, not 3}
\begin{align}
\label{eq:1type-big-enough}
& %\bigvee_{1 \leq i  \leq 3} \psi_i(\vx,\vy)
\psi^1_{\va|\vb}(\vx,\vy) \lor \psi^2_{\va|\vb}(\vx,\vy) \lor \psi^2_{\vb|\va}(\vy,\vx) \lor \psi^3_{\va|\vb}(\vx,\vy),
\end{align}
where each formula $\psi_{\va|\vb}^i(\vx,\vy)$ handles one of the scenarios described above.
Combining  Lemmas~\ref{lem:1type-s1}--\ref{lem:1type-s3},
$\psi_{\va|\vb}(\vx,\vy)$ captures precisely all the  big enough 
sizes $\vM|\vN$ of an  $\va|\vb$-biregular graph.
This is stated formally as Lemma~\ref{lem:1type-big-enough}.

\begin{lemma}
\label{lem:1type-big-enough}
For each degree vectors $\va,\vb$ and
for each size vectors $\vM,\vN$ big enough for $\va|\vb$,
there is an $\va|\vb$-biregular graph with size $\vM|\vN$
if and only if $\psi_{\va|\vb}(\vM,\vN)$ holds in $\cN$.
\end{lemma}

\subsubsection{The formula for the case of not big enough sizes: fixed size encoding}
\label{subsubsec:1type-incomplete-not-big-enough}

Subsection~\ref{subsubsec:1type-incomplete-big-enough} gives a formula that
captures the existence of $1$-color biregular graphs for big enough sizes.
We now turn to sizes that are  not big enough---that is, 
when one of the conditions (a)--(c) is violated.
When condition (a) is violated, we have restricted the total size of the graph, and
thus we can write a formula that simply enumerates all possible valid sizes.

We will first consider the case when (b) is violated,
while (a) and (c) hold. 
If condition (b) is violated, the value of $\normt{\vM}_{\per(\va)}$ is some $r$ between $1$ and 
$\delta(\va,\vb)^2$ and
it suffices to show that, for each fixed $r$ between $1$ and $\delta(\va,\vb)^2$,
we can find a formula that works for this $r$.
The idea is that a fixed number of vertices in a graph can be ``encoded'' as formulas.
We will refer to this technique as {\em fixed size encoding} in the remainder
of the paper.

We will define a formula covering the case where each of the following holds:
\begin{itemize}
\item
$\normt{\vM}_{\nz(\va)} - \normt{\vM}_{\per(\va)} \geq 2\delta(\va,\vb)^2+1$.
\item 
$\normt{\vM}_{\per(\va)}=r$, for some fixed $r$ between $1$ and $\delta(\va,\vb)^2$.
\item 
$\normt{\vN}_{\per(\vb)}=0$ or $\geq \delta(\va,\vb)^2+1$.
\end{itemize}
Note that the first bullet item is a slightly stronger requirement than the one
required in the definition of big enough size.
However, this does not effect the applicability to the case where (b) is violated and both (a) and (c) hold.
If (b) is violated, i.e., $\normt{\vM}_{\per(\va)}=r$ where $1\leq r \leq \delta(\va,\vb)^2$
and if $\normt{\vM}_{\nz(\va)} - \normt{\vM}_{\per(\va)} \leq 2\delta(\va,\vb)^2$,
then $\normt{\vM}_{\nz(\va)} \leq 3\delta(\va,\vb)^2$,
which means that the number of edges is fixed and all possible sizes of $A|B$-biregular graphs
can be simply enumerated.

The formula is defined inductively on $r$, with the base case $r=0$.
Note that when $r=0$, $\normt{\vM}_{\per(\va)}=0$,
which means (b) is no longer violated and it falls under the ``big enough'' case.
We now give the inductive construction.
Let $\va$ and $\vb$ be degree vectors.
For an integer $r\geq 0$, define the formula
$\phi^{r}_{\va|\vb}(\vx,\vy)$ as follows:
\begin{itemize}
\item
when $r=0$, let
\begin{align*}
\phi^{0}_{\va|\vb}(\vx,\vy)
& :=\
\normt{\vx}_{\per(\va)} = 0 \ \wedge \ \psi_{\va|\vb}(\vx,\vy),
\end{align*}
where $\psi_{\va|\vb}(\vx,\vy)$ is defined in \eqref{eq:1type-big-enough};
\item
when $r\geq 1$, let 
\begin{align} \label{eq:1type-not-big-enough}
&
\phi^{r}_{\va|\vb}(\vx,\vy) \ := \
\exists s \exists \vz_0 \exists \vz_1 
\bigvee_{i\in \per(\va)}
\left(
\begin{array}{l}
x_i \neq 0 \ \wedge \
\vz_0+\vz_1 =\vy 
\\
\wedge \ \normt{\vz_1}_{\nz(\vb)} = \offset(a_i) + ps
\\ 
\wedge \ \phi^{r-1}_{\va|(\vb,\vb-\vone)}(\vx-\textbf{e}_i,\vz_0,\vz_1)
\end{array}
\right),
\end{align}
where the lengths of $\vz_0$ and $\vz_1$ are the same as $\vy$, $\textbf{e}_i$ is the unit vector where the $\myth{i}$ component is $1$, and the subtraction $\vb-\vone$ of degree vectors is the usual element-wise subtraction except the cases $\prdp{b}-1 = \prdp{(b-1)}$ for $b > 0$, $\prdp{0}-1 = \prdp{(p-1)}$ and $0-1=0$.
\end{itemize}

\begin{lemma}
\label{lem:1type-not-big-enough}
For every pair of degree vectors $\va,\vb$,
for every pair of size vectors $\vM,\vN$ , and each integer $r\geq 0$
 such that
\begin{itemize}
\item 
$\normt{\vM}_{\nz(\va)} \geq 2\delta(\va,\vb)^2+1+r$, 

\item 
$\normt{\vM}_{\per(\va)}=r$,
\item 
$\normt{\vN}_{\per(\vb)}\geq \delta(\va,\vb)^2+1$,
\end{itemize}
there is an $\va|\vb$-biregular graph
with size $\vM|\vN$ if and only if
$\phi^{r}_{\va|\vb}(\vM,\vN)$ holds in $\cN$.
\end{lemma}
\begin{proof}
Let $\va,\vb$ be degree vectors and
let $\vM,\vN$ be size vectors that satisfy the hypothesis.
The proof is by induction on $r$. The base case, as in the formulas, is
$r=0$, and is straightforward by the characterization of ``big enough''.

Now assume the claim holds inductively for $r-1 \geq 0$, and consider $r$.
We begin with the ``only if'' direction, which provides the intuition for these formulas.
Suppose $G=(U,V,E)$ is an $\va|\vb$-biregular graph with size $\vM|\vN$
that satisfies all the items listed above.
Let $U=U_1\uplus\cdots\uplus U_m$ and $V=V_1\uplus\cdots\uplus V_n$
be witness partitions.

Since $r\neq 0$ and $\normt{\vM}_{\per(\va)}=r$,
there is $i \in \per(\va)$ such that $U_{i}\neq \emptyset$.
Choose $u\in U_i$.
Based on this $u$,  for each $j \in [n]$ we define  $Z_{j}$ to
be the set of vertices in $V_{j}$ adjacent to $u$.
Figure \ref{fig:notbigenough} illustrates the situation. 

If we omit the vertex $u$ and all its adjacent edges,
then, for each  $j\in [n]$, every vertex in $Z_{j}$ has degree $b_j-1$.
Note here that, for each $j$ where $Z_j \neq \emptyset$,  $b_j > 0$ since $u$ is adjacent to the vertices in $Z_j\subseteq V_j$. 
Thus, we are left with an $\va|(\vb,\vb-\vone)$-biregular graph with size 
$(\vM-\textbf{e}_i)|(\vK_{0},\vK_{1})$,
where $\vK_{0}=  (|V_{1}|-|Z_{1}|,\ldots,|V_{n}|-|Z_{n}|)$ and
$\vK_{1}\ = \ (|Z_{1}|,\ldots,|Z_{n}|)$.
Also note that
\begin{itemize}
\item 
$\normt{(\vM-\textbf{e}_i)}_{\nz(\va)} = \normt{\vM}_{\nz(\va)} -1 \geq 2\delta(\va,\vb)^2+1+(r-1)$.

The equality comes from the fact that $i\in \per(\va)$, hence $i\in \nz(\va)$
and $\normt{\textbf{e}_i}_{\nz(\va)}=1$, which implies the equality.
The inequality comes from the assumption that 
$\normt{\vM}_{\nz(\va)} \geq 2\delta(\va,\vb)^2+1+r$.
\item 
$\normt{(\vM-\textbf{e}_i)}_{\per(\va)}=r-1$.

This comes from the assumptions that $\normt{\vM}_{\per(\va)}=r$ and $i \in \per(\va)$.
\item 
$\normt{(\vK_0,\vK_1)}_{\per(\vb,\vb-\vone)} = \normt{\vN}_{\per(\vb)} \geq \delta(\va,\vb)^2+1$.

The equality comes from the fact that $\vN = \vK_{0}+\vK_1$,
while periodic entries in $\vb$ stay periodic in $\vb-\vone$.
The inequality is the assumption that 
$\normt{\vN}_{\per(\vb)} \geq \delta(\va,\vb)^2+1$.
\end{itemize}
The items above tell us that $(\vM-\textbf{e}_i)|(\vK_0,\vK_1)$ satisfies 
the hypothesis of the lemma (w.r.t. to the degree vectors $\va|(\vb,\vb-\vone)$).
Thus, we can apply the induction hypothesis and obtain that 
$\phi^{r-1}_{\va|(\vb,\vb-\vone)}(\vM-\textbf{e}_i,\vK_{0},\vK_{1})$ holds.
Moreover, since $i\in \per(\va)$, the degree of $u$ is $\offset(a_i) + ps$ for some $s$
and hence $\normt{\vK_1} = \offset(a_i) + ps$. 
Since, by construction, every vertex in each $Z_j$ is adjacent to $u$,
 $|Z_j| = 0$ whenever $b_j = 0$. So, $\normt{\vK_1} = \offset(a_i) + ps$ implies that $\normt{\vK_1}_{\nz(\vb)} = \offset(a_i) + ps$, and therefore $\phi^r_{\va|\vb}(\vM,\vN)$ holds, with the witnessing
$\vz_0$ and $\vz_1$ being  $\vK_0$ and $\vK_1$, respectively.

\begin{figure}
\begin{center}

\begin{tikzpicture}
\draw (0,0) ellipse (.4cm and .7cm);
\node  at (-.5,.7) {$U_1$}; 

\draw[red,dashed,dash pattern=on \pgflinewidth off 5pt,line width=0.4mm] (0,-.9) -- (0,-1.6);

\draw (0,-2.5) ellipse (.4cm and .7cm);
\node  at (-.5,-1.8) {$U_i$}; 

\draw[red,dashed,dash pattern=on \pgflinewidth off 5pt,line width=0.4mm] (0,-3.4) -- (0,-4.1);

\draw (0,-5) ellipse (.4cm and .7cm);
\node  at (-.5,-4.3) {$U_m$};

\draw[gray!40] (0,-2.6) -- (8,-0.8);
\draw[gray!40] (0,-2.6) -- node[color=black,above,sloped] {\scriptsize vertices in $V_1$ adjacent to $u$} (8,0);
\draw[gray!40,fill=gray!40] (8,-.4) ellipse (.2cm and .4cm);
\node at (8,-.4) {\scriptsize $Z_1$};

\draw[gray!40] (0,-2.6) -- (8,-3.8);
\draw[gray!40] (0,-2.6) -- node[color=black,below,sloped] {\scriptsize vertices in $V_n$ adjacent to $u$} (8,-4.6);
\draw[gray!40,fill=gray!40] (8,-4.2) ellipse (.2cm and .4cm);
\node at (8,-4.2) {\scriptsize $Z_n$};

\node[circle,fill=blue,inner sep=0pt,minimum size=3pt,label=above:{\small $u$}] (a) at (0,-2.6) {};

\draw (8,-.2) ellipse (.5cm and 1cm);
\node  at (8.5,.7) {$V_1$}; 

\draw[red,dashed,dash pattern=on \pgflinewidth off 5pt,line width=0.4mm] (8,-1.4) -- (8,-2.8);

\draw (8,-4) ellipse (.5cm and 1cm);
\node  at (8.5,-3.1) {$V_n$};

\end{tikzpicture}

\end{center}
\caption{Illustration of why the formula for the ``not big enough'' case is a necessary condition.}  \label{fig:notbigenough}
\end{figure}

For the ``if'' direction,
suppose $\phi^{r}_{\va|\vb}(\vM,\vN)$ holds.
Then we can fix some $s,\vz_{0},\vz_{1}$, and  $i \in \per(\va)$ such that:
\begin{enumerate}[(a)]
\item 
$M_{i}\neq 0$,
\item
$\offset(a_i) + ps = \normt{\vz_1}_{\nz(\vb)}$,
\item
$\vz_{0}+\vz_{1} = \vN$,
\item
$\phi^{r-1}_{\va|(\vb,\vb-\vone)}(\vM-\textbf{e}_i,\vz_0,\vz_1)$ holds.
\end{enumerate}
We prove from this that a biregular graph of the appropriate size exists.

By similar reasoning as in the previous case (i.e., the ``only if'' case),
the following holds:
\begin{itemize}
\item 
$\normt{(\vM-\textbf{e}_i)}_{\nz(\va)} = \normt{\vM}_{\nz(\va)} -1 \geq 2\delta(\va,\vb)^2+1+(r-1)$,
\item 
$\normt{(\vM-\textbf{e}_i)}_{\per(\va)}=r-1$,
\item
$\normt{(\vz_0,\vz_1)}_{\per(\vb,\vb-\vone)} = \normt{\vN}_{\per(\vb)} \geq \delta(\va,\vb)^2+1$.
\end{itemize}
That is, $(\vM-\textbf{e}_i)|(\vz_0,\vz_1)$ satisfies the hypothesis of the lemma 
(w.r.t. to the degree vectors $\va|(\vb,\vb-\vone)$).
Thus, we can apply the induction hypothesis
and obtain an $\va|(\vb,\vb-\vone)$-biregular graph $G=(U,V,E)$ 
with size $(\vM-\textbf{e}_i)|(\vz_0,\vz_1)$.
Let $U=U_1\uplus\cdots\uplus U_m$ and 
$V=V_{0,1}\uplus\cdots\uplus V_{0,n}\uplus V_{1,1}\uplus\cdots\uplus V_{1,n}$ be the witness partitions.
Note that the degrees of the vertices in $V_{1,1}\uplus\cdots\uplus V_{1,n}$ are $\vb-\vone$.

Let $u$ be a fresh vertex.
We construct an $\va|\vb$-biregular graph 
$G'=(U\cup\{u\},V,E')$, by connecting $u$ with every vertex in $\bigcup_{j\in \nz(\vb)} V_{1,j}$.
This makes the degree of the vertices in $V_{1,1}\uplus\cdots\uplus V_{1,n}$ become $\vb$.
The formula states that $\normt{\vz_1}_{\nz(\vb)}=\offset(a_i)+ps$;
thus, the degree of $u$ is $\offset(a_i)+ps$, which satisfies the requirement for a vertex to be in $U_i$.
Moreover, $\vz_0+\vz_1 =\vN$. 
Thus, the graph $G'$ has size $\vM|\vN$. As witness partition for $G'$  we use the $U_j$ on the left, while  on the right  each $V_{0,j} \cup V_{1,j}$ becomes
 a single partition element.
\end{proof}

The case where (a) and (b) hold, but (c) is violated is handled symmetrically.

Next, we consider the case when (a) holds, but both (b) and (c) are violated.
The treatment is similar to the previous case.
We will  define  a formula for the case where all of the following holds.
\begin{itemize}
\item
$\normt{\vM}_{\nz(\va)} - \normt{\vM}_{\per(\va)} \geq 2\delta(\va,\vb)^2+1$.
\item
$\normt{\vN}_{\nz(\vb)} - \normt{\vN}_{\per(\vb)} \geq 2\delta(\va,\vb)^2+1$.
\item 
$\normt{\vM}_{\per(\va)}=r_1$, for some fixed $r_1$ between $0$ and $\delta(\va,\vb)^2$.
\item 
$\normt{\vN}_{\per(\vb)}=r_2$, for some fixed $r_2$ between $0$ and $\delta(\va,\vb)^2$.
\end{itemize}

The formula is defined inductively on $r_2$ with the base case $r_2=0$.
Note that when $r_2=0$, $\normt{\vN}_{\per(\vb)}=0$,
which means (c) is no longer violated and it falls under the previous case.
Define the formula $\phi^{r_1,r_2}_{\va|\vb}(\vx,\vy)$ as follows:
\begin{itemize}
\item
when $r_2=0$, let
\begin{align*}
\phi^{r_1,0}_{\va|\vb}(\vx,\vy)
& :=\
\normt{\vy}_{\per(\vb)} = 0 \ \wedge \ \phi^{r_1}_{\va|\vb}(\vx,\vy),
\end{align*}
where $\phi^{r_1}_{\va|\vb}(\vx,\vy)$ is defined in the previous case;
\item
when $r_2\geq 1$, let 
\begin{align} \label{eq:1type-not-big-enough-b-c-not-hold}
&
\phi^{r_1,r_2}_{\va|\vb}(\vx,\vy) \ := \
\exists s \exists \vz_0 \exists \vz_1 
\bigvee_{i\in \per(\vb)}
\left(
\begin{array}{l}
y_i \neq 0 \ \wedge \
\vz_0+\vz_1 =\vx 
\\
\wedge \ \normt{\vz_1}_{\nz(\va)} = \offset(b_i) + ps
\\ 
\wedge \ \phi^{r_1,r_2-1}_{(\va,\va-\vone)|\vb}(\vz_0,\vz_1,\vy-\textbf{e}_i)
\end{array}
\right),
\end{align}
Here the lengths of $\vz_0$ and $\vz_1$ are the same as $\vx$, 
$\textbf{e}_i$ is the unit vector where the $\myth{i}$ component is $1$, and 
the subtraction $\va-\vone$ of degree vectors is the same as in the earlier case.
\end{itemize}
Note that the formula $\phi^{r_1,r_2}_{\va|\vb}(\vx,\vy)$ is defined  as in the previous case,
but  the roles of $\va,\vx$ and $\vb,\vy$ are reversed and the base case is now the formula $\phi^{r_1,0}(\vx,\vy)$.

\begin{lemma}
\label{lem:1type-not-big-enough-b-c-not-hold}
For every pair of degree vectors $\va,\vb$,
for every pair of size vectors $\vM,\vN$,  and each integer $r_1,r_2\geq 0$ such that
\begin{itemize}
\item 
$\normt{\vM}_{\nz(\va)} \geq 2\delta(\va,\vb)^2+1+r_1$, 
\item 
$\normt{\vN}_{\nz(\vb)} \geq 2\delta(\va,\vb)^2+1+r_2$, 
\item 
$\normt{\vM}_{\per(\va)}=r_1$,
\item 
$\normt{\vN}_{\per(\vb)}=r_2$,
\end{itemize}
there is an $\va|\vb$-biregular graph
with size $\vM|\vN$ if and only if
$\phi^{r_1,r_2}_{\va|\vb}(\vM,\vN)$ holds in $\cN$.
\end{lemma}

The proof of Lemma~\ref{lem:1type-not-big-enough-b-c-not-hold} is similar to Lemma~\ref{lem:1type-not-big-enough},
hence is omitted.

We have completed the case of fixed $r_1,r_2$. As mentioned above,
this suffices to give the entire not big enough case, via enumerating solutions 
for each of the finitely many possible values of $r_1,r_2$.

To wrap up this section, we define the formula
$\bireg_{\va|\vb}(\vx,\vy)$ required in Lemma~\ref{lem:1type-incomplete}
to characterize solutions in the $1$-color case without the completeness requirements 
\begin{align*}
\bireg_{\va|\vb}(\vx,\vy) & := 
\psi_{\va|\vb}(\vx,\vy) \ \vee\ \phi_{\va|\vb}(\vx,\vy)  
\ \vee \
\bigvee^{\delta(\va,\vb)^2}_{r = 1} \left(\phi^r_{\va|\vb}(\vx,\vy) \lor \phi^r_{\vb|\va}(\vy,\vx)\right)
\\
&\qquad \vee \ 
\bigvee_{0\leq r_1,r_2 \leq \delta(\va,\vb)^2} \phi^{r_1,r_2}_{\va|\vb}(\vx,\vy)
\end{align*}
where $\psi_{\va|\vb}(\vx,\vy)$ is defined in \eqref{eq:1type-big-enough}
to deal with the big enough sizes, $\phi_{\va|\vb}(\vx,\vy)$ is the formula enumerating all valid sizes when condition (a) is violated, 
the formulas in the second last disjunction deal with the not big enough cases when 
exactly one of the conditions (b) or (c) is violated as defined in~\eqref{eq:1type-not-big-enough} and
the formulas in the final disjunctions deal with the other not big enough cases as defined in~\eqref{eq:1type-not-big-enough-b-c-not-hold}.
The correctness of the construction follows immediately from 
Lemmas~\ref{lem:1type-big-enough} and \ref{lem:1type-not-big-enough-b-c-not-hold}.

\subsection{The proof in the $1$-color case for biregular digraphs}
\label{subsec:1type-digraphs}

\begin{figure}
\resizebox{.88\textwidth}{!}{%
\noindent
\begin{picture}(420,100)(-50,0)

\put(100,50){\circle*{3}}\put(98,54){$w$}
\put(60,70){\vector(2,-1){37}}\put(50,50){\vector(1,0){47}}\put(60,30){\vector(2,1){37}}
\put(102,51){\vector(3,1){37}}\put(102,49){\vector(3,-1){37}}

\put(210,50){$\Rightarrow$}

\put(300,50){\circle*{3}}\put(298,54){$v$}
\put(260,70){\vector(2,-1){37}}\put(250,50){\vector(1,0){47}}\put(260,30){\vector(2,1){37}}

\put(330,50){\circle*{3}}\put(324,54){$u$}
\put(332,51){\vector(3,1){37}}\put(332,49){\vector(3,-1){37}}

\end{picture} 
}%resize
\caption{Splitting a vertex $w$ in a digraph $G$ into two vertices $u$ and $v$ in $G'$.
One is adjacent to all the outgoing edges and the other to all the incoming edges.}  \label{fig:splitting}
\end{figure}

Recall that we define digraphs as without any self-loop. 
Thus, a digraph can be viewed as a bipartite graph by splitting every vertex $u$
into two vertices,
where one is adjacent to all the incoming edges, and the other to all the outgoing edges; see Figure~\ref{fig:splitting}.
Thus, 
$\va|\vb$-regular digraphs with size $\vM$ can be characterized as $\va|\vb$-biregular graphs with size $\vM|\vM$
(see~\cite[Section~8]{KT15} for a similar construction when the degrees are fixed).

\myparagraph{Some remarks on the general case vs. the~$1$-color case}
To conclude this section, we stress that although the $1$-color case contains many of the key ideas, the multi-color case requires a finer
analysis to deal with the ``big enough'' case, and also may benefit from a reduction that allows one to restrict
to matrices of a very special form (``simple matrices'').  Note that our definition of multi-color
graph requires the edges of different colors to be disjoint, which imposes additional correlations
between sizes on top of those one would get from considering each color in isolation. We will present these details in the following sections.

%!TEX root = fo2siam.tex

\section{Proof of Theorem~\ref{thm:main-lemma-bireg} for the  case of ``simple'' matrices and without the completeness requirement}
\label{sec:proofsimple}

This section will provide the construction of the Presburger formula for the case
where the matrices $A$ and $B$ may have multiple colors, but are what we call {\em simple} matrices,
defined formally in Definition~\ref{def:simple-matrix}, and where the requirement of being complete is dropped.
Here it is useful to recall that a fixed entry is of the form $a \in \bbN$
and a periodic entry is of the form $\prdp{a}$.

\begin{definition}
\label{def:simple-matrix}
A degree matrix $A$ is simple  if every row consists of
either only periodic entries or only fixed entries.
\end{definition}

That is, for every fixed edge color, either each partition is constrained
using fixed degree constraints on each vertex, or each partition is only
``loosely constrained'' with a periodic constraint on each each vertex.
We devote this section to the proof of the following lemma,
which only deals with finite graphs.
The extension to general graphs can be found in Appendix~\ref{app:sec:proofsimple}.

\begin{lemma}
\label{lem:simple-bireg}
For each pair of simple matrices $A\in \bbNp^{t\times m}$ and $B\in \bbNp^{t\times n}$,
there exists an (effectively computable) existential 
Presburger formula $\bireg_{A|B}(\vx,\vy)$ such that, 
for every size vectors $\vM \in \bbN^m$ and $\vN \times \bbN^n$,
there is an $A|B$-biregular graph with size $\vM|\vN$
if and only if $\bireg_{A|B}(\vM,\vN)$ holds in $\cN$.
\end{lemma}

This section is organized as follows.
We introduce the proper notation in Subsection~\ref{subsec:proof-notations}.
In the setting with multiple colors we also need to introduce ``big enough'' sizes and ``extra big enough'' sizes.
The ``big enough'' sizes are defined only for the matrices $A|B$ whose entries are all fixed,
whereas the ``extra big enough'' sizes are defined for the matrices $A|B$ whose entries can be fixed and periodic.
As in the $1$-color case, the formula $\bireg_{A|B}(\vx,\vy)$ is divided into three cases:
\begin{enumerate}[(1)]
\item
For big enough sizes and when the degree matrices contain only fixed entries, dealt with in
Section~\ref{subsec:simple-big-enough-fixed-entries}.
\item 
For extra big enough sizes and when the degree matrices may contain fixed and periodic entries, in
Section~\ref{subsec:simple-bireg-extra-big-enough}.
\item 
For not big/extra big enough sizes, in
Section~\ref{subsec:simple-not-big-enough}.
\end{enumerate}
Lemma~\ref{lem:simple-bireg} can then be proven by  combining
these cases, as we show in Subsection~\ref{subsec:proof-simple-bireg}.

Briefly, the formula for case~(1) is the same as the one in~\cite[Theorem~7.4]{KT15}.
However, the proof we give here is more straightforward and simpler.
The formula for case~(2) is a generalization of the formula for case~(1)
and we will use techniques such as ``{\em creating a phantom partition for the period, then merging}'' (Lemma~\ref{lem:1type-s2});
``{\em move a multiple of the period entries to one side}'' (Lemma~\ref{lem:1type-s3}) and ``{\em edge swapping}'' (Lemma~\ref{lem:1type-s1}).
As with the $1$-color case, the purpose of extra big enough sizes is to enable us to perform these techniques
without violating the requirement of $A|B$-biregularity.
Finally, the formula for not extra big enough sizes is a straightforward generalization
of the ``{\em fixed size encoding}'' presented in Subsection~\ref{subsubsec:1type-incomplete-not-big-enough}.

\subsection{Notation and terminology}
\label{subsec:proof-notations}

As before, the term ``vectors'' means row vectors
and we use $\vx,\vy,\vz$ (possibly indexed) to denote vectors of variables,
and $\vM,\vN$ to denote size vectors.

Since we are now transitioning to general multi-color graphs, we will use matrix notation,
where matrices are primarily used to describe the degrees of vertices.
We will often call the matrices  {\em degree matrices}.
We use  $\cdot$ to denote matrix multiplication.
When we perform matrix multiplication,
we always assume that the sizes of the operands are appropriate.
We write $I_t$ to denote the identity matrix with size $t\times t$.

The transpose of a matrix $A$ is denoted by $A^\tT$.
The entry in row $i$ and column $j$ is $A_{i,j}$.
We write $A_{i,*}$ and $A_{*,j}$ to denote the $i^{th}$ row and $j^{th}$ column of $A$, respectively.
The numbering of the rows and columns of a matrix starts from $1$.
So for a matrix $A$ with $t$ rows and $m$ columns,
the rows are numbered from $1$ to $t$
and the columns from $1$ to $m$.

As before, we call an entry $A_{i,j}$ a {\em fixed} entry, if it is some $a \in \bbN$.
Otherwise, it is called a {\em periodic} entry,
i.e., an entry of the form $\prdp{a}$.
The {\em offset} of $A$, denoted by $\offset(A)$, is the matrix obtained by replacing
every entry $A_{i,j}$ with $\offset(A_{i,j})$.
Of course, if $A$ does not contain any periodic entry, then
$\offset(A)$ is $A$ itself.

For a matrix $A$ (with $t$ rows and $m$ columns) that contains only fixed entries,
its norm is defined as $\norm{A}  =  \max_{j\in [m]}\quad \sum_{i=1}^{t} A_{i,j}$.
This is the standard $1$-norm.
Of course, a vector $\va$ (of fixed entries) can be viewed as a $1$ row matrix.
Thus, for $\va=(a_1,\ldots,a_m)$, its norm is $\norm{\va}=\max(a_1,\ldots,a_m)$ and 
the norm of its transpose is $\norm{\va^\tT} = \sum_{i=1}^m a_i$.
For matrices $A$ and $B$ that contain only fixed entries, 
$\delta(A,B)$ denotes $\max(\norm{A},\norm{B})$.
If they contain periodic entries,
$\delta(A,B)$ denotes $\max(\norm{\offset(A)},\norm{\offset(B)},p)$.
Note that $\delta(A,B)$ is actually the generalization of the 
$\delta(\va,\vb)$
introduced in Section~\ref{subsubsec:1type-notation} for the
$1$-color case.

If $A$ and $B$ are matrices with the same number of columns,
$\begin{pmatrix} A \\ B \end{pmatrix}$ denotes the matrix
where the first sequence of rows  are $A$ and the next sequence of rows are $B$.
Likewise, if $A$ and $B$ have the same number of rows,
$(A,B)$ is
the matrix where the first sequence of columns are $A$ and the next sequence of columns are $B$.

For degree matrices $A$ and $B$ (with entries from $\bbNp$ and the same number of rows),
and for every size vectors $\vM$ and $\vN$,
we say that {\em $\vM|\vN$ is appropriate for $A|B$}, if
the length of $\vM$ is the same as the number of columns of $A$ and
the length of $\vN$ is the same as the number of columns of $B$.
Since we will only use degree matrices $A$ and $B$ to describe $A|B$-biregular graphs (or $A|B$-regular digraphs),
in the rest of the paper, whenever we use the notation $A|B$,
we implicitly assume that $A$ and $B$ have the same number of rows.
Moreover, unless indicated otherwise, entries in degree matrices always come from $\bbNp$.

Next, we generalize the notion of ``big enough'' in Section~\ref{sec:1type}.
The distinction between ``big enough''  and ``not big enough'' size vectors used for the $1$-color case
in Section~\ref{sec:1type} will need to be refined.

\begin{definition}
\label{def:slightly-big-enough}
Let $A$ and $B$ be degree matrices with $t$ rows whose entries are all fixed entries, i.e., from $\bbN$.
For size vectors $\vM$ and $\vN$ where $\vM|\vN$ is appropriate for $A|B$, $\vM|\vN$ is {\em big enough} for $A|B$,
if the following holds, for every $i\in [t]$:
\begin{enumerate}[(a)]
\item $\max(\normt{\vM}_{\nz(A_{i,*})},\normt{\vN}_{\nz(B_{i,*})})\ \geq\ 2\delta(A,B)^2+1$.
\end{enumerate}
\end{definition}

\begin{definition}
\label{def:big-enough}
Let $A$ and $B$ be simple degree matrices with $t$ rows.
Let $\vM$ and $\vN$ be size vectors where $\vM|\vN$ is appropriate for $A|B$.
We say that {\em $\vM|\vN$ is extra big enough for $A|B$}, if each of the following holds, for every $i\in [t]$: 
\begin{enumerate}[(a)]
\item 
$\max(\normt{\vM}_{\nz(A_{i,*})},\normt{\vN}_{\nz(B_{i,*})})\geq 8t^2\delta(A,B)^4+1$,
\item 
$\normt{\vM}_{\per(A_{i,*})}=0$ or $\normt{\vM}_{\per(A_{i,*})}\geq \delta(A,B)^2+1$,
\item 
$\normt{\vN}_{\per(B_{i,*})}=0$ or $\normt{\vN}_{\per(B_{i,*})}\geq \delta(A,B)^2+1$.
\end{enumerate}
\end{definition}

Note that since $A$ is a simple matrix,
for each color $i \in [t]$, either $\per(A_{i,*})=\emptyset$ or $\per(A_{i,*})=[m]$,
where $m$ is the number of columns in $A$.
The first case is equivalent to $\normt{\vM}_{\per(A_{i,*})}=\normt{\vM}=0$ in condition (b),
while the second case is equivalent to $\normt{\vM}_{\per(A_{i,*})}=\normt{\vM}\geq \delta(A,B)^2+1$.
The same property also holds for matrix $B$ and  size vector $\vN$.
Thus, conditions (a)--(c) can be equivalently restated as:
\begin{itemize}
\item
$\max(\normt{\vM}_{\nz(A_{i,*})},\normt{\vN}_{\nz(B_{i,*})})\geq 8t^2\delta(A,B)^4+1$, for every $i\in [t]$.
\item 
If $A$ contains periodic entries, then $\normt{\vM}\geq \delta(A,B)^2+1$,
\item 
If $B$ contains periodic entries, then $\normt{\vN}\geq \delta(A,B)^2+1$.
\end{itemize}
This is the version we will use in arguments below. 
The formulation in Definition~\ref{def:big-enough} was presented only to  highlight
the generalization from the $1$-color case in Definition~\ref{def:1type-big-enough}.

When we say $\vM|\vN$ is big/extra big enough for $A|B$,
we implicitly assume that $\vM|\vN$ is appropriate for $A|B$.

\begin{remark}
\label{rem:big-enough}
Some basic observations:
\begin{itemize}
\item
The notion of ``big enough'' is defined just on matrices $A|B$
which contain only fixed entries.
\item
Definition~\ref{def:slightly-big-enough} is a direct generalization of Definition~\ref{def:1type-big-enough} for the case without periodic degrees,
where $\vM|\vN$ is big enough for $A|B$, if $\vM|\vN$ is big enough for every color,
i.e., $\vM|\vN$ is big enough for degree vector $A_{i,*}|B_{i,*}$ 
(in the sense of Definition~\ref{def:1type-big-enough}),
for every row $i$.
\item
In the notion of ``extra big enough,'' in Definition~\ref{def:big-enough}
condition (a) requires that
$\max(\normt{\vM}_{\nz(A_{i,*})},\normt{\vN}_{\nz(B_{i,*})})$ is at least
$8t^2\delta(A,B)^4+1$, which is quartic in $\delta(A,B)$,
a jump from quadratic for the $1$-color case.
The reason is purely technical, because in multiple color graphs,
in some cases periodic entries can be reduced to fixed entries
but with quadratic blow-up on the matrix entries.
\item 
Of course, extra big enough is stronger than big enough.
\end{itemize}
\end{remark}

Informally, ``big enough'' entries are those that will allow
the analogous results to Lemma \ref{lem:1type-s1} from the $1$-color case, which concerned
fixed-degree constraints, to go through. ``Extra big enough'' will have some additional margin
over ``big enough'', which will allow us to handle the case of matrices with
periodic entries by reduction to the fixed-entry case.

\subsection{Proof of Lemma~\ref{lem:simple-bireg} for big enough sizes, when the degree matrices contain only fixed entries}
\label{subsec:simple-big-enough-fixed-entries}

Let $A$ and $B$ be degree matrices with $t$ rows that contain only fixed entries.
Consider the  formula:
\begin{align}
\label{eq:simple-s1}
\Psi^{1}_{A|B}(\vx,\vy) & := \ A\cdot\vx^\tT\ =\  B\cdot\vy^\tT.
\end{align}

This formula is a generalization of \eqref{eq:1type-s1} to the case of multiple color graphs for $\va$ and $\vb$ without periodic entries.

\begin{lemma}
\label{lem:simple-s1}
For every pair of degree matrices $A, B$ that contain only fixed entries
and for every pair of size vectors $\vM,\vN$ such that $\vM|\vN$ is big enough for $A|B$,
there is an $A|B$-biregular graph with size $\vM|\vN$ if and only if $\Psi^{1}_{A|B}(\vM,\vN)$ holds.
\end{lemma}
\begin{proof}
Let $A$ and $B$ be degree matrices with $t$ rows,  containing only fixed entries.
Let $\vM|\vN$ be big enough for $A|B$.

We argue for the ``only if'' direction.
Let $G=(U,V,E_1,\ldots,E_t)$ be an $A|B$-biregular graph with size $\vM|\vN$.
The equality, as in the analogous $1$-color case, comes from the  ``edge counting equality,''
i.e., both $A\cdot\vM^\tT$ and $B\cdot\vN^\tT$
simply ``count'' the number of edges in each color, i.e.,
$A \cdot \vM^\tT \ = \ (|E_1|,\ldots,|E_t|)^\tT\ = \ B\cdot\vN^\tT$.
Thus, $\Psi^{1}_{A|B}(\vM,\vN)$ holds.

We now show the ``if'' direction.
Suppose $\Psi^{1}_{A|B}(\vM,\vN)$ holds, i.e., $A\cdot\vM^\tT = B\cdot\vN^\tT$.
We will show that there is an $A|B$-biregular graph with size $\vM|\vN$.

The proof is by induction on $t$.
The base case $t=1$ has been shown in Lemma~\ref{lem:1type-s1}.
For the induction hypothesis, we assume the lemma holds when the number of colors is less than $t$.

Let $A'$ and $B'$ be the degree matrices obtained by omitting the last row in $A$ and $B$, respectively.
Since  $\vM|\vN$  is big enough for $A|B$, we infer that $\vM|\vN$ is big enough for $A'|B'$.
Applying the induction hypothesis, there is an $A'|B'$-biregular graph 
$G'=(U',V',E_1,\ldots,E_{t-1})$ with size $\vM|\vN$.

Similarly, since $\vM|\vN$ is big enough for $A|B$,
it is big enough for $A_{t,*}|B_{t,*}$ (in the sense of Definition~\ref{def:1type-big-enough}).
Recall that $A_{t,*}$ and $B_{t,*}$ are the last rows of $A$ and $B$.
Applying the induction hypothesis,
there is an $A_{t,*}|B_{t,*}$-biregular graph $G''=(U'',V'',E_t)$ with size $\vM|\vN$.
Since $G'$ and $G''$ have the same size, we can assume that $U''=U'$ and $V''=V'$.

To obtain the desired $A|B$-biregular graph,
we first merge the two graphs, obtaining a single graph $G=(U,V,E_1,\ldots,E_t)$.
Such a graph $G$ is ``almost'' $A|B$-biregular, except that
it is possible we have an edge $(u,v)$ which is in $E_1\cup\cdots\cup E_{t-1}$
as well as in $E_t$.
Here we will make use of the ``edge swapping'' technique adapted from Lemma~\ref{lem:1type-s1}.

Recall that $\delta(A,B)=\max(\norm{A},\norm{B})$.
Thus, there are only at most $\delta(A,B)^2$ edges 
incident to the neighbors (via any of $E_1,\ldots,E_t$-edges) of vertex $u$.
The same holds  for neighbors of $v$.
Since $\vM|\vN$ is big enough for $A|B$,
there are at least $\max(\normt{\vM}_{\nz(A_{t,*})},\normt{\vN}_{\nz(B_{t,*})}) \geq 2\delta(A,B)^2+1$ $E_t$-edges in $G$.
So there is an $E_t$-edge $(w,w')$ such that both $w,w'$ are not adjacent (via any of $E_1,\ldots,E_t$-edges) to either $u$ or $v$.
We can perform edge swapping where we omit the edges $(u,v),(w,w')$ from $E_t$,
but add $(u,w'),(w,v)$ into $E_t$.
See the illustration in Figure \ref{fig:simpleswap}.
This edge swapping does not effect the degree of any of the vertices $u,v,w,w'$.
%This completes the proof of Lemma~\ref{lem:simple-s1}. the box literally means this, 
\end{proof}

\noindent
\begin{figure}
\resizebox{1.0\textwidth}{!}{%
\begin{picture}(420,100)(-20,0)

\put(60,65){\circle*{3}}\put(54,78){$u$}

\put(62,65){\line(1,0){57}}
\put(120,65){\circle*{3}}\put(124,78){$v$}

\put(87,55){$E_t$}

\qbezier(60,65)(90,80)(120,65)

\put(87,80){$E_i$}

\put(60,15){\circle*{3}}\put(54,18){$w$}
\put(62,15){\line(1,0){57}}
\put(120,15){\circle*{3}}\put(124,18){$w'$}

\put(87,20){$E_t$}

\put(210,50){$\Rightarrow$}

\put(280,75){\circle*{3}}\put(274,78){$u$}
\put(282,73){\line(1,-1){56}} 
\put(340,75){\circle*{3}}\put(344,78){$v$}

\put(283,51){$E_t$}

\qbezier(280,75)(310,90)(340,75) 

\put(307,90){$E_i$}

\put(280,15){\circle*{3}}\put(274,18){$w$}
\put(282,17){\line(1,1){56}}
\put(340,15){\circle*{3}}\put(344,18){$w'$}

\put(325,51){$E_t$}

\end{picture}  \label{fig:simpleswap}
}
\caption{Swapping used in the case of fixed entries, where $i < t$ 
(proof of Lemma \ref{lem:simple-s1}).}
\label{fig:simple-s1-swap}
\end{figure}

\subsection{Proof of Lemma~\ref{lem:simple-bireg} for extra big enough sizes}
\label{subsec:simple-bireg-extra-big-enough}

In this section we will present the construction of the formula for Lemma~\ref{lem:simple-bireg}
that captures all the extra big enough sizes.
For illustration purposes, 
we start with Subsection~\ref{subsubsec:simple-example} where we consider a special case when the degree matrices $A$ and $B$ contain only 1 column and 2 rows,
the proof of which already contains all the essential ideas required for the proof this case.
Then, in Subsection~\ref{subsubsec:simple-bireg-extra-big-enough},
we present the general formula for the extra big enough sizes for Lemma~\ref{lem:simple-bireg}.

\subsubsection{A special case to illustrate the main ideas}
\label{subsubsec:simple-example}

We consider  the two-color case, i.e., $t=2$,
and  the degree matrices
$A_0 = \begin{pmatrix} a_1 \\ \prd{a_2}{p} \end{pmatrix}$ and $B_0 = \begin{pmatrix} \prd{b_1}{p} \\ b_2 \end{pmatrix}$,
where $a_1,a_2,b_1,b_2$ are all non zero.
Both $A_0$ and $B_0$ have only one column -- that is, only a single partition, whose
size will be the size of one side of the bipartite graph.
So it is trivial that
every row contains either only fixed entries or only periodic entries.
Hence both are simple matrices.

We will now present the formula $\psi_0(x,y)$ that captures all possible extra big enough sizes $M|N$ of $A_0|B_0$-biregular graphs:
\begin{align*}
\psi_0(x,y) & := \
\exists z_1\exists z_2 \ a_1x=b_1y+pz_1 \ \wedge \ a_2x+pz_2 = b_2y.
\end{align*}

Equivalently, we can write $\psi_0(x,y)$ in matrix form:
\begin{align*}
\psi_0(x,y) & := \
\exists z_1\exists z_2 \quad
C
\begin{pmatrix}
x
\\
z_2
\end{pmatrix}
\ = \
D\begin{pmatrix}
y
\\
z_1
\end{pmatrix},
\end{align*}
where $C=\begin{pmatrix}
a_1 & 0
\\
a_2 & p
\end{pmatrix}$ and $D=\begin{pmatrix}
b_1 & p
\\
b_2 & 0
\end{pmatrix}$.
Note that $C$ and $D$ contain only fixed entries.
The following lemma will be useful.

\begin{lemma}
\label{lem:example-extra-big-to-big}
For every pair of integers $M,N \geq 0$,
if $M|N$ is extra big enough for $A_0|B_0$,
then, for every integers $K,L\geq 0$,
$(M,K)|(N,L)$ is big enough for $C|D$.
\end{lemma}
\begin{proof}
It is straightforward from the definitions of extra big enough, big enough, $\delta(A_0,B_0)$ and $\delta(C,D)$.
\end{proof}

We now show that $\psi_0(x,y)$ captures all possible extra big enough sizes $M|N$ of $A_0|B_0$-biregular graphs,
stated formally in Lemma~\ref{lem:simple-example}.
The proof actually contains all the essential ideas required for the proof of Lemma~\ref{lem:simple-bireg}.

\begin{lemma}
\label{lem:simple-example}
For every pair of integers $M,N \geq 0$,
if $M|N$ is extra big enough for $A_0|B_0$, then 
there is an $A_0|B_0$-biregular graph with size $M|N$ if and only if $\psi_0(M,N)$ holds in $\cN$.
\end{lemma}
\begin{proof}
Let $M|N$ be extra big enough for $A_0|B_0$.
Again, the ``only if'' direction follows immediately from the edge counting equality.
So, we focus on the ``if'' direction.
Suppose $\psi_0(M,N)$ holds, i.e., there are $K,L\geq 0$ such that:
\begin{align}
\label{eq:example-bound-a}
a_1 M & \ = \ b_1N + pK,
\\
\label{eq:example-bound-b}
a_2M + pL & \ = \ b_2N. 
\end{align}
Since $M|N$ is extra big enough for $A_0|B_0$,
by Lemma~\ref{lem:example-extra-big-to-big},
$(M,L)|(N,K)$ is big enough for $C|D$.

By Lemma~\ref{lem:simple-s1}, there is a $C|D$-biregular graph $G=(U,V,E_1,E_2)$
of size $(M,L)|(N,K)$.
Let $U=U_0\uplus U_1$ and $V=V_0\uplus V_1$ be the witness partitions,
where $(|U_0|,|U_1|)=(M,L)$ and $(|V_0|,|V_1|)=(N,K)$
and the degree of every vertex is as follows:
\begin{itemize}
\item
every vertex in $U_0$ has $E_1$-degree $a_1$ and $E_2$-degree $a_2$,
\item
every vertex in $U_1$ has $E_1$-degree $0$ and $E_2$-degree $p$,
\item
every vertex in $V_0$ has $E_1$-degree $b_1$ and $E_2$-degree $b_2$,
\item
every vertex in $V_1$ has $E_1$-degree $p$ and $E_2$-degree $0$.
\end{itemize}
We will show how to merge every vertex in the ``phantom'' partition $U_1$ with some vertex in $U_0$
and likewise, merge every vertex in the ``phantom'' partition $V_1$ with some vertex in $V_0$.

We consider two cases.
Case~(a): at least one of $K$ or $L$ is zero;
Case~(b): both $K$ and $L$ are not zero.

\subparagraph*{Case~(a): When at least one of $K$ or $L$ is zero}
We may assume that $K=0$, i.e., $V_1=\emptyset$.
Hence we may consider $G=(U,V,E_1,E_2)$ as a $C|\offset(B_0)$-biregular graph of size $(M,L)|N$.
We will use the same merging technique as in Scenario (S2) in Subsection~\ref{subsubsec:1type-incomplete-big-enough}.

Let $w\in U_1$.
The number of vertices in $U_0$ reachable from $w$ in distance $2$ is at most $\delta(A_0,B_0)^2$.
Due to the condition that $(M,L)|N$ is big enough for $C|\offset(B_0)$,
we have $|U_0|=M\geq\delta(A_0,B_0)^2+1$.
Thus, there is a vertex $u\in U_0$ not reachable from $w$ in distance $2$: 
that is, $u$ does not share adjacent vertices with $w$.
We merge $z_0$ and $u$ into one vertex.
Since the $E_1$-degree of $w$ is $0$ and its $E_2$-degree is $p$, 
the merging does not break the $A_0|B_0$-biregularity requirement.
We perform this merging for every vertex in $U_1$
and obtain an $A_0|B_0$-biregular graph of size $M|N$.

\subparagraph*{Case~(b): When both $K$ and $L$ are not zero}
For this case, we first establish that $K\leq \delta(A_0,B_0)^2N$ and $L\leq \delta(A_0,B_0)^2M$,
which will be used to bound the number of vertices in the ``phantom partition'' that are merged with the same vertex in the ``real partition.''

By \eqref{eq:example-bound-a} and~\eqref{eq:example-bound-b}, we have:
\begin{align}
\label{eq:example-bound-c}
0 \ < \ pK\  =\  a_1M - b_1N \ \leq \ a_1M -N,
\\
\label{eq:example-bound-d}
0 \ < \ pL \ = \ b_2N - a_2M \ \leq \ b_2N -M.
\end{align}
Note that \eqref{eq:example-bound-c} implies $N < a_1M$.
Thus, plugging it into \eqref{eq:example-bound-d}, we obtain
$$
pL  \ \leq \ b_2N - M \ \leq \ b_2a_1M - M \ \leq \ b_2a_1M\ \leq \ \delta(A,B)^2 M.
$$
Similarly, \eqref{eq:example-bound-d} implies $M < b_2N$.
Plugging it into \eqref{eq:example-bound-c}, we obtain
$$
pK \ \leq \ a_1 M - N \ \leq \ a_1\cdot b_2N - N \ = \ a_1b_2N \ \leq \ \delta(A,B)^2N.
$$
Hence
\begin{align}
\label{eq:example-useful-bound}
K \leq \delta(A_0,B_0)^2 N/p
\qquad\text{and}\qquad
L\leq \delta(A_0,B_0)^2 M/p.
\end{align}

Now, when we merge every vertex in the ``phantom partition'' with a vertex in the ``real partition'',
the bound $L \leq \delta(A_0,B_0)^2 M/p$ tells us that we can do it in such a way that
every vertex in $U_0$ is merged with at most $\delta(A_0,B_0)^2/p$ vertices in $U_1$.
Likewise, the bound $K \leq \delta(A_0,B_0)^2 N/p$ tells us that 
we can do the merging in such a way that each vertex in $V_0$ is merged with at most $\delta(A_0,B_0)^2/p$ vertices in $V_1$.
After this merging we obtain an ``almost'' $A_0|B_0$-biregular graph $G=(U_0,V_0,E_1,E_2)$ with size $M|N$.
Again the ``almost'' is because it is possible that there are parallel edges between two vertices in $G$. 
The bounds above have controlled the number of parallel edges that we need to worry about.
We again perform the ``edge swapping'' to get rid of the parallel edges without effecting the degree of each vertex.
Note that after the merging the total degree of each vertex increases by $\delta(A_0,B_0)^2$,
since the degree of every vertex in $U_1\cup V_1$ is $p$.
The requirement that $M|N$ is extra big enough ensures that we have enough edges after the merging that we can perform
the needed edge swapping to get rid of $\delta(A_0, B_0)^2$ parallel edges.
\end{proof}

\subsubsection{Proof of Lemma~\ref{lem:simple-bireg} for extra big enough sizes}
\label{subsubsec:simple-bireg-extra-big-enough}

We now give the general construction for extra big enough sizes, extrapolating from
the idea in the prior example.
For simple degree matrices $A$ and $B$ with $t$ rows,
consider the formula $\Psi^{2}_{A|B}(\vx,\vy)$ given by
\begin{align}
\nonumber 
&
\exists z_{1,1}\cdots \exists z_{1,t}\
\exists z_{2,1}\cdots \exists z_{2,t}
\\
\label{eq:simple-s2}
& 
\qquad\quad
\offset(A)\cdot \vx^{\tT} + \begin{pmatrix}\alpha_1 p z_{1,1}\\ \vdots \\ \alpha_t p z_{1,t}\end{pmatrix}
\ = \
\offset(B)\cdot \vy^{\tT} + \begin{pmatrix}\beta_1 p z_{2,1}\\ \vdots \\ \beta_t p z_{2,t}\end{pmatrix},
\end{align}
where $\alpha_i = 1$ if row $i$ in $A$ consists of periodic entries and is $0$ otherwise, 
and similarly $\beta_i = 1$ if row $i$ in $B$ consists of periodic entries and is $0$ otherwise.

This is again an edge counting equality, with the $p$ multiples of $z_{1,i}$ and of $z_{2,i}$ representing additional edges due to the periodic
factors. We can see that
 $\Psi^1_{A|B}(\vx,\vy)$ is a special case of it 
where
all the constants $\alpha_1,\ldots,\alpha_t$,$\beta_1,\ldots,\beta_t$ are zero.
We will show that $\Psi^2_{A|B}(\vx,\vy)$ captures 
all possible extra big enough sizes $\vM|\vN$ of $A|B$-biregular graphs,
as formally stated in Lemma~\ref{lem:simple-s2}.

\begin{lemma}
\label{lem:simple-s2}
For each pair of simple degree matrices $A$, $B$ and
for each pair of size vectors $\vM$, $\vN$ such that $\vM|\vN$ is extra big enough for $A|B$,,
there is a $A|B$-biregular graph with size $\vM|\vN$ 
if and only if
$\Psi_{A|B}^2(\vM,\vN)$ holds in $\cN$.
\end{lemma}
\begin{proof}
Let $A$ and $B$ be simple degree matrices with $t$ rows.
Let $\vM|\vN$ be extra big enough for $A|B$.

The ``only if'' direction is just an edge counting equation for each color.
Suppose there is an $A|B$-biregular graph $G=(U,V,E_1,\ldots,E_t)$ with size $\vM|\vN$.
For each $i \in [t]$,
the number of $E_i$-edges is the sum of $E_i$-degrees of vertices in $U$
which is $\offset(A_{1,*})\cdot \vM + \alpha_i p z_{1,i}$, for some integer $z_{1,i}\geq 0$.
This, of course, must equal the sum of $E_i$-degrees of vertices in $V$, and this
 is $\offset(B_{1,*})\cdot \vN + \beta_i p z_{2,i}$, for some integer $z_{2,i}\geq 0$.
Thus, $\Psi^2_{A|B}(\vM,\vN)$ holds.

We now prove the ``if'' direction.
Suppose $\Psi^2_{A|B}(\vM,\vN)$ holds.
Abusing notation as before, we denote the values assigned to the variables $z_{i,j}$'s
by the variables $z_{i,j}$'s themselves.

We are going to construct an $A|B$-biregular graph with size $\vM|\vN$.
There are two cases -- analogous to Cases~(a) and (b) in Subsection \ref{subsubsec:simple-example}.

(Case 1) $\alpha_iz_{1,i} \geq \beta_i z_{2,i}$,
for every $i\in [t]$.

This case is analogous to scenarios (S2) and (S3) in the $1$-color case,
where we first ``move a multiple of the period entries to one side'' (S3)
and  ``create a phantom partition for the period, then merge'' (S2).
It is also analogous to case~(a) in Lemma~\ref{lem:simple-example}.
First, as in (S3), we move all the multiple of the period entries to one side --
that is, rewrite \eqref{eq:simple-s2} as
\begin{align*}
\offset(A)\cdot \vM^{\tT} + \begin{pmatrix}(\alpha_1 z_{1,1} - \beta_1z_{2,1})p \\ \vdots \\ (\alpha_t z_{1,t}-\beta_tz_{2,t})p\end{pmatrix}
& \ = \
\offset(B)\cdot \vN^{\tT}.
\end{align*}
We further rewrite the left hand side as
\begin{align*}
(\offset(A),pI_t) \cdot \begin{pmatrix} \vM^\tT \\ \alpha_1 z_{1,1} - \beta_1z_{2,1} \\ \vdots \\ \alpha_t z_{1,t}-\beta_tz_{2,t}\end{pmatrix}
& \ =\ 
\offset(B) \cdot\vN^{\tT}.
\end{align*}
Recall that $I_t$ is the identity matrix with size $t\times t$
and that $(\offset(A),pI_t)$ denotes
the matrix where the first sequence of columns are $\offset(A)$ and the next sequence of columns are $pI_t$.
Intuitively, the submatrix $pI_t$ represents $p$ ``phantom'' partitions,
each containing vertices whose $E_i$-degree is $p$ on exactly one color $i$, with the other degrees being $0$.
The vector $(\alpha_1 z_{1,1} - \beta_1z_{2,1},\ldots,\alpha_t z_{1,t}-\beta_tz_{2,t})$ represents the sizes of these phantom partitions.
Note that this is similar to (S2) in the $1$-color case in Lemma~\ref{lem:1type-s2}, except that now we have one phantom partition for each color.

Let $C=(\offset(A), pI_t)$ and $\vK = (\alpha_1z_{1,1}-\beta_1 z_{2,1}, \ldots, \alpha_t z_{1,t}-\beta_t z_{2,t})$.
Since $\vM|\vN$ is extra big enough for $A|B$, $(\vM,\vK)|\vN$ is big enough for $C|\offset(B)$.

Note that $C$ and $\offset(B)$ contain only fixed entries.
By Lemma~\ref{lem:simple-s1},
there is an $(\offset(A), pI_t)| \offset(B)$-biregular graph $G=(U,V,E_1,\ldots,E_t)$
with size $(\vM,\vK)|\vN$.

Let $U=U_1\uplus \cdots \uplus U_m \uplus W_1\uplus \cdots \uplus W_t$ be its witness partition,
-- that is, for every $i\in [t]$:
\begin{itemize}
\item
for every $j\in [m]$, the $E_i$-degree of every vertex in $U_j$ is $\offset(A_{i,j})$,
and $|U_j|=M_j$;
\item
the $E_i$-degree of every vertex in $W_i$ is $p$
and $|W_i|=\alpha_i z_{1,i}-\beta_i z_{2,i}$,
and for every $i'\neq i$, the $E_{i'}$-degree of every vertex in $W_i$ is $0$.
\end{itemize}
Here we actually ``create phantom partitions $W_1,\ldots,W_t$ for the periods.''

Observe that if $W_i\neq \emptyset$,
i.e., $\alpha_i z_{1,i}-\beta_i z_{2,i}\neq 0$, then $\alpha_i \neq 0$.
Since $A$ is a simple matrix, its row $i$ consists of only periodic entries,
hence $\normt{\vM}_{\per(A_{i,*})}=\normt{\vM}$.
Since the sizes are extra big enough, we have $|U_1\uplus\cdots\uplus U_m|=\normt{\vM}\geq \delta(A,B)^2+1$, 
where $\delta(A,B)=\max(\norm{\offset(A)},\norm{\offset(B)},p)$.
For such an $i$, 
we are going to ``merge'' vertices in $W_i$ with vertices in $U_1\uplus\cdots\uplus U_m$ --- analogous to Lemma~\ref{lem:1type-s2}.

Let $w$ be a vertex in $W_i$, where $W_j\neq \emptyset$.
The number of vertices in $G$ reachable by $w$ in distance $2$ (with any edges) is at most $\delta(A,B)^2$.
Since $|U_1\uplus\cdots\uplus U_m| \geq \delta^2+1$,
there is a vertex $u \in U_1\uplus\cdots\uplus U_m$ which is not reachable from $w$ in distance $2$.
We can merge $w$ with $u$.
We perform such merging for every vertex in $W_i$.
Since the $E_i$-degree of every vertex in $W_i$ is $p$,
and the $E_{i'}$-degree of vertices in $W_i$ is $0$, for every $i'\neq i$,
such merging only increases the $E_i$-degree of a vertex in $U$ by $p$.
We continue in this way for every $i$ where $W_i\neq \emptyset$,
resulting in an $A|B$-biregular graph with size $\vM|\vN$.

The case where $\beta_iz_{2,i} \geq \alpha_i z_{1,i}$,
for every $i \in [t]$, 
can be handled symmetrically.

(Case 2)
There are $i,i'\in [t]$ such that $\alpha_iz_{1,i} > \beta_i z_{2,i}$
and $\alpha_{i'}z_{1,i'} < \beta_{i'} z_{2,i'}$.

This case is analogous to case~(b) in Lemma~\ref{lem:simple-example}.
Let $\Gamma_1$ be the set of indexes $i$ such that $\alpha_iz_{1,i} > \beta_i z_{2,i}$
and $\Gamma_2$ be the set of indexes $i$ such that $\alpha_iz_{1,i} < \beta_i z_{2,i}$.
Since $A$ and $B$ are simple matrices, this means:
\begin{itemize}
\item
For every $i\in \Gamma_1$, $\alpha_i \neq 0$, 
i.e., row $i$ in $A$ consists of only periodic entries.
\item
Likewise, for every $i\in \Gamma_2$, $\beta_i \neq 0$,  
i.e., row $i$ in $B$ consists of only periodic entries.
\end{itemize}

First, we can rewrite \eqref{eq:simple-s2}
as
\begin{align*}
\offset(A)\cdot \vM^{\tT} + \begin{pmatrix}p K_1 \\ \vdots \\ pK_t\end{pmatrix}
& \ = \
\offset(B)\cdot \vN^{\tT} + \begin{pmatrix}pL_1 \\ \vdots \\ pL_t\end{pmatrix},
\end{align*}
where each $K_i$ and $L_i$ is defined as:
\begin{align*}
K_i & := 
\left\{
\begin{array}{ll}
\alpha_iz_{1,i} - \beta_i z_{2,i} \quad& \text{if}\ i \in \Gamma_1
\\
0 & \text{if} \ i \notin \Gamma_1,
\end{array}
\right.
\\
L_i & := 
\left\{
\begin{array}{ll}
\beta_iz_{2,i} - \alpha_i z_{1,i} \quad& \text{if}\ i \in \Gamma_2
\\
0 & \text{if} \ i \notin \Gamma_2.
\end{array}
\right.
\end{align*}
We can further rewrite the formula:
\begin{align}
\label{eq:simple-s2-case2}
(\offset(A),pI_t) \cdot \begin{pmatrix} \vM^\tT \\ K_1 \\ \vdots \\ K_t\end{pmatrix}
& \ =\ 
(\offset(B),pI_t) \cdot \begin{pmatrix} \vN^\tT \\ L_1 \\ \vdots \\ L_t\end{pmatrix}.
\end{align}
In the following we let $C = (\offset(A),pI_t)$ and $D= (\offset(B),pI_t)$.
We also let $\vK =(K_1,\ldots,K_t)$ and $\vL =(L_1,\ldots,L_t)$.
Note that since $\vM|\vN$ is extra big enough for $A|B$,
$(\vM,\vK)|(\vN,\vL)$ is big enough for $C|D$.
By Lemma~\ref{lem:simple-s1}, there is $C|D$-biregular graph $G=(U,V,E_1,\ldots,E_t)$
with size $(\vM,\vK)|(\vN,\vL)$.
We let $U = U_{0,1}\uplus\cdots \uplus U_{0,m}\uplus U_{1,1}\uplus \cdots\uplus U_{1,t}$
and $V= V_{0,1}\uplus\cdots \uplus V_{0,n}\uplus V_{1,1}\uplus \cdots\uplus V_{1,t}$
be the witness partition, where
\begin{itemize}
\item
$\vM =(|U_{0,1}|,\ldots,|U_{0,m}|)$ and $\vK=(|U_{1,1}|,\ldots,|U_{1,t}|)$, and
\item 
$\vN =(|V_{0,1}|,\ldots,|V_{0,n}|)$ and $\vL = (|V_{1,1}|,\ldots,|V_{1,t}|)$.
\end{itemize}
The partitions $U_{1,1},\ldots,U_{1,t},V_{1,1}, \ldots,V_{1,t}$ are the ``phantom partitions''
whose vertices are to be merged with the vertices in the ``real partitions'' 
$U_{0,1},\ldots,U_{0,m},V_{0,1},\ldots,V_{0,n}$.

Similar to case~(b) in Lemma~\ref{lem:simple-example},
we can bound the value of each $K_i$ and $L_i$.
For simplicity, we may first assume the following assumptions (a1) and (a2) hold.
\begin{enumerate}[(a1)]
\item[(a1)]
For every $i \in \Gamma_1$, $A_{i,*}$ does not contain a $\prdp{0}$ entry,
or equivalently, $\offset(A_{i,*})$ does not contain a zero entry.
\item[(a2)]
Likewise, for every $i \in \Gamma_2$, $B_{i,*}$ does not contain a $\prdp{0}$ entry.
\end{enumerate}

Note that for every $i\in \Gamma_1$ we have:
\begin{align}
\label{eq:simple-s2-1}
0 < pK_i = p (\alpha_iz_{1,i} - \beta_i z_{2,i}) & = \offset(B_{i,*})\cdot \vN^\tT - \offset(A_{i,*})\cdot \vM^\tT
\\
\nonumber
& \leq \delta(A,B)\norm{\vN^\tT} - \norm{\vM^\tT}.
\end{align}
In the last inequality we use the assumption that $\offset(A_{i,*})$ does not contain a $0$ entry.
Similarly, for every $i\in \Gamma_2$ we have:
\begin{align}
\label{eq:simple-s2-2}
0 < pL_i = p(\beta_{i} z_{2,i} - \alpha_{i}z_{1,i}) & \leq \delta(A,B)\norm{\vM^\tT} - \norm{\vN^\tT}.
\end{align}
From \eqref{eq:simple-s2-2}, we obtain $\norm{\vN^\tT} \leq \delta(A,B)\norm{\vM^\tT}$.
If we plug this into \eqref{eq:simple-s2-1}, we obtain that, for every $i\in \Gamma_1$,
\begin{align}
\label{eq:simple-s2-3}
pK_i & 
\leq \delta(A,B)^2\norm{\vM^\tT} - \norm{\vM^\tT} \leq \delta(A,B)^2\norm{\vM^\tT}.
\end{align}
Symmetrically, for every $i\in \Gamma_2$, we have
\begin{align}
\label{eq:simple-s2-4}
pL_i & 
\leq  \delta(A,B)^2\norm{\vN^\tT}.
\end{align}
Inequalities \eqref{eq:simple-s2-3} state that, for every $i\in \Gamma_1$,
when performing the merging between vertices in the phantom partition $U_{1,i}$ and 
the real partitions $U_{0,1}\uplus\cdots\uplus U_{0,m}$,
we can do in such a way that every vertex in the real partition
is merged with at most $\delta(A,B)^2/p$ vertices in phantom partition.
Likewise, \eqref{eq:simple-s2-3} states similarly for $i\in \Gamma_2$ for the merging between
vertices in the phantom partitions $V_{1,i}$ and 
the real partitions $V_{0,1}\uplus\cdots\uplus V_{0,m}$.

Now we reason as in the illustrative case.
After the merging, we obtain an ``almost'' $A|B$-biregular graph with size $\vM|\vN$.
As in the example, almost is because it is possible that there are parallel edges between two vertices in $G$
and the established bounds above have controlled the number of parallel edges that we need to worry about.
After the merging the total degree of each vertex increases by $t\delta(A_0,B_0)^2$.
We perform the ``edge swapping'' to get rid of the parallel edges without effecting the degree of each vertex.
The requirement that $\vM|\vN$ is extra big enough ensures that we have enough edges to perform the edge swapping.
This completes the proof for case 2 when the assumptions (a1) and (a2) hold.

Now we consider the case when at least one of the assumptions (a1) or (a2) does not hold.
The main idea is to rewrite the $\prdp{0}$ entries in $A$ and $B$ as $\prdp{p}$
in such a way that the bounds in \eqref{eq:simple-s2-3} and~\eqref{eq:simple-s2-4} still hold.

First, we rewrite \eqref{eq:simple-s2-case2}:
\begin{align*}
(\offset(A'),pI_t) \cdot \begin{pmatrix} \vM^\tT \\ K_1' \\ \vdots \\ K_t'\end{pmatrix}
& \ =\ 
(\offset(B'),pI_t) \cdot \begin{pmatrix} \vN^\tT \\ L_1' \\ \vdots \\ L_t'\end{pmatrix},
\end{align*}
where the matrix $A'$ and the integers $K_1',\ldots,K_t'$ are:
\begin{enumerate}[(1)]
\item
For every $i\notin \Gamma_1$,
we let $A'_{i,*}=A_{i,*}$ and $K_i' = K_i$.
\item
For every $i\in \Gamma_1$ such that $A_{i,*}$ does not contain $\prdp{0}$ entries,
the row $A'_{i,*}$ is $A_{i,*}$ and $K_i' = K_i$.
\item 
For every $i\in \Gamma_1$ such that $A_{i,*}$ contains $\prdp{0}$ entries, we let $X=\{j : A_{i,j}=\prdp{0}\}$; moreover,
\begin{enumerate}[(3.a)]
\item 
if $K_i <  \normt{\vM}_X $,
then $A_{i,*}'=A_{i,*}$ and $K_i' = K_i$, and

\item 
if $K_i \geq  \normt{\vM}_X$,
then $A_{i,*}'$ is obtained from $A_{i,*}$ by changing every $\prdp{0}$ entry with $\prdp{p}$
and $K_i' = K_i - \normt{\vM}_X$.
\end{enumerate} 
\end{enumerate}
The matrix $B'$ and the integers $L_1',\ldots,L_t'$ are defined in a similar manner.
\begin{enumerate}[(1)]\setcounter{enumi}{3}
\item
For every $i\notin \Gamma_2$,
we let $B'_{i,*}=B_{i,*}$ and $L_i' = L_i$.
\item
For every $i\in \Gamma_2$ such that $B_{i,*}$ does not contain $\prdp{0}$ entries,
the row $B'_{i,*}$ is $B_{i,*}$ and $L_i' = L_i$.
\item 
For every $i\in \Gamma_2$ such that $B_{i,*}$ contains $\prdp{0}$ entries, we
let $X=\{j : B_{i,j}=\prdp{0}\}$; moreover, 
\begin{enumerate}[(6.a)]
\item 
if $L_i <  \normt{\vN}_X $,
then $B_{i,*}'=B_{i,*}$ and $L_i' = L_i$, and

\item 
if $L_i \geq  \normt{\vN}_X$,
then $B_{i,*}'$ is obtained from $B_{i,*}$ by changing every $\prdp{0}$ entry with $\prdp{p}$
and $L_i' = L_i - \normt{\vN}_X$.
\end{enumerate} 
\end{enumerate}
Note that the only difference between $A$ and $A'$, and between $B$ and $B'$, 
is in (3.b) and (6.b), respectively,
where some $\prdp{0}$ entries are changed into $\prdp{p}$.
Thus, an $A'|B'$-biregular graph is also an $A|B$-biregular graph.

Performing a similar calculation as in \eqref{eq:simple-s2-1}--\eqref{eq:simple-s2-4},
we can show that:
\begin{itemize}
\item
for every $i\in \Gamma_1$, $K_i' \leq \delta(A,B)^2\norm{\vM^\tT}/p$,
\item 
for every $i\in \Gamma_2$, $L_i' \leq \delta(A,B)^2\norm{\vN^\tT}/p$.
\end{itemize}
The construction of an $A'|B'$-biregular graph with size $\vM|\vN$
can be done almost verbatim as above.
\end{proof}

\begin{remark} %~
It  is only in Case~2 in the proof of Lemma~\ref{lem:simple-s2}
that we require the quantity
\[
\max(\normt{\vM}_{\nz(A_{i,*})},\normt{\vN}_{\nz(B_{i,*})}),
\]
which is precisely the number of vertices with non-zero $E_i$-degree in an $A|B$-biregular graph of size $\vM|\vN$,
to be at least quartic in $\delta(A,B)$, and not quadratic as in Section~\ref{sec:1type}.
This is because the total degree of each vertex increases by at most $t\delta(A,B)^2$
after the merging between the vertices in the phantom partitions and real partitions.
Thus, we require the number of edges is at least quartic in $\delta(A,B)$ to ensure
there is enough edges to perform edge swapping to get rid of the parallel edges.
Note also that the restriction of $A$ and $B$ to simple matrices
allows us to merge every vertex in the phantom partition with any vertex in the real partition.
Thus we can perform the merging in such a way that the total degree of each vertex in the real partition increases
by at most $t\delta(A,B)^2$. 
\end{remark}

\subsection{Encoding of not ``big/extra big enough'' components for simple matrices}
\label{subsec:simple-not-big-enough}

Lemma~\ref{lem:simple-s2} gives a formula that captures the existence of 
biregular graphs for extra big enough sizes for simple degree matrices.
We now turn to sizes that are not big/extra big enough. 
Here we will use the same idea of {\em fixed size encoding} as in the $1$-color case.

Note that a ``not extra big enough size'' means that one of the conditions (a)--(c) is violated
and thus some of the entries in the size vectors $\vM,\vN$ are already fixed.
For example, if condition (a) is violated, 
then $\max (\normt{\vM}_{\nz(A_{i,*})},\normt{\vN}_{\nz(B_{i,*})})$ 
is between $1$ and $8t^2\delta(A,B)^4$, for some $i\in [t]$.
So in this case we can fix the values of $\normt{\vM}_{\nz(A_{i,*})}$ and $\normt{\vN}_{\nz(B_{i,*})}$ as some $r_1$ and $r_2$,
where $r_1,r_2$ are in between $1$ and $8t^2\delta(A,B)^4$.
As in the $1$-color case (Lemma~\ref{lem:1type-not-big-enough}),  the idea will be
a fixed number of non-zero degree vertices in a graph can be ``encoded'' as formulas, along the lines of
Subsection~\ref{subsubsec:1type-incomplete-not-big-enough}.

We detail the formula construction for the case where for some color, condition (a) is violated, 
but conditions (b) and (c) hold.
All the other cases can be handled in a similar manner.
We fix degree matrices $A$ and $B$ with $t$ rows,
and let $m$ and $n$ be the number of columns in $A$ and $B$.
For simplicity, we focus on the case where the color where (a) is violated is the $t^{th}$ row.
For integers $r_1,r_2\geq 0$,
we define a formula $\Phi^{r_1,r_2}_{A|B}(\vx,\vy)$ that captures precisely
the sizes $\vM|\vN$ of $A|B$-biregular graph
where $\normt{\vM}_{\nz(A_{t,*})}=r_1$ and $\normt{\vN}_{\nz(B_{t,*})} =r_2$.

The construction is by induction on $r_1+r_2$ and the number of rows in the degree matrices $A$ and $B$.
\begin{itemize}
\item 
When the number of rows in $A,B$ is $1$,
the formula $\Phi^{i,r_1,r_2}_{A|B}(\vx,\vy)$ 
simply enumerates all possible sizes of $A|B$-biregular graphs.

Such an enumeration is possible since the number of vertices with non-zero degree on the left hand side is fixed to $r_1$,
and the number of vertices with non-zero degree on the right hand side is fixed to $r_2$.

\item 
If the $i^{\text{th}}$ row in both $A$ and $B$ contains periodic entries,
the formula\\
 $\Phi^{i,r_1,r_2}_{A|B}(\vx,\vy)$ simply enumerates all possible sizes of $A|B$-biregular graphs
where $r_1$ is the number of vertices on the left hand side
and $r_2$ is the number of vertices on the right hand side.

Here it is useful to recall that $A$ (resp. $B$) is a simple matrix,
hence either the entries in $A$ (resp. $B$) are all fixed entries
or are all periodic entries.

\item 
When $r_1+r_2=0$, we let the formula 
\begin{align*}
 \Phi^{r_1,r_2}_{A|B}(\vx,\vy)
 \ := \ &
\Phi_{\tilde{A}|\tilde{B}}(\vx_0,\vy_0) \ \wedge \ \normt{\vx}_{\nz(A_{t,*})} = 0
\ \wedge \ \normt{\vy}_{\nz(B_{t,*})} = 0
\end{align*}
where $\tilde{A}$ is the matrix $A$ without the $t^{\text{th}}$ row and without the columns in $\nz(A_{t,*})$,
$\tilde{B}$ is the matrix $B$ without the $t^{\text{th}}$ row and without the columns in $\nz(B_{t,*})$,
and $\vx_0$ and $\vy_0$ are the vectors $\vx$ and $\vy$ without the components in $\nz(A_{t,*})$ and $\nz(B_{t,*})$, respectively.

The purpose of the the formula $\Phi_{\tilde{A}|\tilde{B}}(\vx_0,\vy_0)$ is to capture all possible sizes of $\tilde{A}|\tilde{B}$-biregular graphs.
Formally, it is defined as:
$$
\Psi^2_{\tilde{A}|\tilde{B}}(\vx_0,\vy_0)
\ \vee \
\Theta_{\tilde{A}|\tilde{B}}(\vx_0,\vy_0)
$$
where $\Psi^2_{\tilde{A}|\tilde{B}}(\vx_0,\vy_0)$ captures all the extra big enough sizes of $\tilde{A}|\tilde{B}$-biregular graphs
as defined in Subsection~\ref{subsec:simple-bireg-extra-big-enough} and 
$\Theta_{\tilde{A}|\tilde{B}}(\vx_0,\vy_0)$
captures all the not extra big enough sizes of $\tilde{A}|\tilde{B}$-biregular graphs.
Note that the number of rows in $\tilde{A}|\tilde{B}$ is now $t-1$, hence
the formula $\Theta_{\tilde{A}|\tilde{B}}(\vx_0,\vy_0)$ can be defined inductively.
The intuition behind the matrices $\tilde{A}$ and $\tilde{B}$ is that, since $\normt{\vx}_{\nz(A_{t,*})} = r_1=0$
and $\normt{\vy}_{\nz(B_{t,*})} = r_2 = 0$, we can ignore the color $t$, i.e., by removing the $t^{\text{th}}$ row in $A$ and $B$
and all the corresponding columns in $\nz(A_{t,*})$ and $\nz(B_{t,*})$.

\item 
When $r_1+r_2\geq 1$, at least one of $r_1$ or $r_2$ is bigger than or equal to $1$.

When $r_1\geq 1$, we let
\begin{align*}
& \Phi^{r_1,r_2}_{A|B}(\vx,\vy) \ := 
\\ 
& \hspace{1cm} \exists s_1\cdots \exists s_t \ \exists \vz_0 \exists \vz_1 \cdots \exists \vz_t
\\
& \hspace{1.5cm}
\bigvee_{j\in \nz(A_{t,*})}
\left(
\begin{array}{l}
\quad(x_{1,j}\neq 0)\ \wedge \ \vy= \sum_{\ell=0}^t \vz_\ell
\\
\wedge \ \bigwedge_{\ell=1}^t  \normt{\vz_\ell} = \offset(A_{\ell,j})+ \alpha_\ell \cdot p \cdot s_\ell
\\
\wedge \
\Phi^{r_1-1,r_2}_{A|(B,B-J_1,\ldots,B-J_t)}
(\vx-\textbf{e}_j,\vz_0,\vz_1,\ldots,\vz_t)
\end{array}
\right).
\end{align*}
where  each $\alpha_\ell$ is in $\{0,1\}$ with $\alpha_\ell=1$ if and only if $A_{\ell,j}$ is a periodic entry;
each $J_\ell$ is a matrix with size $(t\times m)$ where row $\ell$ consists of all $1$ entries
and all the other rows have only $0$ entries.

When $r_2\geq 1$, the formula can be defined symmetrically with the roles of $A,\vx$ and $B,\vy$ being swapped.
\end{itemize}

The following lemma states the correctness of the formula constructed above.
\begin{lemma}
\label{lem:simple-not-big-enough}
For every pair of simple degree matrices $A,B$ with $t$ rows, 
for every integers $r_1,r_2\geq 0$, for every size vectors $\vM,\vN$,
the formula $\Phi^{r_1,r_2}_{A|B}(\vM,\vN)$ holds in $\cN$
if and only if there is a $A|B$-biregular graph
with size $\vM|\vN$ where $\normt{\vM}_{\nz(A_{t,*})}=r_1$ and $\normt{\vN}_{\nz(B_{t,*})}=r_2$.
\end{lemma}

The proof of Lemma~\ref{lem:simple-not-big-enough} is a straightforward generalization of Lemma~\ref{lem:1type-not-big-enough-b-c-not-hold},
hence we omit it.

The case where (b) or (c) is violated for some color $i\in [t]$ can be treated in a similar manner.
Note that in the case when {\em both} (b) and (c) are violated, i.e., $1\leq \normt{\vM}_{\per(A_{i_1,*})}\leq \delta(A,B)^2$ 
and $1\leq \normt{\vN}_{\per(B_{i_2,*})}\leq \delta(A,B)^2$ for some $i_1,i_2\in [t]$,
the number of vertices is fixed to some $r$ in between $1$ and $2\delta(A,B)^2$, since $\per_{A_{i_1,*}}=[m]$ and $\per(B_{i_2,*})=[n]$
due to $A$ and $B$ being simple matrices.
Thus, in this case all possible sizes of $A|B$-biregular graphs can simply be enumerated.

\begin{remark}
\label{rem:simple-not-big-enough}
The following observations about the formula will be useful in our complexity analysis later on.
By pulling out the disjunction, 
we can rewrite the formula $\Phi^{r_1,r_2}_{A|B}(\vx,\vy)$
as a disjunction $\bigvee_i \varphi_i$ conjoined with $\Phi_{\tilde{A}|\tilde{B}}(\vx_0,\vy_0)$,
where each $\varphi_i$ is a conjunction of $O(t(r_1+r_2))$ (in)equations.
Since $r_1,r_2$ ranges between $1$ and $\max(8t^2\delta(A,B)^4,t\delta(A,B))= 8t^2\delta(A,B)^4$,
each $\varphi_i$ is a conjunction of $O(t^3\delta(A,B)^4)$ (in)equations conjoined with $\Phi_{\tilde{A}|\tilde{B}}(\vx_0,\vy)$.
It is useful to recall that $\tilde{A}|\tilde{B}$ now have one less rows than $A|B$.

By straightforward induction on the number of rows $t$,
we observe that the formula $\Phi^{r_1,r_2}_{A|B}(\vx,\vy)$ can be written as a disjunction 
$\bigvee_i\varphi_i$,
where each $\varphi_i$ is a conjunction of $O(t^4\delta(A,B)^4)$ (in)equations. 
 \end{remark}

\subsection{Proof of Lemma~\ref{lem:simple-bireg}}
\label{subsec:proof-simple-bireg}
To wrap up this section, for simple matrices $A$ and $B$,
we define the formula $\bireg_{A|B}(\vx,\vy)$ required in Lemma~\ref{lem:simple-bireg}
to characterize all the possible sizes of $A|B$-biregular graph, without the completeness requirement:
\begin{align*}
\bireg_{A|B}(\vx,\vy) & := 
\Psi^2_{A|B}(\vx,\vy) \ \vee\
\bigvee_{i\in [\ell]} \Phi_i(\vx,\vy),
\end{align*}
where $\Psi^2_{A|B}(\vx,\vy)$ is defined in \eqref{eq:simple-s2}
to deal with the big enough sizes,
while the disjunction $\bigvee_{i\in [\ell]} \Phi_i(\vx,\vy)$ deals with the not extra big enough sizes
as defined in Subsection~\ref{subsec:simple-not-big-enough}.
Here we assume an enumeration of all the formulas $\Phi_1(\vx,\vy),\ldots,\Phi_{\ell}(\vx,\vy)$ 
that deal with the not extra big enough sizes.
The correctness of the construction follows immediately from 
Lemma~\ref{lem:simple-s2} and~\ref{lem:simple-not-big-enough}.

\begin{remark}
\label{rem:simple-bireg}
Let $t$ be the number of rows in matrices $A$ and $B$
and let $m$ and $n$ be the number of columns in $A$ and $B$, respectively.
By Remark~\ref{rem:simple-not-big-enough}, each $\Phi_i(\vx,\vy)$
is a disjunction of conjunctions of $O(t^4\delta(A,B)^4)$ (in)equations.
Since $\Psi^2_{A|B}(\vx,\vy)$ is a conjunction of $t$ equations,
the formula $\bireg_{A|B}(\vx,\vy)$ can be written as a disjunction $\bigvee_{i} \varphi_i$
where each $\varphi_i$ is a conjunction of $O(t^4\delta(A,B)^4)$ (in)equations.
\end{remark}

%!TEX root = fo2siam.tex

\section{Proof of Theorem~\ref{thm:main-lemma-bireg} for the  case of ``simple'' matrices with the completeness requirement being enforced}
\label{sec:simple-complete}

We will now consider the formula defining possible partition sizes, still restricting
to  simple biregular graphs, but now enforcing
 the completeness restriction.
This will be done
via reduction to the case where the completeness restriction has not been enforced.

We  introduce a further restriction on the matrices that will be useful.
\begin{definition}
\label{def:good-pair}
For a pair of simple matrices $A|B$ (with the same number of rows),
we say that $A|B$ is {\em a good pair} if
there is $i$ such that row $i$ is periodic in both $A$ and $B$.
\end{definition}

\begin{remark}
\label{rem:not-good-pair}
Note that if $A|B$ is not a good pair,
then complete $A|B$-biregular graphs can only have  up to $2\delta(A,B)$ vertices.
Indeed, suppose $G=(U,V,E_1,\ldots,E_t)$ is a complete $A|B$-biregular graph.
Since $A|B$ is not a good pair,
for every $i\in [t]$, the number of edges in $E_i$ is at most $\delta(A,B)|U|$ or $\delta(A,B)|V|$.
Thus, $\sum_{i\in [t]} |E_i|$ is at most $\delta(A,B)(|U|+|V|)$.
On the other hand, the fact that $G$ is complete implies that $\sum_{i\in [t]} |E_i|=|U||V|$
which is strictly bigger than $\delta(A,B)(|U|+|V|)$, 
when $|U|+|V|> 2\delta(A,B)$.
So, when $A|B$ is not a good pair,  to capture all possible sizes of complete $A|B$-biregular graphs,
we simply write a formula that enumerates all possible $\vM|\vN$
where $\normt{\vM}+\normt{\vN}\leq 2\delta(A,B)$.
\end{remark}

So it suffices to define the formula that captures all possible sizes of complete 
(finite) $A|B$-biregular graphs
where $A$ and $B$ are both simple matrices and $A|B$ is a good pair.
Let $\vx=(x_1,\ldots,x_m)$ and $\vy=(y_1,\ldots,y_n)$.
Let $A\in \bbNp^{t\times m}$ and $B\in\bbNp^{t\times n}$ be simple matrices
and $A|B$ is a good pair.
Let $\xi_{A|B}(\vx,\vy)$ be the formula
\begin{align}
\label{eq:simple-complete-xi1}
 \bireg_{A|B}(\vx,\vy) & \quad \wedge {} &
\\ 
\label{eq:simple-complete-xi1-a}  
\bigwedge_{j\in [m]}
x_j\neq 0 \ \to\ \exists z \ \normt{\vy} = \norm{\offset(A_{*,j})}+p z &  \quad \wedge {} &
\\
\label{eq:simple-complete-xi1-b}
 \bigwedge_{j\in [n]}
y_j\neq 0 \ \to \ \exists z \ \normt{\vx} = \norm{\offset(B_{*,j})}+p z. && 
\end{align}
Here $\bireg_{A|B}(\vx, \vy)$ is the formula characterizing the situation without
the completeness requirement.

Intuitively, \eqref{eq:simple-complete-xi1-a} states that
the number of vertices on the right hand side
must equal the total degree of the vertices on the left hand side.
Likewise, \eqref{eq:simple-complete-xi1-b} states that
the number of vertices on the left hand side
must equal the total degree of the vertices on the right hand side.

\begin{lemma} 
\label{lem:completesimplegoodonefinite}
For every good pair of simple matrices $A$ and $B$ such that $A|B$,
for every size vectors $\vM$ and $\vN$,
$\xi_{A|B}(\vM,\vN)$ holds in $\cN$
exactly when  there is a complete $A|B$-biregular graph of size $\vM|\vN$.
\end{lemma}
\begin{proof}
That $\xi_{A|B}(\vM,\vN)$ is a necessary condition
for the existence of complete $A|B$-biregular graph is pretty straightforward.
This follows from the fact that if $G=(U,V,E_1,\ldots,E_t)$ is a complete $A|B$-biregular graph
then the sum of all $E_i$-degrees of every vertex in $U$ must equal $|V|$,
and likewise, the sum of all $E_i$-degrees of every vertex in $V$ must equal  $|U|$.

Now we show that it is also a sufficient condition.
Suppose $\xi_{A|B}(\vM,\vN)$ holds.
Thus, $\bireg_{A|B}(\vM,\vN)$ holds,
and by Lemma~\ref{lem:simple-bireg},
there is a (not necessarily complete) $A|B$-biregular graph 
$G=(U,V,E_1,\ldots,E_t)$ with size $\vM|\vN$.
We will show how to make $G$ complete.

Let $U = U_1\uplus\cdots\uplus U_m$ and $V=V_1\uplus \cdots \uplus V_n$ be the witness partition.
Since $A|B$ is a good pair, there is $i_0$ such that row $i_0$ is periodic in both $A$ and $B$.
Now, for every $(u,v)\notin E_1\cup \cdots \cup E_t$,
we define $(u,v)$ to be in $E_{i_0}$.
Obviously, after adding such $E_{i_0}$-edges, the graph $G$ becomes complete.
We argue that $G$ is still $A|B$-biregular by showing
\begin{enumerate}[(a)]
\item 
for every $j\in[m]$, for every vertex $w\in U_j$, the $E_{i_0}$-degree of $w$ increases by a multiple of $p$;
\item 
for every $j\in[n]$, for every vertex $w\in V_j$, 
the $E_{i_0}$-degree of $w$ increases by a multiple of $p$.
\end{enumerate}
We prove (a), fixing
 $w\in U_j$.
The $E_{i_0}$-degree of $w$ increases by
\begin{align*}
\label{eq:complete-degree-increase}
|V|- \sum_{i\in [t]} \deg_{E_i}(w).
\end{align*}
Note that \eqref{eq:simple-complete-xi1-a} forces $|V|$ to be:
$$
|V| = \prdp{\norm{\offset(A_{*,j})}} = \norm{\offset(A_{*,j})} + (\text{some multiple of $p$}). 
$$
On the other hand, we also have
$$
\sum_{i\in [t]} \deg_{E_i}(w) = \sum_{i\in [t]} A_{i,j} = \norm{\offset(A_{*,j})} + (\text{some multiple of $p$}).
$$
Here it is useful to recall that row $i_0$ in $A$ contains periodic entries, hence the additional term ``some multiple of $p$''.
Thus, the quantity $|V|- \sum_{i\in [t]} \deg_{E_i}(w)$ is a multiple of $p$,
and therefore the $E_{i_0}$-degree of $w$ only increases by a multiple of $p$.
This does not violate the $A|B$-biregularity condition.

Part (b)  can be proven in a similar manner to  \eqref{eq:simple-complete-xi1-b}.
This completes our proof of Lemma~\ref{lem:completesimplegoodonefinite}.
\end{proof}

\begin{remark}
\label{rem:simple-complete}
We will again make some further observations that will be important only for the complexity
analysis, which will be detailed in Section \ref{sec:complexity}.
For each $j\in [m]$, let $a_j=\normt{\offset(*,j)}$.
We first rewrite \eqref{eq:simple-complete-xi1-a} as follows:
$$
\bigvee_{j_1\in [m]} \Bigg(\exists z \ \normt{\vy} = a_{j_1} + pz \ \wedge \ x_{j_1}\neq 0
\ \wedge\ \bigwedge_{j_2\in [m]\ \text{s.t.}\ a_{j_2} \not\equiv a_{j_1} \bmod p} x_{j_2}=0\Bigg),
$$
Indeed, if $x_{j_1},x_{j_2}\neq 0$,
then $\normt{\vy}=\prdp{a_{j_1}}$ and $\normt{\vy}=\prdp{a_{j_2}}$,
which implies $a_{j_1}\equiv a_{j_2} \bmod p$.
Therefore, if $x_{j_1}\neq 0$, then $x_{j_2}=0$ whenever $a_{j_2}\not\equiv a_{j_1}\bmod p$.
We also rewrite \eqref{eq:simple-complete-xi1-b} in a similar manner.

Note that \eqref{eq:simple-complete-xi1-a} yields $O(m)$ equalities,
while the rewriting above transforms it into a disjunction of $O(1)$ (in)equations.\footnote{Here we do not count equations of the form $x=0$
since such variable $x$ can be ignored during the computation,
thus, becomes negligible in the complexity analysis.}
By Remark~\ref{rem:simple-bireg}, the formula $\xi_{A|B}(\vx,\vy)$ is a disjunction of conjunctions of $O(t^4\delta(A,B)^4)$ (in)equations,
where $t$ is the number of rows in matrices $A$ and $B$.
\end{remark}

To wrap up Section~\ref{sec:simple-complete},
we define the formula $\biregc_{A|B}(\vx,\vy)$ for simple matrices $A$ and $B$
as follows:
\begin{align}
\label{eq:finite-simple-complete}
\biregc_{A|B}(\vx,\vy) & :=
\left\{
\begin{array}{ll}
\xi_{A|B}(\vx,\vy) \qquad\quad& \text{if}\ A|B\ \text{is a good pair}
\\
\bigvee_{i} \phi_i(\vx,\vy) & \text{if}\ A|B\ \text{is not a good pair},
\end{array}
\right.
& 
\end{align}
where $\xi_{A|B}(\vx,\vy)$ is defined in \eqref{eq:simple-complete-xi1}--\eqref{eq:simple-complete-xi1-b}
when $A|B$ is a good pair
and the disjunction $\bigvee_{i} \phi_i(\vx,\vy)$ enumerates all possible sizes $\vM|\vN$ when $A|B$ is not a good pair.
Recall that by Remark~\ref{rem:not-good-pair}, when $A|B$ is not a good pair,
complete $A|B$-biregular graphs can only have sizes $\vM|\vN$
where $\normt{\vM}+\normt{\vN}\leq 2\delta(A,B)$.
Since there are only finitely many such sizes, they can be enumerated. 
The correctness of the formula $\biregc_{A|B}(\vx,\vy)$ follows immediately from Lemma~\ref{lem:completesimplegoodonefinite}
and Remark~\ref{rem:not-good-pair}, as stated formally in Lemma~\ref{lem:finite-simple-complete}.

\begin{lemma}
\label{lem:finite-simple-complete}
For every pair of simple matrices $A$ and $B$ and
for every pair of size vectors $\vM$ and $\vN$,
$\biregc_{A|B}(\vM,\vN)$ holds in $\cN$
exactly when  there is a complete $A|B$-biregular graph of size $\vM|\vN$.
\end{lemma}

%!TEX root = fo2siam.tex

\section{Proof of Theorem~\ref{thm:main-lemma-bireg} and~\ref{thm:main-lemma-direg}}
\label{sec:proofnonsimple}

In this section we will present the proof of Theorems~\ref{thm:main-lemma-bireg} and~\ref{thm:main-lemma-direg}.
Recall that Theorem~\ref{thm:main-lemma-bireg} states that
for every arbitrary degree matrices $A$ and $B$,
we can effectively construct a Presburger formula $\biregc_{A|B}(\vx,\vy)$
that captures all possible sizes of complete $A|B$-biregular graphs.
Theorem~\ref{thm:main-lemma-direg} is the analog for the directed graphs.

In Section~\ref{sec:simple-complete} we showed how to construct
Presburger formulas that capture all possible sizes of complete simple $A|B$-biregular graphs,
i.e., where the degree matrices $A$ and $B$ are simple matrices.
In this section we will show how to reduce the non-simple matrices to simple matrices for biregular graphs.
We divide this section into three subsections.
We begin with an example that shows the main idea in Section \ref{subsec:example-non-simple-simple}.
In Section \ref{subsec:general-non-simple-simple} we present the general reduction
from non-simple biregular graphs to simple biregular graphs.
Finally, in Section \ref{subsec:proof-main-lemma-direg}
we deal with the regular digraphs.

\subsection{A special case illustrating the reduction}
\label{subsec:example-non-simple-simple}

Consider the degree matrices $A_0 =(a_1,\prdp{a_2})$ and $B_0=(b_1,\prdp{b_2})$,
where $a_1,a_2,b_1,b_2$ are all non zero integers.
Obviously, they are not simple matrices, since each row contains both fixed and periodic entries.
We will show that every $A_0|B_0$-biregular graph can be
viewed as a collection of four simple biregular graphs,
as stated formally in Theorem~\ref{theo:example-non-simple-simple}.

The main idea is as follows.
Suppose we have $A_0|B_0$-biregular graph $G=(U,V,E)$ with witness partition $U=U_1\uplus U_2$
and $V=V_1\uplus V_2$.
We will decompose the graph into  $4$ induced biparttte subgraphs, each representing
the restriction to one partition on the left and one on the right \footnote{As usual, for a graph $G=(V,E)$ and for a subset $S\subseteq V$,
the notation $G[S]$ denotes the subgraph induced in $G$ by the set $S$.}
We will show below that each such subgraph satisfies a biregularity condition:
\begin{itemize}
\item
The induced subgraph $G[U_1\cup V_1]$
is a $(0,1,\ldots,a_1)|(0,1,\ldots,b_1)$-biregular graph.
\item 
The induced subgraph $G[U_1\cup V_2]$
is a $(a_1,a_1-1,\ldots,0)|(\prdp{0},\prdp{1},\ldots,\prdp{b_2})$-biregular graph.
\item 
The induced subgraph $G[U_2\cup V_1]$
is a $(\prdp{0},\prdp{1},\ldots,\prdp{a_2})|(b_1,b_1-1,\ldots,0)$-biregular graph.
\item 
The induced subgraph $G[U_2\cup V_2]$
is a $(\prdp{a_2},\prdp{(a_2-1)},\ldots,\prdp{0})|(\prdp{b_2},\prdp{(b_2-1)},\ldots,\prdp{0})$-biregular graph.
\end{itemize}
Note that the degree matrices involved are all simple matrices.
For example, the degree matrix $(0,1,\ldots,a_1)$, which has only one row, is simple,
since every row contains only fixed entries.
As another example, the degree matrix $(\prdp{0},\prdp{1},\ldots,\prdp{a_2})$ is also simple,
since every row contains only periodic entries.

We call the decomposition of $G$ into the subgraphs $G[U_1\cup V_1]$, $G[U_1\cup V_2]$, $G[U_2\cup V_1]$ and $G[U_2\cup V_2]$
{\em the degree-based decomposition of $G$}.
We reduce a characterization of sizes of $A_0|B_0$-biregular graphs to characterization of the sizes
of the components of the decomposition.

\begin{theorem}
\label{theo:example-non-simple-simple}
For every pair  $M_1,M_2 \in \bbN^2$ and every pair $N_1,N_2 \in\bbN^2$,
the following are equivalent.
\begin{enumerate}[(a)]
\item
There is an $A_0|B_0$-biregular graph with size $(M_1,M_2)|(N_1,N_2)$.
\item 
There exist size vectors $\vK_1\in \bbN^{a_1+1}$, $\vK_2\in\bbN^{a_2+1}$, $\vL_1\in \bbN^{b_1+1}$, $\vL_2\in \bbN^{b_2+1}$
such that $\normt{\vK_1}=M_1$, $\normt{\vK_2}=M_2$,
$\normt{\vL_1}=N_1$ and $\normt{\vL_2}=N_2$ and 
\begin{itemize}
\item 
a $(0,1,\ldots,a_1)|(0,1,\ldots,b_1)$-biregular graph with size $\vK_1|\vL_1$;
\item 
a $(a_1,a_1-1,\ldots,0)|(\prdp{0},\prdp{1},\ldots,\prdp{b_2})$-biregular graph with size $\vK_1|\vL_2$;
\item 
a $(\prdp{0},\prdp{1},\ldots,\prdp{a_2})|(b_1,b_1-1,\ldots,0)$-biregular graph with size  $\vK_2|\vL_1$;
\item 
a $(\prdp{a_2},\prdp{(a_2-1)},\ldots,\prdp{0})|(\prdp{b_2},\prdp{(b_2-1)},\ldots,\prdp{0})$-biregular graph with size $\vK_2|\vL_2$,
\end{itemize}
\end{enumerate}
\end{theorem}

Note that there can be several vectors $\vK_1 \ldots$ satisfying the conditions on norms in the theorem. But the condition
on sizes can clearly be described in Presburger arithmetic.
So the theorem suffices to give a reduction to generating Presburger invariants for a vector of biregular graph problems
involving simple matrices.

The proof of Theorem~\ref{theo:example-non-simple-simple} is conceptually simple,
but rather technical.
We divide it into two lemmas:
Lemma~\ref{lem:example-non-simple-simple-only-if} which implies the ``only if'' direction
and Lemma~\ref{lem:example-non-simple-simple-if} 
which deals with the ``if'' direction.
Below we let $[0,k]$ denote the set $\{0,1,\ldots,k\}$
for an integer $k\geq 0$.

\begin{figure}
\begin{center}

\begin{tikzpicture}
\draw (-0,0) ellipse (1.4cm and 2cm);
\node  at (-1.3,1.8) {\large $U_1$}; 

\draw[red] (-1.1,1.2) -- (1.1,1.2);
\node at (0,1.5) {\scriptsize $U_{1,0}$};
\draw[red] (-1.35,0.4) -- (1.35,0.4);
\node[red] at (0,.85) {$\vdots$};
\node at (-.5,-.1) {\scriptsize $U_{1,j}$};
\draw[red] (-1.35,-0.4) -- (1.35,-0.4);
\node[red] at (0,-.7) {$\vdots$};
\draw[red] (-1.1,-1.2) -- (1.1,-1.2);
\node at (0,-1.6) {\scriptsize $U_{1,a_1}$};

\draw[gray!40] (.5,0) --  node[color=black,above,sloped] {\scriptsize neighbors of $u$ in $V_1$} (6.5,.7);
\draw[gray!40] (.5,0) -- (6.5,-.7);
\draw[gray!40,fill=gray!40] (6.5,0) ellipse (.3cm and .7cm);

\node at (7,0) {\scriptsize $j$ vertices};

\draw[gray!40] (.5,0) -- (6.5,-3.3);
\draw[gray!40] (.5,0) -- node[color=black,below,sloped] {\scriptsize neighbors of $u$ in $V_2$} (6.45,-4.7);
\draw[gray!40,fill=gray!40] (6.5,-4) ellipse (.3cm and .7cm);
\node at (7,-4) {\scriptsize $a_1-j$ vertices};

\node[circle,fill=blue,inner sep=0pt,minimum size=3pt,label=below:{\scriptsize $u$}] (a) at (0.5,0) {};

\draw (0,-4) ellipse (1.4cm and 1.8cm);
\node  at (-1.5,-2.8) {\large $U_2$}; 
\draw (7,0) ellipse (1.4cm and 2cm);
\node  at (8.2,1.8) {\large $V_1$};
\draw (7,-4) ellipse (1.4cm and 1.8cm);
\node  at (8.4,-2.8) {\large $V_2$};
\end{tikzpicture}

\end{center}
\label{fig:example-non-simple-simple}
\caption{An illustration for the proof of Lemma~\ref{lem:example-non-simple-simple-only-if}.
$G=(U,V,E)$ is an $A_0|B_0$-biregular graph with $U=U_1\uplus U_2$ and $V=V_1\uplus V_2$  the
witness partition.
We partition $U_1 = U_{1,0}\uplus \cdots \uplus U_{1,a_1}$
where for each $j\in [0,a_1]$,
each vertex $u \in U_{1,j}$ has $j$ neighbors in $V_1$ and $(a_1-j)$ neighbors in $V_2$.
Similarly we partition $U_2 = U_{2,0}\uplus \cdots \uplus U_{2,a_2}$,
$V_1 = V_{1,0}\uplus \cdots \uplus V_{1,b_1}$
and 
$V_2 = V_{2,0}\uplus \cdots \uplus V_{2,b_2}$.}
\end{figure}

\begin{lemma}
\label{lem:example-non-simple-simple-only-if}
For every $A_0|B_0$-biregular graph $G=(U,V,E)$ with witness partition
$U=U_1\uplus U_2$ and $V=V_1\uplus V_2$,
there exist size vectors $\vK_1\in \bbN^{a_1+1}$, $\vK_2\in \bbN^{a_2+1}$,
$\vL_1\in \bbN^{b_1+1}$ and $\vL_2\in \bbN^{b_2+1}$
such that:
\begin{enumerate}
\item
The induced subgraph $G[U_1\cup V_1]$
is a $(0,1,\ldots,a_1)|(0,1,\ldots,b_1)$-biregular graph with size $\vK_1|\vL_1$.
\item 
the induced subgraph $G[U_1\cup V_2]$
is a $(a_1,a_1-1,\ldots,0)|(\prdp{0},\prdp{1},\ldots,\prdp{b_2})$-biregular graph
with size $\vK_1|\vL_2$.
\item 
the induced subgraph $G[U_2\cup V_1]$
is a $(\prdp{0},\prdp{1},\ldots,\prdp{a_2})|(b_1,b_1-1,\ldots,0)$-biregular graph
with size $\vK_2|\vL_1$.
\item 
the induced subgraph $G[U_2\cup V_2]$
is a $(\prdp{a_2},\prdp{(a_2-1)},\ldots,\prdp{0})|(\prdp{b_2},\prdp{(b_2-1)},\ldots,\prdp{0})$-biregular graph
with size $\vK_2|\vL_2$.
\end{enumerate}
\end{lemma}
\begin{proof}
Let $G=(U,V,E)$ be $A_0|B_0$-biregular graph with size $(M_1,M_2)|(N_1,N_2)$.
Let $U=U_1\uplus U_2$ and $V=V_1\uplus V_2$ be the witness partition where
\begin{itemize}
\item
every vertex in $U_1$ has degree $a_1$ and 
every vertex in $U_2$ has degree $\prdp{a_2}$;
\item 
every vertex in $V_1$ has degree $b_1$ and 
every vertex in $U_2$ has degree $\prdp{b_2}$.
\end{itemize}
We partition the set $U_1$ as follows:
$$
U_1 \ =\  U_{1,0}\uplus U_{1,1}\uplus\cdots \uplus U_{1,a_1},
$$
where for each $j \in [0,a_1]$, 
the set $U_{1,j}$ is the set of vertices in $U_1$ with $j$ neighbors in $V_1$
and $(a_1-j)$ neighbors in $V_2$.
Formally, we set
$$
U_{1,j} := 
\left\{
\begin{array}{l|l}
u \in U_1 &
 u \ \text{has $j$ neighbors in $V_1$ and $(a_1-j)$ neighbors in $V_2$}
\end{array}
\right\};
$$
see Figure~\ref{fig:example-non-simple-simple} for an illustration.
We repartition the set $U_2$, $V_1$, $V_2$ in a similar manner.
\begin{itemize}
\item
Let $U_2 = U_{2,0}\uplus U_{2,0}\uplus\cdots \uplus U_{2,a_2}$,
where for each $j\in [0,a_2]$, 
the set $U_{2,j}$ is the set of vertices in $U_2$ that has $\prdp{j}$ neighbors in $V_1$
and $\prdp{(a_2-j)}$ neighbors in $V_2$.
Formally, we set
$$
U_{2,j} \ :=\ 
\left\{
\begin{array}{l|l}
u \in U_2 &
\begin{array}{l}
 u \ \text{has $\prdp{j}$ neighbors in $V_1$}
\\
\text{and $\prdp{(a_2-j)}$ neighbors in $V_2$}
\end{array}
\end{array}
\right\}.
$$

\item 
We let $V_1= V_{1,0}\uplus V_{1,1}\uplus\cdots \uplus V_{1,b_1}$, where
for each $j \in [0,b_1]$, 
the set $V_{1,j}$ is the set of vertices in $V_1$ that has $j$ neighbors in $U_1$
and $(b_1-j)$ neighbors in $U_2$.
Formally, we set
$$
V_{1,j} \ := \
\left\{
\begin{array}{l|l}
u \in V_1 &
\begin{array}{l}
 u \ \text{has $j$ neighbors in $U_1$}
\\
\text{and $(b_1-j)$ neighbors in $U_2$}
\end{array}
\end{array}
\right\}.
$$
\item 
We let $V_2= V_{2,0}\uplus V_{2,1}\uplus\cdots \uplus V_{2,a_2}$, where
for each $j \in [0,b_2]$, 
the set $V_{2,j}$ is the set of vertices in $V_2$ that has $\prdp{j}$ neighbors in $U_1$
and $\prdp{(b_2-j)}$ neighbors in $U_2$.
Formally, we set
$$
V_{2,j} \ := \
\left\{
\begin{array}{l|l}
u \in V_2 &
\begin{array}{l}
 u \ \text{has $\prdp{j}$ neighbors in $U_1$}
\\
\text{and $\prdp{(b_2-j)}$ neighbors in $U_2$}
\end{array}
\end{array}
\right\}.
$$
\end{itemize}
Now, we let $\vK_1,\vK_2,\vL_1,\vL_2$ as follows:
\begin{align*}
\vK_1   &:= \ (|U_{1,0}|,|U_{1,1}|,\ldots,|U_{1,a_1}|),
\quad
& \vK_2   & := \ (|U_{2,0}|,|U_{2,1}|,\ldots,|U_{2,a_2}|),
\\
\vL_1  & := \ (|V_{1,0}|,|V_{1,1}|,\ldots,|V_{1,b_1}|),
\quad
& \vL_2  & := \ (|V_{2,0}|,|V_{2,1}|,\ldots,|V_{2,b_2}|).
\end{align*}
To complete the proof of Lemma~\ref{lem:example-non-simple-simple-only-if},
we show:
\begin{enumerate}[(1)]
\item
$G[U_1\cup V_1]$ is $(0,1,\ldots,a_1)|(0,1,\ldots,b_1)$-biregular graph with size $\vK_1|\vL_1$.
\item 
$G[U_1\cup V_2]$
is $(a_1,a_1-1,\ldots,0)|(\prdp{0},\prdp{1},\ldots,\prdp{b_2})$-biregular graph
with size $\vK_1|\vL_2$.
\item 
$G[U_2\cup V_1]$
is $(\prdp{0},\prdp{1},\ldots,\prdp{a_2})|(b_1,b_1-1,\ldots,0)$-biregular graph
with size $\vK_2|\vL_1$.
\item 
$G[U_2\cup V_2]$
is $(\prdp{a_1},\prdp{(a_1-1)},\ldots,\prdp{0})|(\prdp{b_2},\prdp{(b_2-1)},\ldots,\prdp{0})$-biregular graph
with size $\vK_2|\vL_2$.
\end{enumerate}
To prove (1), note that
\begin{itemize}
\item
for each $j_1\in [0,a_1]$, each vertex in $U_{1,j_1}$ has degree $j_1$ in $G[U_1\cup V_1]$;
\item 
for each $j_2\in [0,b_1]$, each vertex in $V_{1,j_2}$ has degree $j_2$ in $G[U_1\cup V_1]$. 
\end{itemize}
Thus, $U_1=U_{1,0}\uplus U_{1,1}\uplus\cdots \uplus U_{1,a_1}$
and $V_1=V_{1,0}\uplus V_{1,1}\uplus\cdots \uplus V_{1,b_1}$
is the witness partition of $(0,1,\ldots,a_1)|(0,1,\ldots,b_1)$-biregularity of $G[U_1\cup V_1]$.
Since $\vK_1= (|U_{1,0}|,|U_{1,1}|,\ldots,|U_{1,a_1}|)$
$\vL_1 = \ (|V_{1,0}|,|V_{1,1}|,\ldots,|V_{1,b_1}|)$,
the subgraph $G[U_1\cup V_1]$ has size $\vK_1|\vL_1$.
The proof of (2)--(4) is similar. 
This completes the proof of Lemma~\ref{lem:example-non-simple-simple-only-if}.
\end{proof}

Next, we will show Lemma~\ref{lem:example-non-simple-simple-if} 
which deals with the ``if'' direction of Theorem~\ref{theo:example-non-simple-simple}.

\begin{figure}
\begin{center}

\begin{tikzpicture}
\draw (0,0) ellipse (1cm and 1.6cm);
\node  at (-1,1.6) {\large $U_1$}; 

\draw (0,-4) ellipse (1cm and 1.6cm);
\node  at (-1,-2.6) {\large $U_2$};

\draw[gray!90] (0,0.5) --  node[xshift=-1.5cm,color=black,align=center,fill=white] {$H_1$} (7,.5);
\draw[gray!90] (0,0) --  node[xshift=-1.5cm,color=black,align=center,fill=white,sloped] {$H_2$} (7,-3.75); 

\draw[gray!90] (0,-3.75) --  node[xshift=-1.5cm,color=black,align=center,fill=white,sloped] {$H_3$} (7,0);
\draw[gray!90] (0,-4.25) --  node[xshift=-1.5cm,color=black,align=center,fill=white,sloped] {$H_4$} (7,-4.25); 

\draw (7,0) ellipse (1cm and 1.6cm);
\node  at (8,1.6) {\large $V_1$};

\draw (7,-4) ellipse (1cm and 1.6cm);
\node  at (8,-2.6) {\large $V_2$};

\end{tikzpicture}

\end{center}
\label{fig:example-non-simple-simple-if}
\caption{An illustration for the proof of Lemma~\ref{lem:example-non-simple-simple-if}.
The graph $H_1$ contains only edges between the vertices in $U_1$ and $V_1$.
the graph $H_2$ contains only edges between the vertices in $U_1$ and $V_2$.
The graph $H_3$ contains only edges between the vertices in $U_2$ and $V_1$.
The graph $H_4$ contains only edges between the vertices in $U_2$ and $V_2$.
Thus, the sets of edges in $H_1,H_2,H_3,H_4$ are pairwise disjoint.
The graph $G$ obtained by combining all four of them 
is $A_0|B_0$-biregular graph.}
\end{figure}

\begin{lemma}
\label{lem:example-non-simple-simple-if}
For every size vectors $\vK_1\in\bbN^{a_1+1}$,
$\vK_2\in \bbN^{a_2+1}$,
$\vL_1\in \bbN^{b_1+1}$ and $\vL_2\in \bbN^{b_2+1}$,
if there are
\begin{enumerate}[(1)]
\item 
a $(0,1,\ldots,a_1)|(0,1,\ldots,b_1)$-biregular graph with size $\vK_1|\vL_1$;
\item 
a $(a_1,a_1-1,\ldots,0)|(\prdp{0},\prdp{1},\ldots,\prdp{b_2})$-biregular graph with size 
$\vK_1|\vL_2$;
\item 
a $(\prdp{0},\prdp{1},\ldots,\prdp{a_2})|(b_1,b_1-1,\ldots,0)$-biregular graph with size 
$\vK_2|\vL_1$;
\item 
a $(\prdp{a_2},\prdp{(a_2-1)},\ldots,\prdp{0})|(\prdp{b_2},\prdp{(b_2-1)},\ldots,\prdp{0})$-biregular graph with size 
$\vK_2|\vL_2$,
\end{enumerate}
then there is an $A_0|B_0$-biregular graph with size $(M_1,M_2)|(N_1,N_2)$,
where $M_1=\normt{\vK_1}$, $M_2=\normt{\vK_2}$,
$N_1=\normt{\vL_1}$ and $N_2= \normt{\vL_2}$.
\end{lemma}
\begin{proof}
Let $\vK_1=(K_{1,0},\ldots,K_{1,a_1})\in \bbN^{a_1+1}$,
$\vK_2=(K_{2,0},\ldots,K_{2,a_2})\in \bbN^{a_2+1}$,
$\vL_1=(L_{1,0},\ldots,L_{1,b_1})\in \bbN^{b_1+1}$,
$\vL_2=(L_{2,0},\ldots,L_{2,b_2})\in \bbN^{b_2+1}$.
Let $U_1,U_2,V_1,V_2$ be pairwise disjoint sets of elements such that
$$
|U_1| = \normt{\vK_1},\qquad
|U_2| = \normt{\vK_2},\qquad
|V_1| = \normt{\vL_1},\qquad
|V_2| = \normt{\vL_2}.
$$
We partition $U_1,U_2,V_1,V_2$ as follows:
\begin{align*}
U_1 & := \ U_{1,0}\uplus U_{1,1}\uplus\cdots \uplus U_{1,a_1},\quad& &\text{where}\ (|U_{1,0}|,|U_{1,1}|,\ldots,|U_{1,a_1}|) = \vK_1,
\\
U_2 & := \ U_{2,0}\uplus U_{2,1}\uplus \cdots \uplus U_{2,a_2},\quad& &\text{where}\ (|U_{2,0}|,|U_{2,1}|,\ldots,|U_{2,a_2}|) = \vK_2,
\\
V_1 & := \ V_{1,0}\uplus V_{1,1}\uplus\cdots \uplus V_{1,b_1},\quad& &\text{where}\ (|V_{1,0}|,|V_{1,1}|,\ldots,|V_{1,b_1}|) = \vL_1,
\\
V_2 & := \ V_{2,0}\uplus U_{2,1}\uplus\cdots \uplus V_{2,b_2},\quad& &\text{where}\ (|V_{2,0}|,|U_{2,1}|,\ldots,|V_{2,b_2}|) = \vL_2.
\end{align*}

Suppose we have biregular graphs 
$H_{1}, H_{2}, H_{3}, H_{4}$, as stated in the hypotheses (1)--(4):
\begin{itemize}
\item 
$H_1$ is a $(0,1,\ldots,a_1)|(0,1,\ldots,b_1)$-biregular graph  with size $\vK_1|\vL_1$;
\item 
$H_{2}$ is a $(a_1,a_1-1,\ldots,0)|(\prdp{0},\prdp{1},\ldots,\prdp{b_2})$-biregular graph  with size $\vK_1|\vL_2$;
\item 
$H_{3}$ is a $(\prdp{0},\prdp{1},\ldots,\prdp{a_2})|(b_1,b_1-1\ldots,0)$-biregular graph  with size 
$\vK_2|\vL_1$;
\item 
$H_4$ is a $(\prdp{a_2},\prdp{(a_2-1)}\ldots,\prdp{0})|(\prdp{b_2},\prdp{(b_2-1)},\ldots,\prdp{0})$-biregular graph  with size 
$\vK_2|\vL_2$.
\end{itemize}
We will combine all these graphs $H_1,H_2,H_3,H_4$ into one $A_0|B_0$-biregular graph $G$ with size $(M_1,M_2)|(N_1,N_2)$.
See Figure~\ref{fig:example-non-simple-simple-if} for an illustration.
First, we make some observations.

\begin{itemize}
\item
Note that $H_1$ is a $(0,1,\ldots,a_1)|(0,1,\ldots,b_1)$-biregular graph with size $\vK_1|\vL_1$, 
matching the sizes of $U_1$ and $V_1$. \
So we may assume that $U_1$ is the set of vertices on the left hand side,
$V_1$ is the set of vertices on the right hand side. We can also assume
 that $U_1=U_{1,0}\uplus U_{1,1}\uplus\cdots \uplus U_{1,a_1}$ and 
$V_1  =  V_{1,0}\uplus V_{1,1}\uplus\cdots \uplus V_{1,b_1}$ is the witness partition
for $(0,1,\ldots,a_1)|(0,1,\ldots,b_1)$-biregularity of $H_{1}$.

Thus $H_1=(U_1,V_1,R_1)$ where $R_1$ is the set of edges.
\item
In a similar manner, since $H_2$ is $(a_1,a_1-1,\ldots,0)|(\prdp{0},\prdp{1},\ldots,\prdp{b_2})$-biregular graph with size $\vK_1|\vL_2$,
we may assume that $U_1$ is the set of vertices on the left hand side,
$V_2$ is the set of vertices on the right hand side,
and that $U_1=U_{1,0}\uplus U_{1,1}\uplus\cdots \uplus U_{1,a_1}$ and 
$V_2  =  V_{2,0}\uplus V_{2,1}\uplus\cdots \uplus V_{2,b_2}$ is the witness partition
of $(a_1,a_1-1,\ldots,0)|(\prdp{0},\prdp{1},\ldots,\prdp{b_2})$-biregularity of $H_{2}$.

We can thus write  $H_2=(U_1,V_2,R_2)$ where $R_2$ is the set of edges.
Note that $R_1$ and $R_2$ are disjoint since $R_1$ contains only edges between vertices in $U_1$ and vertices in $V_1$,
whereas $R_2$ contains only edges between vertices in $U_1$ and vertices in $V_2$.

\item
Analogously to what we observed about $H_2$,
since $H_3$ is a\\
 $(\prdp{0},\prdp{1},\ldots,\prdp{a_2})|(b_1,b_1-1,\ldots,0)$-biregular graph with size $\vK_2|\vL_1$,
we may assume that $U_2$ is the set of vertices on the left side,
$V_1$ is the set of vertices on the right side,
and that $U_2=U_{2,0}\uplus U_{2,1}\uplus\cdots \uplus U_{2,a_2}$ and 
$V_1  =  V_{1,0}\uplus V_{1,1}\uplus \cdots \uplus V_{1,b_1}$ is the witness partition
of $(\prdp{0},\prdp{1},\ldots,\prdp{a_2})|(b_1,b_1-1,\ldots,0)$-biregularity of $H_{3}$.
We write  $H_3=(U_2,V_1,R_3)$ where $R_3$ is the set of edges
and again note that $R_1,R_2,R_3$ are pairwise disjoint.

\item
Finally, since $H_4$ is a $(\prdp{a_2},\prdp{(a_2-1)}\ldots,\prdp{0})|(\prdp{b_2},\prdp{(b_2-1)},\ldots,\prdp{0})$-biregular graph with size $\vK_2|\vL_2$,
we may assume $U_2$ is the set of vertices on the left  side,
$V_2$ is the set of vertices on the right,
and that $U_2= U_{2,0}\uplus U_{2,1}\uplus\cdots \uplus U_{2,a_2}$ and 
$V_2  =  V_{2,0}\uplus V_{2,1}\uplus\cdots \uplus V_{2,b_2}$ is the witness partition
of $(\prdp{a_2},\prdp{(a_2-1)}\ldots,\prdp{0})|(\prdp{b_2},\prdp{(b_2-1)},\ldots,\prdp{0})$-biregularity of $H_{4}$.

We can thus write $H_4=(U_2,V_2,R_4)$ where $R_4$ is the set of edges
and again note that $R_1,R_2,R_3,R_4$ are pairwise disjoint.
\end{itemize}
Let $G=(U_1\cup U_2,V_1\cup V_2,E)$, where $E=R_1\cup R_2\cup R_3\cup R_4$.
That is, $G$ is the graph union of all $H_1,\ldots,H_4$.
In fact, $G[U_1\cup V_1]$ is $H_1$,
$G[U_1\cup V_2]$ is $H_2$,
$G[U_2\cup V_1]$ is $H_3$
and $G[U_2\cup V_2]$ is $H_4$.

We will prove that $G$ is $A_0|B_0$-biregular graph with size $(M_1,M_2)|(N_1,N_2)$,
where $M_1=|U_1|$, $M_2=|U_2|$, $N_1=|V_1|$ and $N_2=|V_2|$
by showing that
\begin{enumerate}[(1)]
\item
every vertex in $U_1$ has degree $a_1$
and every vertex in $U_2$ has degree $\prdp{a_2}$; and
\item 
every vertex in $V_1$ has degree $b_1$
and every vertex in $V_2$ has degree $\prdp{b_2}$.
\end{enumerate}
To prove (1), note that:
\begin{itemize}
\item
Since $H_1$ is a $(0,1,\ldots,a_1)|(0,1,\ldots,b_1)$-biregular graph,
for every $j \in [0,a_1]$,
every vertex $u\in U_{1,j}$ has degree $j$ in $H_1$.

Since $H_2$ is a $(a_1,a_1-1,\ldots,0)|(\prdp{0},\prdp{1},\ldots,\prdp{b_2})$-biregular graph,
for every $j \in [0,a_1]$,
every vertex $u\in U_{1,j}$ has degree $(a_1-j)$ in $H_2$.

Therefore, for each $j \in [0,a_1]$,
every vertex $u\in U_{1,j}$ has degree $j+ (a_1-j)=a_1$ in the graph $G$.

\item
Similarly, since $H_3$ is a $(\prdp{0},\prdp{1},\ldots,\prdp{a_2})|(b_1,b_1-1,\ldots,0)$-biregular graph,
for every $j \in [0,a_2]$,
every vertex $u\in U_{2,j}$ has degree $\prdp{j}$ in $H_3$.

Since $H_4$ is a $(\prdp{a_2},\prdp{(a_2-1)}\ldots,\prdp{0})|(\prdp{b_2},\prdp{(b_2-1)},\ldots,\prdp{0})$-biregular graph,
for every $j \in [0,a_2]$,
every vertex $u\in U_{2,j}$ has degree $\prdp{(a_2-j)}$ in $H_4$.

Therefore, for every $j \in [0,a_2]$,
each vertex $u\in U_{2,j}$ has degree $\prdp{j}+ \prdp{(a_2-j)}=\prdp{a_2}$ in the graph $G$.
\end{itemize}
The proof of (2) is similar.
\end{proof}

\subsection{The general reduction from non-simple to simple}
\label{subsec:general-non-simple-simple}

We now give the general process which makes use of the idea above.
In this section we will deal directly with complete biregular graphs.
Let $A\in \bbNp^{t\times m}$ and $B\in \bbNp^{t\times n}$ be arbitrary degree matrices.
We will show that every complete $A|B$-biregular graph 
can be decomposed into a collection of complete simple biregular graphs.

\begin{figure}
\begin{center}

\begin{tikzpicture}

%%% the set U1,...,Um

\draw (0,0) ellipse (1cm and 1cm);
\node  at (-1.1,0.8) {\large $U_1$}; 

\node at (0,-1.25) {$\vdots$};

\draw (0,-3.5) ellipse (1cm and 1.8cm);
\node  at (-1.1,-2.3) {\large $U_{j}$}; 

\node at (0,-5.6) {$\vdots$};

\draw (0,-7) ellipse (1cm and 1cm);
\node  at (-1.1,-6.1) {\large $U_n$}; 

%%% partition of U_j

\draw[red] (-0.75,-2.3) -- (0.75,-2.3);
\node at (0,-2.1) {\scriptsize $U_{j,g_1}$};

\node[red] at (0,-2.6) {$\vdots$};

\draw[red] (-.95,-3.1) -- (.95,-3.1);
\node at (-.5,-3.6) {\scriptsize $U_{j,g_k}$};
\node[circle,fill=blue,inner sep=0pt,minimum size=3pt,label=below:{\scriptsize $u$}] (a) at (0.5,-3.5) {};
\draw[red] (-.95,-3.9) -- (.95,-3.9);

\node[red] at (0,-4.2) {$\vdots$};

\draw[red] (-0.75,-4.7) -- (0.75,-4.7);
\node at (0,-5) {\scriptsize $U_{j,g_{\ell}}$};

%%% neighbors of u

\draw[gray!40] (.5,-3.5) --  node[color=black,above,sloped] {\scriptsize $E_i$-neighbors of $u$ in $V_1$} (6.5,-0.3);
\draw[gray!40] (.5,-3.5) -- (6.5,-1.7);
\draw[gray!40,fill=gray!40] (6.5,-1) ellipse (.5cm and .7cm);
\node at (7,-1) {\scriptsize $g_k(i,1)$ vertices};

\draw[gray!40] (.5,-3.5) -- (6.5,-5.3);
\draw[gray!40] (.5,-3.5) -- node[color=black,below,sloped] {\scriptsize $E_i$-neighbors of $u$ in $V_n$} (6.45,-6.7);
\draw[gray!40,fill=gray!40] (6.5,-6) ellipse (.5cm and .7cm);
\node at (7,-6) {\scriptsize $g_k(i,n)$ vertices};

%%% the set V1,...,Vn

\draw (7,-1) ellipse (1.2cm and 1.6cm);
\node  at (8.1,0.3) {\large $V_1$};

%\draw[line width=0.4mm,dotted] (7,-3) -- (7,-4);

\node at (7,-3.25) {$\vdots$};
\node at (7,-3.7) {$\vdots$};

\draw (7,-6) ellipse (1.2cm and 1.6cm);
\node  at (8.2,-5) {\large $V_n$};
\end{tikzpicture}

\end{center}
\label{fig:general-non-simple-simple}
\caption{Suppose $G$ is $A|B$-biregular graph with $U=U_1\uplus \cdots \uplus U_m$
and $V=V_1\uplus \cdots\uplus V_m$ being the witness partition.
We partition $U_{j}$ according to the functions $g_1,\ldots,g_{k}:[t]\times[n]\to \{0,1\ldots,q,\prdp{0},\prdp{1},\ldots,\prdp{q}\}$
where for each $\ell \in [k]$,
each vertex in $U_{j,g_\ell}$ has $g_\ell(i,1)$ $E_i$-neighbors in $V_1$,
$g_\ell(i,2)$ $E_i$-neighbors in $V_2$ and so on to
$g_\ell(i,n)$ $E_i$-neighbors in $V_n$.}
\end{figure}

The idea is similar to the one in Subsection \ref{subsec:example-non-simple-simple}.
Let $G=(U,V,E_1,\ldots,E_t)$ be a complete $A|B$-biregular graph.
We let $q$ to be the maximal (finite) offset found in $A$ and $B$.
For each color $i\in [t]$, we call a vertex $v$ an {\em $E_i$-neighbor} of a vertex $u$,
if $v$ is adjacent to $u$ via $E_i$-edges.
 
Suppose $U=U_1\uplus \cdots \uplus U_n$ and $V=V_1\uplus \cdots \uplus V_n$
is the witness partition of $A|B$-biregularity of $G$.
For each $j\in [m]$,
we further partition each $U_{j}$:
$$U_j \ = \
 U_{j,g_1}\uplus \cdots \uplus U_{j,g_k},
$$
where $g_1,\ldots,g_{k}:[t]\times[n]\to \{0,1,\ldots,q,\prdp{0},\prdp{1},\ldots,\prdp{q}\}$
are functions and 
for each color $i\in [t]$,
for each $\ell \in [k]$, each vertex $u\in U_{j,g_\ell}$ has $g_\ell(i,1)$ $E_i$-neighbors in the set $V_1$,
$g_\ell(i,2)$ $E_i$-neighbors in the set $V_2$ and so on to $g_\ell(i,n)$ $E_i$-neighbors in the set $V_n$.
See Figure~\ref{fig:general-non-simple-simple}.\footnote{The partitioning of $U_j$ into $U_{j,g_1}\uplus \cdots \uplus U_{j,g_k}$
is similar to how we partition the set $U_1=U_{1,0}\uplus\cdots\uplus U_{1,a_1}$ 
in Lemma~\ref{lem:example-non-simple-simple-only-if}
where for each $j\in [0,a_1]$, each vertex in $U_{1,j}$ has $j$ neighbors in the set $V_1$
and $(a_1-j)$ neighbors in the set $V_2$.}
To ensure that each vertex in $U_{j}$ has $E_i$-degree $A_{i,j}$ for every color $i\in [t]$,
we require that $g_\ell(i,1)+\cdots g_\ell(i,n)=A_{i,j}$.
Note that if $A_{i,j}$ is a fixed entry, then all $g_\ell(i,1),\ldots,g_\ell(i,n)$ are fixed entries.
If $A_{i,j}$ is a periodic entry, then all $g_\ell(i,1),\ldots,g_\ell(i,n)$ are periodic entries.

In the same way, for each $j'\in [n]$,
we further partition each set $V_{j'}$:
$$
V_{j'}\ =\  V_{j',h_1}\uplus \cdots \uplus V_{j',h_k},
$$
where $h_1,\ldots,h_{k}:[t]\times[m]\to \{0,1,\ldots,q,\prdp{0},\prdp{1},\ldots,\prdp{q}\}$
are functions
and for each color $i\in [t]$, for each $\ell\in [k]$,
every vertex $u\in V_{j',h_\ell}$ has $h_\ell(i,1)$ $E_i$-neighbors in the set $U_1$,
$h_\ell(i,2)$ $E_i$-neighbors in the set $U_2$ and so on to
$h_\ell(i,m)$ $E_i$-neighbors in the set $U_m$.

We will show that every complete $A|B$-biregular graph $G$ with witness partition $U=U_1\uplus \cdots\uplus U_m$
and $V=V_1\uplus \cdots \uplus V_n$ can be decomposed into
complete simple biregular graphs in the sense that
for each $j\in [m]$ and each $j'\in [n]$,
the induced subgraph $G[U_j\cup V_{j'}]$
is a complete simple biregular graph with witness partition
$U_{j}= U_{j,g_1}\uplus \cdots \uplus U_{j,g_k}$
and $V_{j'}= V_{j',h_1}\uplus \cdots \uplus V_{j',h_k}$.
Such decomposition is also sufficient to capture all possible complete $A|B$-biregular graphs.
We will formalize this idea in the next paragraphs.

We first need some terminology.

\begin{definition}
\label{def:behavior-function-non-simple}
For each $j \in [m]$, we define {\em a behavior function of column $j$ in $A$}
to be a function $g:[t]\times [n]\to \{0,1,\ldots,q,\prdp{0},\prdp{1},\ldots,\prdp{q}\}$ such that:
\begin{itemize}
\item
$
A_{*,j} \ = 
\begin{pmatrix}
g(1,1)+\cdots + g(1,n)
\\
g(2,1)+\cdots + g(2,n)
\\
\vdots
\\
g(t,1)+\cdots + g(t,n)
\end{pmatrix};
$
\item
for each color $i\in [t]$,
if $A_{i,j}$ is a fixed entry, then $g(i,1),\ldots,g(i,n)$ are all fixed entries;
\item
for each color $i\in [t]$,
if $A_{i,j}$ is a periodic entry, then $g(i,1),\ldots,g(i,n)$ are all periodic entries.
\end{itemize}
In a similar manner 
for each $j'\in [n]$, we define {\em a behavior function of column $j'$ in $B$}
to be a function $h:[t]\times [m]\to \{0,1,\ldots,q,\prdp{0},\prdp{1},\ldots,\prdp{q}\}$ such that:
\begin{itemize}
\item
$
B_{*,j'} \ = 
\begin{pmatrix}
h(1,1)+\cdots + h(1,m)
\\
h(2,1)+\cdots + h(2,m)
\\
\vdots
\\
h(t,1)+\cdots + h(t,m)
\end{pmatrix};
$
\item
for each color $i\in [t]$,
if $B_{i,j'}$ is a fixed entry, then $h(i,1),\ldots,h(i,n)$ are all fixed entries;
\item
for each color $i\in [t]$,
if $B_{i,j'}$ is a periodic entry, then $h(i,1),\ldots,h(i,n)$ are all periodic entries.
\end{itemize}
\end{definition}

For each $j\in [m]$, let $g_{j,1},\ldots,g_{j,k}$ enumerate all behavior functions of column $j$ in $A$.
Similarly,
for each $j'\in [n]$, let $h_{j',1},\ldots,h_{j',k}$ enumerate all behavior functions of column $j'$ in $B$.
Note that we assume that the number of behavior functions of column $j$ in $A$ 
is the same as the number of behavior functions of column $j'$ in $B$ for every $j\in [m]$ and every $j'\in [n]$.
This is because we may ``repeat'' the same behavior function a few times in the enumeration
$g_{j,1},\ldots,g_{j,k}$ and $h_{j',1},\ldots,h_{j',k}$.

For each $j\in [m]$, for each $j'\in [n]$, define the matrix $C_{j,j'}$ and $D_{j,j'}$:
\begin{align*}
C_{j,j'} & := 
\begin{pmatrix}
g_{j,1}(1,j') & g_{j,2}(1,j') & \cdots & g_{j,k}(1,j')
\\
g_{j,1}(2,j') & g_{j,2}(2,j') & \cdots & g_{j,k}(2,j')
\\
\vdots & \vdots  & \ddots & \vdots
\\
g_{j,1}(t,j') & g_{j,2}(t,j') & \cdots & g_{j,k}(t,j')
\end{pmatrix}
\end{align*}
and
\begin{align*}
D_{j,j'} & := 
\begin{pmatrix}
h_{j',1}(1,j) & h_{j',2}(1,j) & \cdots & h_{j',k}(1,j)
\\
h_{j',1}(2,j) & h_{j',2}(2,j) & \cdots & h_{j',k}(2,j)
\\
\vdots & \vdots  & \ddots & \vdots
\\
h_{j',1}(t,j) & h_{j',2}(t,j) & \cdots & h_{j',k}(t,j)
\end{pmatrix}.
\end{align*}
Note that for each color $i\in [t]$, if $A_{i,j}$ is a fixed entry,
the values $g_{j,\ell}(i,1),\ldots,g_{j,\ell}(i,n)$ are all fixed for each $\ell \in [k]$.
Hence all the values $g_{j,1}(i,j'),\ldots,g_{j,k}(i,j')$ are fixed,
i.e., row $i$ in $C_{j,j'}$ contains only fixed entries.
Similarly, if $A_{i,j}$ is a periodic entry,
the values $g_{j,\ell}(i,1),\ldots,g_{j,\ell}(i,n)$ are all periodic for every $\ell \in [k]$.
Hence all the values $g_{j,1}(i,j'),\ldots,g_{j,k}(i,j')$ are periodic,
i.e., row $i$ in $C_{j,j'}$ contains only periodic entries.
Therefore for each $j\in [m]$ and every $j'\in [n]$ $C_{j,j'}$ is a simple matrix.
In a similar manner, we can argue that each $D_{j,j'}$ is a simple matrix.
% for  $j\in [m]$, $j'\in [n]$.

We will show that every complete $A|B$-biregular graph 
can be decomposed into complete $C_{j,j'}|D_{j,j'}$-biregular graphs
for every $j\in [m]$ and every $j'\in [n]$,
as stated formally in Lemma~\ref{lem:general-non-simple-simple}.

\begin{lemma}
\label{lem:general-non-simple-simple}
For every pair of size vectors $\vM\in \bbN^m$ and $\vN\in\bbN^n$,
the statements (a) and (b) are equivalent.
\begin{enumerate}[(a)]
\item
There is a complete $A|B$-biregular graph with size $\vM|\vN$.
\item 
There are size vectors $\vK_1,\ldots,\vK_m,\vL_1,\ldots,\vL_n \in \bbN^{k}$
such that:
$$
\vM = (\normt{\vK_1},\ldots,\normt{\vK_m})\qquad\text{and}\qquad
\vN = (\normt{\vL_1},\ldots,\normt{\vL_n})
$$
and for every $j\in [m]$ and for every $j'\in [n]$,
there is a complete $C_{j,j'}|D_{j,j'}$-biregular graph with size $\vK_{j}|\vL_{j'}$.
\end{enumerate}
\end{lemma}

The proof is a routine adaptation of Lemma~\ref{theo:example-non-simple-simple},
hence we omit the details.
We describe here the main intuition.
For (a) implies (b), suppose $G=(U,V,E_1,\ldots,E_t)$ is a complete $A|B$-biregular graph with size $\vM|\vN$.
Let $U=U_1\uplus \cdots \uplus U_m$ and $V=V_1\uplus\cdots\uplus V_n$ be the witness partition.
For every $j\in [m]$, for every $j'\in[n]$,
we can  show that each induced subgraph $G[U_j\cup V_{j'}]$ is a complete $C_{j,j'}|D_{j,j'}$-biregular graph
with witness partition $U_{j}=U_{j,g_1}\uplus \cdots \uplus U_{j,g_{k}}$
and $V_{j'}=V_{j',h_1}\uplus \cdots \uplus V_{j',h_{k}}$,
where $\vK_{j}= (|U_{j,g_1}|,\ldots,|U_{j,g_{k}}|)$
and  $\vL_{j'}= (|V_{j',h_1}|,\ldots,|V_{j',h_{k}}|)$.

Conversely, for (b) implies (a), let $\vK_1,\ldots,\vK_m,\vL_1,\ldots,\vL_n \in \bbN^{k}$
be such that
$$
\vM = (\normt{\vK_1},\ldots,\normt{\vK_m})\qquad\text{and}\qquad
\vN = (\normt{\vL_1},\ldots,\normt{\vL_n}).
$$
Suppose for every $j\in [m]$ and for every $j'\in [n]$,
there is a complete $C_{j,j'}|D_{j,j'}$-biregular graph $G_{j,j'}$ with size $\vK_{j}|\vL_{j'}$.
Due to the matching sizes,
we can assume that the set of vertices on the left hand side of $G_{j,j'}$ is $U_{j}$
and the set of vertices on the right hand side of $G_{j,j'}$ is $V_{j'}$,
where $|U_{j}|=\normt{\vK_{j}}$ and $|V_{j'}|=\normt{\vL_{j'}}$.
Taking the disjoint union of all the graphs $G_{1,1}\cup \cdots \cup G_{m,n}$,
we obtain a complete $A|B$-biregular graph $G$ with size $\vM|\vN$.

Using Lemma~\ref{lem:general-non-simple-simple},
we can now define the formula $\biregc_{A|B}(\vx,\vy)$
as required in Theorem~\ref{thm:main-lemma-bireg}.
We first explain the variables of the formula.
\begin{itemize}
\item 
For every $j\in [m]$,
for every behavior function $g$ of column $j$ in $A$, we have a variable $X_{j,g}$.
Let $\vX_{j}=(X_{j,g_1},\ldots,X_{j,g_k})$, where $g_1,\ldots,g_k$ are all the behavior functions of column $j$ in $A$.
\item 
Similarly, 
for every $j'\in [n]$,
for every behavior function $h$ of column $j'$ in $B$, we have a variable $Y_{j',h}$.
Let $\vY_{j'}=(Y_{j',h_1},\ldots,Y_{j',h_{k}})$, where $h_1,\ldots,h_{k}$ are all the behavior functions of column $j'$ in $B$.
\end{itemize}

Consider the formula $\biregc_{A|B}(\vx,\vy)$:
\begin{align}
\label{eq:non-simple-main}
\exists \vX_1\cdots \exists \vX_m\ \exists \vY_1\cdots\exists\vY_n
\qquad &
\vx=(\normt{\vX_1},\ldots,\normt{\vX_m}) \ \wedge\
\vy=(\normt{\vY_1},\ldots,\normt{\vY_n})
\\
\wedge \ &
\bigwedge_{j\in [m]}\; \bigwedge_{j'\in [n]}
\biregc_{C_{j,j'}|D_{j,j'}}(\vX_{j},\vY_{j'}).
\end{align}
Note that $C_{j,j'}$ and $D_{j,j'}$ are simple matrices
and the formula $\biregc_{C_{j,j'}|D_{j,j'}}(\vX_{j},\vY_{j'})$
is as defined in \eqref{eq:finite-simple-complete}.

We show that the formula $\biregc_{A|B}(\vx,\vy)$ is correct,
i.e., it captures all possible sizes of complete $A|B$-biregular graphs,
as stated formally in Theorem~\ref{theo:main-lemma-bireg}.

\begin{theorem}
\label{theo:main-lemma-bireg}
For all degree matrices $A$ and $B$,
for all size vectors $\vM$ and $\vN$,
there is a complete $A|B$-biregular graph with size $\vM|\vN$
if and only if $\biregc_{A|B}(\vM,\vN)$ holds in $\cN$.
\end{theorem}
The proof follows directly from Lemma~\ref{lem:general-non-simple-simple}
and Lemma~\ref{lem:completesimplegoodonefinite}.

\subsection{Proof of Theorem~\ref{thm:main-lemma-direg}: Construction of the Presburger formula for complete regular digraphs}
\label{subsec:proof-main-lemma-direg}

In Section~\ref{subsec:general-non-simple-simple} we have shown that given arbitrary degree matrices $A$ and $B$,
we can construct a Presburger formula that captures precisely the sizes of complete $A|B$-biregular graphs.
The construction is by reducing $A$ and $B$ into a collection of simple matrices.
The proof for the digraph case is very similar to the biregular case.
As in the $1$-color case from Subsection~\ref{subsec:1type-digraphs},
the existence of $A|B$-regular digraphs with size $\vM$
can be reduced to the existence of $A|B$-biregular graphs with size $\vM|\vM$.
Indeed, an $A|B$-regular digraph $G$ with size $\vM$ can be encoded as 
an $A|B$-biregular graph $G'$ with size $\vM|\vM$
by splitting each vertex $w$ in $G$ into two vertices $u$ and $v$ in $G'$
where $u$ is adjacent to all the outgoing edges and $v$ to all the incoming edges.
Thus, $G'$ is a bipartite graph where 
the vertices on the left hand side in $G'$ are all the vertices with the outgoing edges
and the vertices on the right hand side are all the vertices with the incoming edges;
see Figure~\ref{fig:splitting} for an illustration.

The construction of the desired formula $\diregc_{A|B}(\vx)$ that captures all possible sizes of complete $A|B$-regular digraph
can be done similarly to the one for complete biregular graphs.
First, we construct a formula $\diregc_{A|B}(\vx)$ when $A$ and $B$ are simple matrices,
which is similar to Section \ref{sec:simple-complete}.
The reduction from non-simple matrices to simple matrices is
similar to the one in Section \ref{subsec:general-non-simple-simple}.
We omit the details, since they are just a routine adaptation of the ones in 
Section \ref{sec:simple-complete} and Section \ref{subsec:general-non-simple-simple}.

%

%!TEX root = fo2siam.tex

\section{Complexity of the decision procedures}
\label{sec:complexity}

\newcommand{\cM}{\mathcal{M}}
\newcommand{\poly}{\textsf{poly}}
\newcommand{\len}{\textrm{len}}

We now analyze the complexity for each of the problems studied earlier.
We begin with the biregular graph problems.
 We will then turn to the combined complexity
of the decision procedure for the logic. Finally, we consider the complexity
of the decision procedure for the logic when we fix a formula and vary its conjunction
with a collection of ground facts --  data complexity.

\subsection{Complexity of the graph analysis}
\label{subsec:complexity-bireg}

In this section we state the refined versions of the main results concerning biregular
and biregular graph problems, now with complexity upper bounds.  We do not 
have non-trivial lower bounds for these problems.
As before, we only deal with the finite satisfiability.
The analysis of  general satisfiability can be found in the appendix.

\begin{lemma}
\label{lem:algo-for-bireg}
There is a non-deterministic Turing machine $\cM$ that does the following:
on input degree matrices $A\in \bbN_{+p}^{t\times m}$ and $B\in \bbN_{+p}^{t\times n}$,
on every run $r$ of $\cM$, it outputs an existential Presburger formula $\varphi_r(\vx,\vy)$ 
such that
\begin{itemize}
\item 
each $\varphi_r(\vx,\vy)$ is of the form
$\exists\vz\ \widetilde{\varphi}_r(\vx,\vy,\vz)$,
where each $\widetilde{\varphi}_r(\vx,\vy,\vz)$ is a conjunction of $O(mnt^4\delta(A,B)^4)$ linear (in)equations; and
\item
for every $(\vM,\vN)\in \bbN^m\times \bbN^n$, 
there is complete $A|B$-biregular graph with size $\vM|\vN$
if and only if 
there is a run $r$ of $\cM$ such that $\varphi_r(\vM,\vN)$ holds in $\cN$.
\end{itemize}
Moreover, $\cM$ runs in time exponential in the size of $A$ and $B$,
where the coefficients of the input degree matrices and the output formula $\varphi_r$ are in binary.
\end{lemma}
\begin{proof}
For arbitrary degree matrices $A\in \bbN_{+p}^{t\times m}$ and $B\in \bbN_{+p}^{t\times n}$,
recall the formula $\biregc_{A|B}(\vx,\vy)$ defined in \eqref{eq:non-simple-main}:
\begin{align*}
\exists \vX_1\cdots \exists \vX_m\ \exists \vY_1\cdots\exists\vY_n
\qquad &
\vx=(\normt{\vX_1},\ldots,\normt{\vX_m}) \ \wedge\
\vy=(\normt{\vY_1},\ldots,\normt{\vY_n})
\\
\wedge \ &
\bigwedge_{j\in [m]}\; \bigwedge_{j'\in [n]}
\biregc_{C_{j,j'}|D_{j,j'}}(\vX_{j},\vY_{j'}),
\end{align*}
where each $C_{j,j'}$ and $D_{j,j'}$ are simple matrices with $t$ rows.
Note that each variable in each $\vX_j$ is of the form $X_{j,g}$
where $j\in [m]$ and $g:[t]\times[n]\to \{\infty,0,\ldots,q,\prdp{0},\ldots,\prdp{q}\}$
is a function and $q$ is the maximal finite offset in $A$ and $B$.
Hence the number of bits to encode each $X_{j,g}$ is polynomial in 
the length of $A$ and $B$.
Similarly for each variable in each $\vY_{j'}$.

By Remark~\ref{rem:simple-complete}, each $\biregc_{C_{j,j'}|D_{j,j'}}(\vX_{j},\vY_{j'})$
is a disjunction of conjunctions of $O(t^4\delta(A,B)^4)$ (in)equations.
Thus, the formula $\biregc_{A|B}(\vx,\vy)$ is a disjunction of conjunctions of $O(mnt^4\delta(A,B)^4)$ (in)equations.

The desired NTM $\cM$ works as follows. 
On input $A$ and $B$, it constructs the formula $\biregc_{A|B}(\vx,\vy)$,
where on each disjunction, it guesses which disjunct should hold.
It outputs the constructed formula, which is a conjunction of $O(mnt^4\delta(A,B)^4)$ (in)equations
and all the variables that are not in $\vx$ and $\vy$ are existentially quantified.

This by itself, of course, does not guarantee that the running time is only exponential, since
the number of variables in the system may be more than exponential.
Here we invoke results in~\cite{intprogborosh,intprogeisenbrand}, which state that
if a system of linear equations has a solution, it has a solution in which
the number of variables taking non-zero is bounded by a polynomial in the number of equations and in the length of the binary representation 
of the coefficients in the system.\footnote{For example, Corollary~5 in~\cite{intprogeisenbrand}
states that if a system $A\vx =\vb$ has a solution in $\cN$,
then it has a solution $\vx$ such that the number of variables taking non-zero values is at most $2(d+1)(\log (d+1)+s+2)$,
where $d$ is the number of rows of $A$ and $s$ is the largest size of a coefficient in $A$ and $b$ (in binary representation).}
Thus, when our algorithm constructs the formula $\biregc_{A|B}(\vx,\vy)$,
it also guesses the variables that take non-zero values, and ignores the remaining variables.
Finally, applying Theorem~\ref{thm:pres}, our decision procedure runs in (non-deterministic) exponential time.
\end{proof}

Lemma~\ref{lem:algo-for-direg} is the directed graph analogue of  Lemma~\ref{lem:algo-for-bireg},
and the proof is similar.

\begin{lemma}
\label{lem:algo-for-direg}
There is a non-deterministic Turing machine $\cM$ that does the following:
on input degree matrices $A\in \bbN_{+p}^{t\times m}$ and $B\in \bbN_{+p}^{t\times m}$,
on every run $r$ of $\cM$, it outputs an existential Presburger formula $\varphi_r(\vx)$ 
such that 
\begin{itemize}
\item 
each $\varphi_r(\vx)$ is of the form $\exists\vz\ \widetilde{\varphi}_r(\vx,\vz)$,
where each $\widetilde{\varphi}_r(\vx,\vz)$ is a conjunction of $O(m^2t^4\delta(A,B)^4)$ linear (in)equations; and
\item
for every $\vM\in \bbNo^m$, 
there is complete $A|B$-regular digraph with size $\vM$
if and only if 
there is a run $r$ of $\cM$ such that $\varphi_r(\vM)$ holds in $\cN$.
\end{itemize}
Moreover, $\cM$ runs in time exponential in the size of $A$ and $B$,
where the coefficients of the input degree matrices and the output formula $\varphi_r$ are in binary.
\end{lemma}

\subsection{2-$\nexp$ algorithm for the finite satisfiability of $\fotwopres$}
\label{subsec:complexity-dec-proc}

We now give an analysis of the complexity of the decision procedure
for our logic, based on the analysis of the complexity of the corresponding graph problems.

Recall that $\OneTypes$ and $\TwoTypes$ denote the set of $1$- and $2$-types, respectively,
and that $\cE^\darr = \cE\cup \rev{\cE}$ where the elements in $\rev{\cE}$ represent the reversal of each $2$-types in $\cE$.
For finite satisfiability, a behavior function is a function $g:\TwoTypes^\darr\times\OneTypes \to \bbNp$,
where the co-domain is $\{0,\ldots,q,\prdp{0},\ldots,\prdp{q}\}$ and $q$
is the maximal non-$\infty$ offset in the u.p.s. $S_i$'s.
So, the total number of behavior functions is
\begin{align*}
m &\ =\  (2q+2)^{2tn} \ = \ 2^{2tn\log(2q+2)}
\end{align*}
where $t=|\TwoTypes|$ and $n=|\OneTypes|$.

We enumerate all behavior functions $g_1,\ldots,g_{m}$
and all $1$-types $\pi_1,\ldots,\pi_n$.
The Presburger sentence $\PREB_\phi$ is of the form
\begin{align*}
\PREB_\phi & \ := \  \exists \vX \ \const_1(\vX) \ \wedge \ \const_2 (\vX) \wedge (\bigvee_{i\in [n],\ j\in [m]}  X_{(\pi_i, g_j)} \neq 0),
\end{align*}
where $\vX$ is a vector of variables $(X_{(\pi_1, g_1)}, X_{(\pi_1, g_2)}, \ldots , X_{(\pi_n, g_m)}$). 

The formula $\const_1(\vX)$ is
\begin{align*}
\const_1(\vX) & : = \ 
\bigwedge_{\pi\ \text{is incompatible}}\:\sum_{g\in \cG} X_{\pi, g} = 0
\quad\wedge\quad
\bigwedge_{(\pi,g) \in H} X_{\pi, g} = 0,		
\end{align*}
where $H$ is the set of all incompatible $(\pi,g)$.
Checking whether $\pi$ and $(\pi,g)$ is compatible/incompatible can be done in deterministic exponential time.
So, this formula is negligible in our analysis.

Recall that for a $1$-type $\pi$,
 $\vX_{\pi}$ denotes the tuple of variables $(X_{\pi,g_1},\ldots,X_{\pi,g_m})$
The formula $\const_2$ is defined as
\begin{align*}
\const_2(\vX) & : = \ \bigwedge_{1 \leq i \leq n} 
\diregc_{M_{\pi_i}|\rev{M}_{\pi_i}}(\vX_{\pi_i})
\ \wedge \
\bigwedge_{1 \leq i < j \leq n}
\biregc_{L_{\pi_j}|\rev{L}_{\pi_i}}(\vX_{\pi_i},\vX_{\pi_j}),
\end{align*}
where
\begin{itemize}
\item 
$M_{\pi_i}$ and $\rev{M}_{\pi_i}$ are matrices with size $t\times m$, and
\item
$L_{\pi_j}$ and $\rev{L}_{\pi_i}$ are matrices with size $2t\times m$.
\end{itemize}  
Recall that $t$ and $m$ are the number of 2-types and behavior functions, respectively.

Using the Turing machine in Lemmas~\ref{lem:algo-for-bireg} and~\ref{lem:algo-for-direg},
the decision procedure can guess a formula $\const_2(\vx)$ where the total number of (in)equations is
\begin{align}
\label{eq:number-equations}
O(n^2m^2t^4\delta(A,B)^4) & \ = \ 
O(2^{4tn\log(2q+2)} n^2t^4 \delta(A,B)^4),
\end{align}
where $t$ and $n$ are the numbers of 2-types and 1-types, respectively.
That is, the number of (in)equations is doubly-exponential in the size of the input formula.

The Turing machines in Lemmas~\ref{lem:algo-for-bireg} and~\ref{lem:algo-for-direg} runs in time exponential in the size of each $M_{\pi_i}|\rev{M}_{\pi_i}$
and $L_{\pi_j}|\rev{L}_{\pi_i}$, which in turn, is exponential in the size of the input formula.
So, altogether our decision procedure takes doubly-exponential time to construct $\const_2(\vX)$.
Applying Theorem~\ref{thm:pres}, it runs in (non-deterministic) doubly-exponential time.

Note that here we also invoke results in~\cite{intprogborosh,intprogeisenbrand}.
Since the number of (in)equations in $\const_2(\vX)$ is only doubly-exponential,
if it has a solution, it has a solution in which the number of variables taking non-zero is at most doubly-exponential.
Thus, the decision procedure also guesses the variables that take non-zero values, and ignores the remaining variables.

Thus, we have the 2-$\nexp$ upper bound for the finite satisfiability of $\fotwopres$,
as stated formally as Theorem~\ref{theo:fin-sat-fotwopres}.

\begin{theorem}
\label{theo:fin-sat-fotwopres}
The finite satisfiability of $\fotwopres$ is in 2-$\nexp$.
\end{theorem}

\subsection{2-$\nexp$ algorithm for the general satisfiability of $\fotwopres$}
\label{subsec:complexity-dec-proc-general}

In this subsection we will briefly explain that the same upper bound also holds
for the general satisfiability of $\fotwopres$.
First, we have the following lemma which is the analogue of Lemma~\ref{lem:algo-for-bireg}
for the general case.

\begin{lemma}
\label{lem:algo-for-bireg-general}
There is a non-deterministic Turing machine $\cM$ that does the following:
on input degree matrices $A\in \bbNop^{t\times m}$ and $B\in \bbNop^{t\times n}$,
on every run $r$ of $\cM$, it outputs an existential Presburger formula $\varphi_r(\vx,\vy)$ 
such that
\begin{itemize}
\item 
each $\varphi_r(\vx,\vy)$ is of the form
$\exists\vz\ \widetilde{\varphi}_r(\vx,\vy,\vz)$,
where each $\widetilde{\varphi}_r(\vx,\vy,\vz)$ is a conjunction of $O(mn2^tt^4\delta(A,B)^4)$ linear (in)equations; and
\item
for every $(\vM,\vN)\in \bbNo^m\times \bbNo^n$, 
there is complete $A|B$-biregular graph with size $\vM|\vN$
if and only if 
there is a run $r$ of $\cM$ such that $\varphi_r(\vM,\vN)$ holds in $\cN_{\infty}$.
\end{itemize}
Moreover, $\cM$ runs in time exponential in the size of $A$ and $B$,
where the coefficients of the input degree matrices and the output formula $\varphi_r$ are in binary.
\end{lemma}

Note the additional factor $2^t$ in the number of linear (in)equations which is incurred in the construction of the formula $\biregc_{A|B}(\vx,\vy)$
when $A$ and $B$ are simple matrices and may contain $\infty$ entries.
The detailed analysis can be found in the appendix.
The directed graph analogue is stated as Lemma~\ref{lem:algo-for-direg-general}.

\begin{lemma}
\label{lem:algo-for-direg-general}
There is a non-deterministic Turing machine $\cM$ that does the following:
on input degree matrices $A\in \bbNop^{t\times m}$ and $B\in \bbNop^{t\times m}$,
on every run $r$ of $\cM$, it outputs an existential Presburger formula $\varphi_r(\vx)$ 
such that 
\begin{itemize}
\item 
each $\varphi_r(\vx)$ is of the form $\exists\vz\ \widetilde{\varphi}_r(\vx,\vz)$,
where each $\widetilde{\varphi}_r(\vx,\vz)$ is a conjunction of $O(m^22^tt^4\delta(A,B)^4)$ linear (in)equations; and
\item
for every $\vM\in \bbNo^m$, 
there is complete $A|B$-regular digraph with size $\vM$
if and only if 
there is a run $r$ of $\cM$ such that $\varphi_r(\vM)$ holds in $\cN_{\infty}$.
\end{itemize}
Moreover, $\cM$ runs in time exponential in the size of $A$ and $B$,
where the coefficients of the input degree matrices and the output formula $\varphi_r$ are in binary.
\end{lemma}

Another difference between the procedures for the finite and general satisfiability of $\fotwopres$ 
is that the co-domain of a behavior function for the general case is $\{\infty, 0,\ldots,q,\prdp{0},\ldots,\prdp{q}\}$,
where $q$ is the maximal (non $\infty$) offset in the u.p.s. $S_i$'s. 
Then, total number of behavior functions becomes
\begin{align*}
m &\ =\  (2q+3)^{2tn} \ = \ 2^{2tn\log(2q+3)}
\end{align*}
where $t$ is the number of all 2-types and $n$ is the number of all 1-types.

Similar to the finite case, using the Turing machine in Lemmas~\ref{lem:algo-for-bireg-general} and~\ref{lem:algo-for-direg-general},
the decision procedure can guess a formula $\const_2(\vx)$ where the total number of (in)equations is
\begin{align}
\label{eq:number-equations-general}
O(n^2m^22^tt^4\delta(A,B)^4) & \ = \ 
O(2^{t+4tn\log(2q+3)} n^2t^4 \delta(A,B)^4),
\end{align}
where $t$ and $n$ are the numbers of 2-types and 1-types, respectively.
That is, the number of (in)equations is still doubly-exponential in the size of the input formula.
Using the algorithm in Theorem~\ref{thm:pres},
the $2$-$\nexp$ upper bound also holds for the general satisfiability of $\fotwopres$,
as stated formally as Theorem~\ref{theo:gen-sat-fotwopres}.

\begin{theorem}
\label{theo:gen-sat-fotwopres}
The general satisfiability of $\fotwopres$ is in 2-$\nexp$.
\end{theorem}

\subsection{Data complexity of $\fotwopres$ formulas}

We now turn to families of formulas of the form $\phi \wedge \bigwedge_i A_i$ , where
$\phi$ is in the logic and the sets $\{ A_1 \ldots A_n\}$ range  over set finite collection of facts.
We say that  $\phi$ has \emph{$\np$ data complexity of (finite) satisfiability} if there is a
non-deterministic algorithm that takes
as input a set of ground atoms $A$ and determines whether $\phi \wedge \bigwedge A$ is satisfiable, running
in time polynomial in the size of $A$.
Pratt-Hartmann \cite{iandatacomplexity} showed that $\ctwo$ formulas have  $\np$ data complexity
of both satisfiability and finite satisfiability.
Following the general approach to data complexity from \cite{iandatacomplexity}, while plugging
in our Presburger characterization of $\fotwopres$, we  can show that the same data complexity
bound holds for $\fotwopres$.

\begin{theorem} \label{thm:datacomplexity} $\fotwopres$ formulas have $\np$ data complexity of
satisfiability and finite satisfiability.
\end{theorem}
\begin{proof}
We give only the proof for finite satisfiability.
We will follow closely the approach used for $\ctwo$ in Section 4 of \cite{iandatacomplexity},
and the terminology we use below comes from that work.

Given a set of facts $D$,
our algorithm guesses a set of facts (including equalities) on elements of $D$, giving
us a finite set of facts $D^+$ extending $D$, but with the same
domain as $D$. We check that our guess is consistent with the universal part $\alpha$ and such that equality satisfies the
usual transitivity and congruence rules.

Now consider  $1$-types and $2$-types with an additional predicate Observable.
Based on this extended language, we consider good functions as before,
and define the formulas $\const_1$ and $\const_2$ based on them. $1$-types
with that contain the predicate Observable will be referred to as observable $1$-types.
The restriction of a behavior function to observable $1$-types will be called an \emph{observable behavior}.
Given a structure $M$, an observable one-type $\pi$, and an observable behavior function $g_0$,
we let $M_{\pi, g_0}$ be the elements of $M$ having $1$-type $\pi$ and observable behavior $g_0$, and
we analogously let $D_{\pi, g_0}$ be the elements of $D$ whose $1$-type and behavior in $D^+$ match
$\pi$ and $g_0$.

We declare that all elements in $A$ are in the predicate Observable.
We add additional conjuncts to the formulas $\const_1$ and $\const_2$ stating that for each observable $1$-type $\pi$ and
for each observable behavior function $g_0$, the total sum of the
number of elements with $1$-type $\pi$ and  a behavior function $g$ extending $g_0$ (i.e. the
cardinality of $M_{\pi, g_0}$) is the same as $|D_{\pi, g_0}|$. Here
the cardinality is being counted modulo equalities
of $D^+$.

At this point, our algorithm returns true exactly when the sentence
obtained by existentially-quantifying this extended set of conjuncts is satisfiable in
the integers.
The solving procedure is certainly in $\np$. In fact, since
   the number of variables is fixed, with only the constants varying, it is in $\ptime$ \cite{christos}.

We argue for correctness, focusing on the proof that when the algorithm returns true we
have the desired model.
Assuming the constraints above are satisfied,
we get a graph, and from the graph we get a model $M$. $M$ will clearly satisfy
$\phi$, but its domain does not contain the domain of $D$.
Letting $O$ be the  elements of $M$ satisfying Observable, we know, from the additional constraints
imposed, that
the cardinality of $O$  matches the cardinality of the domain of $D$ modulo the equalities in $D^+$, and
for each observable $1$-type $\pi_o$ and observable behavior $g_0$,  $|M_{\pi,g_0}|=|D_{\pi, g_0}$.

Fix an isomorphism $\lambda$ taking each $M_{\pi, g_0}$ to (equality classes of) $D_{\pi,g_0}$
Create $M'$ by redefining $M$ on $O$ by  connecting pairs $(o_1, o_2)$ via $E$ exactly when $\lambda(o_1), \lambda(o_2)$
ise connected via $E$ in $D^+$. We can thus identify $O$ with $D^+$ modulo equalities in $M'$.

Clearly $M'$ now satisfies $D$.
To see that $M'$ satisfies $\phi$, we simply note that since all of the observable behaviors
are unchanged in moving from an element $e$ in $M$ to the corresponding
element $\lambda(e)$ in $M'$, and every such $e$ modified has an observable type,
it follows that the behavior of every element in  $M$ is unchanged in moving from $M$ to $M'$.
Since the $1$-types are also unchanged, $M'$ satisfies $\phi$.
\end{proof}

Note that the data complexity  result here is best possible, 
since even for $\fotwo$ the data complexity can be $\np$-hard 
\cite{iandatacomplexity}.

%!TEX root = fo2siam.tex

\section{The spectrum problem} \label{sec:spectrum}

As mentioned in the introduction, our Presburger definability result gives additional information
about models of $\fotwopres$ sentences, allowing us to characterize the sets that can occur
as cardinalities of models. Recall from the introduction 
that the spectrum of a sentence $\phi$ in any logic is the set of cardinalities of finite models of $\phi$.
We now use the prior tools to characterize the spectra for  $\fotwopres$ sentences.

\begin{theorem} \label{thm:spectrum} From an $\fotwopres$ sentence $\phi$, we can effectively
construct a Presburger formula $\psi(n)$ such that ${\cal N}_\infty \models \psi(n)$  exactly when $n$ is the size of a finite
structure that satisfies $\phi$, and similarly a formula $\psi_\infty(n)$  such that
 ${\cal N}_\infty \models \psi_\infty(n)$ exactly when $n$ is the size of a finite or countably infinite model of
$\phi$.
\end{theorem}

\begin{proof}
A \emph{type/behavior profile} for a model $\cA$ is the vector of  cardinalities of the sets $A_{\pi,g}$ computed in
$M$, where $\pi$ ranges of $1$-types and $g$ over behavior functions (for a fixed $\phi$).
Recall that in the proof of Theorem \ref{theo:main} we actually showed, in Lemma \ref{lem:correct},  that we can obtain
existential Presburger formulas  which define exactly the vectors
of integers that arise as the type/behavior profiles of models of $\phi$.
The domain of the model can be broken up as a disjoint union of  sets $A_{\pi,g}$, 
and thus its cardinality is a sum of numbers in this vector.  We can thus add  one additional integer variable $x_{\totalcard}$ in $\PREB_{\phi}$, which
will be free, with an additional equation stating that $x_{\totalcard}$ is the sum of all $X_{\pi,g}$'s.
This allows us to conclude definability of the spectrum.
\end{proof}

\section{Related work} \label{sec:related}
The  biregular graph method was introduced and applied to  $\ctwo$ in \cite{KT15,ctworevisited}. 
The case of $1$-color is characterized by a Presburger formula that just expresses the equality of the number of edges calculated from
either side of the bipartite graph. The non-trivial direction of correctness is shown via distributing edges and then merging.
The case of fixed degree and multiple colors is done via an induction, using merging and then swapping to eliminate parallel edges.
The case of unfixed degree is handled using a case analysis depending on whether sizes are ``big enough'', but the approach is different from
the one we apply here based on simple matrices followed by a reduction to non-simple ones.

Note that a more restricted version of the method is used
to prove the decidability of $\fotwo$ extended with two equivalence relations~\cite{KieronskiMPT14}.

This work can be seen as a demonstration of the power of the biregular graph method to get new
decidability results. We make heavy use of both techniques and results in \cite{KT15}, adapting
them to the richer logic. The additional expressiveness of the logic requires  the introduction of  additional inductive arguments to handle the interaction
of ordinary counting quantifiers and modulo counting quantification.

An alternative to the biregular graph method is the machinery developed by Pratt-Hartmann 
for analyzing the decidability
and complexity of $\ctwo$ \cite{ctwobinary,ctworevisited}, its fragments \cite{iancomplexitytwovarguarded},  and its extensions
\cite{ctwocountequiv,ctwoandforests}.  It is clear that the approaches are closely related, despite
the differing terminology.
In \cite{ctworevisited}
binary relationships that are tied to fixed numerical bounds are associated with ``feature functions'', 
while relationships that are not constrained realize ``silent $2$-types''.  
 At this point we cannot provide a more precise mapping, nor can we say whether it would be possible
to extend the approach of  \cite{ctwobinary} to our logic.
An advantage of the biregular graph method  is that it is transparent how to extract more information about the shape of witness
structures. 
While we imagine that results on  spectra of formula can be shown via either method,  with an understanding of biregular graph problems related to a logic in hand,
it  is completely straightforward to draw conclusions about the spectrum.
From an expository point of view, the biregular graph approach has the advantage that
one deals with the combinatorics of the underlying problems with the logic abstracted away early on.
But admittedly,  the current arguments are complex in both approaches.

Characterizing  the spectrum for general first order
formulas is quite a difficult problem, with ties to major open questions in complexity theory \cite{spectrumfifty}.
There are other logics, incomparable
in expressiveness with $\fotwopres$, where periodicity of the spectrum has been proven \cite{gurevichshelah}. The
arguments have a different feel, since in these logics
one can reduce to reasoning about forests.

The paper \cite{bartosztonyfsstcs} shows decidability for a  logic with incomparable expressiveness: the quantification allows a more powerful
quantitative comparison, but must be \emph{guarded} -- restricting the counts  only of  sets of elements that are adjacent to a given element.
Counting extensions of $1$-variable logics are studied in \cite{bartosztcs}.

% conclusion

\section{Conclusion} \label{sec:conc}

We have shown the Presburger definability of the solution set to certain graph problems.
Using this,  we show that can extend the powerful language two-variable logic with counting to 
include ultimately periodic counting quantifiers without sacrificing decidability, and without losing the effective
definability of the spectrum of formulas within Presburger arithmetic.

A number of complexity questions are left open by our work.
We have obtained a $\twonexp$ bound on complexity of deciding satisfiability of the logic.
However the only lower bound we know of is $\nexp$, inherited
from $\fotwo$.

A natural question left open by our work is the connection with other extensions of two-variable logic with counting.
It has been shown that two-variable logic with counting remains decidable in the presence of a linear
order \cite{twovariablecountorder}. It has also been shown that decidability is maintained when one of the relations
is restricted to 
be an equivalence relation \cite{ctwocountequiv}.
One would like to know if there is a common decidable extension of our logic and one of (or ideally, both of)
these logics.

We also leave open a number of other complexity questions for  biregular graph 
analysis problems. In particular, the line
between $\ptime$ and $\np$ for the membership problem of Subsection \ref{subsec:graphs} (with cardinalities
in unary) is open.

\bibliography{fo2siam}
\newpage
\appendix

% Scott normal form

\section{Scott normal form} 
\label{sec:scott-normal-form}

In this appendix we prove that every $\fotwopres$ formula
can be converted into the normal form used in the body of the paper:
\begin{align*}
%\phi :=  
\forall x \forall y\  \alpha(x,y)
\ \wedge\ 
\bigwedge_{i=1}^{k} \forall x \exists^{S_i} y\ \beta_i(x,y) \wedge x \neq y,
\end{align*}
where $\alpha(x,y)$ is a quantifier free formula, each $\beta_i(x,y)$ is an atomic formula and 
each $S_i$ is an u.p.s.
Moreover, the conversion preserves the satisfiability and the spectra of $\fotwopres$ sentences.

We will first give a couple of lemmas.

\begin{lemma}
\label{lem:scott-no-zero}
Let $S\subseteq \bbN_{\infty}$ where $0\notin S$
and $q$ be a unary predicate.
Let $\phi(x,y)$ be a formula with free variables $x$ and $y$.
The sentence $\Psi_1$ that is defined as
\begin{align*}
\Psi_1 &:= \
\forall x \ \big(q(x)\ \to \ \exists^S y\ \phi(x,y)\big)
\end{align*}
is equivalent to the sentence $\Psi_2$ that is defined as
\begin{align*}
\Psi_2 &:= \
\forall x \ \exists^{S\cup\{0\}}y \ \big(q(x)\wedge \phi(x,y)\big)
\ \wedge \
\forall x \ \exists^{\bbN_{\infty}-\{0\}}y \ \big(q(x)\to \phi(x,y)\big).
\end{align*}
\end{lemma}

\begin{proof}
It is worth noting that $q(x)\wedge \phi(x,y)$ is equivalent to
$(q(x)\to \phi(x,y))\wedge (\neg q(x)\to \myF)$.

Let $\cA$ be a structure.
For an element $a\in A$, 
define $W_{a,\phi(x,y)}$ as follows:
\begin{align*}
W_{a,\phi(x,y)} & := \ \{b\in A \mid (\cA,x/a, y/b) \models \phi(x,y)\};
\end{align*}
that is, $W_{a,\phi(x,y)}$ is the set of elements that can be assigned to $y$ 
so that $\phi(x,y)$ holds, when $x$ is assigned with element $a$.

Suppose $\cA \models \Psi_1$.
So, for every $a\in q^{\cA}$, $|W_{a,\phi}|\in S$. Thus we have
\begin{align}
\label{eq:a-in-q}
\cA,x/a  \models  \exists^{S}y\ q(x)\to \phi(x,y)
\quad \text{and} \quad 
\cA,x/a  \models \exists^{S}y\ q(x)\wedge \phi(x,y).
\end{align}
For every $a\notin q^{\cA}$, the following holds:
\begin{align}
\label{eq:a-not-in-q}
\cA,x/a \models \exists^{|A|} y \ q(x)\to \phi(x,y)
\quad \text{and} \quad
\cA,x/a  \models  \exists^{0} y \ q(x)\wedge \phi(x,y).
\end{align} 
Combining (\ref{eq:a-in-q}) and (\ref{eq:a-not-in-q}), we have $\cA\models \Psi_2$.

For the other direction, suppose $\cA\models \Psi_2$.
Since $\cA\models \forall x \ \exists^{S\cup\{0\}}y \ \big(q(x)\wedge \phi(x,y)\big)$,
for every $a\in A$, either $|W_{a,\phi(x,y)}|=0$ or $|W_{a,\phi(x,y)}|\in S$.
Since $\cA\models \forall x \ \exists^{\bbN_{\infty}-\{0\}}y \ (q(x)\to \phi(x,y))$,
the following holds, for every $a\in q^{\cA}$:
$$
|W_{a,\phi(x,y)}|\neq 0.
$$
Thus, for every $a\in q^{\cA}$, $|W_{a,\phi(x,y)}|\in S$.
Therefore, $\cA \models \Psi_1$.
\end{proof}

The next  lemma is proven in a similar manner.
\begin{lemma}
\label{lem:scott-with-zero}
Let $S\subseteq \bbN_{\infty}$ where $0\in S$
and $q$ be a unary predicate.
Let $\phi(x,y)$ is a formula with free variables $x$ and $y$.
The sentence $\Psi_1$ defined as
\begin{align*}
\Psi_1 &:= \
\forall x \ \big(q(x)\ \to \ \exists^S y\ \phi(x,y)\big)
\end{align*}
is equivalent to the sentence $\Psi_2$ defined as
\begin{align*}
\Psi_2 &:= \
\forall x \ \exists^{S}y \ \big(q(x)\wedge \phi(x,y)\big).
\end{align*}
\end{lemma}

Obviously, Lemma~\ref{lem:scott-no-zero} and~\ref{lem:scott-with-zero}
can be modified easily when $q(x)$ is any quantifier-free formula
with free variable $x$.

\subparagraph*{Conversion into ``almost'' Scott normal form}
We will first show how to convert an $\fotwopres$ sentence into an
equisatisfiable sentence in ``almost'' Scott normal form:
\begin{align}
\label{eq:snf-without-inequality}
& \forall x \forall y \ \alpha(x,y)
\ \wedge \ \bigwedge_{i=1}^k
\forall x \exists^{S_i} y \ \beta_i(x,y).
\end{align}
That is, the requirement $x\neq y$ is dropped for $\beta_i(x,y)$ to hold.
In fact,  we get more than equisatisfiability: each model of our sentence can
be expanded to a model of the normal form. This will be important for our result
about the spectrum. In the remainder of this section we omit similar statements for brevity. 

The conversion is a rather standard renaming technique from two-variable logic.
Let $\Psi$ be an $\fotwopres$ sentence.
We first assume that $\Psi$ does not contain any subformula of the form $\forall x\, \phi$,
by rewriting them into the form $\exists^0 x\, \neg \phi$.

Whenever there is a subformula $\psi(x)$ in $\Psi$ of the form
$\exists^S y \ \phi(x,y)$,
where $\phi(x,y)$ is quantifier free and $S$ is a u.p.s., 
we perform a transformation.
Let $q$ be a fresh unary predicate, and replace the subformula $\psi(x)$ in $\Psi$ with atomic $q(x)$,
and add a sentence which states that $q(x)$ is equivalent to $\psi(x)$:
\begin{align*}
& \forall x \ \big(q(x) \ \leftrightarrow \ \psi(x)\big)
\end{align*}
which is equivalent to
\begin{align*}
& \forall x \ \big(q(x) \ \to \ \exists^S y \ \phi(x,y)\big)
\quad \wedge \quad
\forall x \ \big(\neg q(x) \ \to \ \exists^{\bbN_{\infty}-S} y \ \phi(x,y)\big)
\end{align*}
which, in turn, by Lemmas~\ref{lem:scott-no-zero} and~\ref{lem:scott-with-zero},
can be converted into sentences of the form~(\ref{eq:snf-without-inequality}).
We iterate this  procedure until $\Psi$ is in the ``almost'' Scott normal form described
above.

\subparagraph*{Conversion into Scott normal form in~(\ref{eq:snf})}
Now we provide the conversion from ``almost'' Scott normal form into Scott normal form.
Note that
\begin{align*}
\forall x \exists^S y \ \beta(x,y)
\end{align*}
is equivalent to
\begin{align*}
& \forall x \big(\neg \beta(x,x) \ \to\ \exists^{S}y \ \beta(x,y) \wedge  x\neq y\big)
\quad\wedge \quad
\forall x \big(\beta(x,x) \ \to\ \exists^{S-1}y \ \beta(x,y) \wedge  x\neq y\big),
\end{align*}
where $S-1$ denotes the set $\{i-1 \mid i \in S\}$.

Applying Lemmas~\ref{lem:scott-no-zero} and~\ref{lem:scott-with-zero},
a sentence of form~(\ref{eq:snf-without-inequality}) can be converted into
an equisatisfiable sentence of the form:
\begin{align*}
& \forall x \forall y \ \alpha(x,y)
\quad \wedge \quad \bigwedge_{i=1}^k
\forall x \exists^{S_i} y \ \beta_i(x,y) \wedge x\neq y,
\end{align*}
where each $\beta_i(x,y)$ is quantifier free.
To make it into Scott normal form,
we introduce a new predicate $\gamma_i(x,y)$, for each $1\leq i \leq k$,
and rewrite the sentence as follows:
\begin{align*}
& \forall x \forall y \ \Big(\alpha(x,y)\wedge \bigwedge_{i=1}^k\big(\gamma_i(x,y)\leftrightarrow \beta_i(x,y)\big)\Big)
\ \wedge \ \bigwedge_{i=1}^k
\forall x \exists^{S_i} y \ \gamma_i(x,y) \wedge x\neq y.
\end{align*}

The conversion described above takes $O(Cn)$ time where
$n$ is the length of the original $\fotwopres$ sentence 
and the factor $C$ is the complexity of computing the complement $\bbN_{\infty}-S$
of a u.p.s. $S$, which of course, depends on the representation of a u.p.s.
However, we should note that the number of new atomic predicates introduced
is linear in $n$.

%!TEX root = fo2siam.tex

\section{The extension of Section~\ref{sec:1type}, the $1$-color case, to handle infinite graphs}
\label{app:sec:1type}

In this appendix we will extend the formulas in Section ~\ref{sec:1type} to
handle all possible (finite and infinite) sizes of $1$-color $A|B$-biregular graphs.

\begin{lemma}
\label{app:lem:1type-incomplete}
For every $A\in \bbN_{\infty,+p}^{1\times m}$ and $B\in \bbN_{\infty,+p}^{1\times n}$,
there exists (effectively computable) existential 
Presburger formula $\bireg_{A|B}(\vx,\vy)$ such that 
for every $(\vM,\vN)\in \bbNo^m\times \bbNo^n$:
there is an $A|B$-biregular graph with size $\vM|\vN$
if and only if $\bireg_{A|B}(\vM,\vN)$ holds in $\cN_{\infty}$.
\end{lemma}

\subsection{Notation and terminology}
\label{app:subsec:1type-notation}

We regard $\infty$ as a periodic entry, since $\infty$ is considered the same as $\prdp{\infty}$.
Intuitively, the reason is that when a vertex has degree $\infty$,
adding $p$ (or any arbitrary number) of additional new edges adjacent to it 
still make its degree $\infty$.
A periodic entry which is not $\infty$ is called a {\em finite} periodic entry.
We define $\offset(\infty)$ to be $\infty$.

For degree vectors $\va$ and $\vb$ that contain $\infty$ entries,
we write $\delta(\va,\vb)$ to denote the maximal {\em finite} entry in $(\offset(\va),\offset(\vb),p)$.
For example, if $\va=(3,\infty)$ and $\vb=(2^{+5},4)$,
then $\delta(\va,\vb)$ is the maximal finite entry in $(3,\infty,2,4,5)$, which is $5$.
We let $\per(\va)$ to denote the set of indexes $j$ where $a_j$ is a finite periodic entry
and $\inf(\va)$ to denote the sets of indexes $j$ where $a_j=\infty$.
As before, $\nz(\va)$ denotes the set of indexes $j$ where $a_j\neq 0$.

We  redefine the notion of big enough when the degree vectors contain $\infty$ entries.

\begin{definition}
\label{app:def:1type-big-enough} 
Let $\va$ and $\vb$ be degree vectors and 
let $\vM$ and $\vN$ be size vectors with the same length as $\va$ and $\vb$, respectively.
We say that {\em $\vM|\vN$ is big enough for $\va|\vb$},
if each of the following holds:
\begin{enumerate}[(a)]
\item 
$\max(\normt{\vM}_{\nz(\va)},\normt{\vN}_{\nz(\vb)})\geq 2\delta(\va,\vb)^2+1$,
\item 
$\normt{\vM}_{\per(\va)}=0$ or $\geq \delta(\va,\vb)^2+1$,
\item 
$\normt{\vM}_{\inf(\va)}=0$ or $\geq \delta(\va,\vb)$,
\item 
$\normt{\vN}_{\per(\vb)}=0$ or $\geq \delta(\va,\vb)^2+1$,
\item 
$\normt{\vN}_{\inf(\vb)}=0$ or $\geq \delta(\va,\vb)$.
\end{enumerate}
\end{definition}

\subsection{The formula for the case of big enough sizes}
\label{app:subsec:1type-incomplete-big-enough}

We consider four scenarios for the sizes $\vM|\vN$ of $\va|\vb$-biregular graphs:
\begin{enumerate}[(GS1)]
\item $\normt{\vM}_{\per(\va)}=\normt{\vM}_{\inf(\va)}=\normt{\vN}_{\per(\vb)}=\normt{\vN}_{\inf(\vb)}=0$
 (i.e., there are only vertices with fixed degree),
\item $\normt{\vM}_{\per(\va)}\neq 0$ and $\normt{\vM}_{\inf(\va)}=\normt{\vN}_{\per(\vb)}=\normt{\vN}_{\inf(\vb)}=0$
(i.e., there are vertices with finite periodic degrees on exactly one side, but no vertex with $\infty$ degree),
\item $\normt{\vM}_{\per(\va)},\normt{\vN}_{\per(\vb)}\neq 0$ and $\normt{\vM}_{\inf(\va)}=\normt{\vN}_{\inf(\vb)}=0$ 
(i.e., there are vertices with finite periodic degrees on both sides, but no vertex with $\infty$ degree),
\item $\normt{\vM}_{\inf(\va)}\neq 0$ or $\normt{\vN}_{\inf(\vb)}\neq 0$
(i.e., there are vertices with infinite degree).
\end{enumerate}

The rest of this section is devoted to the formulas for each of the cases above.
Scenarios (GS1)--(GS3) are similar to (S1)--(S3) in Section~\ref{sec:1type}.
For completeness, we present the formulas for them, but without the correctness proofs.
Scenario (GS4) is a new scenario that is not present in the finite biregular graph case.

\paragraph{The formula and argument for scenario (GS1)}
Consider the formula \\
$\text{Gen-}\psi^{1}_{\va|\vb}(\vx,\vy)$:
\begin{align*}
 \offset(\va)\cdot\vx = \offset(\vb)\cdot\vy\ 
\wedge \ 
\normt{\vx}_{\per(\va)}=\normt{\vx}_{\inf(\va)}=\normt{\vy}_{\per(\vb)}=\normt{\vy}_{\inf(\vb)}=0.
\end{align*}
The last conjunct simply states that (GS1) holds.

\begin{lemma}
\label{app:lem:1type-s1}
For every pair of degree vectors $\va,\vb$ and
for every pair of size vectors $\vM,\vN$ such that $\vM|\vN$ is big enough for $\va|\vb$,
there is an $\va|\vb$-biregular graph with size $\vM|\vN$ where (GS1) holds
if and only if $\text{Gen-}\psi^{1}_{\va|\vb}(\vM,\vN)$ holds.
\end{lemma}
\begin{proof}
The proof is similar to Lemma~\ref{lem:1type-s1}.
\end{proof}

\paragraph{The formula and argument for scenario (GS2)}
Recall that (GS2) states that ``there are vertices with finite periodic degrees on exactly one side, but no vertex with $\infty$ degree''.
By symmetry, we may assume that the vertices with finite periodic degrees are on the left.
Consider the formula $\text{Gen-}\psi^{2}_{\va|\vb}(\vx,\vy)$:
\begin{align*}
& \exists z \big(  z\neq \infty \ \wedge \
\offset(\va)\cdot\vx \ + \ pz \ =\ \offset(\vb)\cdot \vy\big)
\\
\wedge \ & \normt{\vx}_{\per(\va)}\neq 0 \ \wedge \ 
\normt{\vx}_{\inf(\va)}=\normt{\vy}_{\per(\vb)}=\normt{\vy}_{\inf(\vb)}=0.
\end{align*}
The last two conjuncts state that (GS2) holds.
\begin{lemma}
\label{app:lem:1type-s2}
For every pair of degree vectors $\va,\vb$ and
for every pair of size vectors $\vM,\vN$ such that $\vM|\vN$ is big enough for $\va|\vb$, 
there is an $\va|\vb$-biregular graph with size $\vM|\vN$ where (GS2) holds
if and only if $\text{Gen-}\psi^{2}_{\va|\vb}(\vM,\vN)$ holds.
\end{lemma}
\begin{proof}
The proof is similar to Lemma~\ref{lem:1type-s2}.
\end{proof}

\paragraph{The formula and argument for scenario (GS3)}
Recall that (GS3) states that ``there are vertices with finite periodic degrees on both sides, but no vertex with $\infty$ degree''.
Consider the formula $\text{Gen-}\psi^{3}_{\va|\vb}(\vx,\vy)$:
\begin{align*}
& 
\exists z_1,z_2 \big( z_1\neq \infty \ \wedge \ z_2\neq \infty\ \wedge \
\offset(\va)\cdot\vx  +  pz_1  = \offset(\vb)\cdot \vy  +  pz_2 \big)
\\
\wedge \ &  
\normt{\vx}_{\per(\va)}\neq 0 \ \wedge \ \normt{\vy}_{\per(\vb)}\neq 0 \ \wedge \ \normt{\vx}_{\inf(\va)}=\normt{\vy}_{\inf(\vb)}=0.
\end{align*}

\begin{lemma}
\label{app:lem:1type-s3}
For every pair of degree vectors $\va,\vb$ and
for every pair of size vectors $\vM,\vN$ such that $\vM|\vN$ is big enough for $\va|\vb$, 
there is an $\va|\vb$-biregular graph with size $\vM|\vN$ where (GS3) holds
if and only if $\text{Gen-}\psi^{3}_{\va|\vb}(\vM,\vN)$ holds.
\end{lemma}

\begin{proof}
The proof is similar to Lemma~\ref{lem:1type-s3}.
\end{proof}

\paragraph{The formula and argument for scenario (GS4)}
Recall that (GS4) states that ``there are vertices with infinite degree''.
Consider the formula $\text{Gen-}\psi^{4}_{\va|\vb}(\vx,\vy)$:
\begin{align}
\label{app:eq:1type-s4}
& \big(\normt{\vx}_{\inf(\va)}\neq 0 \ \vee \ \normt{\vy}_{\inf(\vb)}\neq 0 \big)
\\
\label{app:eq:1type-s4-a}
\wedge \ & 
\big(\normt{\vx}_{\inf(\va)}\neq 0 \ \to \ \normt{\vy}_{\nz(\vb)}=\infty\big)
\ \wedge \
\big(\normt{\vy}_{\inf(\vb)}\neq 0 \ \to \ \normt{\vx}_{\nz(\va)}=\infty\big).
\end{align}

Notice that, unlike the previous scenarios, this formula does not involve any edge-counting on the finite entries. Instead, we will make use
of the fact that infinite-degree vertices give us a lot of flexibility in forming graphs that meet our specification.
\begin{lemma}
\label{app:lem:1type-s4}
For every pair of degree vectors $\va,\vb$ and
for every pair of size vectors $\vM,\vN$ such that $\vM|\vN$ is big enough for $\va|\vb$, 
there is an $\va|\vb$-biregular graph with size $\vM|\vN$ where (GS4) holds
if and only if $\text{Gen-}\psi^{4}_{\va|\vb}(\vM,\vN)$ holds in $\cN_{\infty}$.
\end{lemma}
\begin{proof}
Let $\va,\vb$  be degree vectors and
$\vM|\vN$ be big enough for $\va|\vb$.
For the ``only if'' direction, suppose there is  an
$\va|\vb$-biregular graph with size $\vM|\vN$ where (GS4) holds.
Thus, $\normt{\vM}_{\inf(\va)}\neq 0$ or $\normt{\vN}_{\inf(\va)}\neq 0$.
If there is a vertex on the left with $\infty$ degree,
there are infinitely many vertices with non-zero degree on the right.
Symmetrically, if there is a vertex on the right with $\infty$ degree,
there are infinitely many vertices with non-zero degree on the left.
Therefore, $\text{Gen-}\psi^4_{\va|\vb}(\vM,\vN)$ holds.

We now prove the ``if'' direction, 
assuming $\text{Gen-}\psi^{4}_{\va|\vb}(\vM,\vN)$ 
and  constructing an $\va|\vb$-biregular graph $G=(U,V,E)$ with size $\vM|\vN$.
Let $m$ be the length of $\va$ and $n$ be the length of $\vb$. 
First, we pick pairwise disjoint sets $U_1,\ldots,U_m$, where each $|U_j|=M_j$
and pairwise disjoint sets $V_1,\ldots,V_n$, where each $|V_j|=N_j$.
We define the set of vertices of our graph as  $U=U_1\cup\cdots\cup U_m$ and $V=V_1\cup\cdots\cup V_n$.

We know  $\normt{\vM}_{\inf(\va)}\neq 0$ or $\normt{\vN}_{\inf(\va)}\neq 0$.
Hence we have at least one of 
$\normt{\vM}_{\inf(\va)}\neq 0$ and $\normt{\vN}_{\nz(\vb)}=\infty$ or
$\normt{\vN}_{\inf(\vb)}\neq 0$ and $\normt{\vM}_{\nz(\va)}=\infty$..

We can break this down further into three cases:
\begin{enumerate}[(a)]
\item 
$U$ is infinite and $V$ is finite.
\item 
$U$ is finite and $V$ is infinite.
\item 
$U$ is infinite and $V$ is infinite.
\end{enumerate}

{\bf (Case~(a))}
We perform the following two steps.
\begin{itemize}
\item
Step~1: Making the degrees of vertices in $V$ correct.

Let $k$ be any  index such that $U_k$ is infinite.
For every $j\in [n]$,
for every vertex $v\in V_j$, we ensure that its degree is $\offset(b_j)$
by connecting $v$ with some ``non-adjacent'' vertices from the set $U_k$ --- that is,
vertices in $U_k$ that are not yet adjacent to any vertices in $V$.
Since $U_k$ has an infinite supply of vertices, 
there are always such ``non-adjacent'' vertices for each vertex $v$.
The purpose of  picking ``non-adjacent'' vertices is that,
after this step, every vertex in $U$ has degree either $0$ or $1$.

\item 
Step~2: Making the degrees of vertices in $U$ correct.

Let $V^{\infty}= \bigcup_{j\in\inf(\vb)} V_j$, i.e.,
the set of vertices in $V$ that are supposed to have $\infty$ degree.
Since $\normt{\vN}_{\inf(\vb)}\neq 0$, the set $V^{\infty}$ is not empty.
Moreover, since $\vM|\vN$ is big enough, the cardinality $|V^{\infty}|\geq \delta(\va,\vb)$.

Note that the degree of every vertex in $U$ is at most $1$.
For every $j\in [m]$, for every vertex $u\in U_j$,
we ensure its degree is $\offset(b_j)$
by connecting $u$ with some vertices in $V^{\infty}$.
This is possible since $\offset(b_j)\leq \delta(\va,\vb)$ for every $j\in [m]$.
\end{itemize}

{\bf (Case~(b))} is clearly symmetric to case~(a).

{\bf (Case~(c))} 
We enumerate the elements $u_1,u_2,\ldots$ and $v_1,v_2,\ldots$ in $U$ and $V$, respectively.
We  construct an $\va|\vb$-biregular graph $G=(U,V,E)$ by iterating through all $\ell=1,2,\ldots$,
where on each iteration $\ell$, we do the following.
\begin{itemize}
\item
We make the degree of $u_{\ell}$ ``correct'' in the sense that
if $j$ is the index where $u_\ell \in U_j$,
we make its degree become $\offset(a_j)$.
\item 
We make the degree of $v_{\ell}$ ``correct'' in the sense that
if $j$ is the index where $v_\ell \in U_j$,
we make its degree become $\offset(b_j)$.
\end{itemize}
At the same time, while making the degrees of $u_{\ell}$ and $v_{\ell}$,
we ensure that:
\begin{enumerate}
\item
The degrees of the vertices $u_1,\ldots,u_{\ell-1}$ do not change
and are already ``correct'' in the sense that
for every $u \in \{u_1,\ldots,u_{\ell-1}\}$,
if $j$ is the index where $u \in U_j$,
its degree is already $\offset(a_j)$.
\item 
The degrees of the vertices $v_1,\ldots,v_{\ell-1}$ do not change
and are already ``correct'' in the sense that
for every $v \in \{v_1,\ldots,v_{\ell-1}\}$,
if $j$ is the index where $v \in V_j$,
its degree is already $\offset(b_j)$.
\item 
The degree of each vertex in $\{u_{\ell+1},u_{\ell+2},\ldots\}\cup \{v_{\ell+1},v_{\ell+2},\ldots\}$
is $0$ or $1$.
\item 
There are infinitely many vertices in $\{u_{\ell+1},u_{\ell+2},\ldots\}$ with degree $0$.
\item 
There are infinitely many vertices in $\{v_{\ell+1},v_{\ell+2},\ldots\}$ with degree $0$.
\end{enumerate}
Since $U$ (resp. $V$) is countable, every vertex $u\in U$ (resp. $v\in V$)
has a finite index $\ell$ such that $u_\ell=u$ (resp. $v_\ell=v$).
After the $\ell^{\text{th}}$ iteration
the degree of $u_\ell$ (resp. $v_{\ell}$) does not change any more.
Thus, as the iteration index $\ell$ goes to $\infty$, 
the degree of every vertex is ``correct'' and we obtain $\va|\vb$-biregular graph $G$.

We now describe how to make the degree of $u_{\ell}$ correct.
At the beginning of the $\myth{\ell}$ iteration,
the degree of $u_{\ell}$ is either $0$ or $1$.
We make it correct by picking some zero degree vertices in $\{v_{\ell+1},v_{\ell+2}\}$ 
and connecting them to $u_\ell$.
Such zero degree vertices exist and there are infinitely many of them.
Of course, if the degree of $u_{\ell}$ is supposed to $1$, 
we do not need to pick any vertices.
If the degree of $u_{\ell}$ is supposed to be $\infty$,
we also make sure that there are still infinitely many zero degree vertices left in $U$.
Observe also that the degrees of the vertices $u_1,\ldots,u_{\ell-1},v_1,\ldots,v_{\ell-1}$ do not change.
Making the degree of $v_{\ell}$ correct can be done symmetrically.
\end{proof}

As in Subsection~\ref{subsubsec:1type-incomplete-big-enough}, 
to capture all possible sizes of $\va|\vb$-biregular graphs
there are only some fixed $k$ cases to consider, where each case is either equal to or symmetric to one of the scenarios (GS1)--(GS4).
We can enumerate all the formulas $\varphi_1(\vx,\vy),\ldots,\varphi_k(\vx,\vy)$ that deal with each of the cases
and define the formula $\text{Gen-}\psi_{\va|\vb}(\vx,\vy)$:
\begin{align}
\label{app:eq:1type-big-enough}
& \bigvee_{i=1}^k \varphi_i(\vx,\vy).
\end{align}
Combining Lemmas~\ref{app:lem:1type-s1}--\ref{app:lem:1type-s4},
$\text{Gen-}\psi_{\va|\vb}(\vx,\vy)$ captures precisely all the  big enough 
sizes $\vM|\vN$ of an  $\va|\vb$-biregular graph.

\subsection{The formula for the case of not big enough sizes: fixed size encoding}
\label{app:subsec:1type-incomplete-not-big-enough}

To capture the not big enough sizes,
we use the same ``fixed size encoding'' technique as in Subsection~\ref{subsubsec:1type-incomplete-not-big-enough}.
Note that ``not big enough'' sizes mean that one of the conditions (a) -- (e) is violated.
So, either we have restricted the total size of the graphs (when condition (a) is violated)
or at least one of the norms $\normt{\vM}_{\per(\va)}$, $\normt{\vM}_{\inf(\va)}$,
$\normt{\vN}_{\per(\vb)}$, $\normt{\vN}_{\inf(\vb)}$ are fixed to some number.
Since we can deal with the first option by enumeration, we focus on the second.
 The idea is that we can use fixed size enumeration as in Subsection~\ref{subsubsec:1type-incomplete-not-big-enough},
with additional minor extensions to handle vertices with $\infty$ degree. 
To illustrate, we will show the construction when the first two of the four norms above are fixed to some number,  while the second two still satisfy the 
corresponding condition in ``big enough''.
This corresponds 
to (b) and (c) being violated, while (a), (d) and (e) hold, in the definition of big enough.
In this case  we will have vertices with periodic and infinite degrees on the left hand side,
but not too many.

Let $\va,\vb$ be degree vectors.
We will give the formula $\text{Gen-}\phi^{r_1,r_2}_{\va|\vb}(\vx,\vy)$
to capture the sizes $\vM|\vN$ of all possible $\va|\vb$-biregular graphs
where each of the following holds.
\begin{itemize}
\item
$\normt{\vM}_{\nz(\va)}-r_1-r_2 \geq 2\delta(\va,\vb)^2+1$.
\item 
$\normt{\vM}_{\per(\va)}=r_1 \leq \delta(\va,\vb)^2$.
\item 
$\normt{\vM}_{\inf(\va)}=r_2 \leq \delta(\va,\vb)-1$.
\item 
$\normt{\vN}_{\per(\vb)}=0$ or $\geq \delta(\va,\vb)^2+1$.
\item 
$\normt{\vN}_{\inf(\vb)}=0$ or $\geq \delta(\va,\vb)$.
\end{itemize}
If the first bullet item does not hold, the number of edges is at most $3\delta(\va,\vb)^2+\delta(\va,\vb)$,
and the sizes of all these graphs can simply enumerated.
The formula is defined inductively on $r_1+r_2$ with the base case $r_1+r_2=0$.
Note that when $r_1+r_2=0$, $\normt{\vM}_{\per(\va)} =  \normt{\vM}_{\inf(\va)}=0$,
which means (b) and (c) are no longer violated.

For an integer $r_1,r_2\geq 0$, we define the formula
$\text{Gen-}\phi^{r_1,r_2}_{\va|\vb}(\vx,\vy)$ as follows.
\begin{itemize}
\item
When $r_1=r_2=0$, 
$\text{Gen-}\phi^{r_1,r_2}_{\va|\vb}(\vx,\vy)$
is defined as in Lemma~\ref{lem:1type-big-enough}.
\item
When $r_1\geq 1$, let 
\begin{align*}
&
\phi^{r_1-1,r_2}_{\va|\vb}(\vx,\vy) \ := \
\exists s \exists \vz_0 \exists \vz_1 
\bigvee_{i\in \per(\va)}
\left(
\begin{array}{l}
x_i \neq 0 \ \wedge \
\vz_0+\vz_1 =\vy 
\\
\wedge \ s\neq \infty 
\\ \wedge \ \normt{\vz_1}_{\nz(\vb)} = \offset(a_i) + ps
\\ 
\wedge \ \phi^{r_1-1,r_2}_{\va|(\vb,\vb-\vone)}(\vx-\textbf{e}_i,\vz_0,\vz_1)
\end{array}
\right),
\end{align*}
where $\textbf{e}_i$ is the unit vector where the $\myth{i}$ component is $1$.
and the lengths of $\vz_0$ and $\vz_1$ are the same as $\vy$
The vector subtraction $\vb-\vone$ is defined as in Subsection~\ref{subsubsec:1type-incomplete-not-big-enough}
extended with $\infty-1=\infty$.

\item
When $r_2\geq 1$, let 
\begin{align*}
&
\phi^{r_1,r_2-1}_{\va|\vb}(\vx,\vy) \ := \
\exists s \exists \vz_0 \exists \vz_1 
\bigvee_{i\in \inf(\va)}
\left(
\begin{array}{l}
x_i \neq 0 \ \wedge \
\vz_0+\vz_1 =\vy 
\\ \wedge \  \normt{\vz_1}_{\nz(\vb)} = \infty
\\ 
\wedge \ \phi^{r_1,r_2-1}_{\va|(\vb,\vb-\vone)}(\vx-\textbf{e}_i,\vz_0,\vz_1)
\end{array}
\right),
\end{align*}
where $\textbf{e}_i$ is as in the previous case and
the lengths of $\vz_0$ and $\vz_1$ are the same as $\vy$.
The vector subtraction $\vb-\vone$ is defined as in the previous case.
\end{itemize}

\begin{lemma}
\label{app:lem:1type-not-big-enough}
For every integer $r_1,r_2\geq 0$,
for every pair of degree vectors $\va,\vb$,
for every pair of size vectors $\vM,\vN$ such that:
\begin{itemize}
\item
$\normt{\vM}_{\nz(\va)}-r_1-r_2 \geq 2\delta(\va,\vb)^2+1$,
\item 
$\normt{\vM}_{\per(\va)}=r_1$,
\item 
$\normt{\vM}_{\inf(\va)}=r_2$,
\item 
$\normt{\vN}_{\per(\vb)}=0$ or $\geq \delta(\va,\vb)^2+1$,
\item 
$\normt{\vN}_{\inf(\vb)}=0$ or $\geq \delta(\va,\vb)$.
\end{itemize}
there is an $\va|\vb$-biregular graph
with size $\vM|\vN$ if and only if
$\phi^{r_1,r_2}_{\va|\vb}(\vM,\vN)$ holds.
\end{lemma}
\begin{proof}
The proof is by induction on $r_1+r_2$
and is a routine adaptation of Lemma~\ref{lem:1type-not-big-enough}.
\end{proof}

As mentioned  in Subsection~\ref{subsubsec:1type-incomplete-not-big-enough},
the remaining  not big enough cases can be captured by formulas similar to the  one given above.

\subsection{The proof in the $1$-color case for regular digraphs}
\label{app:subsec:1type-digraphs}

Recall that we define digraphs so that they have no self-loops. 
Similar to what was done in Subsection ~\ref{subsec:1type-digraphs},
a digraph can be viewed as a bipartite graph by splitting every vertex $u$ into two vertices,
where one is adjacent to all the incoming edges, and the other to all the outgoing edges.
Thus, $\va|\vb$-biregular digraphs with size $\vM$ can be characterized as $\va|\vb$-biregular graphs with size $\vM|\vM$.
The construction of the formula for all the sizes of $\va|\vb$-regular digraphs
can be done by a routine adaptation of the one in
Section ~\ref{app:subsec:1type-incomplete-big-enough} and~\ref{app:subsec:1type-incomplete-not-big-enough}.

%!TEX root = fo2siam.tex

\newcommand{\fnzr}{{F_{\text{nz-right}}}}
\newcommand{\fnzl}{{F_{\text{nz-left}}}}

\section{The extension of Section~\ref{sec:proofsimple} (simple multi-color graphs) to infinite graphs}
\label{app:sec:proofsimple}

We will extend the formulas in Section ~\ref{sec:proofsimple}
to accommodate  all possible (finite and infinite) sizes of $A|B$-biregular graphs,
where $A$ and $B$ are simple degree matrices, as stated formally in Lemma~\ref{app:lem:simple-bireg}.

\begin{lemma}
\label{app:lem:simple-bireg}
For every pair of simple matrices $A\in \bbN_{\infty,+p}^{t\times m}$ and $B\in \bbN_{\infty,+p}^{t\times m}$,
there exists an (effectively computable) existential 
Presburger formula $\bireg_{A|B}(\vx,\vy)$ such that 
for every pair of size vectors $\vM\in \bbNo^m$ and $\vN \in \bbNo^n$,
there is an $A|B$-biregular graph with size $\vM|\vN$
if and only if $\bireg_{A|B}(\vM,\vN)$ holds in $\cN_{\infty}$.
\end{lemma}

This section is organized as follows.
To deal with an  $\infty$ entry, we need some new notation, introduced in Section~\ref{app:subsec:proof-notations}.
Section \ref{app:subsec:simple-bireg-big-enough} contains the construction of the formula for extra big enough sizes ---
a generalization of the ones in Section~\ref{subsec:simple-big-enough-fixed-entries} and Section ~\ref{subsec:simple-bireg-extra-big-enough}. Here there
is a new case which is specific to an $\infty$ entry.
We discuss the formula for the sizes that are not extra big enough --- where no new ideas are needed --- in Section \ref{app:subsec:simple-not-big-enough}.

\subsection{Notation and terminology}
\label{app:subsec:proof-notations}

Let $A$ be a degree matrix with $t$ rows and $m$ columns.
For non-empty subsets $R\subseteq [t]$,
we write $A_{R,*}$ to denote the matrix obtained by keeping only the rows
with indices in $R$, with no column being omitted.
Likewise, for $J\subseteq [m]$, $A_{*,J}$ denotes 
the matrix obtained by keeping only the columns with indices in $J$, with no rows being omitted.

Recall that we regard an $\infty$ entry as a periodic entry.
The {\em finite offset} of $A$, denoted by $\finoffset(A)$
is the matrix obtained by replacing every $\infty$ entry in $\offset(A)$ with $0$.
That is, in $\finoffset(A)$ we are concerned only with the non-$\infty$ entries.
For example, if $A = \begin{pmatrix}\prdp{2} & \infty \\ 0 &\prdp{3}\end{pmatrix}$,
then $\offset(A)=\begin{pmatrix}2 & \infty \\ 0 & 3\end{pmatrix}$
and $\finoffset(A)= \begin{pmatrix}2 & 0 \\ 0 & 3\end{pmatrix}$.
Obviously, if $A$ does not contain any $\infty$ entry, $\offset(A)=\finoffset(A)$.

If $A$ and $B$ contain periodic or infinite entries,
$\delta(A,B)$ denotes\\
 $\max(\norm{\finoffset(A)},\norm{\finoffset(B)},p)$.
We also have to modify the notion of ``extra big enough'' in Section ~\ref{subsec:proof-notations}
in order to take the $\infty$ entries into account.

\begin{definition}
\label{app:def:big-enough}
Let $A$ and $B$ be simple degree matrices with $t$ rows.
Let $\vM$ and $\vN$ be size vectors where $\vM|\vN$ is appropriate for $A|B$.
We say that {\em $\vM|\vN$ is extra big enough for $A|B$}, if for every $i\in [t]$:
\begin{enumerate}[(a)]
\item 
$\max(\normt{\vM}_{\nz(A_{i,*})},\normt{\vN}_{\nz(B_{i,*})})\geq 8t^2\delta(A,B)^4+1$,
\item 
$\normt{\vM}_{\per(A_{i,*})}=0$ or $\geq \delta(A,B)^2+1$,
\item 
$\normt{\vM}_{\inf(A_{i,*})}=0$ or $\geq t\delta(A,B)$,
\item 
$\normt{\vN}_{\per(B_{i,*})}=0$ or $\geq \delta(A,B)^2+1$,
\item 
$\normt{\vN}_{\inf(B_{i,*})}=0$ or $\geq t\delta(A,B)$.
\end{enumerate}
\end{definition}

\subsection{Proof of Lemma~\ref{app:lem:simple-bireg} for big enough sizes}
\label{app:subsec:simple-bireg-big-enough}

We divide the proof into three  scenarios.
\begin{enumerate}[(GM1)]
\item 
$\normt{\vM}_{\nz(A_{i,*})},\normt{\vN}_{\nz(B_{i,*})}\neq \infty$, for every $i\in [t]$ 
(i.e., the number of vertices with non-zero degree is finite, which means that the degree of every vertex must be finite).
\item 
$\normt{\vM}_{\inf(A_{i,*})}=\normt{\vN}_{\inf(B_{i,*})}=0$, for every $i\in [t]$ 
(i.e., the degree of every vertex is finite, but there maybe infinitely many vertices).
\item 
(the general case).
\end{enumerate}
Note that (GM1) is strictly subsumed by (GM2), since in (GM2) 
the number of vertices with non-zero (finite) degree may be infinite.
(GM2) is clearly strictly subsumed by (GM3).
The rest of this section is devoted to the formulas for each of the scenarios above.
The formula for Scenario (GM${i}$) will be used by the formula for Scenario (GM${i+1}$).
Only (GM3), which deals with the possibility of vertices with infinite degree, will require substantial new work.

\paragraph{The formula and argument for scenario (GM1)}
For simple degree matrices $A$ and $B$ with $t$ rows,
consider the formula $\text{Gen-}\Psi^{1}_{A|B}(\vx,\vy)$ given by:
\begin{align}
\label{app:eq:simple-s2}
&
\exists z_{1,1}\cdots \exists z_{1,t}\
\exists z_{2,1}\cdots \exists z_{2,t}
\\
\nonumber
& 
\quad
\offset(A)\cdot \vx^{\tT} + \begin{pmatrix}\alpha_1 p z_{1,1}\\ \vdots \\ \alpha_t p z_{1,t}\end{pmatrix}
\ = \
\offset(B)\cdot \vy^{\tT} + \begin{pmatrix}\beta_1 p z_{2,1}\\ \vdots \\ \beta_t p z_{2,t}\end{pmatrix}
\\
\nonumber
& 
\quad
\wedge \
\bigwedge_{i\in [t]}\ z_{1,i} \neq \infty\  \wedge\  z_{2,i} \neq \infty \ \wedge \ \normt{\vx}_{\nz(A_{i,*})}\neq \infty
\ \wedge \ \normt{\vy}_{\nz(B_{i,*})}\neq \infty
\\
\nonumber
& 
\quad
\wedge \
\bigwedge_{i\in [t]} \normt{\vx}_{\inf(A_{i,*})}= \normt{\vy}_{\inf(B_{i,*})}=0
\end{align}
where $\alpha_i = 1$ if row $i$ in $A$ consists of periodic entries and is $0$ otherwise,.
Similarly $\beta_i = 1$ if row $i$ in $B$ consists of periodic entries and is $0$ otherwise.

We claim that $\text{Gen-}\Psi^1_{A|B}(\vx,\vy)$ captures all possible big enough sizes $\vM|\vN$ of $A|B$-biregular graphs
where (GM1) holds, as stated in Lemma~\ref{app:lem:simple-s2}.

\begin{lemma}
\label{app:lem:simple-s2}
For each pair of simple degree matrices $A$ and $B$, and
each pair of size vectors $\vM,\vN$ such that $\vM|\vN$ is big enough for $A|B$, 
there is a $A|B$-biregular graph with size $\vM|\vN$ where (GM1) holds if and only if
$\Psi_{A|B}^1(\vM,\vN)$ holds.
\end{lemma}
\begin{proof}
This is similar to Lemma~\ref{lem:simple-s2}.
\end{proof}

\paragraph{The formula and argument for scenario (GM2)}
Recall that (GM2) states that there is no vertex with $\infty$ degree,
but the number of edges may be infinite.
The main idea for this scenario is to partition the edge colors into two kinds,
depending on whether the number of edges is finite or infinite.

For simple matrices $A$ and $B$ with $t$ rows, for a subset $R\subseteq [t]$, 
consider the formula $\text{Gen-}\Psi^{2,R}_{A|B}(\vx,\vy)$ given by:
\begin{align}
\label{app:eq:simple-s3-aux-a}
\text{Gen-}\Psi^1_{A_{R,*}|B_{R,*}}(\vx,\vy) 
\ 
\wedge \ &
\bigwedge_{i\in[t]} \normt{\vx}_{\inf(A_{i,*})}= \normt{\vy}_{\inf(B_{i,*})}=0
\\
\label{app:eq:simple-s3-aux-b}
\wedge \ &
\bigwedge_{i\notin R} \normt{\vx}_{\nz(A_{i,*})}= \normt{\vy}_{\nz(B_{i,*})}=\infty,
\end{align}
where
$\text{Gen-}\Psi^1_{A_{R,*}|B_{R,*}}(\vx,\vy)$ is as defined in \eqref{app:eq:simple-s2}.
Recall that $A_{R,*}$ is the matrix obtained from $A$ by omitting all the rows not in $R$.
When $R=[t]$, the formula $\text{Gen-}\Psi^{2,R}_{A|B}(\vx,\vy)$ is the same as 
$\text{Gen-}\Psi^1_{A|B}(\vx,\vy)$ defined for Scenario (GM1).

Intuitively, 
$\text{Gen-}\Psi^{2,R}_{A|B}(\vx,\vy)$ captures all the big enough sizes of $A|B$-biregular graphs
where (GM2) holds and $i\in R$ if and only if the number of $E_i$-edges is finite.
This is stated formally as Lemma~\ref{app:lem:simple-s3-aux}.

\begin{lemma}
\label{app:lem:simple-s3-aux}
For every pair of simple matrices $A$ and $B$ with $t$ rows,
for every $R\subseteq [t]$ and
for every $\vM|\vN$ big enough for $A|B$, 
the following are equivalent.
\begin{enumerate}
\item 
There is an $A|B$-biregular graph $G=(U,V,E_1,\ldots,E_t)$ with size $\vM|\vN$
where (GM2) holds and $R=\{i : |E_i| \neq \infty\}$. 
\item 
$\text{Gen-}\Psi^{2,R}_{A|B}(\vM,\vN)$ holds in $\cN_{\infty}$.
\end{enumerate}
\end{lemma}
\begin{proof}
Let $A$ and $B$ be simple matrices with $t$ rows
and $R\subseteq [t]$.
Let $\vM|\vN$ be big enough for $A|B$. 

We first prove ``1 implies 2''.
Suppose $G=(U,V,E_1,\ldots,E_t)$ is a $A|B$-biregular graph with size $\vM|\vN$,
where (GM2) holds and $R =\{i : |E_i|\neq \infty\}$.
For every $i\in R$, since $|E_i|\neq \infty$,
both $\normt{\vM}_{\nz(A_{i,*})}$ and $\normt{\vN}_{\nz(B_{i,*})}$ are not $\infty$.
This means that $G$ is an $A_{R,*}|B_{R,*}$-biregular graph where (GM1) holds.
By Lemma \ref{lem:simple-s2}, 
%the following holds:
%\begin{align}
%\label{app:eq:simple-s3-1}
$\text{Gen-}\Psi^2_{A_{R,*}|B_{R,*}}(\vM,\vN)$ holds.
%\end{align}

Since (GM2) holds in $G$, there is no vertex with $\infty$ degree.
Thus,  we have:
\begin{align*}
\bigwedge_{i\in[t]} \normt{\vM}_{\inf(A_{i,*})}= \normt{\vN}_{\inf(B_{i,*})}=0
\end{align*}
Since every vertex has finite $E_i$-degree and $|E_i|=\infty$, for every $i\notin R$,
the following conjunction holds:
\begin{align*}
\bigwedge_{i\notin R} \normt{\vM}_{\nz(A_{i,*})}=\normt{\vN}_{\nz(B_{i,*})}=\infty
\end{align*}
Combining all the assertions above, 
we see that
$\text{Gen-}\Psi^{2,R}_{A|B}(\vM,\vN)$ holds.

Now we prove ``2 implies 1''.
Suppose $\text{Gen-}\Psi^{2,R}_{A|B}(\vM,\vN)$ holds.
For simplicity, let $R=[\ell]$.
Similar to Lemma~\ref{app:lem:simple-s2},
since $\normt{\vx}_{\inf(A_{i,*})}= \normt{\vy}_{\inf(B_{i,*})}=0$
for every $i\in [t]$,
we may assume that $A$ and $B$ do not contain $\infty$ entry.

Since $\text{Gen-}\Psi^1_{A_{R,*}|B_{R,*}}(\vM,\vN)$ holds,
by Lemma~\ref{app:lem:simple-s2}
there is an $A_{R,*}|B_{R,*}$-biregular graph $G_0=(U,V,E_1,\ldots,E_{\ell})$
with size $\vM|\vN$ where (GM2) holds.
Moreover, by \eqref{app:eq:simple-s3-aux-b}, we have:
$$
\bigwedge_{i\notin R} \normt{\vM}_{\nz(A_{i,*})}= \normt{\vN}_{\nz(B_{i,*})}=\infty.
$$
Hence for every $i\notin R$ we have:
\begin{align*}
\offset(A_{i,*})\cdot \vM \ = \ \offset(B_{i,*})\cdot \vN.
\end{align*}
By Lemma~\ref{app:lem:1type-s1},
there is an $\offset(A_{i,*})|\offset(B_{i,*})$-biregular graph $G_i=(U,V,E_i)$
with size $\vM|\vN$, for every color $i\notin R$.

Consider the graph $G=(U,V,E_1,\ldots,E_t)$ with size $\vM|\vN$.
This graph $G$ is almost $A|B$-biregular except that
there may be an edge $(u,v)\in E_{i_1}\cap E_{i_2}$, for some $i_1\neq i_2$.
We use the ``edge swapping'' as in Lemma~\ref{lem:simple-s1}, 
to remove all such parallel edges.

Note that since $G_0$ is already $A_{R,*}|B_{R,*}$-biregular,
at least one of $i_1,i_2$ is not in $R$.
Suppose $i_1\notin R$.
Since $|E_{i_1}|=\infty$ and the degree of every vertex in $G$ is finite,
there is an $E_{i_1}$-edge $(w,w')$ that is not incident to any of the neighbors of $u$ and $v$.
We can perform edge swapping (see also Fig.~\ref{fig:simple-s1-swap})
so that $(u,v)$ is no longer an $E_{i_1}$-edge, without effecting the degree of each vertex.
We perform edge swapping until there is no more parallel edge.
This completes the proof of Lemma~\ref{app:lem:simple-s3-aux}. 
\end{proof}

To wrap up scenario (GM2),
we define formula $\text{Gen-}\Psi^2_{A|B}(\vx,\vy)$ for simple matrices $A$ and $B$ as
\begin{align}
\label{app:eq:simple-s3}
& \bigvee_{R\subseteq [t]}\quad  \text{Gen-}\Psi^{2,R}_{A|B}(\vx,\vy),
\end{align}
where each $\text{Gen-}\Psi^{2,R}_{A|B}(\vx,\vy)$ is defined in \eqref{app:eq:simple-s3-aux-a}--\eqref{app:eq:simple-s3-aux-b}.
This formula $\text{Gen-}\Psi^2_{A|B}(\vx,\vy)$ captures precisely all the extra big enough sizes of $A|B$-biregular graphs
where (GM2) holds, as stated formally as Lemma~\ref{app:lem:simple-s3}.

\begin{lemma}
\label{app:lem:simple-s3}
For each pair of simple matrices $A$ and $B$, and
for each pair of size vectors $\vM,\vN$ such that $\vM|\vN$ is big enough for $A|B$, 
there is an $A|B$-biregular graph with size $\vM|\vN$
where (GM2) holds iff
$\text{Gen-}\Psi^2_{A|B}(\vM,\vN)$ holds in $\cN_{\infty}$.
\end{lemma}
\begin{proof}
Let $A$ and $B$ be simple matrices $A$ and $B$ with $t$ rows
and $\vM|\vN$ be big enough for $A|B$.

Suppose there is an $A|B$-biregular graph $G=(U,V,E_1,\ldots,E_t)$ with size $\vM|\vN$
where (GM2) holds.
Let $R =\{i : |E_i|\neq \infty\}$.
By Lemma~\ref{app:lem:simple-s3-aux}, $\text{Gen-}\Psi^{2,R}_{A|B}(\vM,\vN)$ holds.
Thus, $\text{Gen-}\Psi^2_{A|B}(\vM,\vN)$ holds.

For the converse direction, suppose $\text{Gen-}\Psi^2_{A|B}(\vM,\vN)$ holds.
Let $R$ be such that $\text{Gen-}\Psi^{2,R}_{A|B}(\vM,\vN)$ holds.
By Lemma~\ref{app:lem:simple-s3-aux}, there is an $A|B$-biregular graph
$G=(U,V,E_1,\ldots,E_t)$ with size $\vM|\vN$
where (GM2) holds and $R =\{i : |E_i|\neq \infty\}$.
\end{proof}

\paragraph{The formula and argument for scenario (GM3)}
Recall that (GM3) is the general case where there may be vertices with infinite degree.
The main idea here is to partition the edge colors $E_i$ into two kinds, but this
time depending on whether there are vertices with infinite $E_i$-degree.
Let $A$ and $B$ be simple matrices with $t$ rows
and $R\subseteq [t]$.
Consider the formula $\text{Gen-}\Psi^{3,R}_{A|B}(\vx,\vy)$ given by:
\begin{align}
\label{app:eq:simple-s4-aux}
 \text{Gen-}\Psi^2_{A_{R,*}|B_{R,*}}(\vx,\vy)\  \wedge \  & \bigwedge_{i\notin R} 
\big(\normt{\vx}_{\inf(A_{i,*})}\neq 0 \ \vee \ \normt{\vy}_{\inf(B_{i,*})} \neq 0\big)
\\
\label{app:eq:simple-s4-aux1}
\wedge \ &
\bigwedge_{i\notin R} 
\normt{\vx}_{\inf(A_{i,*})}\neq 0 \ \to \ \normt{\vy}_{\nz(B_{i,*})} =\infty 
\\
\label{app:eq:simple-s4-aux2}
\wedge \ &
\bigwedge_{i\notin R} 
\normt{\vy}_{\inf(B_{i,*})}\neq 0 \ \to \ \normt{\vx}_{\nz(A_{i,*})} =\infty,
\end{align}
where $\text{Gen-}\Psi^2_{A_{R,*}|B_{R,*}}(\vx,\vy)$ is as defined in \eqref{app:eq:simple-s3}.
When $R=[t]$, $\text{Gen-}\Psi^{3,R}_{A|B}(\vx,\vy)$ 
is the same as $\text{Gen-}\Psi^2_{A|B}(\vx,\vy)$ defined for Scenario (GM2).

Intuitively, $\text{Gen-}\Psi^{3,R}_{A|B}(\vx,\vy)$ captures all the big enough sizes of $A|B$-biregular graphs
where $R$ is the set of colors $i$ such that every vertex has finite $E_i$-degree.
We state this formally in Lemma~\ref{app:lem:simple-s4-aux}.

\begin{lemma}
\label{app:lem:simple-s4-aux}
For every pair of simple matrices $A$ and $B$ with $t$ rows,
for every $R\subseteq [t]$, and
for every pair of size vectors $\vM,\vN$ such that $\vM|\vN$ is extra big enough for $A|B$, 
the following are equivalent.
\begin{enumerate}
\item 
There is an $A|B$-biregular graph $G=(U,V,E_1,\ldots,E_t)$ with size $\vM|\vN$
where $R=\{i : \text{every vertex in $G$ has finite $E_i$-degree}\}$. 
\item 
$\text{Gen-}\Psi^{3,R}_{A|B}(\vM,\vN)$ holds in $\cN_{\infty}$.
\end{enumerate}
\end{lemma}
\begin{proof}
Let $A$ and $B$ be simple matrices with $t$ rows
and let $R\subseteq [t]$.
Let $\vM|\vN$ be extra big enough for $A|B$. 

We first prove ``1 implies 2''.
Let $G$ be an $A|B$-biregular graph with size $\vM|\vN$
and $R=\{i : \text{every vertex in $G$ has finite $E_i$-degree}\}$. 
Thus, $G$ is also an $A_{R,*}|B_{R,*}$-biregular graph where (GM2) holds.
By Lemma~\ref{app:lem:simple-s3}, $\text{Gen-}\Psi^{2}_{A_{R,*}|B_{R,*}}(\vM,\vN)$ holds.

By definition of $R$, for every $i\notin R$, we have:
$$
\normt{\vM}_{\inf(A_{i,*})}\neq 0\quad\text{or}\quad\normt{\vN}_{\inf(B_{i,*})}\neq 0.
$$
If there is a vertex on one side with $E_i$-degree $\infty$,
then there must be infinitely many vertices on the other side with non-zero $E_i$-degree.
In other words, for every $i\notin R$,
\begin{align*}
\normt{\vM}_{\inf(A_{i,*})}\neq 0 \ \to \ \normt{\vN}_{\nz(B_{i,*})} =\infty, \text{ and } 
\\
\normt{\vN}_{\inf(B_{i,*})}\neq 0 \ \to \ \normt{\vM}_{\nz(A_{i,*})} =\infty.
\end{align*}
Therefore, the formula $\text{Gen-}\Psi^{3,R}_{A|B}(\vM,\vN)$ holds.

We now prove ``2 implies 1''.
Suppose $\text{Gen-}\Psi^{3,R}_{A|B}(\vM,\vN)$ holds.
For simplicity, we may assume $R=[\ell]$.
Since $\text{Gen-}\Psi^2_{A_{R,*}|B_{R,*}}(\vM,\vN)$ holds, by Lemma~\ref{app:lem:simple-s3},
there is an $A_{R,*}|B_{R,*}$-biregular graph $G_0=(U,V,E_1,\ldots,E_{\ell})$
with size $\vM|\vN$ where (GM2) holds.
We will show how to extend $G_0$ to an $A|B$-biregular graph $G$ with size $\vM|\vN$.
Without loss of generality, we may assume that $R\neq [t]$.
Otherwise, $G_0$ is already $A|B$-biregular and we are done.

In the following, we fix $U=U_1\uplus\cdots\uplus U_m$ and $V=V_1\uplus \cdots \uplus V_n$
as the witness partition of $A_{R,*}|B_{R,*}$-biregularity of the graph $G_0$.
We will add new edges to make $G_0$ into an  $A|B$-biregular graph.
In the following, when we say {\em ``we make the $E_i$-degree of a vertex $u\in U$ correct''},
we mean that we will add $E_{i}$-edges adjacent to $u$ so that its $E_i$-degree becomes $\offset(A_{i,j})$
where $j$ is the index such that $u \in U_j$.
Similarly for vertex $v\in V$.

We can break this down further into three cases -- analogous to Scenario~(GS4):
\begin{enumerate}[(a)]
\item 
$U$ is infinite and $V$ is finite.
\item 
$U$ is finite and $V$ is infinite.
\item 
$U$ is infinite and $V$ is infinite.
\end{enumerate}
Case~(b) is symmetric to case~(a), so we will only consider case~(a) and~(c).

{\bf (Case~(a))}
This case is a straightforward generalization of case~(a) in (GS4).
We perform the following two steps.
\begin{itemize}
\item
Step~1: Making the $E_i$-degrees of vertices in $V$ correct for every $i\notin R$.

For each $i\notin R$,
let $k_i$ be any index such that $U_{k_i}$ is infinite.
For every $j\in [n]$,
for every vertex $v\in V_j$, we ensure that its degree is $\offset(B_{i,j})$
by connecting $v$ with some ``non-adjacent'' vertices from the set $U_{k_i}$ --- that is,
vertices in $U_{k_i}$ that are not yet adjacent to any vertices in $V$.
Since $U_{k_i}$ has an infinite supply of vertices, 
there are always such ``non-adjacent'' vertices for each vertex $v$.
The purpose of  picking ``non-adjacent'' vertices is that,
after this step, for every vertex $u \in U$ the sum $\sum_{i\notin R}$ (the $E_i$-degree of $u$)
is either $0$ or $1$.

\item 
Step~2: Making the degrees of vertices in $U$ correct.

For each $i\notin R$, let $V^{i,\infty}= \bigcup_{j\in\inf(B_{i,*})} V_j$, i.e.,
the set of vertices in $V$ that are supposed to have $\infty$ $E_i$-degree.
Since $\normt{\vN}_{\inf(B_{i,*})}\neq 0$, the set $V^{i,\infty}$ is not empty.
Moreover, since $\vM|\vN$ is extra big enough, the cardinality $|V^{i,\infty}|\geq t\delta(A,B)$.

Note that for each $i\notin R$, 
the sum $\sum_{i\notin R}$ (the $E_i$-degree of $u$) of every vertex $u \in U$
is at most $1$.
For every $j\in [m]$, for every vertex $u\in U_j$,
we ensure its $E_i$-degree is $\offset(B_{i,j})$,
by connecting $u$ with some vertices in $V^{i,\infty}$, for every $i\notin R$.
This is possible since $\offset(B_{i,j})\leq \delta(A,B)$ for every $j\in [m]$.
\end{itemize}

{\bf (Case~(c))} 
%The following notations will be useful.
For each $i\in [t]$, define the sets.
\begin{align*}
U^{i,\nz}:= & \bigcup_{j\in \nz(A_{i,*})} U_j
\qquad\qquad\text{and}\qquad\qquad
V^{i,\nz}:= & \bigcup_{j\in \nz(B_{i,*})} V_j,
\\
U^{i,\infty}:=& \bigcup_{j\in \inf(A_{i,*})} U_j
\qquad\qquad\text{and}\qquad\qquad
V^{i,\infty}:=& \bigcup_{j\in \inf(B_{i,*})} V_j;
\end{align*}
Informally, $U^{i,\nz}$ and $V^{i,\nz}$ are the sets of vertices in $U$ and $V$ whose $E_i$-degree is supposed to be non-zero,
while $U^{i,\infty}$ and $V^{i,\infty}$ are the vertices in $U$ and $V$ whose $E_i$-degree is supposed to be $\infty$.

Note that for each $i\notin R$,
there are supposed to be vertices with infinite $E_i$-degree, which gives us a lot of flexibility in constructing the $E_i$-edges.
We can use a technique similar to the one in Scenario~(GS4) from the previous appendix, which handled the case where some vertex has
infinite degree in the single-color case.
Note that for each $i\notin R$, we have one of:
\begin{enumerate}[(a)]
\item 
$U^{i,\nz}$ is infinite and $V^{i,\nz}$ is finite.
\item 
$U^{i,\nz}$ is finite and $V^{i,\nz}$ is infinite.
\item 
$U^{i,\nz}$ is infinite and $V^{i,\nz}$ is infinite.
\end{enumerate}
That is, (a) holds for a subset of the  colors,
(b) holds for another subset,
and (c) holds for the remaining colors.
Constructing the $E_i$-edges by itself for each $i\notin R$ is comparatively easy, as shown in Scenario (GS4).
The main technical difficulty occurs when we try to make sure that the sets of constructed edges are still pairwise disjoint.
Note also that here we do not have any guarantees about how big the partitions and degrees are in $G_0$. This limits us in
using techniques such as ``edge swapping'', which rely on having sufficiently many available edges.

In the following paragraphs, we will illustrate the new obstacle that arises.
Suppose there are $i_1,i_2\notin R$, where $i_1\neq i_2$, such that:
\begin{itemize}
\item
$U^{i_1,\nz}$ is finite and $U^{i_2,\nz}$ is infinite.
\item
$V^{i_1,\nz}$ is infinite and $V^{i_2,\nz}$ is finite.
\end{itemize}
Since $i_1\notin R$,
the set $U^{i_1,\nz}$ contains vertices that are supposed to have infinite $E_{i_1}$-degree.
Similarly, since $i_2\notin R$,
the set $V^{i_2,\nz}$ contain vertices that are supposed to have infinite $E_{i_2}$-degree.
Assume, for convenience, that $U^{i_1,\nz}\subseteq U^{i_2,\nz}$
and $V^{i_2,\nz}\subseteq V^{i_1,\nz}$.
See Figure~\ref{fig:infinity-GM4} for an illustration.

\begin{figure}
\begin{center}

\begin{tikzpicture}
\draw (0,0) ellipse (1cm and 2cm);
\node  at (-1,1.8) {$U^{i_2,\nz}$}; 

\draw[gray!40,fill=gray!40] (0,-0.5) ellipse (.6cm and 1cm);
\node at (0,-0.5) {$U^{i_1,\nz}$};

\draw[->,red] (0,-0.8) to [bend right] (0,-2.5);
\node at (0,-2.7) {\footnotesize $U^{i_1,\nz}$ is finite and};
\node at (0,-3.1) {\footnotesize some vertices have $\infty$ $E_{i_1}$-degree};
%\node at (0,-3.5) {\footnotesize have $\infty$ $E_i$-degree};

\draw (7,0) ellipse (1cm and 2cm);
\node  at (8,1.8) {$V^{i_1,\nz}$}; 

\draw[gray!40,fill=gray!40] (7,-0.5) ellipse (.6cm and 1cm);
\node at (7,-0.5) {$V^{i_2,\nz}$};

\draw[->,red] (7,-0.8) to [bend left] (7,-2.5);
\node at (7,-2.7) {\footnotesize $V^{i_2,\nz}$ is finite and};
\node at (7,-3.1) {\footnotesize some vertices have $\infty$ $E_{i_2}$-degree};

\end{tikzpicture}

\end{center}
\label{fig:infinity-GM4}
\caption{An illustration of the sets $U^{i_1,\nz}$, $U^{i_2,\nz}$, $V^{i_1,\nz}$ and $V^{i_2,\nz}$.
When constructing the $E_{i_1}$-edges, we exploit the infinite $E_{i_1}$-degree vertices
in $U^{i_1,\nz}$.
Similarly, when constructing the $E_{i_2}$-edges, we exploit the infinite $E_{i_2}$-degree vertices
in $V^{i_2,\nz}$.
However, we have to make sure to avoid the possibility that 
every vertex in $U^{i_1,\nz}$ is already adjacent to every vertex in $V^{i_2,\nz}$ via $E_{i_1}$-edges,
thus, leaving ``no room'' to connect them via $E_{i_2}$-edges.}
\end{figure}

If we construct the $E_{i_1}$-edges as in Scenario (GS4),
by connecting the vertices in $V^{i_2,\nz}$ with the vertices in $U^{i_1,\nz}$ with $E_{i_1}$-edges,
there is a possibility that every vertex in $U^{i_1,\nz}$ is adjacent to every vertex in $V^{i_2,\nz}$ via $E_{i_1}$-edges.
Thus, when we want to construct the $E_{i_2}$-edges,
we can no longer connect the vertices in $U^{i_1,\nz}$ with the vertices in $V^{i_2,\nz}$ with $E_{i_2}$-edges, 
but the vertices in $V^{i_2,\nz}$ are the only vertices in $G$ that are supposed to have non-zero $E_{i_2}$-degree.
In other words, there is no more ``room'' to construct the $E_{i_2}$-edges.
This issue will be circumvented by partitioning $U^{i_1,\infty}= X_0\uplus X_1$
and $V^{i_2,\infty} = Y_0\cup Y_1 $ and
 constructing the $E_{i_1}$-edges so that:
\begin{itemize}
\item
Vertices in $X_0$ are connected by $E_{i_1}$-edges only to  vertices in $Y_0$.
\item 
Vertices in $X_1$ are connected by $E_{i_1}$-edges only to  vertices in $Y_1$.
\end{itemize}
Then when we construct the $E_{i_2}$-edges,
we will connect the vertices in $X_0$ with the vertices in $Y_1$
and the vertices in $X_1$ with the vertices in $Y_0$.
The rest of the proof is devoted to the details of the construction.

Due to the technical difficulty described above, 
the following two sets of colors $\fnzl,\fnzr\subseteq [t]$ will need some special care.
\begin{align*}
\fnzl := &\ \{i\notin R : \normt{\vM}_{\nz(A_{i,*})}\ \text{is finite}\},
\\
\fnzr := &\ \{i\notin R : \normt{\vN}_{\nz(B_{i,*})}\ \text{is finite}\}.
\end{align*}
Intuitively, the set $\fnzl$ is the set of color $i\notin R$
where there will only be {\bf F}initely many vertices with {\bf n}on-{\bf z}ero $E_i$-degree on the {\bf left} hand side.
The set $\fnzr$ has the same intuitive meaning w.r.t to the vertices on the right hand side.

We argue that $\fnzl$ and $\fnzr$ are disjoint.
For every $i\notin R$, at least one of $\normt{\vM}_{\nz(A_{i,*})}$ or
$\normt{\vN}_{\nz(B_{i,*})}$ is infinite.
Moreover, since $\text{Gen-}\Psi^{3,R}_{A|B}(\vM,\vN)$ holds, we have:
\begin{itemize}
\item
$\normt{\vN}_{\inf(B_{i,*})}=0$, $\normt{\vM}_{\inf(A_{i,*})}\neq 0$
and $\normt{\vN}_{\nz(B_{i,*})}=\infty$, for every $i\in \fnzl$.

This is because for every $i\in \fnzl$, $\normt{\vM}_{\nz(A_{i,*})}$ is finite.
Thus, the $E_i$-degree of every vertex in $V$ must be finite,
i.e., $\normt{\vN}_{\inf(B_{i,*})}=0$.
Since $i\notin R$, this means there are vertices on the left with $\infty$ $E_i$-degree,
i.e., $\normt{\vM}_{\inf(A_{i,*})}\neq 0$.
Therefore, the number of vertices on the right with non-zero $E_i$-degree must be infinite,
i.e., $\normt{\vN}_{\nz(B_{i,*})}=\infty$.
\item 
Similarly, $\normt{\vM}_{\inf(A_{i,*})}=0$, 
$\normt{\vN}_{\inf(B_{i,*})}\neq 0$ and 
$\normt{\vM}_{\nz(A_{i,*})}=\infty$, for every $i\in \fnzr$.
\end{itemize}
Therefore, $\fnzl$ and $\fnzr$ are disjoint.

Define  the sets:
\begin{align*}
U^{\fnzl,\infty} :=  \bigcup_{i\in \fnzl} U^{i,\infty}
\qquad \text{and}\qquad
V^{\fnzr,\infty} :=  \bigcup_{i\in \fnzr} V^{i,\infty}.
\end{align*}
Note that for every $i\in \fnzl$,  the set $U^{i,\nz}$ is finite.
Since $U^{i,\infty}\subseteq U^{i,\nz}$, the set $U^{i,\infty}$ is finite
and hence so is the set $U^{\fnzl,\infty}$.
By analogous reasoning, $V^{\fnzr,\infty}$ is also finite.
Because $\vM|\vN$ is extra big enough for $A|B$,
for every $i\in \fnzl$, $|U^{i,\infty}|\geq t\delta(A,B)$ holds.
Similarly, for every $i\in \fnzr$,  $|V^{i,\infty}|\geq t\delta(A,B)$.

The claim below is the formalization of the partition $X_0\uplus X_1$
and $Y_0\uplus Y_1$ described above.
\begin{claim}
\label{app:cl:simple-s4}
Suppose $\fnzl,\fnzr\neq \emptyset$. Then:
\begin{itemize}
\item
There is a partition $X_0 \uplus X_1$ of $U^{\fnzl,\infty}$ such that
for every $i\in \fnzl$, both the sets
$X_0\cap U^{i,\infty}$ and
$X_1\cap U^{i,\infty}$ contain at least $\delta(A,B)$ vertices.
\item 
There is a partition $Y_0 \uplus Y_1$ of $V^{\fnzr,\infty}$ such that
for every $i\in \fnzr$, both 
$Y_0\cap V^{i,\infty}$ and
$Y_1\cap V^{i,\infty}$ contain at least $\delta(A,B)$ vertices.
 \end{itemize}
\end{claim}

As explained above,
the main difficulty in constructing the $E_i$-edges for color $i\in \fnzl$
is that $|U^{i,\infty}|$ is finite but $|V^{i,\nz}|$ is infinite,
and for color $i\in \fnzr$,  $|V^{i,\infty}|$ is finite but $|U^{i,\nz}|$ is infinite.
The claim implies that there are sets $X_0,X_1,Y_0,Y_1$ -- each with enough vertices -- allowing us
to construct the $E_i$-edges for color $i\in \fnzl$ as follows.
\begin{itemize}
\item
To make every vertex in $Y_0$ have the correct $E_i$-degree,
we connect it by $E_i$-edges only to the vertices in $X_0\cap U^{i,\infty}$.
\item
To make every vertex in $Y_1$ have the correct $E_i$-degree,
we connect it by $E_i$-edges only to the vertices in $X_1\cap U^{i,\infty}$.
\end{itemize}
Note that for any color $i\in \fnzl$, the set $U^{i,\nz}$ is finite, 
hence the degree of each vertex in $V^{i,\nz}$ must be finite, and bounded by $\delta(A,B)$.
Since for every color $i\in \fnzl$,
the cardinalities of $X_0\cap U^{i,\infty}$ and
$X_1\cap U^{i,\infty}$ are at least $\delta(A,B)$,
there are ``enough'' vertices to connect vertices in $Y_0$ only with the vertices in $X_0\cap U^{i,\infty}$
and vertices in $Y_1$ only with vertices in $X_1\cap U^{i,\infty}$ via $E_{i}$-edges.
After this construction, vertices in $X_0$  are not connected via $E_i$ to any vertex in $Y_1$, for any color $i\in \fnzl$.
Likewise, vertices in $X_1$ are not connected via $E_i$ to any vertex in $Y_0$, for any color $i\in \fnzl$.
This leaves some ``room'' for the construction of $E_i$-edges for each color $i'\in \fnzr$
where we connect vertices in $X_0$ only to vertices in $Y_1\cap V^{i',\infty}$
and vertices in $X_1$ only to vertices in $Y_0\cap V^{i',\infty}$.
See Figure~\ref{fig:X0-X1-Y0-Y1} for an illustration.

\begin{proof} (of Claim~\ref{app:cl:simple-s4})
We prove the first item. The second one is similar.
Initially, $X_0=X_1=\emptyset$.
To achieve the desired property,
we will add vertices to $X_0$ and $X_1$ 
by iterating on every $i\in \fnzl$.
On each iteration, we add at most $\delta(A,B)$ vertices to  $X_0$ and $X_1$.

Suppose we are now iterating on some $i\in \fnzl$.
There are $4$ cases:
\begin{itemize}
\item
Case 1: 
$|X_0\cap U^{i,\infty}|\geq \delta(A,B)$ and
$|X_1\cap U^{i,\infty}|\geq \delta(A,B)$.
In this case, we do nothing and move on to the next $i\in \fnzl$.
\item
Case 2: 
$|X_0\cap U^{i,\infty}|< \delta(A,B)$ and
$|X_1\cap U^{i,\infty}|< \delta(A,B)$.
Observe that $U^{i,\infty}$ contains $t\delta(A,B)\geq 2\delta(A,B)$ vertices.
Thus, we can add some vertices from $U^{i,\infty}$ to $X_0$ and $X_1$
so that $X_0$ and $X_1$ are still disjoint and 
\begin{align*}
|X_0\cap U^{i,\infty}| \ = \
|X_1\cap U^{i,\infty}|\ = \ \delta(A,B).
\end{align*}

\item
Case 3: 
$|X_0\cap U^{i,\infty}|< \delta(A,B)$ and
$|X_1\cap U^{i,\infty}|\geq \delta(A,B)$.
Here, we see that
\begin{align*}
|U^{i,\infty}|\ \geq \ t\delta(A,B)
\ > \ |\fnzl|\delta(A,B) \ > \ (|\fnzl|-1)\delta(A,B) \ \geq \ |X_1|.
\end{align*}
Thus, $U^{i,\infty}$ contains at least $\delta(A,B)$ vertices 
that are not yet in $X_0\cup X_1$.
We add some of these vertices into $X_0$ so that
$|X_0\cap U^{i,\infty}|= \delta(A,B)$.

\item
Case 4: 
$|X_0\cap U^{i,\infty}|\geq  \delta(A,B)$ and
$|X_1\cap U^{i,\infty}|< \delta(A,B)$.
This case is symmetric to Case 3.
\end{itemize}
\end{proof}

Now we are ready to extend  $G_0$ to an $A|B$-biregular graph $G$.
Recall that $G_0$ is an $A_{R,*}|B_{R,*}$-biregular graph where $R\neq [t]$.
By \eqref{app:eq:simple-s4-aux1} and \eqref{app:eq:simple-s4-aux2},
at least one of $U$ and $V$ is infinite.
We will show how to construct the $E_i$-edges in $G$ for each $i\in [t]- R$.
The construction will yield an $A|B$-biregular graph $G$ with witness partition 
$U=U_1\uplus\cdots\uplus U_m$ and $V=V_1\uplus \cdots\uplus V_n$ that has the properties:
\begin{enumerate}[(P1)]
\item 
For every $j\in [n]$ where $|V_j|=\infty$, for every vertex $u\in U$, 
there are infinitely many vertices in $V_j$ that are {\em not} adjacent to $u$
via any $E_i$-edge.
\item 
Similarly, for every $j\in [m]$ where $|U_j|=\infty$, 
for every vertex $v\in V$, 
there are infinitely many vertices in $U_j$ that are {\em not} adjacent to $v$
via any $E_i$-edge.
\end{enumerate}
An infinite set $V_j$/$U_j$ that satisfies (P1)/(P2) is called a {\em strongly infinite} set in $G$.
An infinite $A|B$-biregular graph $G$ is called \emph{strongly partitioned},
if it has a witness partition whose infinite sets are all strongly infinite.

Note that $G_0$ is an infinite $A_{R,*}|B_{R,*}$-biregular graph
and every vertex has a finite degree.
Thus $G_0$ is already strongly partitioned.

There are two cases to consider, depending on
whether both $\fnzl$ and $\fnzr$ are not empty, or
at least one of $\fnzl$ and $\fnzr$ is empty.
We first consider the case when both $\fnzl$ and $\fnzr$ are not empty.

Let $X_0\uplus X_1$ and $Y_0\uplus Y_1$ be the partition of $U^{\fnzl,\infty}$ and $V^{\fnzr,\infty}$
in Claim~\ref{app:cl:simple-s4}.
There are three steps.
\begin{description}
\item[Step~1:]
Construct the $E_i$-edges for each color $i\in \fnzl$,
similarly to Lemma~\ref{app:lem:1type-s4}.
\item[Step~2:]
Construct the $E_i$-edges for each color $i\in \fnzr$ in a manner symmetric to Step~1.
\item[Step~3:]
Construct the $E_i$-edges for each color $i\notin R\cup \fnzl\cup \fnzr$.
\end{description}
We detail each of these steps in the next paragraphs.

\bigskip

\noindent{\bf Step~1:} Make the $E_i$-degree of every vertex correct for every $i\in \fnzl$.
This step is divided into three substeps.
The first two are similar to case~(a) in Lemma~\ref{app:lem:1type-s4}
and the third is needed to leave enough ``room'' for the construction of the edges of colors in $\fnzr$.
\begin{enumerate}[(a)]
\item
Make the $E_i$-degree of every vertex in $U^{i,\nz}$ correct for every $i\in \fnzl$.

For every $u\in U^{i,\nz}$, we ensure that its degree is correct
by connecting it via $E_i$-edges with some ``non-adjacent'' vertices from the set $V^{i,\nz} - V^{\fnzr,\infty}$---that is,
vertices that are not yet adjacent to $u$ via any $E_i$-edges where $i\in \fnzl\cup R$.

Note that $U^{i,\nz}$ is finite.
Since $V^{i,\nz}-V^{\fnzr,\infty}$ is infinite and $G_0$ is strongly partitioned,
there is an infinite supply of vertices.
So such ``non-adjacent'' vertices always exist for every vertex $u\in U^{i,\nz}$.
We also make sure that when we add the new edges,
there are still infinitely many vertices in each $V_j$ that are not yet adjacent to $u$,
for every $j\in [m]$ where $V_j$ is infinite (which is possible since $V_j$ is infinite).
Thus, the graph stays strongly partitioned.

After this step, the degree of every vertex in $V$ increases by at most $1$.
That is, $\sum_{i\in \fnzl} \deg_{E_i}(v)$ is either $0$ or $1$,
for every $v\in V-V^{\fnzr,\infty}$.
Note also that the degrees of vertices in $V^{\fnzr,\infty}$ do not increase.

\item
Make the $E_i$-degree of every vertex in $V^{i,\nz}- V^{\fnzr,\infty}$ correct for every $i\in \fnzl$.

Since $U^{i,\nz}$ is finite,
the $E_i$-degree of every vertex in $V^{i,\nz}$ is supposed to be finite.
Due to the size being extra big enough, $U^{i,\infty}$ contains at least $\delta(A,B)$
vertices.
So, for every vertex $v\in V^{i,\nz}- V^{\fnzr,\infty}$,
we can add ``new'' $E_i$-edges to make its $E_i$-degree correct by connecting it via $E_i$-edges with vertices in $U^{i,\infty}$,
for every $i\in \fnzl$.
Note that by definition, for every $i\in \fnzl$, vertices in $U^{i,\infty}$ have $\infty$ $E_i$-degrees.
So the new $E_i$-edges in this step will violate their degree requirement.

\item
Make the $E_i$-degree of every vertex in $V^{\fnzr,\infty}$ correct for every $i\in \fnzl$.

Here it is useful to recall that for every $i\in \fnzl$, every vertex in $V^{\fnzr,\infty}$
is supposed to have finite $E_i$-degree since $U^{i,\nz}$ is finite.
This step is similar to step~(b),
except that we connect via $E_i$-edges the vertices in $Y_0$ to some vertices in $X_0$,
and the vertices in $Y_1$ to some vertices in $X_1$, for every $i\in \fnzl$.
Since $X_0\cap U^{i,\infty}$ and $X_1\cap U^{i,\infty}$ contains at least $\delta(A,B)$ vertices,
there are enough vertices in $X_0\cap U^{i,\infty}$ and $X_1\cap U^{i,\infty}$
that we may connect each $v\in V^{\fnzr,\infty}$ with to make the $E_i$-degree of $v$ correct,
for every $i\in \fnzl$.

After this step,
vertices in $X_0$ are not adjacent via $E_i$-edges to vertices in $Y_1$, for every $i\in \fnzl$.
Similarly,
vertices in $X_1$ are not adjacent via $E_i$-edges to vertices in $Y_0$, for every $i\in \fnzl$.
This observation will be important in the next step.
See Figure~\ref{fig:X0-X1-Y0-Y1} for an illustration.
\end{enumerate}

\bigskip

\noindent {\bf Step~2:} Make the $E_i$-degree of every vertex correct for every $i\in \fnzr$.
This step consists of three steps (2.a)--(2.c)
which are symmetric to Steps~(1.a)--(1.c), where the role of $U^{i,\nz}$ is replaced with $V^{i,\nz}$,
$U^{i,\infty}$ with $V^{i,\infty}$ and $V^{\fnzr,\infty}$ with $U^{\fnzl,\infty}$.
The difference is in Step (1.c).
To make the $E_i$-degree of vertices in $U^{\fnzl,\infty}$ correct, for every $i\in \fnzr$,
we connect the vertices in $X_0$ to some vertices in $Y_1$,
and the vertices in $X_1$ to some vertices in $Y_0$.
Here it is important that 
vertices in $X_0$ are not adjacent to vertices in $Y_1$ via $E_i$-edges for any $i\in \fnzl$.
Since $Y_1\cap V^{i,\infty}$ contains at least $\delta(A,B)$ vertices,
there are still enough vertices in $Y_1$ that can be connected to each $u\in X_0$ to make 
the $E_i$-degree of $u$ correct, for every $i\in \fnzr$.

\bigskip
\noindent {\bf Step~3:} Make the $E_i$-degree of every vertex correct for every $i\notin R\cup \fnzl\cup \fnzr$. 
This step is similar to case~(c) in  Lemma~\ref{app:lem:1type-s4}.
Let $u_1,u_2,\ldots$ and $v_1,v_2,\ldots$ be an enumeration of the vertices in $U$ and $V$.
After Step~2, the graph $G$ is still strongly partitioned.
In the following we fix a color $i \notin R\cup \fnzl\cup \fnzr$.
We make the $E_i$-degree of each vertex $u_\ell$ and $v_\ell$ correct, where $\ell$ ranges from $1$ to $\infty$. 
We work by induction on $\ell$, where the inductive invariant is that
after the $\ell^{th}$ iteration, the $E_i$-degrees of $u_l,v_1,\ldots,u_\ell,v_\ell$ are already correct.
The process is as follows:
\begin{itemize}
\item
We pick some vertices in $V^{i,\nz}$ that are not yet adjacent to $u_\ell$ via any $E$-edges.
We call these vertices the ``non-adjacent'' vertices
and we pick some of them and connect them to $u_\ell$ via $E_i$-edges to make the $E_i$-degree of $u_\ell$ correct. 
For this purpose, we can choose any vertices are not $v_1,\ldots,v_\ell$ and are not adjacent to any of $u_1,\ldots,u_\ell$.
Such vertices always exist, since the graph $G$ is strongly partitioned.

Note that if the $E_i$-degree of $u_\ell$ is supposed to be infinite,
we have to pick infinitely many ``non-adjacent'' vertices.
So when we pick these vertices,
we also make sure that there are still infinitely many vertices in each $V_j$ that are still not adjacent to all 
$u_1,\ldots,u_\ell$, for every $j\in [m]$ where $V_j$ is infinite. 
Thus, the graph is strongly partitioned.

\item 
Similarly, we pick ``non-adjacent'' vertices in $U^{i,\nz}$
and connect them to $v_\ell$ to make the $E_i$-degree of $v_\ell$ correct,
where ``non-adjacent'' vertices are those in  $u_1,\ldots,u_\ell$  are not adjacent to any of $v_1,\ldots,v_\ell$.

Again, such ``non-adjacent'' vertices always exist since the graph is strongly partitioned,
and we can always pick the new vertices so that the graph stays strongly partitioned
after this iteration. 
\end{itemize}
We perform the iteration above for each color $i \notin R\cup \fnzl\cup \fnzr$.
This completes the construction of an $A|B$-biregular graph $G$ with size $\vM|\vN$
and witness partition $U=U_1\uplus\cdots \uplus U_m$
and $V=V_1\uplus \cdots \uplus V_n$ for the case when both $\fnzl$ and $\fnzr$ are not empty.
For the case when $\fnzl=\emptyset$, we can do as above, but skip Step~1.
Reasoning along the same lines, for the case when $\fnzr=\emptyset$, we can skip Step~2.
This completes the proof of Lemma~\ref{app:lem:simple-s4-aux}.
\end{proof}

\begin{figure}
\begin{center}

\begin{tikzpicture}
\draw (0,0) ellipse (1cm and 1.6cm);
\node  at (0,1.8) {\small $X_0$};

\draw[blue] (0,-.3) ellipse (.5cm and .9cm); 
\node[blue]  at (-0,.8) {\footnotesize $X_0\cap U^{i,\infty}$};

\draw[gray!40,fill=gray!40] (0,-.5) ellipse (.3cm and .5cm); 
\draw[gray!40] (7,0) -- node[color=black,above,sloped] {\scriptsize adjacent via $E_i$-edges} (0,0); 
\draw[gray!40] (7,0) -- (0,-1);

\draw (-0,-3.4) ellipse (0.8cm and 1.6cm);
\node  at (-0.7,-2.1) {\footnotesize $X_1$}; 

\draw[blue] (0,-3.9) ellipse (.5cm and .9cm); 
\node[blue]  at (0,-2.8) {\footnotesize $X_1\cap U^{i,\infty}$};

\draw[gray!40,fill=gray!40] (0,-3.9) ellipse (.3cm and .5cm); 
\draw[gray!40] (7,-3.2) -- node[color=black,below,sloped] {\scriptsize adjacent via $E_i$-edges} (0,-4.4); 
\draw[gray!40] (7,-3.2) -- (0,-3.4);

\node[circle,fill=blue,inner sep=0pt,minimum size=3pt,label=above:{\footnotesize $v$}] (a) at (7,0) {};

\node[circle,fill=blue,inner sep=0pt,minimum size=3pt,label=above:{\footnotesize $v'$}] (a) at (7,-3.2) {};

\draw (7,0) ellipse (0.8cm and 1.4cm);
\node  at (7.7,1.2) {\footnotesize $Y_0$};

\draw (7,-3.2) ellipse (0.8cm and 1.4cm);
\node  at (7.7,-2) {\footnotesize $Y_1$};

\end{tikzpicture}

\end{center}
\label{fig:X0-X1-Y0-Y1}
\caption{An illustration for the construction of edges between $X_0\cup X_1$ and $Y_0\cup Y_1$.
For each color $i\in F_{\text{nz-left}}$, 
each $v\in Y_0$ is connected by $E_i$-edges only to vertices in $X_0\cap U^{i,\infty}$
and each $v'\in Y_1$ is connected by $E_i$-edge only to vertices in $X_1\cap U^{i,\infty}$.
This makes vertices in $X_0$  not connected via $E_i$-edges to any vertex in $Y_1$, for any color $i\in F_{\text{nz-left}}$.
Likewise, vertices in $X_1$ are not connected via $E_i$ to any vertex in $Y_0$, for any color $i\in F_{\text{nz-left}}$.
This leaves some ``space'' for the construction of $E_i$-edges for color $i\in F_{\text{nz-right}}$
where the vertices in $X_0$ will be connected to some vertices in $Y_1$
and the vertices in $X_1$ to some vertices in $Y_0$.}
\end{figure}

\begin{remark}
\label{app:rem:strong-infinite}
Note that in the construction of the $A|B$-biregular graph $G$ in Lemma~\ref{app:lem:simple-s4-aux},
we construct the $E_i$-edges for every $i\notin R$ by
iterating on every vertex in $G$.
On each iteration, we preserve the ``strongly partitioned'' property of the graph $G$.
This implies that for every finite subset $W$ of vertices in $G$, we have:
\begin{itemize}
\item
For every $V_j$ such that $V_j$ is infinite,
there are infinitely many vertices in $V_j$ that are not adjacent to any vertex in $W$.
\item
Similarly, for every $U_j$ such that $U_j$ is infinite,
there are infinitely many vertices in $U_j$ that are not adjacent to any vertex in $W$.
\end{itemize}
This property will be useful in Section~\ref{app:sec:simple-complete}
when the completeness requirement is enforced.
\end{remark}

To wrap up Section~\ref{app:subsec:simple-bireg-big-enough},
we define formula $\text{Gen-}\Psi_{A|B}(\vx,\vy)$ for simple matrices $A$ and $B$
as follows:
\begin{align}
\label{app:eq:simple-s4}
& \bigvee_{R\subseteq [t]}\quad  \text{Gen-}\Psi^{4,R}_{A|B}(\vx,\vy),
\end{align}
where each $\text{Gen-}\Psi^{4,R}_{A|B}(\vx,\vy)$ is defined in \eqref{app:eq:simple-s4-aux}.
This formula $\text{Gen-}\Psi_{A|B}(\vx,\vy)$ captures all the big enough sizes of $A|B$-biregular graphs, 
as stated formally as Lemma~\ref{app:lem:simple-s4}.

\begin{lemma}
\label{app:lem:simple-s4}
For every pair of simple matrices $A$ and  $B$, and
for each pair of size vectors $\vM,\vN$ such that $\vM|\vN$ is extra big enough for $A|B$, the following holds:
there is an $A|B$-biregular graph with size $\vM|\vN$ if and only if
$\text{Gen-}\Psi_{A|B}(\vM,\vN)$ holds in $\cN_{\infty}$. 
\end{lemma}

\begin{proof}
Let $A$ and $B$ be simple matrices $A$ and $B$ with $t$ rows
and $\vM|\vN$ be big enough for $A|B$.
Suppose there is an $A|B$-biregular graph $G=(U,V,E_1,\ldots,E_t)$ with size $\vM|\vN$.
Let $R=\{i : \text{every vertex in $V$ has finite $E_i$-degree}\}$. 
By Lemma~\ref{app:lem:simple-s4-aux}, $\Psi^{4,R}_{A|B}(\vM,\vN)$ holds.
Thus, $\text{Gen-}\Psi^4_{A|B}(\vM,\vN)$ holds.

Conversely, suppose $\text{Gen-}\Psi_{A|B}(\vM,\vN)$ holds.
Let $R$ be such that $\text{Gen-}\Psi^{4,R}(\vM,\vN)$ holds.
By Lemma~\ref{app:lem:simple-s4-aux}, 
there is an $A|B$-biregular graph with size $\vM|\vN$
where $R=\{i : \text{every vertex in $V$ has finite $E_i$-degree}\}$.
\end{proof}

The remark below will be useful for the complexity analysis later on.

\begin{remark}
\label{app:rem:simple-big-enough}
Every formula $\text{Gen-}\Psi_{A|B}^{4,R}(\vx,\vy)$ is a disjunction of 
the formulas $\text{Gen-}\Psi^{3,R'}_{A_{R,*}|B_{R,*}}(\vx,\vy)$
for every subset $R'$ of the rows in $A_{R,*}|B_{R,*}$.
In turn, each formula $\text{Gen-}\Psi^{3,R'}_{A_{R,*}|B_{R,*}}(\vx,\vy)$ is a disjunction of 
the formulas $\text{Gen-}\Psi^{2,R''}_{A_{R',*}|B_{R',*}}(\vx,\vy)$
for every subset $R''$ of the rows in $A_{R',*}|B_{R',*}$.
Pulling out all the disjunctions, the formula $\text{Gen-}\Psi_{A|B}(\vx,\vy)$ can be written as a disjunction $\bigvee_i \psi_i(\vx,\vy)$,
where each $\psi_i(\vx,\vy)$ is a conjunction of $O(t)$ equations and inequations
(for short, ``conjunction of (in)equations'').
\end{remark}

\subsection{Encoding of not ``extra big enough'' components for simple matrices}
\label{app:subsec:simple-not-big-enough}

The encoding of not ``extra big enough'' for general case is
a routine adaptation of the one in Subsection ~\ref{subsec:simple-not-big-enough}.
We omit the details.
Disjoining this to the ``extra big enough'' formula completes the description for simple matrices without the completeness
requirement.

The remark below will also be used in  the complexity analysis later on.

\begin{remark}
\label{app:rem:simple-bireg}
Let $t$ be the number of rows in matrices $A$ and $B$.
By Remark~\ref{app:rem:simple-big-enough}, the formula $\text{Gen-}\Psi_{A|B}(\vx,\vy)$
is a disjunction of conjunctions of $O(t)$ (in)equations.
By an adaptation of Remark~\ref{rem:simple-not-big-enough}, each formula $\Phi_i(\vx,\vy)$ for the ``not extra big enough'' cases 
is a disjunction of conjunctions of $O(t^4\delta(A,B)^4)$ (in)equations.
Thus, the formula $\bireg_{A|B}(\vx,\vy)$ can be written as a disjunction of conjunctions of $O(t^4\delta(A,B)^4)$ (in)equations.
\end{remark}

%!TEX root = fo2siam.tex

\section{The extension of Section~\ref{sec:simple-complete} to the general case}
\label{app:sec:simple-complete}

In Section~\ref{sec:simple-complete} we constructed a  formula that
captures all possible finite sizes of \emph{complete} $A|B$-biregular graphs,
where $A$ and $B$ are simple degree matrices.
There we argued that it was sufficient to consider only the case when $A|B$ is a good pair, (as defined in Definition~\ref{def:good-pair}), since
otherwise  there are only some fixed number of possible sizes of $A|B$-biregular graphs,
which can be enumerated. 

In this appendix we will extend the formulas in Section~\ref{sec:simple-complete} to the case where graphs may be infinite.
As in Section~\ref{sec:simple-complete}, we only need to consider the cases where $A|B$ is a good pair.
There are two new cases to consider.
The first  (Subsection \ref{app:subsec:simple-complete-b}) concerns
graphs where  both sides have infinitely many vertices, while
in the  second  (Subsection \ref{app:subsec:simple-complete-a}) 
there are infinitely many vertices on exactly one side.

\subsection{The case when both sides have infinitely many vertices}
\label{app:subsec:simple-complete-b}

Let $A$ be a matrix with $t$ rows.
We write $\col_{\infty}(A_{*,j})$ to denote the set $\{i : A_{i,j}=\infty\}$.
For $R\subseteq [t]$, we let $J(R,A) =\{j : \col_{\infty}(A_{*,j})=R\}$.

\begin{definition}
\label{app:def:good-color-size}
Let $A$ and $B$ be simple matrices with $t$ rows.
Let $m$ and $n$ be the number of columns of $A$ and $B$.
For size vectors $\vM=(M_1,\ldots,M_m)$ and $\vN=(N_1,\ldots,N_n)$,
we say that {\em $\vM|\vN$ is a good color size for $A|B$}, if
for every $R\subseteq [t]$, we have:
\begin{itemize}
\item
$\normt{\vM}_{J(R,A)}=0$ or $\geq \delta(A,B)+1$,
\item 
$\normt{\vN}_{J(R,B)}=0$ or $\geq \delta(A,B)+1$.
\end{itemize}
\end{definition}

We will show that a ``good color size'' of a complete $A|B$-biregular graph 
implies a certain property of the matrices $A$ and $B$ which will be useful later on:

\begin{lemma}
\label{app:lem:completesimplegood-infty-infty-sufficient}
Let $A|B$ be a pair of simple matrices with $t$ rows.
Suppose there is a complete $A|B$-biregular graph $G=(U,V,E_1,\ldots,E_t)$ with size $\vM|\vN$
where $\normt{\vM}=\normt{\vN}=\infty$
and $\vM|\vN$ is a good color size for $A|B$.
Then:
\begin{itemize}
\item
for every vertex $u \in U$,
there is $i\in [t]$ such that $\deg_{E_i}(u)=\infty$ 
and row $i$ in $B$ contains a periodic entry;
\item
 for every vertex $v \in V$,
there is $i' \in [t]$ such that $\deg_{E_{i'}}(v)=\infty$ 
and row $i'$ in $A$ contains a periodic entry.
\end{itemize}
\end{lemma}
Note that since $A$ and $B$ are simple matrices and $\infty$ is regarded as a periodic entry, 
the conclusion implies that row $i$ and $i'$ in $A$ and $B$ can contain \emph{only} periodic entries.
So the lemma implies that $A|B$ is a good pair of simple matrices.
\begin{proof}
We first prove an easy combinatorial claim.
\begin{claim}
Let $\cU$ be an infinite set and let $\cZ$ be a (not necessarily finite) family of subsets of $\cU$
such that every $Z\in \cZ$ is cofinite in $\cU$ (i.e., $\cU-Z$ is finite).
Then every finite subset of $\cZ$ has non-empty intersection.
\end{claim}
\begin{proof}
Let $Z_1,\ldots,Z_n \in \cZ$.
By de Morgan's law, $\bigcap_{j=1}^n Z_{j} = \cU - (\bigcup_{j=1}^n \cU-Z_{j})$.
Since each $Z_{j}$ is cofinite in $\cU$, the claim follows immediately.
\end{proof}

We prove the first bullet item of the lemma, with 
the second proven analogously.
%Let $G=(U,V,E_1,\ldots,E_t)$ be a complete $A|B$-biregular graph 
%with size $\vM|\vN$ where $\normt{\vM}=\normt{\vN}=\infty$
%and $\vM|\vN$ is a good color size for $A|B$.
Let $U=U_1\uplus\cdots\uplus U_m$ and $V=V_1\uplus\cdots\uplus V_n$
be the witness partition.
For a vertex $u\in U$ in $G$,
let $\Gamma(u)$ denote the set of vertices adjacent to $u$ by some $E_i$-edge where the $E_i$-degree of $u$ is infinite.
For each $u \in U$,  every element of $V$ is connected to $u$ by some $E_i$ edge (since $G$ is a complete bipartite graph), 
and thus the number of elements  of $V$ not in $\Gamma(u)$ is finite.

Suppose $u\in U_j$ for some $j\in [m]$ and let $R=\col_{\infty}(A_{*,j})$.
Since $\vM|\vN$ is a good color size for $A|B$,
we have $\normt{\vM}_{J(R,A)}\geq \delta(A,B)+1$.
We pick $k$ vertices $u_1,\ldots,u_k\in \bigcup_{j\in J(R,A)}U_j$ where $u_1=u$ and $k\geq \delta(A,B)+1$.
By the claim above, there is a vertex $v$ in the  intersection $\bigcap_{j\in [k]} \Gamma(u_j)$.
This means $v$ is adjacent to all vertices $u_1,\ldots,u_k$ via some $E_i$-edges where $i\in R$.
Since $k> \norm{\finoffset(B)}$, there is $E_i \in R$ where $\deg_{E_i}(v)$ is a periodic entry of $B$.
Since $B$ is a simple matrix, this implies that  row $i$ in $B$ contains only periodic entries.
Note that $E_i\in R$, so  the $E_i$-degree of $u$ is $\infty$.
This completes the proof of the first bullet item.
\end{proof}

Let $A$ and $B$ be simple matrices with $m$ and $n$ columns, respectively.
We denote by $(\cC1)$ and $(\cC2)$ the  constraints:
\begin{enumerate}
\item[$(\cC1)$]
For every $j\in [m]$ where $M_j \neq 0$, there is a color $i\in [t]$ such that $A_{i,j}=\infty$
and $B_{i,*}$ contains only periodic entries.
\item[$(\cC2)$]
For every $j\in [n]$ where $N_j \neq 0$, there is a color $i\in [t]$ such that $B_{i,j}=\infty$
and $A_{i,*}$ contains only periodic entries.
\end{enumerate}
Note that both $(\cC1)$ and $(\cC2)$ are Presburger definable
and formalize the properties o
from items (1) and (2) in Lemma~\ref{app:lem:completesimplegood-infty-infty-sufficient}.

We define $\xi^{(1)}_{A|B}(\vx,\vy)$, where $\vx=(x_1,\ldots,x_m)$ and $\vy=(y_1,\ldots,y_n)$, as follows:
\begin{align}
\label{app:eq:simple-complete-xi2}
 \bireg_{A|B}(\vx,\vy)  \ \wedge \ & \normt{\vx}=\normt{\vy}=\infty \ \wedge \ (\cC1) \ \wedge \ (\cC2)
\\
\label{app:eq:simple-complete-xi2-a}
\wedge \ &
\bigwedge_{R\subseteq [t]} 
\Big(
\normt{\vx}_{J(R,A)} = 0 
\ \vee\
\normt{\vx}_{J(R,A)} \geq \delta(A,B)+1
\Big) 
\\
\label{app:eq:simple-complete-xi2-b}
\wedge \ &
\bigwedge_{R\subseteq [t]}
\Big(\normt{\vy}_{J(R,B)} = 0 \ \vee\
\normt{\vy}_{J(R,B)} \geq \delta(A,B)+1 
\Big).
\end{align}
Above $ \bireg_{A|B}(\vx,\vy)$ is the characterizing formula for not-necessarily-complete biregular graphs.
Intuitively, \eqref{app:eq:simple-complete-xi2-a} and \eqref{app:eq:simple-complete-xi2-b}
state that $\vx|\vy$ is a good color size.

\begin{lemma}
\label{app:lem:completesimplegood-infty-infty}
For every pair of simple matrices $A$, $B$ and 
for every pair of size vectors $\vM$, $\vN$, 
the formula $\xi^{(1)}_{A|B}(\vM,\vN)$ holds 
if and only if there is a complete $A|B$-biregular graph of size $\vM|\vN$
where $\normt{\vM}=\normt{\vN}=\infty$ and $\vM|\vN$ is a good color size for $A|B$.
\end{lemma}
\begin{proof}
That  $\xi^{(1)}_{A|B}(\vM,\vN)$ holding is necessary follows from Lemma~\ref{app:lem:completesimplegood-infty-infty-sufficient}.
Now we show that it is also a sufficient condition.
Suppose $\xi^{(1)}_{A|B}(\vM,\vN)$ holds,
which implies there is a (not necessarily complete) $A|B$-biregular graph 
$G=(U,V,E_1,\ldots,E_t)$ with size $\vM|\vN$.
Let $U=U_1\uplus\cdots\uplus U_m$ and $V=V_1\uplus\cdots\uplus V_n$ be the witness partition.
By Remark~\ref{app:rem:strong-infinite}, the graph $G$ has the property:
\begin{enumerate}[($\cP$)]
\item
{\em For every finite subset $W\subseteq U$,
there are infinitely many vertices in $V$ that are not adjacent to any vertex in $W$.

While for every finite subset $W\subseteq V$,
there are infinitely many vertices in $U$ that are not adjacent to any vertex in $W$.}
\end{enumerate}

We enumerate the elements $u_1,u_2,\ldots$ and $v_1,v_2,\ldots$ in $U$ and $V$, respectively.
We will make $G$ a complete $A|B$-biregular graph by iterating through all $\ell=1,2,\ldots$,
where on each iteration $\ell$, we first add ``new'' edges so that
$u_\ell$ is adjacent to all the vertices $v_\ell,v_{\ell+1},\ldots$ and
then some more ``new'' edges so that $v_\ell$ is adjacent to all the vertices $u_{\ell+1},u_{\ell+2},\ldots$.
These new edges will preserve the $A|B$-biregularity of the graph $G$
and as the iteration index $\ell$ goes to $\infty$, the graph $G$ becomes complete.

Before we proceed to the construction, we first explain the main idea behind
 making $u_\ell$ adjacent to all the vertices $v_\ell,v_{\ell+1},\ldots$.
Choose $i_0 \in [t]$ such that the $E_{i_0}$-degree of $u_\ell$ is $\infty$
and the row $B_{i_0,*}$ contains only periodic entries -- such $i_0$ exists due to $(\mathcal{C}1)$.
We add new $E_{i_0}$-edges so that:
\begin{enumerate}[(a)]
\item
$u_\ell$ is adjacent to all the vertices $v_\ell,v_{\ell+1},\ldots$ (that are not yet adjacent to $u_{\ell}$)
via $E_{i_0}$-edges.
\item
For every vertex $v_h \in \{v_{\ell},v_{\ell+1},\ldots\}$:
\begin{itemize}
\item
If the $E_{i_0}$-degree of $v_h$ is not $\infty$,
the new $E_{i_0}$-edges increase it by $p$.
\item 
If the $E_{i_0}$-degree of $v_h$ is $\infty$,
there are either $1$ or $p$ new $E_{i_0}$-edges adjacent to $v_h$; in particular the degree is still infinite.
\end{itemize}
\item
For every vertex $u_h \in \{u_{\ell+1},v_{\ell+2},\ldots\}$,
either there are no new $E_{i_0}$-edges added,
or the $E_{i_0}$-degree increases by a multiple of $p$.
\end{enumerate}
Adding new edges to make $v_\ell$ adjacent to all the vertices $u_{\ell+1},u_{\ell+2},\ldots$
can be done in the same manner.
The purpose of (a) is to make $G$ complete
while the purpose of (b) and (c) is to preserve the $A|B$-biregularity of $G$.

Since $U$ (resp. $V$) is countable, every vertex $u\in U$ (resp. $v\in V$)
has a finite index $\ell$ such that $u_\ell=u$ (resp. $v_\ell=v$).
After the $\ell^{\text{th}}$ iteration
we do not add any more edges adjacent to $u_\ell$ and $v_\ell$.
Therefore, for every vertex $w \in U\cup V$,
for every color $i\in [t]$, if $\deg_{E_i}(w)$ is finite in the original graph $G$,
it stays finite as the iteration index $\ell$ goes to $\infty$.
If $\deg_{E_i}(u)$ is $\infty$ in the original graph $G$,
it stays $\infty$, since we are only adding edges.
Thus, if the original graph $G$ is $A|B$-biregular,
as the iteration index $\ell$ goes to $\infty$,
the resulting graph is still $A|B$-biregular.

Note also that due to (b) and (c),
after the $\ell^{\text{th}}$ iteration, 
the degree of every vertex in $\{u_{\ell+1},u_{\ell+2},\ldots\}\cup\{v_{\ell+1},v_{\ell+2},\ldots\}$
increases only by some finite number, i.e., by $0$, $1$ or a multiple of $p$.
Thus, property ($\cP$) still holds 
for every finite subset $W\subseteq \{u_{\ell+1},u_{\ell+2},\ldots\}\cup \{v_{\ell+1},v_{\ell+2},\ldots\}$
in the sense that:
\begin{quote}
For every finite subset $W\subseteq \{u_{\ell+1},u_{\ell+2},\ldots\}$,
there are infinitely many vertices in $V$ that are not adjacent to any vertex in $W$.
While for every finite subset $W\subseteq \{v_{\ell+1},v_{\ell+2},\ldots\}$,
there are infinitely many vertices in $U$ that are not adjacent to any vertex in $W$.
\end{quote}
We will call this {\em the non-adjacency invariant}.

We devote the rest of the proof to the details on how to add edges adjacent to vertex $u_\ell$. 
The argument for vertex $v_\ell$ is handled  symmetrically.
Let $u_\ell$ be a vertex in $U_j$.
By constraint $(\cC)$, there is some color $i_0\in [t]$ where $A_{i_0,j}=\infty$ and 
$B_{i_0,*}$ contains only periodic entries.
Since $A$ is a simple matrix, row $i_0$ in $A$ also contains periodic (possibly infinite) entries.
We will add $E_{i_0}$-edges so that $u_\ell$ is adjacent to every vertex in $V$.
Note, however, that some care is needed, since the $E_{i_0}$-degree of some vertices  --- those
with finite $E_{i_0}$-degree bound ---
can only increase by a multiple of $p$.

Let $Z$ denote the set of vertices in $V$ that are not adjacent to vertex $u_\ell$.
By the non-adjacency invariant, the set $Z$ is infinite.
Let $Z=Z_{\fin}\cup Z_{\infty}$ be a partition of $Z$ where
every vertex in $Z_{\fin}$ has finite $E_{i_0}$-degree
and every vertex in $Z_{\infty}$ has infinite $E_{i_0}$-degree.

First, we add $E_{i_0}$-edges between $u_\ell$ and every vertex in $Z$.
At this point, 
vertex $u_\ell$ is already adjacent to every vertex in $V$.
Note that the $E_{i_0}$-degree of each vertex in $Z_{\infty}$ stays infinite.
However, the $E_{i_0}$-degree of vertices in $Z_{\fin}$ increases by $1$.
So we need to add additional edges to make it increase further by $(p-1)$.
There are two cases.
\begin{itemize}
\item
Case 1: $|Z_{\fin}|$ is finite.
Since the set $Z$ is infinite, we infer that $|Z_{\infty}|$ is infinite.
Let $Y$ be a finite subset of $Z_{\infty}$ so that 
the sum $|Z_{\fin}|+|Y|$ is some multiple of $p$.
By the non-adjacency invariant, 
there are infinitely many vertices in $U$ that are not adjacent to any vertex in $Z_{\fin}\cup Y$.
We pick $(p-1)$ such vertices $w_1,\ldots,w_{p-1}$ and 
add $E_{i_0}$-edges for every pair in $\{w_1,\ldots,w_{p-1}\} \times (Z_{\fin}\cup Y)$.
That is, $\{w_1,\ldots,w_{p-1}\} \times (Z_{\fin}\cup Y)$ becomes a complete bipartite graph of $E_{i_0}$-edges.
Note that the $E_{i_0}$-degrees of vertices $w_1,\ldots,w_{p-1}$ increase by a multiple of $p$, 
since $|Z_{\fin}\cup Y|$ is a multiple of $p$.
Moreover, the $E_{i_0}$-degrees of vertices in $Z_{\fin}$ increase further by $(p-1)$.
The $E_{i_0}$-degrees of vertices in $Y$ remain infinite. 
Thus, after this construction $G$ is still $A|B$-biregular.
See Figure~\ref{fig:finite-Z-fin} for an illustration.

\begin{figure}
\begin{center}

\begin{tikzpicture}

\draw (0,0) ellipse (1cm and 3cm);
\node  at (-.9,2.8) {\large $U$}; 

\node[circle,fill=blue,inner sep=0pt,minimum size=3pt,label=left:{\footnotesize $u_1$}] (a) at (0,2.5) {};

\draw[red,dashed,dash pattern=on \pgflinewidth off 5pt,line width=0.4mm] (0,2.2) -- (0,1);

\node[circle,fill=blue,inner sep=0pt,minimum size=3pt,label=left:{\footnotesize $u_\ell$}] (a) at (0,0.7) {};

\node[circle,fill=blue,inner sep=0pt,minimum size=3pt,label=left:{\footnotesize $w_1$}] (a) at (0,-.5) {};

\draw[red,dashed,dash pattern=on \pgflinewidth off 5pt,line width=0.4mm] (0,-.8) -- (0,-1.5);

\node[circle,fill=blue,inner sep=0pt,minimum size=3pt,label=below:{\footnotesize $w_{p-1}$}] (a) at (0,-1.8) {};

\draw (7,0) ellipse (1cm and 3cm);
\node  at (7.8,2.8) {\large $V$}; 

\node[circle,fill=blue,inner sep=0pt,minimum size=3pt,label=left:{\footnotesize $v_1$}] (a) at (7,2.5) {};

\draw[red,dashed,dash pattern=on \pgflinewidth off 5pt,line width=0.4mm] (7,2.2) -- (7,1);

\node[circle,fill=blue,inner sep=0pt,minimum size=3pt,label=left:{\footnotesize $v_\ell$}] (a) at (7,0.7) {};

\draw[gray!40,fill=gray!40] (7,-0.2) ellipse (.45cm and .65cm);
\node at (7,-0.3) {\small $Z_{\fin}$};

\draw[gray!40,fill=gray!40] (7,-1.5) ellipse (.15cm and .3cm);
\node at (7,-1.5) {\small $Y$};

\draw (7,-1.7) ellipse (.45cm and .75cm);
\node at (7,-2.2) {\small $Z_{\infty}$};

\end{tikzpicture}

\end{center}
\label{fig:finite-Z-fin}
\caption{An illustration for the choices of $w_1,\ldots,w_{p-1}$ and $Y\subseteq Z_{\infty}$ to 
construct the complete $A|B$-biregular graph when $|Z_{\fin}|$ is finite.
First, we connect $u_\ell$ with all the vertices in $Z_{\fin}\cup Z_{\infty}$ via an $E_{i_0}$-edge.
Then, to ensure the degrees in $Z_{\fin}$ increases by a multiple of $p$,
we pick $w_1,\ldots,w_{p-1}$ and connect them via $E_{i_0}$-edges with all the vertices in $Z_{\fin}\cup Y$, where $Y\subseteq Z_{\infty}$ such that
$|Z_{\fin}|+|Y|$ is a multiple of $p$.}
\end{figure}

\item
Case 2: $|Z_{\fin}|$ is infinite.
We partition $Z_{\fin}$ into infinitely many pairwise disjoint sets $Z_1\uplus Z_2 \uplus \cdots $,
where $|Z_h|=p$ for each $h=1,2,\ldots$.
We increase the $E_{i_0}$-degrees of vertices in each $Z_h$
by iterating the following process for each $h=1,2,\ldots$:
We pick a finite set $X\subseteq U$ such that $|X|=p-1$
and every vertex in $X$ is not adjacent to any vertex in $Z_1\cup\cdots\cup Z_h$.
Such a set $X$ exists, by the non-adjacency invariant.
Then we add $E_{i_0}$-edges between every pair in $Z_h \times X$.
That is, $Z_h \times X$ becomes a complete bipartite graph of $E_{i_0}$-edges.
See Figure~\ref{fig:infinite-Z-fin} for an illustration.
\begin{figure}
\begin{center}

\begin{tikzpicture}

\draw (0,0) ellipse (1cm and 3cm);
\node  at (-.9,2.8) {\large $U$}; 

\node[circle,fill=blue,inner sep=0pt,minimum size=3pt,label=left:{\footnotesize $u_1$}] (a) at (0,2.5) {};

\draw[red,dashed,dash pattern=on \pgflinewidth off 5pt,line width=0.4mm] (0,2.2) -- (0,1);

\node[circle,fill=blue,inner sep=0pt,minimum size=3pt,label=left:{\footnotesize $u_\ell$}] (a) at (0,0.7) {};

\draw[gray!40,fill=gray!40] (0,-0.8) ellipse (.2cm and .4cm);
\node at (0,-0.8) {\small $X$};

\draw (7,0) ellipse (1cm and 3cm);
\node  at (7.8,2.8) {\large $V$}; 

\node[circle,fill=blue,inner sep=0pt,minimum size=3pt,label=left:{\footnotesize $v_1$}] (a) at (7,2.5) {};

\draw[red,dashed,dash pattern=on \pgflinewidth off 5pt,line width=0.4mm] (7,2.2) -- (7,1);

\node[circle,fill=blue,inner sep=0pt,minimum size=3pt,label=left:{\footnotesize $v_\ell$}] (a) at (7,0.7) {};

\draw[blue] (7,-.6) ellipse (.45cm and .9cm);
\node[blue] at (7.6,0.2) {\small $Z_{\fin}$};

\draw[red] (6.62,-0.1) -- (7.38,-0.1);
\node at (7,0.1) {\scriptsize $Z_1$};

\draw[red,dashed,dash pattern=on \pgflinewidth off 2pt,line width=0.2mm] (7,-.2) -- (7,-.5);

\draw[red] (6.55,-0.6) -- (7.45,-0.6);
\node at (7,-0.8) {\scriptsize $Z_h$};
\draw[red] (6.6,-1) -- (7.4,-1);

\draw[red,dashed,dash pattern=on \pgflinewidth off 2pt,line width=0.2mm] (7,-1.1) -- (7,-1.4);

\draw (7,-2.2) ellipse (.4cm and .6cm);
\node at (7,-2.5) {\small $Z_{\infty}$};

\end{tikzpicture}

\end{center}
\label{fig:infinite-Z-fin}
\caption{An illustration for the choices of  $w_1,\ldots,w_{p-1}$ and $Y\subseteq Z_{\infty}$ to 
construct the complete $A|B$-biregular graph 
when $|Z_{\fin}|$ is infinite.
First, we connect $u_l$ with all the vertices in $Z_{\fin}\cup Z_{\infty}$ via an $E_{i_0}$-edge,
thus, increasing the $E_i$-degree of vertices in $Z_{\fin}\cup Z_{\infty}$ by $1$.
Then, we partition $Z_{\fin}$ into $Z_1\uplus Z_2 \uplus \cdots $ where each $Z_h$ has cardinality $p$.
To make sure that the $E_{i_0}$-degrees in $Z_{\fin}$ increase by a multiple of $p$,
For each $h=1,2,\ldots$,
we pick a set $X\subseteq U$ s.t. $|X|=p-1$ and every vertex in $X$ is not adjacent to any vertex in $Z_1\cup \cdots \cup Z_h$.
Then, we connect every vertex in $X$ with every vertex in $Z_h$ via $E_{i_0}$-edges.}
\end{figure} 
After this construction, the $E_{i_0}$-degree of each vertex in each $Z_{h}$ increases further by $(p-1)$.
Since each $|Z_h|=p$,
we also increase the $E_{i_0}$-degrees of some vertices in $U$ by $p$.
Thus, $G$ is still $A|B$-biregular.
\end{itemize}
\end{proof}

Lemmas~\ref{app:lem:completesimplegood-infty-infty} and~\ref{app:lem:completesimplegoodonefinite} deal with all the  $\vM|\vN$ 
that are good color sizes for $A|B$.
To capture the sizes that are not good color sizes,
we can use ``fixed size encoding'', as in Subsection~\ref{subsec:simple-not-big-enough}.
Note that if $\vM|\vN$ is not a good color size for $A|B$,
there is $R\subseteq [t]$ such that
\begin{align*}
1 \leq \normt{\vx}_{J(R,A)} \leq \delta(A,B)
\qquad\text{or}\qquad
1 \leq \normt{\vy}_{J(R,B)} \leq \delta(A,B).
\end{align*}
Thus, we can fix $\normt{\vx}_{J(R,A)}$ or $\normt{\vy}_{J(R,B)}$
to some $r$ where $1\leq r\leq \delta(A,B)$.
Recall that $J(R,A)$ is  a subset of columns of $A$,
while  $J(R,B)$ is a subset of the columns of $B$.
Thus in fixing one of these norms,  we are focusing on complete $A|B$-biregular graphs $G=(U,V,E_1,\ldots,E_t)$
with sizes $\vM|\vN$ where the sum of some components in $\vM$ (or $\vN$) is fixed to $r\leq \delta(A,B)$.
For example, we can define the formula $\Phi^r_{A|B}(\vx,\vy)$ such that for every $\vM|\vN$,
$\Phi^r_{A|B}(\vM,\vN)$ holds if and only if
there is a complete $A|B$-biregular graph with size $\vM|\vN$ where $\normt{\vx}_{J(R,A)}=r$.
The construction of $\Phi^r_{A|B}(\vx,\vy)$ is very similar to the one in Section~\ref{subsec:simple-not-big-enough}, so we omit it.

To wrap up this subsection,
we define the formula $\biregc^{\infty,\infty}_{A|B}(\vx,\vy)$ for simple matrices $A$ and $B$
as follows:
\begin{align}
\label{eq:simple-complete}
&  \xi^{(1)}_{A|B}(\vx,\vy) \ \vee \
\varphi(\vx,\vy) \ \vee \
\bigvee_{j} \phi_j(\vx,\vy),
\end{align}
where $\xi^{(1)}_{A|B}(\vx,\vy)$ is defined in \eqref{app:eq:simple-complete-xi2}--\eqref{app:eq:simple-complete-xi2-b},
$\varphi(\vx,\vy)$ captures all the sizes $\vM|\vN$ that are not good color sizes for $A|B$, and
the disjunction $\bigvee_{j} \phi_j(\vx,\vy)$ enumerates all possible sizes $\vM|\vN$ when $A|B$ is not a good pair.
By Remark~\ref{rem:not-good-pair}, when $A|B$ is not a good pair,
complete $A|B$-biregular graphs can only have sizes $\vM|\vN$
where $\normt{\vM}+\normt{\vN}\leq 2\delta(A,B)$. It is clear that this remark holds regardless of whether
an $\infty$ entry is allowed. 
Since there are only finitely many sizes satisfying this upper bound, they can be enumerated. 
The formula $\biregc^{\infty,\infty}_{A|B}(\vx,\vy)$ captures the sizes
of all possible $A|B$-biregular graphs where both sides have infinitely many vertices.

\begin{remark}
\label{app:rem:simple-complete}
We will again make some further observations that will be important only for the complexity
analysis.
Suppose $t$ is the number of rows in matrices $A$ and $B$.
By Remark~\ref{app:rem:simple-bireg},
$\xi^{(1)}_{A|B}(\vx,\vy)$ is a disjunction of conjunctions of $O(t^4\delta(A,B)^4)$ (in)equations.

As in Remark~\ref{rem:simple-not-big-enough}, 
the encoding of components of a fixed size $r$ yields $O(rt)$ (in)equations.
Since $r\leq \delta(A,B)$ and there are $2^t$ subsets $R\subseteq [t]$,
the formula for the ``{\em fixed size encoding}'' can be written as a disjunction of conjunctions of $O(2^tt\delta(A,B))$ (in)equations.
So, the whole formula $\biregc_{A|B}^{\infty,\infty}(\vx,\vy)$ can be rewritten as a disjunction of conjunctions 
of $O(2^tt^4\delta(A,B)^4)$ (in)equations.
\end{remark}

\subsection{The case when exactly one side has only finitely many vertices}
\label{app:subsec:simple-complete-a}

In this subsection we will give the formula that captures the sizes of all possible $A|B$-biregular graphs
where on the left hand side there are infinitely many vertices
and on the right hand side there are only finitely many vertices.
Here the degree matrices $A$ and $B$ can be arbitrary degree matrices, i.e., we drop the assumption that
they must be simple matrices.

In a first step (Subsection~\ref{app:subsubsec:simple-complete-a-special})
we consider the case where the degree matrix $B$ is restricted to a very special form
and the size vectors on the left contain only $\infty$.
In a  second step (Subsection~\ref{app:subsubsec:complete-infinite-finite})
we show that capturing the sizes of $A|B$-biregular graphs where exactly one side has infinitely many vertices
can be reduced to the finite case and the case in Subsection~\ref{app:subsubsec:simple-complete-a-special}.

\subsubsection{A special case}
\label{app:subsubsec:simple-complete-a-special}

We fix matrices $A$ and $B$ (with $t$ rows) with the following properties:
\begin{itemize}
\item 
$A$ contains only finite entries.
\item 
Each entry in $B$ is either $0$ or $\infty$.
\item
Every row and every column in $B$ has $\infty$ entry.
\end{itemize}
We note that for such $A$ and $B$,
in a complete $A|B$-biregular graph it is necessary that
the left side has infinitely many vertices
and the right side has only finitely many.
We will define a formula that captures all possible size vectors $\vN$
where $\normt{\vN}\neq\infty$ and there is a complete $A|B$-biregular graph with size $(\infty,\ldots,\infty)|\vN$.

Let $m$ and $n$ be the number of columns in $A$ and $B$.
We start with a simple observation:

\begin{remark}
\label{app:rem:complete-special}
Let $G=(U,V,E_1,\ldots,E_t)$ be a complete $A|B$-biregular graph
with witness partition $U=U_1\uplus \cdots \uplus U_m$
and $V=V_1\uplus \cdots \uplus V_n$.
Let $u\in U$ and let $j$ be the index such that $u\in U_j$.

For each $i\in [t]$, let $Z_i$ be the set of vertices adjacent to $u$ via $E_i$-edges.
Since $G$ is complete,  $Z_1\uplus \cdots \uplus Z_t$ partitions the set $V$.
Moreover, since $G$ is $A|B$-biregular, for every $i\in [t]$:
\begin{itemize}
\item
$|Z_i|= E_i\text{-degree of $u$}=A_{i,j}$.
\item 
Every vertex in $Z_i$ has $\infty$ $E_i$-degree.

Recall that every entry in $B$ can only be either $0$ or $\infty$,
hence, the $E_i$-degree of every vertex in $V$ can only be $0$ or $\infty$.
\end{itemize}
We will call the partition $Z_1\uplus Z_2\uplus \cdots\uplus Z_t$,
{\em the partition of $V$ according to $u$}.
As we will see later,
we can construct a Presburger formula
that defines the sizes of the partitions of $V$ according to vertices in $U_j$,
for every vertex in $U_j$.
\end{remark}

For $j\in [m]$, for each $k\in [t]$, 
define the formula $\varphi_{k,j}(z_1,\ldots,z_k,s_1,\ldots,s_n)$ inductively on $k$ as follows.
\begin{itemize}
\item
When $k=1$, $\varphi_{1,j}(z_1,s_1,\ldots,s_n)$ is given by
$$
\quad z_1 = s_1+\cdots +s_n = A_{1,j} 
\  \wedge \
\bigwedge_{h\in [n]} s_h\neq 0 \to B_{1,h}=\infty
$$
\item 
When $k\geq 2$, $\varphi_{k,j}(z_1,\ldots,z_k,s_1,\ldots,s_n)$ is given by \\
\begin{align*}
\exists c_1\cdots \exists c_n \qquad  c_1+\cdots+c_n=z_k 
 \wedge 
\bigwedge_{h\in [n]} c_h\neq 0\ \to\ B_{k,h}=\infty \\
\wedge
\varphi_{k-1,j}(z_1,\ldots,z_{k-1},s_1-c_1,\ldots,y_n-s_n) 
\end{align*}
\end{itemize}

Finally, define the formula $\xi_{A|B}(\vy)$, where $\vy=(y_1,\ldots,y_n)$:
\begin{align}
\label{app:eq:complete-special}
\xi_{A|B}(\vy) \ := \ & \bigwedge_{j\in[m]} \exists z_1\cdots \exists z_t \ \varphi_{t,j}(z_1,\ldots,z_t,\vy)
\end{align}

We will show that $\xi_{A|B}(\vy)$ captures all size vectors $\vN$
such that there are complete $A|B$-biregular graph with size $(\infty,\ldots,\infty)|\vN$.
The variables $z_1,\ldots,z_t$ in the formula $\varphi_{t,j}(z_1,\ldots,z_t,y_1,\ldots,y_n)$
represent the cardinalities $|Z_1|,\ldots,|Z_t|$,
for the partition $Z_1\uplus \cdots \uplus Z_t$ according to a vertex in $U_j$.
We start with an easy lemma, proven by induction on $k$:

\begin{lemma}
\label{app:lem:equal-sum}
For every $k\in [t]$,
for every $z_1,\ldots,z_k,s_1,\ldots,s_n\in \bbN$,
\\
if $\varphi_{k,j}(z_1,\ldots,z_k,s_1,\ldots,s_n)$ holds,
then $z_1+\cdots+z_k=s_1+\cdots+s_n$.
\end{lemma}

\begin{lemma}
\label{app:lem:complete-special}
For every size vector $\vN$ where $\normt{\vN}\neq\infty$,
the formula $\xi_{A|B}(\vN)$ holds precisely when
there is a complete $A|B$-biregular graph with size $(\infty,\ldots,\infty)|\vN$.
\end{lemma}
\begin{proof}
Let $\vN=(N_1,\ldots,N_n)$ be a size vector where none of $N_1,\ldots,N_n$ are $\infty$.
We first show that $\xi_{A|B}(\vN)$ holding is a necessary condition.
Suppose there is a complete $A|B$-biregular graph $G=(U,V,E_1,\ldots,E_t)$ with size $(\infty,\ldots,\infty)|\vN$.
Let $U=U_1\uplus \cdots \uplus U_m$ and $V=V_1\uplus \cdots \uplus V_n$ be the witness partition.

We will show that for every $j\in [m]$,
$\varphi_{t,j}(z_1,\ldots,z_t,\vN)$ holds, for some $z_1,\ldots,z_t$.
To this end, let $j\in [m]$.
We pick a vertex $u\in U_j$ and let $Z_1\uplus \cdots \uplus Z_t$ be the partition of $V$
according to $u$.
Let $z_i=|Z_i|$, for every $i\in [t]$.

The next claim can be proven by straightforward induction on $k$.

\begin{claim}
For every $k\in [t]$, the formula $\varphi_{k,j}(z_1,\ldots,z_k,s_1,\ldots,s_n)$ holds 
where $s_h=|V_h\cap (Z_1\cup \cdots \cup Z_k)|$ for each $h\in [n]$.
\end{claim}

In particular, when $k=t$,
$s_h = |V_h\cap (Z_1\cup \cdots \cup Z_t)|= |V_h|=N_h$, for each $h\in [n]$,
since $Z_1\cup \cdots \cup Z_t = V$.
Therefore, $\varphi_{t,j}(z_1,\ldots,z_t,\vN)$ holds.
Thus, $\xi_{A|B}{\vN}$ holds.

We now prove that $\xi_{A|B}(\vN)$ holding is a sufficient condition. 
Suppose $\xi_{A|B}(\vN)$ holds, where $\vN = (N_1,\ldots,N_n)$.
Let $U_1,\ldots,U_m$ be pairwise disjoint infinite sets
and let $V_1,\ldots,V_n$ be pairwise disjoint sets where $|V_h|=N_h$ for each $h\in [n]$.

We will construct a complete $A|B$-biregular graph $G=(U,V,E_1,\ldots,E_t)$
with size $(\infty,\ldots,\infty)|\vN$
and witness partition $U = U_1\uplus \cdots\uplus U_m$
and $V=V_1\uplus \cdots\uplus V_n$.

Let $j\in [m]$.
Since $\xi_{A|B}(\vN)$ holds,
there is $z_1,\ldots,z_t$ such that $\varphi_{t,j}(z_1,\ldots,z_t,\vN)$ holds.

The following claim is proven by straightforward induction on $k$.
\begin{claim}
For every $k\in [t]$,
there are pairwise disjoint sets $Z_1,\ldots,Z_k\subseteq V$ such that
for every $i\in [k]$ 
$z_i = |Z_i| = A_{i,j}$ and $Z_i \subseteq \bigcup_{h \in \inf(B_{i,*})} V_{h}$.
\end{claim}

In particular, when $k=t$,
we have pairwise disjoint sets $Z_1,\ldots,Z_t\subseteq V$ such that
for every $i\in [t]$
$z_i = |Z_i| = A_{i,j}$ and
$Z_i \subseteq \bigcup_{h \in \inf(B_{i,*})} V_{h}$.
By Lemma~\ref{app:lem:equal-sum}, the sum $z_1+\cdots +z_t = |Z_1|+\cdots + |Z_t|=N_1+\cdots+N_n$.
Hence, $Z_1\uplus \cdots \uplus Z_t$ is a partition of $V$.
For every $i\in [t]$,
we connect every vertex $u\in U_j$ with every vertex in $Z_i$ via an  $E_i$-edge.
Thus after this step, every vertex in $U_j$ is adjacent to every vertex in $V$.

Note that the $E_i$-degree of every vertex in $U_j$ is $|Z_i|=A_{i,j}$.
Moreover, we connect $u$ with a vertex $v\in V$ only when the $E_i$-degree of $v$ is supposed to be $\infty$
-- since $Z_i \subseteq \bigcup_{h \in \inf(B_{i,*})} V_{h}$.
Thus, the resulting graph is $A|B$-biregular.
By repeating the above process for every $j\in [m]$,
we obtain a complete $A|B$-biregular graph.
\end{proof}

\subsubsection{The formula for the case with infinitely many vertices on the left and finitely many vertices on the right}
\label{app:subsubsec:complete-infinite-finite}

In this subsection we will define the formula
that captures precisely the sizes of all possible $A|B$-biregular graphs
where the left hand side has infinitely many vertices
and the right hand side has only finitely many vertices.
Here we do not require the degree matrices to be {\em simple} matrices -- as defined in Definition~\ref{def:simple-matrix}.

In the following lemma, we fix degree matrices $A\in \bbNop^{t\times m}$ and $B\in\bbNop^{t\times n}$.

\begin{lemma}
\label{app:lem:complete-infinite-finite}
Suppose $G=(U,V,E_1,\ldots,E_t)$ is a complete $A|B$-biregular graph
with witness partition $U_1\uplus \cdots \uplus U_m$ and $V_1\uplus \cdots \uplus V_n$.
Suppose $U$ is infinite and $V$ is finite.
Let $R = \{i \in [t] \mid |E_i| = \infty\}$
and let $J = \{j \in [m] \mid |U_j|=\infty \}$.
Then, 
\begin{enumerate}[(1)]
\item
For every color $i \notin R$, for every $j \in J$,
 $A_{i,j}$ is $0$ or $\prdp{0}$.
\item 
For every $j \in [n]$, there is $i\in R$  with $B_{i,j}=\infty$.
\item 
For every $i \in R$, the row $B_{i,*}$ contains an $\infty$ entry.
\item 
For every $j\in J$, for every $i \notin R$, all but finitely many vertices in $U_j$
have zero $E_i$-degree. 
\item 
There are only finitely many vertices in $U$ for which there is $v \in V$ adjacent to the vertex by
an $E_i$-edge and the $E_i$-degree of $v$ is finite.
\end{enumerate}
\end{lemma}
\begin{proof}
To prove (1), let $j\in J$, i.e., the set $U_j$ is infinite.
If there were $i\notin R$ such that $A_{i,j}\neq 0$ or $\neq \prdp{0}$,
the number of edges in $E_i$ is infinite,
which contradicts the assumption that $i\notin R$.

For (2), let $j\in [n]$.
Since $G$ is  a complete $A|B$-biregular graph, 
the total degree of each vertex $v \in V_j$ must equal $|U|$, i.e:
$$
\sum_{i\in [t]} (E_i\text{-degree of}\ v)\ = \
\sum_{i\in R} (E_i\text{-degree of}\ v)\ + \
\sum_{i\notin R} (E_i\text{-degree of}\ v)\ = \
|U|
$$
By definition of $R$, the sum $\sum_{i\notin R} (E_i\text{-degree of}\ v)$ is finite.
Since $U$ is infinite, the sum $\sum_{i\in R} (E_i\text{-degree of}\ v)$ must be infinite.
Therefore, there is $i\in R$ such that $B_{i,j}=\infty$.

For (3),  let $i\in R$.
The cardinality $E_i$ is $|E_i| \ = \ \sum_{v\in V} (E_i\text{-degree of}\ v)$.
Since $i\in R$, the cardinality $|E_i|=\infty$.
Thus, there is $v\in V$ with $E_i$-degree $\infty$.
Therefore, row $B_{i,*}$ must contain $\infty$.

For (4), let $j\in J$ and $i\notin R$.
There can only be finitely many many vertices in $U_j$ with non-zero $E_i$-degree.
Otherwise, $|E_i|=\infty$, which contradicts the assumption that $i\notin R$.

For (5), let $v\in V$.
Obviously there are only finitely many vertices in $U$ that are adjacent to $v$
via some $E_i$-edge, where the $E_i$-degree of $v$ is finite.
Since $V$ is finite, (5) follows immediately.
\end{proof}

Intuitively, (1)--(3) state the properties matrices $A$ and $B$ should have
when considering $A|B$-biregular graphs for the case considered in this subsection,
which also allows us to identify a subgraph whose biregularity can be characterized using Subsection~\ref{app:subsubsec:simple-complete-a-special}.
See Figure~\ref{fig:matrix-infinite-one-side} for an illustration of the decomposition of matrices $A$ and $B$.
We will use (4) and (5) to identify a corresponding subgraph whose biregularity is 
characterized using  the {\em finite} case covered in Theorem~\ref{theo:main-lemma-bireg}.
For an arbitrary graph, we let $R$ and $J$ be as defined in Lemma~\ref{app:lem:complete-infinite-finite}.

\begin{figure}
\begin{center}

\begin{tikzpicture}

\node at (1.5,-1.5) {$\underbrace{\hspace{1.5cm}}_{\text{columns in $J$}}$};

\node at (3.5,0.6) {\footnotesize$\left.\begin{array}{cc} & \\ & \\ &  \end{array}\right\}$ rows not in $R$};

\node at (3.2,-0.7) {\footnotesize$\left.\begin{array}{cc} & \\ & \\ &  \end{array}\right\}$ rows in $R$};

\node  at (0,0) {$A\ :=\ \left(\begin{array}{c|c}
\begin{array}{ccc}
& & 
\\
& A_1 &
\\
& &
\end{array}
&
\text{\em $0$ or $\prdp{0}$}
\\
\hline
\begin{array}{ccc}
& & 
\\
& A_2 &
\\
& &
\end{array}
&
\begin{array}{ccc}
& & 
\\
& A_3 &
\\
& &
\end{array}
\end{array}
\right)$};

%%%% matrix B

\node at (0.2,-5.6) {$\underbrace{\hspace{2.8cm}}_{\text{every column contains $\infty$ in some row in $R$}}$};

\node at (2.8,-3.4) {\footnotesize$\left.\begin{array}{cc} & \\ & \\ &  \end{array}\right\}$ rows not in $R$};

\node at (2.5,-4.7) {\footnotesize$\left.\begin{array}{cc} & \\ & \\ & \end{array}\right\}$ rows in $R$};

\node  at (-.4,-4) {$B\ :=\ \left(\begin{array}{ccccc}
~~ & ~~ & ~~ & ~~ & ~~
\\
& & B_1 & &
\\
& & & &
\\
\hline
 &  &  &  &
\\
 &  &  B_2 &  &
\\
 &  &  &  &
\end{array}
\right)$}; 

\end{tikzpicture}

\end{center}
\label{fig:matrix-infinite-one-side}
\caption{An illustration of the matrices $A$ and $B$ for the case
when there is a complete $A|B$-biregular graph $G=(U,V,E_1,\ldots,E_t)$ with infinitely many vertices
on the left hand side and only finitely many vertices on the right hand side.
Suppose $U=U_1\uplus \cdots \uplus U_n$ and $V=V_1\uplus \cdots \uplus V_n$
is the witness partition.
$R$ is the set of color $i$ where $|E_i|=\infty$
and $J$ is the set of column $j$ where $U_j$ is infinite.}
\end{figure}

Let $C$ be the matrix obtained by replacing every $\infty$ entry in $B$ with $0^{+1}$.
Intuitively, we replace $\infty$ with some finite value.\footnote{Technically 
we cannot simply replace $\infty$ with $0^{+1}$
since we insist that every periodic entry in a degree matrix has period $p$.
Instead we can replace it with $\prdp{0},\prdp{1},\ldots,\prdp{(p-1)}$
by repeating the columns.
For example, a column $\begin{pmatrix}
\infty \\ a
\end{pmatrix}$
becomes $\begin{pmatrix}
\prdp{0}\\ a
\end{pmatrix},
\begin{pmatrix}
\prdp{1}\\ a
\end{pmatrix},\ldots, \begin{pmatrix}\prdp{(p-1)}\\ a\end{pmatrix}$.
We allow the matrix to have $0^{+1}$ entries.} 
Let $A_3$ be the matrix obtained from $A$ by keeping only the rows in $R$
and the columns in $J$.
Let $B_2$ be the matrix obtained from $B$ by keeping only the rows in $R$
and  $D$ be the matrix obtained from $B_2$ by replacing every non-$\infty$ entry with $0$.

Let $U^*$ be the set of vertices in $\bigcup_{j\in J} U_j$ adjacent to some $v\in V$ via some $E_i$-edges
where the $E_i$-degree of $v$ is finite.
For each $j\in J$, let $U_{j,\fin}$ be the set of vertices in $U_j$
with non-zero $E_i$-degree for some $i\notin R$.
Define the sets:
\begin{align*}
U^{\fin} \ := \ & \ U^*\ \cup \ \bigcup_{j\notin J} U_j \ \cup \ \bigcup_{j\in J} U_{j,\fin}
\\
U^{\infty} \ := \ & \ U - U^{\fin}
\end{align*}
By~(4) in Lemma~\ref{app:lem:complete-infinite-finite},
the set $U_{j,\fin}$ is finite for every $j \in J$.
By (5), the set $U^*$ is finite.
Thus, the set $U^{\fin}$ is finite.
By definition of $U^{\infty}$, a vertex in $U^{\infty}$
has non-zero $E_i$-degree only when $i\in R$.
See Figure \ref{fig:complete-infinite-finite-charac} for an illustration.

The following lemma will provide our reduction.

\begin{lemma}
\label{app:lem:complete-infinite-finite-charac}
Suppose $G=((U,V,E_1,\ldots,E_t)$ is a complete  $A|B$-biregular graph  with size $\vM|\vN$. Let
$U^{\fin}$ and $U^{\infty}$  be as defined above,
$G_{\fin}$ denote the induced subgraph $G[U^{\fin}\cup V]$,
and $G_{\infty}$ denote the induced subgraph $G[U^{\infty}\cup V]$.
Then:
\begin{itemize}
\item
$G_{\fin}$ is a (finite) complete $A|C$-biregular graph
with size $\vK|\vN$ for some $\vK=(K_1,\ldots,K_m)$
where $K_j=M_j$, if $j\notin J$ and $K_j$ is some finite value, if $j\in J$.

\item 
$G_{\infty}$ is a complete $A_3|D$-biregular graph
with size $(\infty,\ldots,\infty)|\vN$.
\end{itemize}
\end{lemma}
\begin{proof}
For each vertex $w\in U\cup V$,
for each color $i\in [t]$, we say that {\em the $E_i$-degree of $w$ is affected in $G_{\fin}$ (resp. $G_{\infty}$)},
if its $E_i$-degree in $G_{\fin}$ (resp. $G_{\infty}$) is different from its $E_i$-degree in $G$.
Otherwise, we say that the $E_i$-degree of $w$ is \emph{unaffected} in $G_{\fin}$ (resp. $G_{\infty}$).

Towards proving the first bullet item,
note that for each $i\in [t]$,
the $E_i$-degree of every vertex $u$ is unaffected in $G_{\fin}$,
since $V$ is still the set of vertices on the left hand side of $G_{\fin}$.
On the other hand, for each vertex $v\in V$,
and  color $i\in [t]$, if the $E_i$-degree of $v$ is finite in $G$,
then its $E_i$-degree is unaffected in $G_{\fin}$.
This is because if $(u,v)\in E_i$ and the $E_i$-degree of $v$ is finite,
then by definition, $u\in U^*$, and hence, $u\in U^{\fin}$.
So, the $E_i$-degree of $v$ is affected in $G_{\fin}$
only when the $E_i$-degree of $v$ is $\infty$ in $G$,
which has now become finite in $G_{\fin}$.
Since every $\infty$ entry in $B$ has now becomes $0^{+1}$ in $C$,
it follows immediately that $G_{\fin}$ is a complete  $A|C$-biregular graph.

Turning to the second bullet item,
the $E_i$-degree of every vertex in $U^{\infty}$ is obviously unaffected in $G_{\infty}$.
Moreover, every vertex in $U^{\infty}$ has non-zero $E_i$-degree in $G$ only when $i\in R$.
Thus, the colors of the edges in $G_{\infty}$ are only those in $R$.
For every vertex $v\in V$:
\begin{itemize}
\item
If its $E_i$-degree is $\infty$ in $G$,
its $E_i$-degree is unaffected in $G_{\infty}$.
\item 
If its $E_i$-degree is finite in $G$,
its $E_i$-degree becomes $0$ in $G_{\infty}$.

This is because if $u$ and $v$ are adjacent via an $E_i$-edge
and the $E_i$-degree of $v$ is finite,
then $u\in U^{*}$, hence, $u\in U^{\fin}$.
\end{itemize}
Since every finite entry in $B_2$ becomes $0$ in $D$,
it follows immediately that $G_{\infty}$ is $A_3|D$-biregular.
\end{proof}

Lemma~\ref{app:lem:complete-infinite-finite-charac} reduces characterization of the sizes
of $A|B$-biregular graphs to characterizations of finite complete biregular graphs (which we have provided in the body)
and characterization of infinite $A_3|D$-biregular graphs, whose sizes are of the form $(\infty,\ldots,\infty)|\vN$, 
i.e., the components in the size vectors on the left are all $\infty$
and every entry in $D$ is either $0$ or $\infty$.

\begin{figure}
\begin{center}

\begin{tikzpicture}

%%% the set U1,...,Um

\draw (0,0) ellipse (.7cm and 1cm);
\node  at (-1.2,0.8) {$\bigcup_{j\notin J} U_j$}; 
\node  at (0,0) {finite};

\node  at (-1.8,-1.8) {$U^* \cup \bigcup_{j\in J} U_{j,\fin}$}; 
\draw (0,-2.8) ellipse (.7cm and 1.3cm);
\draw[red] (-0.65,-2.5) -- (0.7,-2.5);
\node  at (-1.2,-4) {$U^{\infty}$}; 

\node at (0,-2.1) {finite};
\node at (0,-3.2) {$\infty$};

%%% the set V

\draw (6,-1) ellipse (.7cm and 1cm);
\node  at (6.8,0) {$V$};
\node at (6,-1) {finite};

\end{tikzpicture}

\end{center}
\label{fig:complete-infinite-finite-charac}
\caption{Illustration of the sett $\bigcup_{j\notin J} U_j$, 
$U^*\cup\bigcup_{j\in J} U_{j,\fin}$
and $U^{\infty}$, with their sizes.
The set $U^{\fin}$ is the union $U^*\cup \bigcup_{j\in J} U_{j,\fin} \cup \bigcup_{j\notin J} U_j$,
which is finite.
In the induced graph $G_{\fin}=G[U^{\fin}\cup V]$, the vertices in $V$ have finite total degrees.
In the induced graph $G_{\infty}=G[U^{\infty}\cup V]$, there are no $E_i$-edges for $i\notin R$.}
\end{figure}

We will next define  formulas that capture  the sizes $\vM|\vN$ of complete $A|B$-biregular graphs, assuming that 
the left hand side has infinitely many vertices
and the right hand side has finitely many vertices.
Let $R$ be the set of colors $i$ where the number of $E_i$-edges is infinite.
and $J\subseteq [m]$ be the set of indexes $j$ where $M_j=\infty$.

Let $\vx=(x_1,\ldots,x_m)$, $\vy=(y_1,\ldots,y_n)$ and $\vz=(z_1,\ldots,z_m)$.
Let $\xi^{J,R}_{A|B}(\vx,\vy)$ be the formula:
\begin{align}
\label{app:eq:simple-complete-xi1}
& \normt{\vx}=\infty \ \wedge \ \normt{\vy}\neq \infty
\\
\label{app:eq:simple-complete-xi1-a}
 \wedge \quad &
\bigwedge_{i\in R} \normt{\vx}_{\nz(A_{i,*})} = \infty \ \wedge \
\bigwedge_{i\notin R} \normt{\vx}_{\nz(A_{i,*})} \neq \infty 
\\
\label{app:eq:simple-complete-xi1-b}
 \wedge \quad &
\bigwedge_{j\in J} x_j = \infty \ \wedge \
\bigwedge_{j\notin J} x_j \neq \infty \ \wedge \
\\
\wedge \quad & \exists \vz \ \ \biregc_{A|C}(\vz,\vy) \ \wedge \ \normt{\vz}\neq \infty \ \wedge \ \bigwedge_{j\notin J} z_i=x_i
\\
\wedge \quad & \xi_{A_3|D}(\vy)
\end{align}
where formula $\biregc_{A|C}(\vz,\vy)$ captures 
the sizes of finite complete $A|C$-biregular graph as defined in Theorem~\ref{theo:main-lemma-bireg}
and $\xi_{A_3|D}(\vy)$ is as defined in \eqref{app:eq:complete-special}.

Intuitively, \eqref{app:eq:simple-complete-xi1}--\eqref{app:eq:simple-complete-xi1-b} 
state that there are infinitely many vertices on the left and only finitely many on the right, and that $R$ and $J$ are as defined above.
The next  lemma follows immediately from Lemma~\ref{app:lem:complete-special},
Theorem~\ref{theo:main-lemma-bireg}, and Lemma~\ref{app:lem:complete-infinite-finite-charac}.

\begin{lemma} \label{app:lem:completesimplegoodonefinite}
For every pair of matrices $A$, $B$ and 
every pair of size vectors $\vM$, $\vN$  with infinitely many vertices on the left, finitely many on the right, $R$ and $J$ defined as above, then
$\xi^{J,R}_{A|B}(\vM,\vN)$ holds in $\cN_{\infty}$
if and only if there is a complete $A|B$-biregular graph of size $\vM|\vN$.
\end{lemma}

To wrap up this subsection, we define the formula:
$$
\biregc_{A|B}^{\infty,\fin}(\vx,\vy) \ := \
\bigvee_{J\subseteq[m], R\subseteq[t]} \xi^{J,R}_{A|B}(\vx,\vy)
$$
that captures the sizes of all possible $A|B$-biregular graph
where the left hand side has infinitely many vertices
and the right hand side has finitely many vertices.

\begin{remark}
By Lemma~\ref{lem:algo-for-bireg} for finite graphs,
$\biregc_{A|C}(\vz,\vy)$ is a disjunction of conjunctions of $O(mnt^4\delta(A,B)^4)$ (in)equations.
By definition, $\xi_{A_3|D}(\vy)$ is a disjunction of conjunctions of $O(tmn)$ (in)equations.
Thus, $\biregc_{A|B}^{\infty,\fin}(\vx,\vy)$ is a disjunction of conjunctions of $O(mnt^4\delta(A,B)^4)$ (in)equations.

For arbitrary degree matrices $A$ and $B$,
we can define a formula $\biregc_{A|B}(\vx,\vy)$ capturing the sizes of all possible $A|B$-biregular graphs
as a disjunction of the formulas for each of the following four cases:
\begin{itemize}
\item 
Both sides have finitely many vertices, which by Lemma~\ref{lem:algo-for-bireg}
is a disjunction of conjunctions of $O(mnt^4\delta(A,B)^4)$ (in)equations.
\item
The left hand side has infinitely many vertices and
the right hand side has finitely many vertices,
which as explained above is a disjunction of conjunctions of $O(mnt^4\delta(A,B)^4)$ (in)equations.
\item
The left hand side has finitely many vertices and
the right hand side has infinitely many vertices, which is symmetric to the previous case.
\item 
Both sides have infinitely many vertices.

By Remark~\ref{app:rem:simple-complete},
the formula when the degree matrices are simple matrices is a disjunction of conjunctions of $O(2^tt^4\delta(A,B))$ (in)equations.
Since the transformation from non-simple to simple requires a blow-up of $O(mn)$ factor,
this case is a disjunction of conjunctions of $O(mnt^4\delta(A,B)^4)$ (in)equations.
\end{itemize}
So $\biregc_{A|B}(\vx,\vy)$ is a disjunction of conjunctions of $O(mnt^4\delta(A,B)^4)$ (in)equations.
\end{remark}

%!TEX root = fo2siam.tex

\section{The extension of Section~\ref{sec:proofnonsimple} to the general case}
\label{app:sec:proofnonsimple}

In this appendix we explain briefly how to extend the reduction from non-simple degree matrices 
to simple degree matrices in Section \ref{sec:proofnonsimple}, 
now allowing the degree matrices to contain $\infty$ entries and the sizes of the partitions to be infinite.
This reduction is only applied to the finite case and Case $1$ from the prior appendix, where  there are infinite-degree vertices on both sides. In the case where
exactly one side had an infinite degree vertex, we  did not make use of the simple restriction.
The reduction we give below can actually apply to all cases: but making use of it in the last case above would not give the desired complexity.

We need to modify  the definition of behavior functions in Definition \ref{def:behavior-function-non-simple} a little bit,
to take into account that the entry in a degree matrix can be $\infty$.

\begin{definition}
\label{app:def:behavior-function-non-simple}
For each $j \in [m]$, we define {\em a behavior function of column $j$ in $A$}
to be a function $g:[t]\times [n]\to \{0,1,\ldots,q,\prdp{0},\prdp{1},\ldots,\prdp{q},\infty\}$ such that the following holds:
\begin{itemize}
\item
$
A_{*,j} \ = 
\begin{pmatrix}
g(1,1)+\cdots + g(1,n)
\\
g(2,1)+\cdots + g(2,n)
\\
\vdots
\\
g(t,1)+\cdots + g(t,n)
\end{pmatrix};
$
\item
for each color $i\in [t]$,
if $A_{i,j}$ is a fixed entry, then $g(i,1),\ldots,g(i,n)$ are all fixed entries;
\item
for each color $i\in [t]$,
if $A_{i,j}$ is a periodic entry, then $g(i,1),\ldots,g(i,n)$ are all periodic entries;
\item 
for each color $i\in [t]$,
if $A_{i,j}$ is an $\infty$ entry, then $g(i,1),\ldots,g(i,n)$ are all periodic entries
and at least one of them is $\infty$.
\end{itemize}
\end{definition}

Note that the difference between Definition \ref{def:behavior-function-non-simple} and Definition \ref{app:def:behavior-function-non-simple}
is the addition of the third item, where the entry $A_{i,j}$ can be $\infty$.
The definition of a behavior function of column $j'$ in $B$ is also modified in a similar manner.
The reduction from  non-simple matrices to simple matrices can now be obtained in exactly the same manner
as in Subsection \ref{subsec:general-non-simple-simple}.

\end{document}